\documentclass[british]{article}

\usepackage{titling}

\usepackage[utf8]{inputenc}
\usepackage[affil-it]{authblk}
\usepackage{amsmath}
\usepackage{amsthm}
\usepackage{biolinum}

\usepackage[libertine]{newtxmath}
\usepackage[letterpaper,top=2cm,bottom=2cm,left=2.5cm,right=2.5cm,marginparwidth=1.75cm]{geometry}

\usepackage{mathtools}
\usepackage{thmtools}
\usepackage{thm-restate}
\usepackage{subeqnarray}
\usepackage{setspace}
\usepackage{graphics,graphicx,color}
\usepackage{url}
\usepackage{enumerate}
\usepackage{float} 
\usepackage{multirow}
\usepackage{hhline}
\usepackage{dsfont} 
\usepackage{multirow}
\usepackage{hhline}
\usepackage[makeroom]{cancel}
\usepackage{ifthen}
\usepackage{comment}
\usepackage{breakcites}
\usepackage[cal=boondoxo, scr=boondoxo, scrscaled=1.05]{mathalfa}
\usepackage{csquotes}

\usepackage{color}
\usepackage{prettyref}
\usepackage{refstyle}
\usepackage{cancel}
\usepackage{graphicx}
\usepackage{xspace}
\usepackage{ifthen}
\usepackage{booktabs}
\usepackage{MnSymbol}

\makeatletter

\AtBeginDocument{\providecommand\Figref[1]{\ref{Fig:#1}}}
\AtBeginDocument{\providecommand\Defref[1]{\ref{Def:#1}}}
\AtBeginDocument{\providecommand\Notaref[1]{\ref{Nota:#1}}}
\AtBeginDocument{\providecommand\Lemref[1]{\ref{Lem:#1}}}
\AtBeginDocument{\providecommand\Thmref[1]{\ref{Thm:#1}}}
\AtBeginDocument{\providecommand\Remref[1]{\ref{Rem:#1}}}
\AtBeginDocument{\providecommand\Algref[1]{\ref{Alg:#1}}}
\AtBeginDocument{\providecommand\Claimref[1]{\ref{Claim:#1}}}
\AtBeginDocument{\providecommand\Corref[1]{\ref{Cor:#1}}}
\AtBeginDocument{\providecommand\Eqref[1]{\ref{Eq:#1}}}
\AtBeginDocument{\providecommand\algref[1]{\ref{alg:#1}}}
\AtBeginDocument{\providecommand\Secref[1]{\ref{Sec:#1}}}
\AtBeginDocument{\providecommand\Propref[1]{\ref{Prop:#1}}}
\AtBeginDocument{\providecommand\Subsecref[1]{\ref{Subsec:#1}}}
\AtBeginDocument{\providecommand\Factref[1]{\ref{Fact:#1}}}
\AtBeginDocument{\providecommand\Exaref[1]{\ref{Exa:#1}}}
\AtBeginDocument{\providecommand\Partref[1]{\ref{Part:#1}}}
\AtBeginDocument{\providecommand\Tabref[1]{\ref{Tab:#1}}}

\RS@ifundefined{subsecref}
  {\newref{subsec}{name = \RSsectxt}}
  {}
\RS@ifundefined{thmref}
  {\def\RSthmtxt{theorem~}\newref{thm}{name = \RSthmtxt}}
  {}
\RS@ifundefined{lemref}
  {\def\RSlemtxt{lemma~}\newref{lem}{name = \RSlemtxt}}
  {}
\RS@ifundefined{exaref}
  {\def\RSlemtxt{example~}\newref{exa}{name = \RSexatxt}}
  {}

\theoremstyle{plain}
\newtheorem{thm}{\protect\theoremname}
\theoremstyle{remark}
\newtheorem{notation}[thm]{\protect\notationname}
\theoremstyle{definition}
\newtheorem{defn}[thm]{\protect\definitionname}
\theoremstyle{definition}
\newtheorem*{defn*}{\protect\definitionname}
\theoremstyle{remark}
\newtheorem{rem}[thm]{\protect\remarkname}
\theoremstyle{plain}
\newtheorem{fact}[thm]{\protect\factname}
\theoremstyle{plain}
\newtheorem{lem}[thm]{\protect\lemmaname}
\theoremstyle{plain}
\newtheorem{cor}[thm]{\protect\corollaryname}
\theoremstyle{remark}
\newtheorem{claim}[thm]{\protect\claimname}
\theoremstyle{plain}
\newtheorem{lyxalgorithm}[thm]{\protect\algorithmname}
\theoremstyle{plain}
\newtheorem{prop}[thm]{\protect\propositionname}
\theoremstyle{plain}
\newtheorem*{prop*}{\protect\propositionname}
\theoremstyle{definition}
\newtheorem{example}[thm]{\protect\examplename}
\theoremstyle{remark}
\newtheorem{note}[thm]{\protect\notename}

\theoremstyle{plain}

\newtheorem{algorithm}{Algorithm}

\usepackage{footnote}

\usepackage{tabu}

\usepackage{tikz}
\usetikzlibrary{calc}

\usepackage{enumitem}
\setlistdepth{20}
\renewlist{itemize}{itemize}{20}
\setlist[itemize]{label=$\cdot$}

\setlist[itemize,1]{label=\textbullet}
\setlist[itemize,2]{label=--}
\setlist[itemize,3]{label=*}
\setlist[itemize,4]{label=$\circ$}
\setlist[itemize,5]{label=$\square$}

\usepackage{multicol}

\usepackage{qtree}

\usepackage{hyperref}

\hypersetup{
    colorlinks=true,
    linkcolor=blue,
    filecolor=magenta,      
    urlcolor=cyan,
	citecolor=blue,
    }

\hypersetup{hidelinks,
breaklinks=true,
colorlinks=true,%
urlcolor=blue,
bookmarksopen=false}

\usepackage[super]{nth}

\usepackage{braket}

\usepackage{flafter}

\usepackage{psvectorian}

\newref{thm}{name=theorem~,Name=Theorem~,names=theorems~,Names=Theorems~}
\newref{def}{name=definition~,Name=Definition~,names=definitions~,Names=Definitions~}
\newref{cor}{name=corollary~,Name=Corollary~,names=corollaries~,Names=Corollaries~}
\newref{lem}{name=lemma~,Name=Lemma~,names=lemmas~,Names=Lemmas~}
\newref{claim}{name=claim~,Name=Claim~,names=claims~,Names=Claims~}
\newref{sec}{name=section~,Name=Section~,names=sections~,Names=Sections~}
\newref{subsec}{name=subsection~,Name=Subsection~,names=subsections~,Names=Subsections~}
\newref{prop}{name=proposition~,Name=Proposition~,names=propositions~,Names=Propositions~}
\newref{rem}{name=remark~,Name=Remark~,names=remarks~,Names=Remarks~}
\newref{alg}{name=algorithm~,Name=Algorithm~,names=algorithms~,Names=Algorithms~}
\newref{nota}{name=notation~,Name=Notation~,names=notations~,Names=Notations~}
\newref{exa}{name=example~,Name=Example~,names=examples~,Names=Examples~}
\newref{tab}{name=table~,Name=Table~,names=tables~,Names=Tables~}
\newref{fact}{name=fact~,Name=Fact~,names=facts~,Names=Facts~}
\newref{part}{name=part~,Name=Part~,names=parts~,Names=Parts~}

\usepackage{color}
\definecolor{purple}{RGB}{0,0,0}

\newcommand\branchcolor[2]{{\color{black} #2}}

\usepackage{tikzsymbols}

\usepackage[bottom]{footmisc}

\@ifundefined{showcaptionsetup}{}{%
 \PassOptionsToPackage{caption=false}{subfig}}
\usepackage{subfig}
\makeatother

\usepackage{babel}
\usepackage[style=alphabetic,citestyle=ieee-alphabetic,maxnames=4,minnames=3,maxbibnames=99,backend=bibtex8]{biblatex}

\providecommand{\algorithmname}{Algorithm}
\providecommand{\claimname}{Claim}
\providecommand{\corollaryname}{Corollary}
\providecommand{\definitionname}{Definition}
\providecommand{\examplename}{Example}
\providecommand{\factname}{Fact}
\providecommand{\lemmaname}{Lemma}
\providecommand{\notationname}{Notation}
\providecommand{\propositionname}{Proposition}
\providecommand{\remarkname}{Remark}
\providecommand{\theoremname}{Theorem}
\providecommand{\notename}{Note}

\global\long\def\polylog#1{{\rm polylog}(#1)}%
\global\long\def\poly{{\rm poly}(n)}%
\global\long\def\ply#1{{\rm poly}(#1)}%
\global\long\def\tr{{\rm tr}}%
\global\long\def\perm#1#2{\!_{#1}P_{#2}}%
\global\long\def\comb#1#2{\,{}_{#1}C_{#2}}%
\global\long\def\paths{{\rm paths}}%
\global\long\def\parts{{\rm parts}}%
\global\long\def\td{{\rm TD}}%
\global\long\def\TD{{\rm TD}}%
\global\long\def\f{\mathcal{F}}%
\global\long\def\g{\mathcal{G}}%
\global\long\def\negl{{\rm negl}(n)}%
\global\long\def\ngl#1{{\rm negl}(#1)}%
\global\long\def\CQd{\ensuremath{\mathsf{CQ_{d}}}\xspace}%
\global\long\def\QCd{\ensuremath{\mathsf{QC_{d}}}\xspace}%
\global\long\def\BQP{\ensuremath{\mathsf{BQP}}\xspace}%
\global\long\def\BPP{\ensuremath{\mathsf{BPP}}\xspace}%
\global\long\def\QNC{\ensuremath{\mathsf{QNC}}\xspace}%
\global\long\def\BQNC{\ensuremath{\mathsf{QNC}}\xspace}%
\global\long\def\QNCd{\ensuremath{\mathsf{QNC_d}}\xspace}%
\global\long\def\NC{\ensuremath{\mathsf{NC}}\xspace}%
\global\long\def\P{\ensuremath{\mathsf{P}}\xspace}%
\global\long\def\NP{\ensuremath{\mathsf{NP}}\xspace}%
\global\long\def\CQ#1{\ensuremath{\mathsf{CQ_{#1}}}\xspace}%
\global\long\def\QC#1{\ensuremath{\mathsf{QC_{#1}}}\xspace}%
\global\long\def\CH#1{#1\text{-}\mathsf{CodeHashing}}%
\global\long\def\recursive#1{#1\text{-}\mathsf{Rec}}%
\global\long\def\CQC#1{\mathsf{CQC}_{#1}}%

\global\long\def\collisionhashing{\mathsf{CollisionHashing}}%

\global\long\def\hcollisionhashing#1{#1\text{-}\mathsf{hCollisionHashing}}%

\global\long\def\serial#1{#1\text{-}\mathsf{Ser}}%

\global\long\def\classCQC#1{\BPP^{\QNC_{#1}^{\BPP}}}%

\global\long\def\classQC#1{\QNC_{#1}^{\BPP}}%

\global\long\def\classCQ#1{\BPP^{\QNC_{#1}}}%

\global\long\def\Pivalid{\Pi_{{\rm valid}}}%

\newcommand{\calO}{\mathcal{O}}
\newcommand{\calQ}{\mathcal{Q}}
\newcommand{\calL}{\mathcal{L}}

\newcommand{\calA}{\mathcal{A}}
\newcommand{\calB}{\mathcal{B}}
\newcommand{\calC}{\mathcal{C}}
\newcommand{\calD}{\mathcal{D}}
\newcommand{\calG}{\mathcal{G}}
\newcommand{\calM}{\mathcal{M}}

\newcommand{\hw}{\mathsf{hw}}

\newcommand{\bit}{\{0,1\}}

\newcommand{\secpar}{\lambda}
\newcommand{\gen}{\mathsf{Gen}}
\newcommand{\ver}{\mathsf{Verify}}
\newcommand{\verify}{\ver}
\newcommand{\prov}{\mathsf{Prove}}
\newcommand{\prove}{\prov}
\newcommand{\sk}{\mathsf{sk}}
\newcommand{\pk}{\mathsf{pk}}
\newcommand{\FBPP}{\mathsf{FBPP}}
\newcommand{\FBQP}{\mathsf{FBQP}}
\newcommand{\onelambda}{1^{\lambda}}
\newcommand{\poqd}{\mathsf{PoQD}}
\newcommand{\dom}{\mathsf{dom}}

\newcommand{\Sumvalid}{\sum_{\mathbf{x}\in X_{\rm valid}}}
\newcommand{\CodeHashing}{\mathsf{CodeHashing}}
\newcommand{\codehashing}{\mathsf{CodeHashing}}
\newcommand{\const}{\mathcal{O}(1)}

\newcommand{\ignore}[1]{}

\newboolean{keepOldProofs}
\setboolean{keepOldProofs}{true}

\newcommand{\beq}{\begin{eqnarray}}
\newcommand{\eeq}{\end{eqnarray}}

\newcommand{\QPT}{\mathsf{QPT}}

\newcommand{\polyA}{{\rm poly}}
\newcommand{\neglA}{\mathsf{negl}}

\newcommand{\B}{\mathcal{B}}

\DeclareMathOperator{\E}{\mathbf{E}}
\newcommand{\niton}{\not\owns}
\newcommand{\comp}{\mathsf{comp}}
\newcommand{\problem}{\textsf{Problem}}
\newcommand{\subproblem}{\textsf{subProblem}}

\newcommand{\twoone}{\mathsf{TwoToOne}}
\newcommand{\G}{G}

\newcommand{\qnc}{\mathsf{QNC}}
\newcommand{\bpp}{\mathsf{BPP}}
\newcommand{\eqref}[1]{Equation~\ref{#1}}

\usepackage{tikz}

\definecolor{dred}{rgb}{.7,.2,.2}

\definecolor{amethyst}{rgb}{0.6, 0.4, 0.8}

\newcommand{\anote}[1]{}

\newcommand{\hendrik}[1]{\textcolor{orange}{HW: #1}}

\newcommand{\atul}[1]{\textcolor{orange}{ASA: #1}}
\newcommand{\uttam}[1]{\textcolor{green}{US: #1}}

\begin{document}
\title{Quantum Depth in the Random Oracle Model} %
\date{}

\thanksmarkseries{arabic}

\author{
Atul Singh Arora,\thanks{Institute for Quantum Information and Matter, California Institute of Technology\newline 
Department of Computing and Mathematical Sciences, California Institute of Technology\newline \href{mailto:atul.singh.arora@gmail.com}{atul.singh.arora@gmail.com}} 
~~Andrea Coladangelo,\thanks{University of California, Berkeley\newline 
Simons Institute for the Theory of Computing \newline \href{mailto:andrea.coladangelo@gmail.com}{andrea.coladangelo@gmail.com}}
~~Matthew Coudron,\thanks{National Institute of Standards and Technology (NIST)/ Joint Center for Quantum Information and Computer Science (QuICS)\newline
Department of Computer Science, University of Maryland\newline \href{mailto:mcoudron@umd.edu}{mcoudron@umd.edu}}\linebreak 
~~Alexandru Gheorghiu,\thanks{Department of Computer Science and Engineering, Chalmers University of Technology \newline
Institute for Theoretical Studies, ETH Z{\"u}rich\newline \href{mailto:alexandru.gheorghiu@chalmers.se}{alexandru.gheorghiu@chalmers.se}} 
~~Uttam Singh,\thanks{Center for Theoretical Physics, Polish Academy of Sciences\newline
\href{mailto:uttam@cft.edu.pl}{uttam@cft.edu.pl}} ~~and 
Hendrik Waldner\thanks{Department of Computer Science, University of Maryland \newline
Max Planck Institute for Security and Privacy\newline
\href{mailto:hwaldner@umd.edu}{hwaldner@umd.edu}}

}

\sloppy

\maketitle
\thispagestyle{empty}

\begin{abstract}

We give a comprehensive characterization of the computational power of shallow quantum circuits combined with classical computation. Specifically, for classes of \emph{search problems}, we show that the following statements hold, relative to a \emph{random oracle}:
\begin{itemize}
\item[(a)] $\BPP^{\QNC^{\BPP}} \neq \BQP$. This refutes Jozsa's conjecture~\cite{Jozsa06} in the random oracle model. As a result, this gives the first \emph{instantiatable} separation between the classes by replacing the oracle with a cryptographic hash function, yielding a resolution to one of Aaronson's ten semi-grand challenges in quantum computing~\cite{aaronson_ten_2005}.
\item[(b)] $\BPP^{\QNC} \not\subseteq \BQNC^{\BPP}$ and $\BQNC^{\BPP} \not\subseteq \BPP^{\QNC}$. This shows that there is a subtle interplay between classical computation and shallow quantum computation. In fact, for the second separation, we establish that, for some problems, the ability to perform adaptive measurements in a \emph{single} shallow quantum circuit, is more useful than the ability to perform \emph{polynomially many} shallow quantum circuits without adaptive measurements.
\item[(c)] There exists a $2$-message \emph{proof of quantum depth} protocol. Such a protocol allows a classical verifier to efficiently certify that a prover must be performing a computation of some minimum quantum depth. Our proof of quantum depth can be instantiated using the recent proof of quantumness construction by Yamakawa and Zhandry \cite{yamakawa_verifiable_2022}.
\end{itemize}
\end{abstract}

\newpage 
\thispagestyle{empty}

\tableofcontents{}

\pagebreak

\section{Introduction}

High depth circuits are believed to be strictly more powerful than low depth circuits, in the sense that having deeper circuits allows one to solve a larger set of problems. Indeed, this is a well established fact for both classical and quantum circuits of depth sub-logarithmic in the size of the input~\cite{furst1984parity, hastad1986almost, bravyi2018quantum, watts2019exponential}. However, for circuits of (poly)logarithmic depth and general polynomial depth, proving any sort of \emph{unconditional} separation is challenging~\cite{razborov1994natural}. In fact, there is not even an unconditional proof that the set of problems that can be solved by polylog-depth classical circuits, \NC, is a \emph{strict subset} of the set of problems solvable by poly-depth classical circuits, \P (or \BPP when allowing for randomness). The same is believed to be the case for the quantum analogues of these classes, \QNC and \BQP, respectively. Nevertheless, the strict containments $\NC \subsetneq \P$ and $\QNC \subsetneq \BQP$ are known to hold in the oracle setting and, in particular, relative to a \emph{random oracle}~\cite{miltersen_circuit-depth_1992}.\footnote{Technically~\cite{miltersen_circuit-depth_1992} only shows the strict containment $\NC \subsetneq \P$, relative to a random oracle. However, the quantum version $\QNC \subsetneq \BQP$ can also be shown as a straightforward extension of that result.} This is a strong indication that there are problems in \P (\BQP) which cannot be parallelized so as to be solvable in \NC (\QNC). \anote{removed explanation of the random oracle heuristic}Under the \emph{random oracle heuristic}, by replacing the random oracle with a cryptographic hash function, one can even provide concrete instantiations of such problems. A further indication of the separation between low and high depth computations is provided by certain inherently sequential cryptographic constructions such as time-lock puzzles and verifiable delay functions~\cite{rsw, boneh_verifiable-delay_2018}. 

The study of circuit depth can also yield insights into the subtle relationship between quantum and classical computation by considering \emph{hybrid circuit models} that combine quantum and classical computation~\cite{chia_need_2020-1, coudron_computations_2020-2, arora_oracle_2022, hasegawa_oracle_2022}. In this setting, one can ask the question: how powerful are poly-depth classical circuits, when augmented with polylog-depth quantum circuits? Could it be the case that interspersing \BPP with \QNC computations captures the full power of \BQP computations? Jozsa famously conjectured that the answer is yes~\cite{Jozsa06}. Indeed, there is some evidence to support this conjecture, as the quantum Fourier transform, a central building block for many quantum algorithms, was shown to be implementable with log-depth quantum circuits~\cite{cleve_fast_2000}. This also implies that Shor's algorithm can be performed by a $\BPP^{\QNC}$ machine, a polynomial-time classical computer having the ability to invoke a (poly)log depth quantum computer.\footnote{Note that here and throughout the paper, the $\QNC$ oracle can output a string, unlike a decision oracle which outputs a bit.} %
Moreover, in the oracle setting, a number of problems yielding exponential separations between quantum and classical computation require only constant quantum-depth to solve, providing further support for Jozsa's conjecture~\cite{Simon1997, forrelation_one, forrelation_two}. 

Despite the evidence in support of Jozsa's conjecture, it was recently shown that, in the oracle setting, the conjecture is false~\cite{chia_need_2020-1, coudron_computations_2020-2}. Specifically, the results of~\cite{chia_need_2020-1} (hereafter referred to as CCL) and~\cite{coudron_computations_2020-2} (hereafter referred to as CM) considered two ways of interspersing poly-depth classical computation with $d$-depth quantum computation. The first is $\BPP^{\QNCd}$, denoting problems solvable by a \BPP machine that can invoke $d$-depth quantum circuits (whose outputs are measured in the computational basis). The second, $\QNCd^{\BPP}$, denotes problems solvable by a $d$-depth quantum circuit that can invoke a \BPP machine at each layer in the computation.\footnote{Note that the \BPP oracle is not invoked coherently. Instead, it is invoked on outcomes resulting from intermediate measurements performed in the layers of the \QNCd circuit.} Later, borrowing terminology from~\cite{chia_need_2020-1, arora_oracle_2022}, we will refer to the former circuit model as \CQd and the latter as \QCd. However, for the purposes of this introduction, we will stick to the more familiar notation using complexity classes. 
Intuitively, $\BPP^{\QNCd}$ captures the setting of a classical computer that can invoke a $d$-depth quantum computater several times. Examples of this include quantum machine learning algorithms such as VQE or QAOA~\cite{vqe, qaoa}, though as mentioned, Shor's algorithm is also of this type.
On the other hand, $\QNCd^{\BPP}$ captures a $d$-depth \emph{measurement-based quantum computation}~\cite{raussendorf2001one, briegel2009measurement}, where intermediate measurements are performed after each layer in the quantum computation. The outcomes of those measurements are processed by a poly-depth classical computation and the results are ``fed'' into the next quantum layer.
CCL and CM showed that there exists an oracle relative to which ${\BPP^{\QNCd}} \cup {\QNCd^{\BPP}} \subsetneq \BQP$, for any $d = \polylog{n}$, with $n$ denoting the size of the input. Notably, each work considered a different oracle for showing the separation. For CM, the oracle is the same one as for Childs' glued trees problem~\cite{glued_trees}. For CCL, the oracle is a modified version of the oracle used for Simon's problem~\cite{Simon1997}, where the modification involves performing a sequence of permutations, allowing them to enforce high quantum depth. 

CCL and CM were the first results to provide a convincing counterpoint to Jozsa's conjecture. However, the main drawback of the CCL and CM results is that they are relative to oracles that are highly structured and it is unclear if they can be explicitly instantiated based on some cryptographic assumption. Indeed, in his ``Ten Semi-Grand Challenges for Quantum Computing Theory'', Aaronson emphasizes this important distinction, and asks whether there is some \emph{instantiatable} function that separates the hybrid models from \BQP. In this work, we resolve Aaronson's question in the affirmative for the \emph{search} variants of these classes.

In contrast to separations between different models of computation running in polynomial time, such as \P and \NP or \BPP and \BQP, where several plausible candidates exist for separating the classes, the case for depth separations is much more subtle. As was already observed in~\cite{bitansky_time-lock_2015}, no standard cryptographic assumption is known to yield a separation between \NC and \P. The best candidates for such a separation are sequential compositions of hash functions (under the random oracle heuristic) as shown in~\cite{miltersen_circuit-depth_1992} and the iterated exponentiation scheme of Rivest, Shamir and Wagner~\cite{rsw}. Thus, informally, the best we could hope for in terms of an instantiatable separation between the hybrid models and \BQP is a separation in the random oracle model which could then be instantiated using cryptographic hash functions. 

Our work is concerned not only with separations between the hybrid models and \BQP in the random oracle model, but also with giving a comprehensive characterization of quantum depth in that model. To that end, we first re-examine Jozsa's conjecture and argue that the natural class associated to ``$d$-depth quantum computation combined with polynomial-time classical computation'' is not ${\BPP^{\QNCd}} \cup {\QNCd^{\BPP}}$, but $\BPP^{\QNCd^{\BPP}}$. This is because, if one has the ability to perform $\QNCd^{\BPP}$ computations, certainly it should also be possible to repeat this polynomially-many times as well as perform classical processing in between the runs. Note that ${\BPP^{\QNCd}} \cup {\QNCd^{\BPP}} \subseteq \BPP^{\QNCd^{\BPP}}$. The separation we then obtain, relative to a random oracle, is $\BPP^{\QNCd^{\BPP}} \subsetneq \BQP$, for any fixed $d \leq \poly$. %
Going beyond this separation, we also show that the hybrid models ${\BPP^{\QNCd}}$ and ${\QNCd^{\BPP}}$ are separate from each other in both directions, relative to a random oracle (in fact, we show that $\BPP^{\QNC_{\mathcal{O}(1)}} \not\subseteq \QNCd^{\BPP}$ and  $\QNC_{\mathcal{O}(1)}^{\BPP} \not\subseteq \BPP^{\QNCd}$), illustrating the subtle interplay between short-depth quantum computation and classical computation. Lastly, by combining the techniques that we develop with previous results on \emph{proof of quantumness} protocols, we obtain \emph{proof of quantum depth} protocols---protocols in which a $\BPP$ verifier, exchanging 2 messages\footnote{2 messages in total or a 1 round protocol.} with an untrusted quantum prover, can certify that the prover has the ability to perform quantum computations of a minimum depth.

\begin{figure}[h!]
    \begin{centering}
        \subfloat[Motivating the various hybrid quantum depth classes.\label{fig:ListOfClasses}]{
            \begin{tabular}{c}
                \toprule
                $\begin{array}{ccc}
                     \QNC_{d}                                 & \underset{\text{w/ BPP processing}}{\overset{\text{adaptive measurements}}{\rightarrow}} & \classQC d                    \\
                     \\
                     \begin{array}{c}
                      {\scriptstyle \text{poly-many invokes}} \\
                      {\scriptstyle \text{w/ BPP processing}}
                    \end{array}\downarrow &                                                                                          & \downarrow\begin{array}{c}
                                                                                                                                                   {\scriptstyle \text{poly-many invokes}} \\
                                                                                                                                                   {\scriptstyle \text{w/ BPP processing}}
                                                                                                                                                 \end{array} \\
                     \\
                     \classCQ d                               &                                                                                          & \classCQC d
                   \end{array}$ \\
                \bottomrule
              \end{tabular}
              \par
        \begin{centering}
        \end{centering}
        }

    \end{centering}

    \begin{centering}
        \subfloat[Illustration of $\QNC_d$ and $\classQC d$ circuits.\label{fig:classQNCd}]{\begin{centering}
                \hspace{-1.5cm}\includegraphics[width=8.5cm]{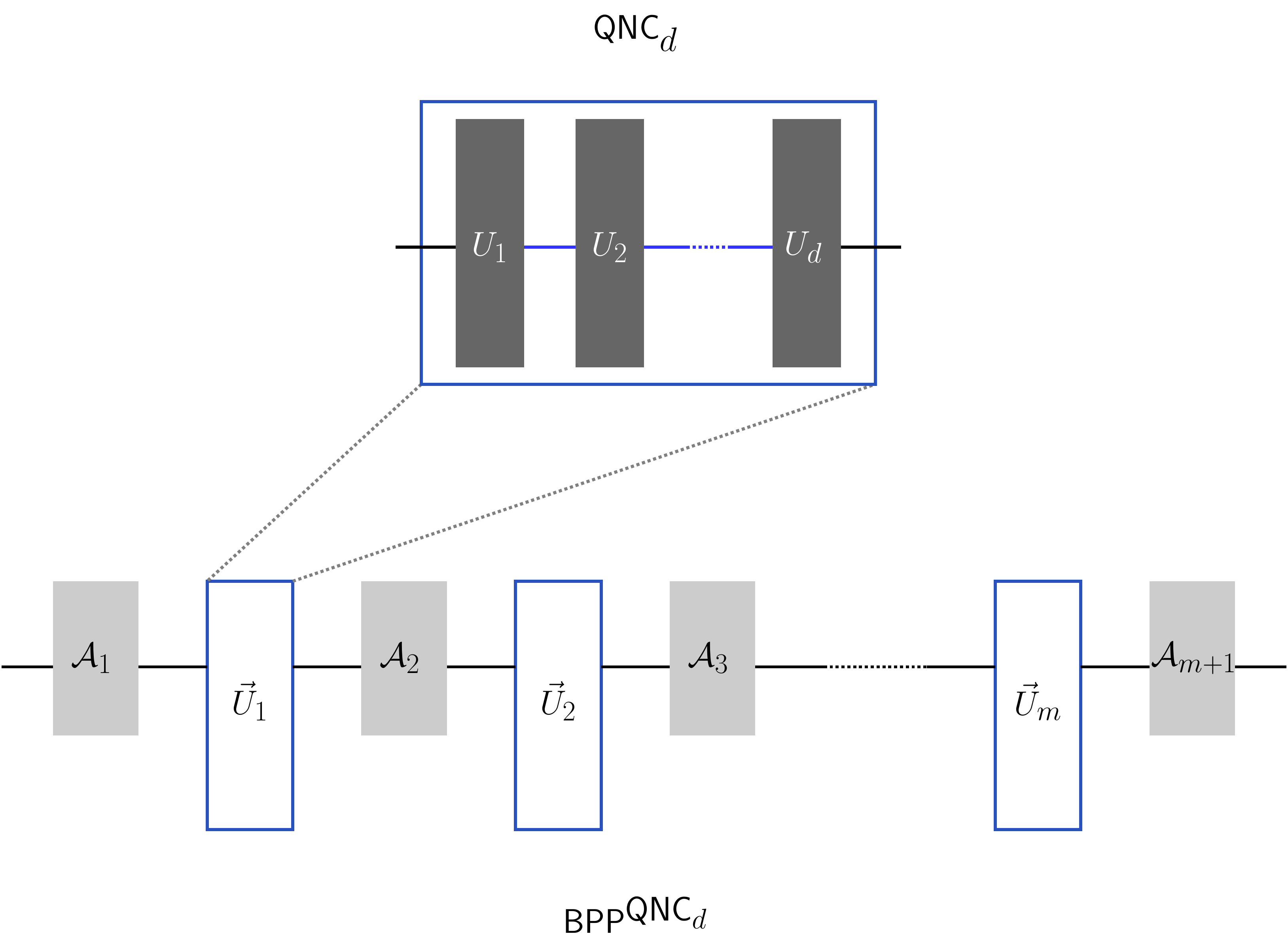}
                \par\end{centering}
        }\enskip{}\subfloat[Illustration of $\classQC d$ and $\classCQC d$ circuits. \label{fig:classQCd}]{\centering{}\includegraphics[width=8.5cm]{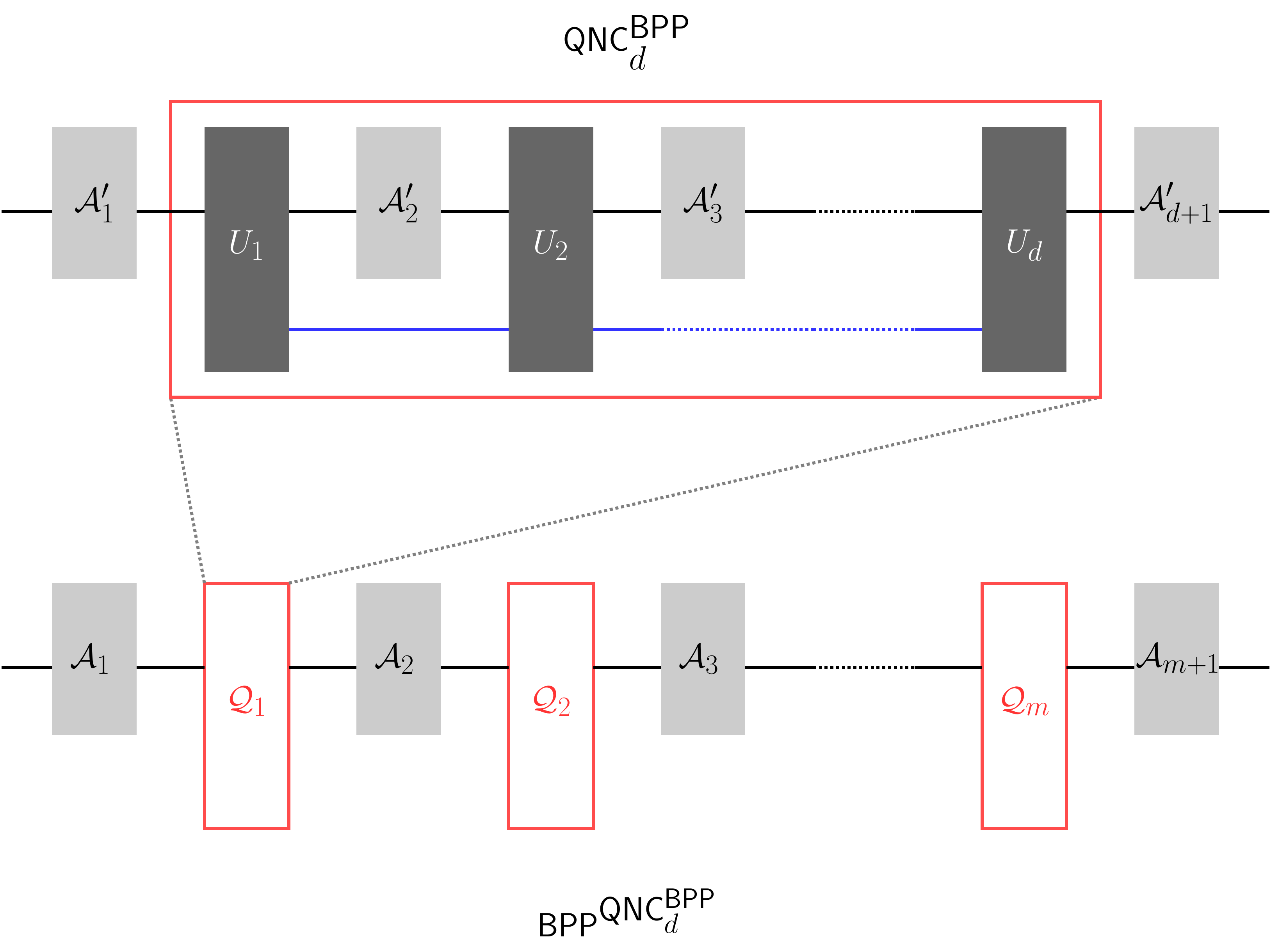} \hspace{-1.5cm}}
        \par\end{centering}

    \caption{The four hybrid quantum depth classes we consider. Blue wires carry qubits, black wires carry bits. Measurements are implicit and performed in the standard basis. $U_i$s denote depth $1$ unitaries, ${\cal A}_i$ and ${\cal A}'_i$ denote poly time classical algorithms. \label{fig:hybridQuantumDepthClasses}}
\end{figure}

\subsection{Main Results\label{subsec:main_results_intro}}
We now state our results more formally and provide some intuition about the proofs. From here on, we abuse the notation slightly and use the standard \emph{decision} complexity class names to refer to their \emph{search} variants.

\subsubsection{Lower bounds on quantum depth}

We first show the following separation.
\begin{thm}[informal]
  \label{thm:inf_BQP_notin_CQC_d} Fix any function $d\le\poly$. Then,
  relative to a random oracle,\footnote{Here, as well as in all subsequent results, the statements hold with probability 1 over the choice of the random oracle. In addition, queries to the oracle are viewed as having depth 1.} it holds that $\classCQC d\subsetneq\BQP$.
\end{thm}

As motivated earlier, we take the class $\BPP^{\QNCd^{\BPP}}$ to capture computations performed by a combination of $d$-depth quantum computation and polynomial-depth classical computation. The interpretation of our result is that $\BPP^{\QNCd^{\BPP}}$ can be separated from $\BQP$ \emph{using the least structured oracle possible}, a random oracle. Together with the (quantum) random oracle heuristic, by instantiating the oracle with a cryptographic hash function like SHA-2 or SHA-3, this yields the first plausible instantiation of a problem solvable in $\BQP$ but not in $\BPP^{\QNCd^{\BPP}}$. This provides a resolution to Aaronson's challenge. The main technical innovation that allows us to achieve the separation is a general lifting lemma that takes any problem separating $\BPP$ from $\BQP$ in the random oracle model, which additionally satisfies a property that we call \emph{classical query soundness}, and constructs a problem separating $\BPP^{\QNCd^{\BPP}}$ and $\BQP$. We show that several known problems satisfy this property. Our lifting lemma is inspired by~\cite{chia_need_2020-1}, and crucially extends their analysis beyond highly structured oracles. We describe this lifting lemma more precisely in \Subsecref{liftingLemma}.

\subsubsection{Proofs of quantum depth\label{subsec:PoQD_inf}}
It is natural to wonder whether Theorem \ref{thm:inf_BQP_notin_CQC_d} yields an efficient test to certify quantum depth, i.e.\ a \emph{proof of quantum depth}. A proof of quantum depth is a more fine-grained version of a proof of quantumness: rather than distinguishing between quantum and classical computation, a proof of quantum depth protocol can distinguish between provers having large or small quantum depth. We show that  instantiating our lifting lemma with a problem whose solution is \emph{efficiently verifiable} immediately yields a proof of quantum depth. One such problem\footnote{We remark that, if one is only concerned with the complexity-theoretic separation of Theorem \ref{thm:inf_BQP_notin_CQC_d}, and not with efficient verification, then a much simpler problem suffices (see $\collisionhashing$ in \Tabref{ProblemsWeConsider}).} is due to Yamakawa and Zhandry~\cite{yamakawa_verifiable_2022}. More precisely, we have the following.

\begin{thm}[informal]
  \label{thm:inf_ProofOfDepth_BQPcompleteness}Let $n$ be the security
  parameter and fix any function $d\le\poly$. In the random oracle
  model, there exists a two-message protocol between a poly-time classical
  verifier and a quantum prover such that,
  \begin{itemize}
    \item Completeness: There is a $\BQP$ prover which makes the verifier accept
          with probability at least $1-\negl$
    \item Soundness: No malicious $\classCQC d$ prover can make the verifier
          accept with probability greater than $\negl$.
  \end{itemize}
\end{thm}

We emphasise that considering protocols with more than two messages
leads to difficulties in formalising the notion of quantum depth. 
For instance, 
one can construct protocols where the prover is forced
to hold $r$ single qubit states and subsequently measures them. Information about the basis
in which to measure each of these qubits is sent one at a time by the verifier over $r$ messages (the verifier waits for the response to each measurement, before sending the next basis).
The measurement results are used by the verifier to ensure soundness (each qubit is measured in its preparation basis and so the outcomes are completely determined).
It is not hard to show that if the prover measures these qubits without
knowing the measurement basis, it cannot succeed except with negligible probability. If one attempts to model the prover as a $\classCQ {d}$ or $\classQC {d}$ circuit, then, because of the delay between messages, it appears that $d \geq r$ is necessary.
However, this can be seen as an artefact of the modelling choice: in practice, the prover only needs $d$ single qubit
quantum computers with quantum depth $1$ where the last gate can
be delayed until the appropriate message is received in order to pass the test. 
Essentially, this approach only tests the prover's ability to maintain the coherence of the qubits it received, without actually testing the depth of the circuit it has to perform. In \Subsecref{Discussion}, we discuss a possible resolution that captures quantum depth in the interactive setting. 

\subsubsection{Tighter bounds}

While \Thmref{inf_BQP_notin_CQC_d} establishes that $\classCQC d$
does not capture the computational power of $\BQP$ for any fixed
$d\le\poly$, it is not a priori clear if, for instance, $\classCQC{2d+{\cal O}(1)}$
is strictly larger than $\classCQC d$. %
Indeed, we show that the answer is affirmative.
\begin{thm}[informal]
  \label{thm:inf_depth_hierarchy}Fix any function $d\le\poly$. Relative
  to a random oracle, it holds that\footnote{and more generally, that $\QNC_{2d+\mathcal{O}(1)} \nsubseteq \classCQC d$.} $\classCQC d\subsetneq\classCQC{2d+{\cal O}(1)}$.
\end{thm}

Formally, \Thmref{inf_BQP_notin_CQC_d} treats a call to the quantum random oracle as a depth-1 quantum gate. In practice, if instead the gate requires depth $\ell$, then $d$ can be replaced by $d\ell$. We remark that there exist hash functions that are thought to be quantum-secure which require only logarithmic depth to evaluate~\cite{ajtai1996generating, peikert2019noninteractive}. Further, there is reason to believe that such hash functions could also be constructed in $\ell={\calO}(1)$ depth. In particular, if one is only concerned with specific cryptographic properties (such as collision resistance), then generic constructions are known which convert log-depth hash functions into ones that require only constant depth \cite{applebaum_cryptography_2006}.%

\subsubsection{Separations between hybrid quantum depth classes}
While both $\bpp^{\qnc}$ and $\qnc^{\bpp}$ capture some notion of a hybrid between efficient classical computation and shallow quantum computation, the relationship between the two is not immediately clear. To get a slightly better intuition about the two models, one can think of $\bpp^\qnc$ as capturing an efficient computation that contains \emph{polynomially many} shallow quantum circuits (separated by measurements and classical computation). On the other hand, one can think of $\qnc^\bpp$ as a \emph{single} shallow quantum circuit, where one is allowed to make partial measurements of some of the wires, and choose the next gates \emph{adaptively}.  While it may not be surprising that there exist problems that can be solved in $\bpp^\qnc$ but not in $\qnc^\bpp$, 
it turns out that the two classes are in fact incomparable---each class
contains problems that the other does not, relative to a random oracle.
\begin{thm}[informal]
  \label{thm:QCdiffCQ}Fix any function $d\le\poly$. Relative to a
  random oracle, it holds that $\classCQ{{\cal O}(1)}\nsubseteq\classQC d$
  and $\classQC{{\cal O}(1)}\nsubseteq\classCQ d$.
\end{thm}

The second separation is arguably more surprising. It says that, relative to a random oracle, there are problems that can be solved by a \emph{single} shallow (in fact, constant-depth) quantum circuit \emph{with} adaptive measurements but cannot be solved by circuits with \emph{polynomially many} shallow quantum circuits \emph{without} adaptive measurements. The problem that shows $\classQC{{\cal O}(1)}\nsubseteq\classCQ d$ is a variant of the proof of quantumness from \cite{brakerski_simpler_2020}. The key technical innovation to achieve this separation is a theorem that characterises the structure of strategies that succeed in the protocol of \cite{brakerski_simpler_2020} (this is discussed further in Section \ref{sec:technical contributions} ). This ``structure theorem'' crucially strengthens a similar theorem from \cite{coladangelo2022deniable}, and may be of independent interest. %

Finally, we examine the relationship between $\classCQ d\cup\classQC d$ and $\classCQC d$.
By definition, it is manifest that $\classCQ d\cup\classQC d\subseteq\classCQC d$.
Even though $\classQC d$ and $\classCQ d$ are incomparable, it is
conceivable that their union captures any reasonable notion of quantum
depth $d$. We show that this is not the case.
\begin{thm}[informal]
  \label{thm:QCcupCQnotenough}Fix any function $d\le\poly$. Relative
  to a random oracle, it holds that $\classCQC{{\cal O}(1)}\nsubseteq\classCQ d\cup\classQC d$.
\end{thm}

In words, the latter theorem asserts that a computation consisting of polynomially many layers of \emph{constant}-depth quantum circuits with adaptive control cannot be simulated by
quantum circuits with $d$ depth which are either adaptive (but
consisting of a single $d$-depth quantum circuit) or consisting of many $d$-depth quantum circuits (but without adaptive control). 

\subsubsection{Summary }
\begin{table}[h!]
    \begin{centering}
      \begin{tabular}{cl}
        \toprule
        Result                                          & Remarks\tabularnewline
        \midrule
        $\classCQC{}\subsetneq\BQP$                     & Refutes Jozsa's conjecture in the random oracle model\tabularnewline
        \midrule
        $\classCQC d\subsetneq\classCQC{2d+{\cal O}(1)}$ & Fine grained advantage of quantum depth\tabularnewline
        \bottomrule
      \end{tabular}
      \par\end{centering}
    \caption{\label{tab:bounds_on_quantum_depth}(Simplified) Bounds on quantum
      depth. Separations are with respect to the random oracle and $d\le\protect\poly$
      is any fixed function of the input size. }
  
  \end{table}

  \begin{table}[h!]
    \begin{centering}
      \begin{tabular}{c>{\raggedright}p{9cm}}
        \toprule
        Result                                                     & Physical Interpretation\tabularnewline
        \midrule
        $\classCQ{{\cal O}(1)}\nsubseteq\classQC{}$                & Running \emph{poly many} \emph{constant} depth
        quantum circuits (with \emph{no} adaptive measurements) cannot be simulated
        by running a \emph{single} $\log$ depth quantum circuit \emph{with}
        adaptive measurements.\tabularnewline
        \midrule
        $\classQC{{\cal O}(1)}\nsubseteq\classCQ{}$                & Running a \emph{single} \emph{constant} depth quantum circuit
        \emph{with} adaptive measurements cannot be simulated by running \emph{poly many} $\log$ depth
        quantum circuits (with \emph{no} adaptive measurements).\tabularnewline
        \midrule
        $\classCQC{{\cal O}(1)}\nsubseteq\classCQ{}\cup\classQC{}$ & Evidence that it is not enough to consider $\classCQ{}$ and $\classQC{}$
        when studying quantum depth.                                                                                                          \\
        Running \emph{poly many} \emph{constant} depth quantum circuits
        \emph{with} adaptive measurements cannot be simulated using either
        (a) \emph{poly many} $\log$ depth
        quantum circuits with \emph{no} adaptive measurements, or by (b) a \emph{single
          $\log$ depth} quantum circuit \emph{with} adaptive measurements.\tabularnewline
        \bottomrule
      \end{tabular}
      \par\end{centering}
    \caption{\label{tab:separatingHybridClasses}(Simplified) Separations of hybrid
      quantum depth with respect to the random oracle. The results hold,
      not only for $\log$ but for any fixed polynomially-bounded function.}
  \end{table}

\Tabref{bounds_on_quantum_depth} lists our lower bounds on quantum depth,
and \Tabref{separatingHybridClasses} lists the separations among the
hybrid classes.

\subsection{Main technical contributions}
\subsubsection{Lifting Lemmas\label{subsec:liftingLemma}}

One of the main technical contributions of our work is to prove \emph{two} general lifting lemmas. These
lemmas take problems, defined relative to a random oracle, that are classically
hard (in a stronger sense, defined next) and create new problems which
are, in addition, hard for specific hybrid quantum depth classes. We describe these lifting lemmas a bit more precisely.

We say that a problem (defined
with respect to the random oracle) is \emph{classical query sound}
if the following holds: any (potentially unbounded time) algorithm which makes only polynomially many 
\emph{classical} queries to the random oracle (i.e.\ no superposition queries),
succeeds at solving the problem with at most negligible probability.
It turns out that the problem introduced by YZ satisfies this property. Another problem which satisfies this property is inspired by the proof of quantumness protocol defined by Brakerski et al.~\cite{brakerski_simpler_2020} (hereafter referred to as BKVV).\footnote{Which we refer to as $\collisionhashing$ later.} 
For such
problems, the following holds. 
\begin{lem}[informal, simplified]
  \label{lem:inf_dRec}There is a procedure\footnote{$\recursive d[\cdot]$ is meant to be short for $d$-Recursive.}
  that takes a classical query sound problem ${\cal P} \in \BQP$ and
  creates a new problem ${\cal P}':=\recursive d[{\cal P}]$, such
  that ${\cal P}'\notin\classCQC d$ and ${\cal P}'\in \BQP$.
\end{lem}
Observe that this lemma makes the problem hard for \emph{the most general notion of quantum depth} we have considered.
To give some intuition about how it is derived, suppose we have a problem $\cal P$ which is classical query sound and denote the random oracle as $H$. Then ${\cal P'}=\recursive d[\cal P]$ is the same problem, defined with respect to \emph{a sequential composition of $d + 1$ random oracles}, $\tilde{H} = H_{d} \circ \dots \circ H_{0}$. In essence, we have substituted $H$ with $\tilde{H}$. This new problem will retain classical query soundness, as $\tilde{H}$ behaves like a random oracle. But in addition, we have now made it so that querying $\tilde{H}$ effectively requires depth $d + 1$. As $\QNC_d$ has depth $d$, only the $\BPP$ parts of $\classCQC d$ will be able to query $\tilde{H}$. We can therefore simulate the $\classCQC d$ algorithm with an exponential time algorithm that is limited to polynomially many queries to $\tilde{H}$. By classical query soundness, such an algorithm cannot solve ${\cal P'}$, which yields the desired result.  

This was a simplified description of our result. In fact, we show a more refined statement that relates the depth required to solve $\cal P'$ to the depth required to solve $\cal P$. In addition, arguing that $\tilde{H}$ behaves like a random oracle and that $\QNC_d$ cannot query $\tilde{H}$ requires a careful and more involved analysis.  We use \Lemref{inf_dRec}
to establish \Thmref{inf_depth_hierarchy}.

Our second lifting lemma produces a problem that is hard for $\classQC d$, starting from
a problem that satisfies what we call \emph{offline soundness}. Consider
a two phase algorithm consisting of: \emph{an online phase} which is a poly-time
classical algorithm \emph{with access} to the random oracle followed
by \emph{an offline phase }which is an unbounded(-time) algorithm
with \emph{no access} to the random oracle. Then, \emph{offline soundness} requires
that no such two phase algorithm succeeds at solving the problem with
non-negligible probability. It turns out, again, that both YZ and BKVV
satisfy this property.
\begin{lem}[informal]
  \label{lem:inf_dSerial}There is a procedure\footnote{$\serial d[\cdot]$ is meant to be short for $d$-Serial.}
  which takes a problem ${\cal P}\in \QNC_{{\cal O}(1)}$ with offline soundness and creates
  a new problem ${\cal P}':=\serial d[{\cal P}]$ such that ${\cal P}'\notin\classQC d$ and ${\cal P'}\in \classCQ{{\cal O}(1)}$.
\end{lem}

Again, we actually show a slightly more general upper bound which depends on the depth required to solve $\cal P$. We use Lemma \ref{lem:inf_dSerial} to establish $\classCQ{{\cal O}(1)}\nsubseteq\classQC d$
(first separation of \Thmref{QCdiffCQ}). Establishing the other direction
($\classQC{{\cal O}(1)}\nsubseteq\classCQ d$) is quite involved and relies heavily on the structure of the problem we consider (explained below).
Consequently,
it is unclear whether there exists a general lifting lemma that yields hardness for $\classCQ d$. %

We remark that, by using \Lemref{inf_dSerial} to lift the problem that yields $\classQC{{\cal O}(1)}\nsubseteq\classCQ d$, we also obtain \Thmref{QCcupCQnotenough}, i.e. $\classCQC 1\nsubseteq\classCQ d\cup\classQC d$.

\subsubsection{A structure theorem for \cite{brakerski_simpler_2020}}
\label{sec:technical contributions}

Another technical contribution of this work, which may be of independent interest, is to prove a theorem characterizing the structure of strategies that are successful at the proof of quantumness from \cite{brakerski_simpler_2020}. This theorem is a crucial strengthening of a theorem from \cite{coladangelo2022deniable}. We employ this theorem as an intermediate step to establish the hybrid separation, $\classQC{{\cal O}(1)}\nsubseteq\classCQ d$. 

Recall, informally, that the proof of quantumness from \cite{brakerski_simpler_2020} requires the prover to succeed at the following task: given access to a 2-to-1 function $g$, and to a random oracle $H$ with a one-bit output, find a pair $(y,r)$ such that $$r \cdot (x_0 \oplus x_1) \oplus H(x_0) \oplus H(x_1) = 0\,,$$
where $\{x_0,x_1\} = g^{-1}(y)$.
This can be solved in $\QNC_{\mathcal{O}(1)}$ as follows: 
\begin{itemize}
    \item[(i)] Evaluate $g$ on a uniform superposition of inputs, yielding $\sum_x \ket{x} \ket{g(x)}$,
    \item[(ii)] Measure the image register obtaining some outcome $y$ and a state $(\ket{x_0} + \ket{x_1})\ket{y}$,
    \item[(iii)] Query a phase oracle for $H$ to obtain $((-1)^{H(x_0)} \ket{x_0} + (-1)^{H(x_1)} \ket{x_1})\ket{y}$,
    \item[(iv)] Make a Hadamard basis measurement of the first register, obtaining outcome $r$. 
\end{itemize} 

Informally, our structure theorem establishes that querying at a superposition of pre-images is essentially \emph{the only way to succeed} (provided finding a collision for $g$ is hard---this is the case when $g$ is a trapdoor claw-free function, as in \cite{brakerski_simpler_2020}, but more generally our theorem also holds e.g.\ when $g$ is a uniformly random 2-to-1 function). Denote by $n$ the bit-length of strings in the domain of $g$.

\begin{thm}[informal]
Let $P$ be any $\BQP$ prover that succeeds with $1-\negl$ probability at the proof of quantumness protocol from \cite{brakerski_simpler_2020}, by making $q$ queries to the oracle $H$. Then, with $1-\negl$ probability over pairs $(H, y)$, the following holds. Let $p_{y|H}$ be the probability that $P^H$ outputs $y$, and let $x_0$,$x_1$ be the pre-images of $y$. Then, for all $b \in \{0,1\}$, there exists $i \in [q]$ such that the state of the query register of $P^H$ right before the $i$-th query has weight $\frac12 p_{y|H}\cdot (1-\negl)$ on $x_b$.
\end{thm} 

Note that a version of the above theorem that applies to provers who win with probability non-negligibly greater than $\frac12$ also holds (but we stated the close-to-ideal version for simplicity). We provide a sketch of how this theorem is used in the proof of $\classQC{{\cal O}(1)}\nsubseteq\classCQ d$ in Subsection~\ref{sec: andrea tech}. We refer to Corollary \ref{cor: bkvv} for a formal statement of the theorem.

\subsection{Discussion and open problems}
\label{subsec:Discussion}

\paragraph{Further questions in the random oracle model.}

Our separations are with respect to \emph{search} problems. The main question left open by our work is whether the same separations can be shown with respect to \emph{decision} problems. Recall that our approach to proving the separations is to \emph{lift} a problem that separates $\BPP$ and $\BQP$ in the random oracle model (for example a \emph{proof of quantumness}) to a problem that requires at least a certain amount of quantum depth. However, we note that this approach is unlikely to yield depth separations for decision problems. This is because the Aaronson-Ambainis conjecture \cite{aaronson2009need} states that one cannot separate the decision versions of $\BPP$ and $\BQP$ in the random oracle model. Thus, a different approach is likely to be necessary.

Another interesting related question is the following. When we instantiate our lifting lemma with the proof of quantumness from YZ, the resulting problem inherits the property that solutions can be publicly verified. We thus obtain a proof of quantum depth that is publicly verifiable. Can we further push this quantum soundness to obtain verification of \BQP with a \BPP verifier relative to a random oracle?

We have also seen that making use of a problem inspired by the Brakerski et al.~\cite{brakerski_simpler_2020} proof of quantumness allows us to prove more fine grained separations between hybrid classes. It is then natural to ask, whether these separations also yield \emph{finer grained proofs of quantum depth} (which are sound against $\BPP^{\QNC_{d}^{\BPP}}$ provers and complete for a $\BPP^{\QNC_{2d + {\cal O}(1)}^{\BPP}}$ prover). This does not immediately follow from our results, as the problem we construct from BKVV is not efficiently verifiable, and our current techniques do not directly extend to the computationally-bounded setting. We therefore leave this as an open problem.

\paragraph{Separations without the random oracle.}

Our work gives the first instantiatable quantum depth separation by virtue of being in the random oracle model. It is natural to ask if one can establish this separation in the plain model. Unfortunately, a separation in the random oracle model seems to be the best that one can hope for, given that even for classical depth there are no known separations that rely on standard cryptographic assumptions (other than the random oracle). In some sense this is peculiar, since one would imagine that using more structured problems would allow one to prove stronger separations. The random oracle is the least structured type of oracle, but the fact that it is an oracle helps in establishing provable lower bounds.

\paragraph{Generalizing beyond $\classCQC{d}$.}

We have argued that $\classCQC{d}$ is the most natural class capturing the notion of $d$-depth quantum computation, combined with polynomial-depth classical computation. However, for the purpose of \emph{certifying} quantum depth, as we have mentioned earlier (and as we discuss in more detail in \Exaref{interactiveDepthIssue}), the situation becomes more subtle when the certification protocol involves \emph{interaction}. We therefore propose that any protocol which establishes quantum depth $d$ and uses $r$ rounds of interaction should be sound against at least an $r$ level generalization of $\classCQC d$ (e.g.\ a 2 level generalization with quantum depth $d$ would be $\BPP^{{\QNC_d}^{\classCQC{d}}}$ --- here $2$ counts the number of times $\QNC_d$ appears in the tower of complexity classes, so that an $r$ level generalisation would have $r$ appearances of $\QNC_d$). In our case, since the proof of depth protocols are single-round, we show the necessary soundness against a $\classCQC{d}$ prover. %

Of course, there are other possible ways to define hybrid $d$-depth quantum-classical computation. For instance, one can define the class $\mathsf{QDepth}_{d}$ of problems solved by polynomial sized circuits with quantum and classical gates where the key constraint is that \emph{the longest path connecting quantum gates (with quantum wires) is at most $d$}. We expect that the union over all $r$ level generalizations of $\classCQC{d}$ (where $r$ is polynomially bounded) equals $\mathsf{QDepth}_d$. We also expect our separating problems (and $\recursive{d}[{\cal P}]$ in general, for classical query sound $\cal P$) to not be in $\mathsf{QDepth}_d$, but we leave the proof to future work.

\subsection{Previous work\label{subsec:previouswork}}
We compare our results to the previous works \cite{chia_need_2020-1}, \cite{coudron_computations_2020-2}, \cite{arora_oracle_2022}, and \cite{chiadepthverification2022}.
    \paragraph{Comparison to \cite{chia_need_2020-1}, \cite{coudron_computations_2020-2} and \cite{arora_oracle_2022}.} Compared to previous work on the topic, our work gives a comprehensive treatment of the complexity of hybrid quantum-classical computation.

    As mentioned earlier, the primary difference compared to~\cite{chia_need_2020-1} and \cite{coudron_computations_2020-2} is that all of our separations are with respect to a random oracle, rather than with respect to highly structured oracles. However, one caveat is that our separations are for search problems. Our contribution is also conceptual. We propose $\classCQC{d}$ as the appropriate model to capture ``$d$-depth quantum computation combined with polynomial-time classical computation''. While \cite{chia_need_2020-1} and \cite{coudron_computations_2020-2} showed that $\classCQ{d} \cup \classQC {d} \nsubseteq \BQP$, we show the stronger result that $\classCQC{d} \nsubseteq \BQP$. %
    
    Our work also shows separations between different hybrid models. Such separations were considered in \cite{arora_oracle_2022}, where they are again proven only with respect to highly structured oracles.
    
    In terms of techniques, we take inspiration and ideas %
    from both \cite{chia_need_2020-1} and \cite{arora_oracle_2022}. In particular we build on two key ideas---sampling argument and domain hiding. One of the main contribution of our analysis is to abstract and generalise these techniques beyond their original scope which was tailored to specific promise problems. While most of our results build on these techniques, we also point out that to prove the separation between the hybrid models $\classQC{{\cal O}(1)}\nsubseteq\classCQ d$ we use entirely different ideas. In particular, as an intermediate step, we establish a theorem that characterizes the structure of strategies that succeed at the proof of quantumness in BKVV, which may be of independent interest.

    \paragraph{Comparison to \cite{chiadepthverification2022}.} 
    The work of \cite{chiadepthverification2022} was the first to consider proofs of quantum depth. However, the notion of soundness that they propose, and their corresponding protocol (in the single prover setting), suffers from the issues that we discussed after~\Thmref{inf_ProofOfDepth_BQPcompleteness} (and in \Exaref{interactiveDepthIssue} below).  
    
    In particular, their protocol can be spoofed by a $d$ level tower of $\classCQC {{\cal O}(1)}$ (as described in \Subsecref{Discussion}). In practical terms, this means that it can be spoofed by running several constant depth quantum computers in parallel, provided the ``idle coherence time'' of each quantum computer is longer than the time that elapses between messages in the protocol. In contrast, our proof of depth protocol does not suffer from this issue and can be used to certify that the prover is able to perform computations ``beyond'' $\classCQC d$.

\subsection*{Acknowledgments}
We are thankful to Joseph Slote, Ulysse Chabaud and Thomas Vidick for various insightful discussions.
While at ETH, AG was supported by Dr.\ Max R{\"o}ssler, the Walter Haefner Foundation and the ETH Z{\"u}rich Foundation. AC is a Quantum Postdoctoral Fellow at the Simons Institute for the Theory of Computing supported by NSF QLCI Grant No. 2016245, and by DARPA under agreement No. HR00112020023. Any opinions, findings and conclusions or recommendations expressed in this material are those of the author(s) and do not necessarily reflect the views of the United States
Government or DARPA. US acknowledges the support by Polish National Science Center (NCN) (Grant No.\ 2019/35/B/ST2/01896). HW is supported by an MC2 postdoctoral fellowship.

\pagebreak

\section{Technical Overview\label{sec:intro_tech_overview}}
Here we give a high level technical overview of the paper.

\subsection{Bounds on quantum depth | \texorpdfstring{$\protect\classCQC{}\subsetneq\protect\BQP$}{Bounds on quantum depth | Separating BPP\^{}\{QNC\_d\}\^{}BPP and BQP}\label{subsec:introBoundsOnQuantumDepth}}

In this subsection, we describe the proof of \Thmref{inf_BQP_notin_CQC_d}. As mentioned previously, our main technical contribution is a general lifting lemma that takes any problem separating $\BPP$ from $\BQP$ in the random oracle model, which additionally satisfies a property that we call \emph{classical query soundness}, and constructs a problem separating $\BPP^{\QNCd^{\BPP}}$ and $\BQP$. We first explain the key idea behind this construction. To be concrete, after describing the key idea, we restrict to an $\NP$ search problem due to Yamakawa and Zhandry~\cite{yamakawa_verifiable_2022}, which satisfies classical query soundness (this problem is particularly appealing because it is in $\NP$, and thus solutions can be publicly verified, however we emphasize that other known search problems that are not in $\NP$ can also be used for the separation). We then build towards a proof that this problem is not in $\classCQC d$ by considering hardness for the three special cases $\QNC_d$, $\classQC{d}$ and $\classCQ{d}$. The desired result is obtained by combining the ideas in these three cases.

Let $\mathcal{P}$ be a (search) problem, defined relative to a random oracle $H$, that separates $\BPP$ from $\BQP$. Suppose that $\mathcal{P}$ is such that it requires \emph{quantum} access to $H$ in order to be solved with polynomially many queries (\emph{classical query soundness} will eventually require a bit more than this). As mentioned in Subsection~\ref{subsec:liftingLemma}, the first natural idea to lift this to a separation between low quantum depth and polynomial quantum depth is to \emph{replace the evaluation of $H$ with a sequential evaluation of random oracles}. For example, suppose that originally $H: \Sigma \rightarrow \{0,1\}^n$.
Then, let $H_0,\dots, H_{d-1}: \Sigma \rightarrow \Sigma$, and $H_d:  \Sigma \rightarrow \{0,1\}^n$ be random oracles. Define $\tilde{H} = H_d \circ \dots \circ H_0$. Now, let $\mathcal{P}'$ be the problem that is identical to $\mathcal{P}$ except that it is relative to $\tilde{H}$. Then, it is natural to imagine that $\mathcal{P}'$ requires quantum depth at least $d+1$ to solve. This idea does not quite work right away, since $\tilde{H}$, as defined, is not actually a uniformly random oracle any more. This is because with every $H_i$ that is added, the number of collisions in $\tilde{H}$ increases (on average). To remedy this, one could assume that $H_0, \dots, H_{d-1}$ are random \emph{permutations} (although note that random permutations cannot be generically constructed from random oracles). A similar idea works in a different setting, for arguing about the post-quantum security of ``proofs of sequential work'' \cite{blocki_security_2021}. However, in our case, the analysis is complicated by the fact that we consider hybrid models. CCL were the first to consider a variant of sequential hashing (sequential permutations), in the context of hybrid models. However, their analysis only works for certain structured oracles. In this work, we adapt their ideas to the random oracle setting and overcome these difficulties.

\paragraph{Lifting ${\cal P} \notin \BPP$ to $\tilde{\cal P} \notin \classCQC d$. \label{dRec}}
Given a problem $\cal P$ with respect to $H$, we define the problem $\tilde{\cal P} = \recursive{d}[{\cal P}]$ to be $\cal P$ with respect to $\tilde H = H_d \circ \dots \circ H_0$ where $H_0, \dots, H_d$ are independent random oracles with the following domains and co-domains: $H_{0}:\Sigma\to\Sigma^{d'}$,
$H_{i}:\Sigma^{d'}\to\Sigma^{d'}$ for $i \in\{1\dots d-1\}$, and $H_{d}:\Sigma^{d'}\to\{0,1\}^{n}$ with $d'=2d+5$.

Notice that $H_0$ is not surjective, as its codomain is much larger than its image.\footnote{We sometimes refer to this fact by saying that the function is ``expanding''.} In fact, this is also true for $H_i \circ \dots \circ H_0$, for all $i < d$. This and the fact that the $H_i$ functions are random, have two important consequences. First, it means that with high probability $H_{d-1} \circ \dots \circ H_0$ is injective and so $\tilde{H}$ behaves like a random oracle. Consequently, $\cal P'$ inherits the soundness and completeness of $\cal P$. Second, it means that one can apply a ``domain hiding'' technique, which, at a high level, works as follows. 
One way of evaluating $\tilde{H}$ at $x\in\Sigma$ is to sequentially compose $H_{0}, H_{1},\dots, H_{d}$ which would require depth $d+1$. Intuitively, it seems unlikely that there is a more depth efficient way of evaluating $\tilde{H}$ because the domain on which the $H_{i}$'s
need to be evaluated (which is $H_{i-1} \circ \dots \circ H_0(\Sigma)$) is getting shuffled and lost in an exponentially
larger domain (which is $\Sigma^{d'}$). 
Therefore, even though one has access
to all ${\cal L}=(H_{0}, H_{1},\dots, H_{d})$ oracles at the first layer of depth, one only knows
that $H_{0}$ needs to be queried at $\Sigma$ but the algorithm has
no information about where the relevant domains of $H_{1}\dots H_{d}$
are. At the second depth layer, the algorithm can learn $H_{0}(\Sigma)$
and so learns where to query $H_{1}$ but, and this needs to be shown,
it still does not know where the relevant domains of $H_{2},\dots H_{d}$
are. By starting with a sufficiently large expansion, i.e.\ a sufficiently large $d' >d$, this argument can be repeated until depth $d$ where the relevant domain of $H_{d}$
still remains hidden. Thus, even though $\cal P'$ can potentially be solved with $d+1$ depth, it cannot be solved with depth $d$. This is the basic idea behind why the problem
is not in $\QNC_{d}$. Instead of working with $\cal P$ and $\recursive{d}[\cal P]$ abstractly, %
we consider the following concrete problem. %

\subsubsection{\textnormal{$\CH{d}$} | The problem}
We refer to the problem introduced by Yamakawa and Zhandry~\cite{yamakawa_verifiable_2022} as $\codehashing$ in this work. The problem is stated in
terms of a family of error-correcting codes called \emph{suitable codes}. For our purposes, it suffices to think
of suitable codes as a family of sets $\{C_{\lambda}\}_{\lambda}$
where each $C_{\lambda}$ is a set of codewords $\{(\mathbf{x}_{1},\dots\mathbf{x}_{n})\}$
with each coordinate $\mathbf{x}_{i}$ belonging to some alphabet
$\Sigma$. The size of this alphabet, $|\Sigma|=2^{\lambda^{\Theta(1)}}$
is exponential in $\lambda$, and the number of components $n=\Theta(\lambda)$
essentially equal to $\lambda$. 
$\codehashing$ is defined as follows.
\begin{defn}[$\codehashing$; informal]
  Let $\{C_{\lambda}\}_{\lambda}$ be a suitable code and let $H:\{0,1\}^{\log n}\times\Sigma\to\{0,1\}$
  be a random oracle. Given a description of the suitable code (e.g.\
  as parity check matrices) and oracle access to $H$, on input $1^{\lambda}$,
  the problem is to find a codeword $\mathbf{x}=(\mathbf{x}_{1}\dots\mathbf{x}_{n})\in C_{\lambda}$ 
 such that\footnote{We use $a||b$ to mean concatenation of $a$ and $b$.} $H(i||\mathbf{x}_{i})=1$ for all $i\in\{1\dots n\}$.
\end{defn}

Note that $\codehashing$ is an $\NP$ search problem, since from, e.g. the parity check
matrix of the code, it is easy to verify that $\mathbf{x}$ is indeed
a codeword and with a single parallel query ($n$ queries in total) to $H$, one can
check that it hashes correctly.

YZ shows that $\codehashing$ satisfies the following two properties.
\begin{lem}[Paraphrased from YZ]
  \label{lem:YZ_intro_paraphrased}The following hold.
  \begin{itemize}
    \item \emph{Completeness:} There is a QPT machine which solves $\codehashing$
          with probability $1-\ngl{\lambda}$ and makes only one parallel query to $H$. %
    \item \emph{Soundness:} Every (potentially unbounded time) classical circuit
          which makes at most $2^{\lambda^{c}}$ queries to $H$, with $c < 1$, solves $\codehashing$
          with probability at most $2^{-\Omega(\lambda)}$.
  \end{itemize}
\end{lem}

The fact that soundness holds against \emph{unbounded time} classical
circuits which make only poly-many queries to the random oracle is essential
in proving that $\classCQC{}\subsetneq\BQP$. Applying our lifting map, $\recursive{d}[\cal P]$ on $\codehashing$ we obtain the following.\footnote{We used $\rm{bit}_i [\tilde{H} (\cdot)] = 1$ instead of $\tilde{H}(i||\cdot)=1$ for notational convenience later.} %

\begin{defn}[$\CH d$; informal]
  Let $\{C_{\lambda}\}_{\lambda}$ be a suitable code, and
  $\tilde{H}:=H_{d}\circ \dots\circ H_{1}\circ H_{0}$, 
  where $H_0, \dots, H_d$ are as in Section~\ref{dRec}. %
  Given a description of the suitable code, access to random oracles
  ${\cal L}=(H_{0}\dots H_{d})$, on input $1^{\lambda}$, find a codeword
  $\mathbf{x}=(\mathbf{x}_{1}\dots\mathbf{x}_{n})\in C_{\lambda}$ such that %
  ${\rm bit}_{i}[\tilde{H}(\mathbf{x}_{i})]=1$
  for all $i\in\{1\dots n\}$.
\end{defn}

To convey the key ideas behind the proof that $\CH{d} \notin \classCQC{d}$, we first consider the $\QNC_d$ case in some more detail, and extend the analysis to $\classQC d$. We then analyse the $\classCQ d$ case, which uses a technique called the ``sampling argument'' due to \cite{EC:CorettiDGS18}. These ideas were first considered in the structured oracle setting by \cite{chia_need_2020-1} and \cite{arora_oracle_2022}. We adapt them to show $\CH{d}\notin \classCQC{d}$ relative to a random oracle. 

\subsubsection{\textnormal{$\CH{d} \notin \QNC_d$}}
\paragraph{Base sets.} %
We started our discussion in \Subsecref{introBoundsOnQuantumDepth} by observing that the analysis is simplified by taking $H_0\dots H_{d-1}$ to be injective functions. However, for a large enough $d'$, it is not
hard to see that this is indeed the case on an appropriately restricted
domain. The sets which describe this restricted domain are chosen
randomly. We call them \emph{base sets }and denote them by $S_{01},\dots S_{0d}$
(corresponding to $H_{1},\dots H_{d}$ respectively). Observe that
$H_{0}$ maps $\Sigma$ to $\Sigma^{d'}$ (which is exponentially
larger than $\Sigma$; recall that $|\Sigma|=2^{\lambda^{\Theta(1)}}$)
and, since $H_{0}$ is a random function, the probability that this
mapping is injective is $1-\ngl{\lambda}$. Pick any set $S_{01}\subseteq\Sigma^{d'}$
uniformly at random in the domain of $H_{1}$ subject to two constraints:
(1) it includes $H_{0}(\Sigma)$, i.e. the domain of $H_{1}$ on which
the value of $\tilde{H}$ depends, and (2) its size is $|S_{01}|=|\Sigma|^{d+2}$.
The first constraint ensures that the domain we care about is included
in the base sets and the second ensures that: (a) $|S_{01}|$ is exponentially
smaller than $|\Sigma|^{d'}$ and (b) $|S_{01}|$ is large enough
for applying ``domain hiding'' as mentioned above. Define $S_{0i}:=H_{i-1}(\dots H_{1}(S_{01})\dots)$
to be the image of $S_{01}$ through the first $1$ to $(i-1)$'th oracles
for $i\in\{2\dots d\}$. Let $E$ denote the event that $H_{0}$ is
injective and $H_{1}\dots H_{d-1}$ are injective on the base sets.
We show that $E$ (given our choice for $d'$), occurs with overwhelming
probability. In the subsequent discussion, we assume that base sets have
been selected and that $E$ occurs.

\paragraph{Proof idea.} We describe the proof that $\CH{d} \notin \QNC_d$ in some more detail, which implements the previously described ``domain hiding'' idea and proceeds via a hybrid argument.
Denote a $\QNC_{d}$ circuit that makes $d$ parallel calls to the oracle
${\cal L}=(H_{0},\dots H_{d})$ by $U_{d+1}\circ{\cal L}\circ U_{d}\dots U_{2}\circ{\cal L}\circ U_{1}\circ\rho_{0}$.
Here, $\rho_{0}$ is some initial state, $U_{i}$ are single layered
unitaries, and the composition is meant to act as conjugation, i.e.\
$U_{1}\circ\rho_{0}=U_{1}\rho_{0}U_{1}^{\dagger}$.  We show that the
behaviour of such a circuit, i.e.\ its probability of outputting a
valid answer, is negligibly close to the behaviour of another circuit
$U_{d+1}\circ{\cal M}_{d}\circ U_{d}\dots U_{2}\circ{\cal M}_{1}\circ U_{1}\circ\rho_{0}$
where ${\cal M}_{1},\dots{\cal M}_{d}$ are ``shadow oracles'' corresponding
to ${\cal L}$ that contain no information about the values taken by $\tilde{H}$ on $\Sigma$. Clearly then, this
circuit cannot be solving $\CH d$ because it never queries $\tilde{H}$.
This in turn means that the original circuit also cannot
solve $\CH d$, which implies $\CH d\notin\QNC_{d}$.
It remains to define ${\cal M}_{1}\dots{\cal M}_{d}$ and to argue
that the two circuits have essentially the same behaviour. Using a hybrid argument, one can establish the latter by showing
that the following are close in trace distance: (1) ${\cal L}\circ U_{1}\circ\rho_{0}$
and ${\cal M}_{1}\circ U_{1}\circ\rho_{0}$, (2) ${\cal L}\circ U_{2}\circ{\cal M}_{1}\circ U_{1}\circ\rho_{0}$
and ${\cal M}_{2}\circ U_{2}\circ{\cal M}_{1}\circ U_{1}\circ\rho_{0}$,
and so on. To convey intuition, we sketch these steps one at a time, and we define ${\cal M}_{1}\dots{\cal M}_{d}$
as we proceed. We restrict to base sets $S_{01}\dots S_{0d}$ as described
above.

\emph{Hybrid 1}. ${\cal L}\circ U_{1}\circ\rho_{0} \approx {\cal M}_{1}\circ U_{1}\circ\rho_{0}$. %
\\
Let $S_{11}\subseteq S_{01}$ be a random subset of $S_{01}$, subject %
to the constraints that (a) it includes $S_{1}:=H_{0}(\Sigma)$ and
(b) $|S_{11}|/|S_{01}|=1/|\Sigma|=\ngl{\lambda}$. Let $S_{1j}:=H_{j-1}(S_{1,j-1})$
be the propagation of $S_{11}$ through $H_{1}$ to $H_{j-1}$. Here, we are trying to define a sequence
of sets $(S_{11},\dots S_{1d})$ on which we require that ${\cal M}_{1}$
outputs $\bot$ and outside of these sets, we require that ${\cal M}_{1}$
behaves just like ${\cal L}$, i.e.\ if one denotes ${\cal M}_1 = (H_0,M_{11},\dots M_{1d})$, then we require that $M_{1i}$ behaves as $H_i$ outside $S_{1i}$ and outputs $\perp$ inside $S_{1i}$. To be concise, we will say that ${\cal M}_{1}$
is a shadow oracle of ${\cal L}$ with respect to $(S_{11}\dots S_{1d})$.
Why do we want this behaviour? For $S_{i}:=H_{i-1}(\dots H_{0}(\Sigma)\dots)$, 
${\cal M}_{1}$ clearly contains no information about $\tilde{H}$ on $\Sigma$, since $S_j \subseteq S_{1j}$. But why couldn't
we just have chosen $(S_{1}\dots S_{d})$ instead of $(S_{11}\dots S_{1d})$
to define ${\cal M}_{1}$? Briefly, this is because choosing to hide
an exponentially larger set (note that $|S_{11}|=|\Sigma|^{d+1}$
while $|S_{1}|=|\Sigma|$) allows us to easily apply similar arguments
in the subsequent hybrids. This will become evident shortly. Recalling
our goal, we want to establish that ${\cal L}\circ U_{1}\circ\rho_{0}$
and ${\cal M}_{1}\circ U_{1}\circ\rho_{0}$ are close in trace distance.
To do this, we use the so-called one-way to hiding (O2H) lemma \parencite{ambainis_quantum_2018}.
Informally, the lemma, as applied to our situation, says that if (a)
the input state $\rho_{0}$ contains no information about the set
where ${\cal L}$ and ${\cal M}_{1}$ behave differently, and (b)
the probability of finding any element inside this set is negligible,
then the trace distance between the two states of interest is negligible.
The lemma clearly applies in our case because (a) initially the algorithm
contains no information about ${\cal L}$ (it has not yet made any
queries) and (b) the probability of finding any element in the set
$S_{1i}$ where ${\cal L}$ and ${\cal M}_{1}$ behave differently,
without knowing anything about ${\cal L}$, is at most $|S_{1i}|/|S_{0i}|=\ngl{\lambda}$,
for each $i\in\{1\dots d\}$, and thus still negligible by a union bound.

\emph{Hybrid 2}. ${\cal L}\circ U_{2}\circ\rho_{1} \approx {\cal M}_{2}\circ U_{2}\circ\rho_{1}$
where $\rho_{1}={\cal M}_{1}\circ U_{1}\circ\rho_{0}$.\\
In this step, we will see the advantage of having chosen
a sequence of sufficiently large sets $(S_{11},\dots S_{1d})$ where
${\cal M}_{1}$ outputs $\perp$. Let us begin with examining the
information contained in $\rho_{1}$ about ${\cal L}$. In the previous
case, $\rho_{0}$ contained no information about ${\cal L}$. Since
$\rho_{1}$ only learns about ${\cal L}$ by querying ${\cal M}_{1}$,
it suffices to examine the information contained in ${\cal M}_{1}$.
Since ${\cal M}_{1}$ does not hide any information about $H_{0}$,
$\rho_{1}$ could have learnt $S_{1}=H_{0}(\Sigma)$. Recall also
that $S_{1}\subseteq S_{11}$. This means that if one were to take
${\cal M}_{2}$ equal to ${\cal M}_{1}$, then one cannot expect ${\cal L}\circ U_{2}\circ\rho_{1}$
to be close to ${\cal M}_{2}\circ U_{2}\circ\rho_{1}$ in general
because $U_{2}$ could query the oracle at $S_{1}$ and the outputs
of the two circuits would be different with probability one---${\cal M}_{1}$
outputs $\perp$ while ${\cal L}$ does not. Consequently, when constructing ${\cal M}_2$, we do not hide anything about $H_1$. As for $H_2\dots H_d$, note that, ${\cal M}_{1}$
 contains no information about the behaviour of ${\cal L}$ inside
$S_{12},S_{13}\dots S_{1d}$. %
We can therefore, treat $S_{12}\dots S_{1d}$
as the new ``base sets'' and proceed analogously. Let $S_{22}\subseteq S_{12}$
be a random subset of $S_{12}$, subject to the constraint (as before)
that (a) it includes $S_{2}=H_{1}(H_{0}(\Sigma))$ and (b) $|S_{22}|/|S_{12}|=1/|\Sigma|=\ngl{\lambda}$.
Defining ${\cal M}_{2}$ to be the shadow oracle of ${\cal L}$ with
respect to $(\emptyset,S_{22},\dots S_{2d})$, one can again apply
the O2H lemma to conclude that ${\cal L}\circ U_{2}\circ\rho_{1}$
and ${\cal M}_{2}\circ U_{2}\circ\rho_{1}$ are close in trace distance.
Note that it is crucial that $|S_{12}|$ is sufficiently large
such that condition (b) above is satisfied.

Generalising the argument above, one sees that the sets $S_{ij}$
constitute a triangular matrix (where the $i$-th row corresponds to sets
on which ${\cal M}_{i}$ outputs $\perp$)
\[
  \left[\begin{array}{ccccc}
      S_{11}    & H_{1}(S_{11}) & H_{2}(H_{1}(S_{11}) & \dots  & H_{d}(\dots H_{1}(S_{11})\dots) \\
      \emptyset & S_{22}        & H_{2}(S_{22})       & \dots  & H_{d}(\dots H_{2}(S_{22})\dots) \\
      \emptyset & \emptyset     & S_{33}              & \dots  & H_{d}(\dots H_{3}(S_{33})\dots) \\
                &               &                     & \ddots                                   \\
      \emptyset & \emptyset     & \emptyset           &        & S_{dd}
    \end{array}\right]
\]
which clarifies why the argument can only be applied for $d$ steps
(as we expect). To see this, note that at the $d$th step, all oracles except the last have been completely revealed (last row). Crucially, the last oracle is blocked at $S_d \subseteq S_{dd}$ and therefore reveals no information about $\tilde H(\Sigma)$. If one proceeds with the $(d+1)$-th step, all oracles are revealed and one can no longer argue that the algorithm does not access $\tilde H(\Sigma)$.%

Observe that so far, we have not used the fact that $\codehashing$
is classically hard, only that without access to the oracle, the problem
cannot be solved. The classical hardness comes into play once $\BPP$
computations are allowed.

\subsubsection{$\CH{d} \notin \protect\classQC d$ \label{subsec:infQCcodehashing}} 

We now sketch how one goes from arguing $\CH{d} \notin \QNC_{d}$ to arguing $\CH{d} \notin \classQC d$. Denote circuits corresponding
to $\classQC d$ by ${\cal A}_{d+1}\circ{\cal B}_{d}^{{\cal L}}\circ\dots\circ{\cal B}_{1}^{{\cal L}}\circ\rho_{0}$
where ${\cal B}_{i}^{{\cal L}}:=\Pi_{i}\circ{\cal L}\circ U_{i}\circ{\cal A}_{i}^{{\cal L}}$,
${\cal A}_{i}^{{\cal L}}$ denotes a classical algorithm, and $\Pi_{i}$
denotes a (possibly partial) measurement. The analogous circuit with
shadow oracles is denoted by ${\cal A}_{d+1}\circ{\cal B}_{d}^{{\cal M}_{d}}\circ\dots{\cal B}_{1}^{{\cal M}_{1}}\circ\rho_{0}$
where ${\cal B}_{i}^{{\cal M}_{i}}:=\Pi_{i}\circ{\cal M}_{i}\circ U_{i}\circ{\cal A}_{i}^{{\cal L}}$.
The idea, again, is to establish, via a hybrid argument, that the two
circuits are close in trace distance. In the $\QNC_{d}$ case, thanks to the depth of the circuit being $d$, we were able to argue that any $\QNC_d$ algorithm behaves equivalently if we take away its access to $\tilde{H}$. When trying to argue that a $\classQC d$ algorithm cannot solve the problem, we have to be more careful because
the $\bpp$ part has sufficient depth to make queries to $\tilde{H}$. In our argument, this will affect how the shadow oracles ${\cal M}_{i}$ are defined.

In some more detail, we allow the classical algorithm to make ``path queries''---which
intuitively just means that if $H_{i}$ is queried at $x_{i}$, the
algorithm also learns $(x_{0},x_{1}\dots x_{d})$ such that\footnote{Two caveats: (1) $H_{0}:\Sigma\to\Sigma^{d'}$ therefore some of the
paths will not have well defined first components and (2) we only
care about queries made inside the base sets where conditioned on
$E$, $H_{1}\dots H_{d-1}$ behave as permutations.} $x_{j+1}=H_{j}(x_{j})$ for all $j$. This of course can only help the algorithm. %

The key idea is that we account for the ``paths'' that have been queried classically until depth $i$ and define ${\cal M}_i$ to be consistent with those (i.e.\ it never outputs $\perp$ on these paths). As before, we can replace queries to $\cal L$ with queries to ${\cal M}_i$ that contain no information about $\tilde{H}$ except for the paths which were classically queried. Appealing to the soundness of $\codehashing$, such an algorithm cannot succeed. %
This is because $\codehashing$ has the property that even an unbounded classical algorithm cannot succeed if it only makes polynomially many queries to the oracle.

\subsubsection{$\CH{d} \notin \protect\classCQ d$}

\global\long\def\varstar{\llcorner\lrcorner}%

Observe that a poly depth quantum circuit can access $\tilde H$ and since a $\classCQ d$ circuit has poly many $\QNC_d$ circuits, it is not a priori clear that $\classCQ d$ cannot also access $\tilde H$. This is why the approach we used to prove that $\CH{d} \notin \QNC_d$ cannot be applied directly. Crucially, to argue that the problem is not in $\classCQ d$, one must use the fact that the contents of each $\QNC_d$ circuit are measured entirely, and that each $\QNC_d$ circuit takes only classical inputs. %
In order to handle the classical information that each $\QNC_d$ circuit receives as input, we use a technique
called the ``sampling argument''. In essence, this says that if
${\cal L}$ has high entropy (which is to say that the oracles being queried are sufficiently random), then conditioned on any string $s$
correlated with it, the resulting ${\cal L}|s$ behaves as a ``convex
combination'' of high entropy distributions with a small fraction of
their values completely fixed. This allows us to reduce the analysis to that of a particular set of paths being exposed, which we can handle by proceeding as in the $\classQC{d}$ case. 

A similar argument was used by CCL to establish that a problem is not in $\classCQ{d}$ with respect to a (structured) oracle. Their analysis used a sequence of permutation oracles and was simplified by viewing the oracles, equivalently, as distributions over paths (as opposed to a sequence of functions assigning values to individual points). The paths viewpoint was particularly helpful when considering the ``sampling argument'' (the version we use is derived from~\cite{EC:CorettiDGS18}). \cite{arora_oracle_2022} showed that such a sampling argument can be obtained for almost any oracle which can be viewed as a distribution over paths. In our setting, since the oracles are random, paths can collide. Thus, one needs to define a suitable notion of ``paths'' in this setting. We provide more details in the next three paragraphs. However, since these are relatively more technical, one may wish to skip directly to \Subsecref{CQCdintro} on a first read.

\paragraph{Sampling argument for Permutations.}

Suppose $t$ is a permutation over $N$ elements labelled $\{0,\dots,N-1\}$.
This permutation $t$ is ordinarily viewed as a function, $t(x)$
specifying how $x$ is mapped. However, one could equivalently view
$t$ as a collection of pairs (or tuples later) $(x,y)$ such that $t(x)=y$. We call
such a pair a ``path''. %

Now consider distributions over permutations. Let's begin with a uniform
distribution $\mathbb{F}$ over all permutations $u$. One may characterise
$\mathbb{F}$ as follows: for any $u\sim\mathbb{F}$, i.e.\ any $u$
sampled from $\mathbb{F}$, it holds that $\Pr[u(x)=y]=\Pr[(x,y)\in\paths(u)]$. %

We first state a basic version of the sampling argument. To this end,
we define a \emph{$(p,\delta)$ non-uniform distribution}, $\mathbb{F}^{(p,\delta)}$,
which is closely related to the uniform distribution $\mathbb{F}$.
At a high level, $\mathbb{F}^{(p,\delta)}$ is ``$\delta$ close
to'' $\mathbb{F}$ with at most $p$ many paths fixed. What does
``$\delta$ closeness'' mean? Let $\Pr[S\subseteq\paths(u)]$ denote the probability that a collection $S$ of (non-colliding) paths is in $u$. Then, for any distribution $\mathbb{G}$
(over permutations), a distribution $\mathbb{G}^{\delta}$ is $\delta$
close to it if the following holds: when $t'\sim\mathbb{G}^{\delta}$
and $t\sim\mathbb{G}$, one has $\Pr[S\subseteq\paths(t')]\le2^{\delta|S|}\Pr[S\subseteq\paths(t)]$
for all $S$.

We are almost ready to state the basic sampling argument. We need
the notion of a ``convex combination'' of random variables. We say
a random variable (such as our permutation) $t$ is a convex combination
of random variables $t_{i}$, denoted by $t\equiv\sum_{i}\alpha_{i}t_{i}$
(where $\sum_{i}\alpha_{i}=1$ and $\alpha_{i}\ge0$), if the following
holds for all $t'$: $\Pr[t=t']=\sum_{i}\alpha_{i}\Pr[t_{i}=t']$.

Informally, the basic sampling argument is a statement about a uniform
permutation $u\sim\mathbb{F}$ and how the distribution $\mathbb{F}$
changes if we are given some ``advice'' about this permutation which
is simply a function $g(u)$. Roughly speaking, given that $g(u)$
evaluates to $r$ with probability at least $2^{-m}$, the distribution
$\mathbb{F}$ conditioned on $r$ is a convex combination\footnote{In the convex combination, there is a small component, of weight at
  most $2^{-m}$, of some arbitrary distribution.} of $\mathbb{F}^{(p,\delta)}$ distributions where the number of paths
fixed is at most $p=2m/\delta$. Here $\delta$ is a free parameter.
We slightly abuse the notation and write this basic sampling argument
as
\[
  \mathbb{F}|r\equiv{\rm conv}(\mathbb{F}^{(p,\delta)}).
\]
If we view $g(u)$ as the output of the first quantum part of the
circuit for $\classCQ d$, and $u$ as the oracle of interest (details
are in the next section), it is suggestive that $u|g(u)$ will be
the oracle for the second quantum part of the circuit. We can use
the sampling argument above and re-use our analysis because $\mathbb{F}$
and $\mathbb{F}^{(p,\delta)}$ have very similar statistical properties.
However, it is unclear how to use the sampling argument thereafter
as the basic sampling argument seems to only apply to $\mathbb{F}$
(and not to $\mathbb{F}^{(p,\delta)}$). It turns out that one can
extend the sampling argument to obtain
\[
  \mathbb{F}^{(p',\delta')}|r\equiv{\rm conv}(\mathbb{F}^{(p+p',\delta'+\delta)}).
\]
Consequently, if the procedure is successively applied $\tilde{n}\le\poly$
times (starting with $\mathbb{F}$), the convex combination would
be over distributions of the form $\mathbb{F}^{(\tilde{n}p,\tilde{n}\delta)}$.
The parameters can be appropriately chosen to ensure that at most
polynomially many paths are exposed but we omit the details in this
overview.

\paragraph{Sampling argument for Injective Shufflers.}

The proofs of the previously mentioned statements do not rely on any special property
of the distribution $\mathbb{F}$ nor do they depend on the fact that
we were considering permutations. Any object for which we can describe
a ``reasonable'' notion of ``paths'' admits such a sampling argument.
Therefore, as we did for permutations, to describe the sampling argument,
we change our viewpoint and consider
``paths'' in $\mathcal{L}=(H_{0},\dots H_{d})$ instead of individual values taken by the  $H_{i}$'s.
Recall that a ``path'' was a tuple of the form $(x_{0},x_{1}\dots)$
such that $x_{i}=H_{i-1}(x_{i-1})$ for all $i$. 

This viewpoint is inadequate for capturing the probabilistic behaviour
of $\mathcal{L}$ due to two reasons (which are not hard to rectify).
\emph{First}, since $H_{0}:\Sigma\to\Sigma^{d'}$, it is clear that
at least $\left|\Sigma^{d'-1}\right|$ many points will never be contained
in any ``path'' as described above. Therefore the behaviour of most
points in $H_{i}$ (for $i\in\{1\dots d\}$) will not be captured
by the ``paths'' viewpoint. \emph{Second}, even though $H_{i}$
maps $\Sigma^{d'}\to\Sigma^{d'}$ for $i\in\{1,\dots d-1\}$, $H_{i}$
may not be injective and therefore the paths might collide, which
again would mean the behaviour of many points would not be captured
by the ``paths'' viewpoint. %

To rectify the \emph{second} issue, we can select base sets $(S_{01},\dots S_{0d})=:\bar{S}_{0}$
and condition on the event $E$. Since in our proofs, we only care
about the behaviour of $\mathcal{L}$ on $\bar{S}_{0}$, it suffices
to restrict our attention to $\bar{S}_{0}$. Recall that $\mathcal{L}|E$
behaves as a permutation on $\bar{S}_{0}$. Therefore no ``path''
inside $\bar{S}_{0}$ collides. To rectify the \emph{first} issue,
we consider two kinds of paths---Type 0 paths and Type 1 paths.\footnote{The 0 and 1 represent where the first non-$\varstar$ component sits.}
A \emph{Type 0 path} is what we described earlier: a tuple of the
form $(x_{0},x_{1}\dots)$ such that $x_{i}=H_{i-1}(x_{i-1})$ for
all $i$. A \emph{Type 1 path} is a tuple of the form $(\varstar,x_{1},x_{2}\dots)$
such that $x_{1}\notin H_{0}(\Sigma)$ (i.e. $\nexists x_{0}$ st
$H_{0}(x_{0})=x_{1}$) and $x_{i}=H_{i-1}(x_{i-1})$ for all $i\in\{2,3\dots\}$. 

Observe that, restricted to $\bar{S}_{0}$ and conditioned\footnote{Recall, $E$ is the event that the oracles $H_0$ and $H_1\dots H_d$ are injective on $\Sigma$ and $\bar{S}_0$ resp.} on $E$,
we have the following equivalence: given $\Pr[H_{i}(x)=x']$ for all
$i$, $x$ and $x'$, one can compute the probability associated with
both types of paths and conversely, given probabilities associated
with the paths, one can compute $\Pr[H_{i}(x)=x']$ for all $i$,
$x$ and $x'$.

As is evident, working with ${\cal L}$ directly is cumbersome  and
we therefore define a simpler object, the \emph{injective shuffler}.
Fix sets $S_{0i}\subseteq\Sigma^{d'}$ of size $|\Sigma^{d+2}|$ for
all $i\in\{1,\dots d\}$. Let $H_{0}':\Sigma\to S_{01}$, $H_{i}':S_{0i}\to S_{0,i+1}$
for all $i\in\{1,\dots d-1\}$ be injective functions and let $H_{d}':S_{0d}\to\{0,1\}^{n}\cup\{\perp\}$
(which may not be injective) such that $H'_{d}$ outputs $\perp$
for all paths originating from $\Sigma$ (and no other).\footnote{i.e. $H'_{d}(x_{d})=\perp$ iff $(x_{0},x_{1},\dots x_{d},x_{d+1})$
is a Type 0 path (therefore $x_{d+1}=\perp$).} We define the \emph{injective shuffler}, $\mathcal{K}$ as $(H_{0}',\dots H_{d}')$.

Think of $\mathcal{K}$ as a simpler way to denote the relevant object
associated with $\mathcal{L}|E$. What do we mean by the relevant object---not
only is it injective, it also never reveals any information\footnote{Except for polynomially possibly many paths exposed by classical queries;
  we handle these shortly.} about the values taken by $\tilde{H}$ in $\Sigma$. As alluded to
at the beginning of this subsection, since the strings $s_{i}$ arise
from quantum parts which only get access to ${\cal L}$ via shadow
oracles, the sampling argument only needs to be applied to parts of ${\cal L}$
outside of paths in $\tilde{H}$.

To state the sampling argument for the injective shuffler, we define
$(p,\delta)$ non-$\beta$-uniform distributions $\mathbb{F}_{{\rm inj}}^{(p,\delta)|\beta}$
for the injective shuffler (analogous to the way we defined them for
permutations). We begin with the uniform distribution---it is simply
a distribution which assigns equal probabilities to all the possible
injective shufflers, given the sets $(S_{0i})_{i}$. As for $\beta$-uniform
distributions, $\mathbb{F}_{{\rm inj}}^{|\beta}$, we first need to
define the ``paths'', $\beta$. Here, $\beta$ will again be a set
of ``non-colliding paths'' but formalising this requires some care
(see \Subsecref{SamplingArgumentInjShuffler}). Then a $\beta$-uniform distribution is the same
as the uniform distribution except that the paths in $\beta$ are
fixed. Omitting further details, one can define $\mathbb{F}^{(p,\delta)|\beta}$
to be a distribution which is ``$\delta$ close to'' the $\beta$-uniform
distribution with at most $p$ many paths fixed (in addition to $\beta$).

The sampling argument for injective shufflers is the following. Suppose
we start with $t\sim\mathbb{F}_{{\rm inj}}^{\delta'|\beta}$ (i.e.
a distribution which is ``$\delta'$ close to'' $\beta$-uniform)
and are given some advice $h(t)$ which happens to be $r$ with probability
at least $2^{-m}$. Then the distribution $\mathbb{F}_{{\rm inj}}^{\delta'|\beta}$
conditioned on $r$ is, roughly speaking, a convex combination\footnote{Again, neglecting a component with weight at most $2^{-m}$.}
of $\mathbb{F}_{{\rm inj}}^{(p,\delta+\delta')|\beta}$ distributions
where the number of paths fixed (in addition to $\beta$) is at most
$p=2m/\delta$ and $\delta$ again is a free parameter. Using the
previous shorthand, we have
\[
  \mathbb{F}_{{\rm inj}}^{\delta'|\beta}|r\equiv{\rm conv}(\mathbb{F}_{{\rm inj}}^{(p,\delta+\delta')|\beta}).
\]

\paragraph{Stitching everything together}

As asserted before we described the sampling argument, one can replace
all the oracles ${\cal L}$ in the quantum part of the circuit for
$\classCQ d$ with appropriate shadow oracles. Let ${\cal M}_{11},\dots{\cal M}_{1d}$
denote the shadow oracles for the first quantum part, ${\cal M}_{21}\dots{\cal M}_{2d}$
for the second quantum part and so on. Suppose the paths queried by
the $i$th classical part were $\beta_{i}$, the string outputted
by the $i$th quantum part be $s_{i}$. Suppose ${\cal M}_{11}\dots {\cal M}_{1d}\dots {\cal M}_{i-1,1}\dots{\cal M}_{i-1,d}$
have been specified. Now, conditioned on $s_{i}$, the sampling argument
says ${\cal L}|s_{i}$ behaves as a convex combination of injective
shufflers with certain paths exposed, when restricted to base sets.
Let $\beta(s_{i})$ be the random variable which specifies these paths
and occurs with the weights specified in the convex combination. One
can define ${\cal M}_{i1}\dots{\cal M}_{id}$ as in the $\QNC_{d}$
case, ensuring the paths $\beta_{1}\dots\beta_{i-1}$ and $\beta(s_{1})\dots\beta(s_{i-1})$
have been exposed. Note crucially that $s_{i}$ is obtained by a quantum
part which only had access to ${\cal L}$ via shadow oracles so it
does not change the distribution over $\tilde{H}$ (except for polynomially
many paths which were already exposed, $\beta_{1}\dots\beta_{i-1}$
and $\beta(s_{1})\dots\beta(s_{i-1})$). Using a hybrid argument as
in the $\QNC_{d}$ case, and using properties of the injective shuffler
which is ``$\delta$ close'' to being uniform, one can apply the
O2H lemma and conclude that the hybrids (again, defined as in the
$\QNC_{d}$ case) are close in trace distance. Eventually, this yields
that the initial circuit is close in trace distance to the circuit
which only accesses ${\cal L}$ via the shadows ${\cal M}_{11}\dots{\cal M}_{1d}\dots{\cal M}_{m1}\dots{\cal M}_{md}$
in the quantum part (denote the number of quantum parts by $m\le\ply{\lambda}$).
The latter circuit cannot solve $\CH d$ again, because $\tilde{H}$
is only accessed by the classical parts of this circuit. More precisely,
$\tilde{H}$ is only queried at at most $|\beta_{1}\cup\dots\beta_{m}\cup\beta(s_{1})\cup\dots\beta(s_{m})|\le\ply{\lambda}$ locations %
and therefore the whole circuit can be simulated while only making
polynomially many classical queries to $\tilde{H}$. From the
soundness of $\codehashing$, this entails $\CH d$ cannot be solved.

\begin{figure}
  \begin{centering}
    \includegraphics[width=8cm]{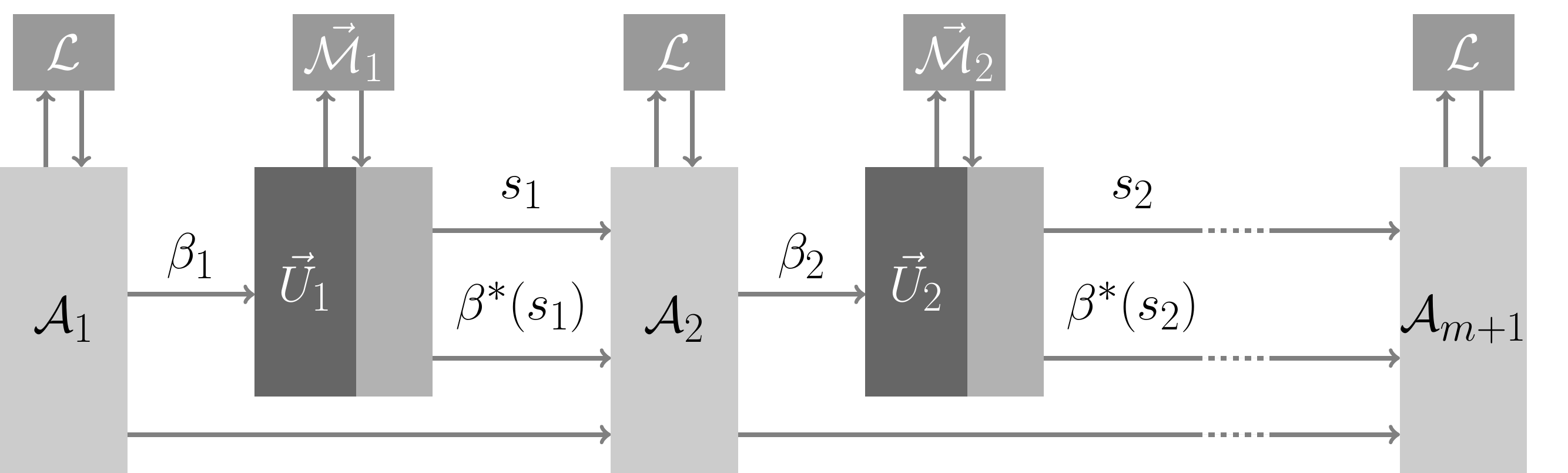}
    \par\end{centering}
  \caption{Here ${\vec{\cal M}}_{i}$ denotes the shadow oracles $({\cal M}_{i1},\dots{\cal M}_{id})$. }
\end{figure}

\subsubsection{$\CH{d} \notin \protect\classCQC d$ \label{subsec:CQCdintro}}
Just as the analysis of the $\classCQ d$ case built on the
$\QNC_{d}$ case, one can analyze the $\classCQC d$ case
by building on the $\classQC d$ case. While
the high level idea stays the same, the details are more involved. This
is partly because, in the $\QNC_{d}$ case, one could construct the
shadow oracles ${\cal M}_{1}\dots{\cal M}_{d}$ ``all at once'' since
we were assuming the ``worst case'', i.e.\ the quantum algorithm learns
everything there is to learn from the shadow oracles. However, in
the $\classQC d$ case, to define ${\cal M}_{i}$, one had to know
the behaviour of the classical algorithms in the hybrid circuits which
involved ${\cal M}_{1}\dots{\cal M}_{i-1}$ (in particular one has to know the ``paths'' that have been exposed). We show how one can account for this, but we leave the details to the main body.

\subsubsection{Proof of quantum depth}

In this subsection, we discuss how our complexity-theoretic separations also yield protocols for certifying quantum depth, i.e.\ \emph{proofs of quantum depth}, in a way that is insensitive to classical polynomial depth. First, let us be a bit more precise about what we mean by proof of quantum depth. 

\begin{defn*}[informal]
  A proof of $d$ quantum depth is a two-message protocol involving
  two parties, a verifier and a prover. Both parties are assumed to
  have access to the random oracle $H$. The verifier is a PPT machine.
  The protocol satisfies the following, where $\lambda$ is the security parameter.
\end{defn*}
\begin{itemize}
  \item Completeness: There is a prover in $\BQP$ which makes the verifier accept with probability $1-\ngl{\lambda}$.
  \item Soundness: No prover in $\classCQC{d}$ makes the verifier accept with probability more than $\ngl{\lambda}$.
\end{itemize}

Let $d$ be at most a fixed polynomial. Since $\CH d$ is in $\NP$, it immediately yields a proof of $d$ quantum depth.

We conclude this discussion by illustrating the subtlety of considering proofs of quantum depth with more than two messages. Consider
the following protocol.
\begin{example}\label{exa:interactiveDepthIssue}
  The verifier, Alice, prepares BB84 states $\left|b_{i}\right\rangle _{\theta_{i}}:=H^{\theta_{i}}\left|b_{i}\right\rangle $
  ($b_{i},\theta_{i}$ are both chosen uniformly at random) for $i\in\{1,\dots n\}$
  where $H$ is the Hadamard operation (not to be confused with the
  random oracle). She sends them all to the prover, Bob.

  Alice and Bob then engage in an $n$ round protocol. In the $i$-th
  round, Alice sends $\theta_{i}$ and Bob sends $b'_{i}$. Alice accepts
  if $b_{1}=b'_{1},\dots b_{n}=b'_{n}$.
\end{example}

In this example,\footnote{While we used quantum communication in the protocol, one could (using
  known results) delegate the production of these states to the prover
  (under computational assumptions) and run a similar protocol using
  classical communication.} it is not hard to see that Bob has to have $n$ layers of unitaries.
Could this simple construction already constitute a proof of quantum
depth? Consider the following observations.
\begin{itemize}
  \item \emph{Spoofed by $n$ single quantum depth devices.} It is easy to
        see that Bob can pass this test using $n$-many single-qubit quantum
        devices, each of which need only apply one quantum gate and make one
        computational basis measurement. The protocol works by simply delaying the
        application of the quantum gate and subsequent measurement. It is
        therefore difficult to call this a proof of quantum depth in any meaningful
        way.
  \item \emph{Interaction seems superfluous.} The only use of the interaction
        is to introduce a delay. The same effect could be achieved with a
        single round protocol where Alice delays sending her message. Therefore,
        this procedure, at best, certifies ``idle coherence'' time.
\end{itemize}
The example shows how defining quantum depth in interactive settings
can be quite subtle. We refer the reader back to the discussion in Section \ref{subsec:Discussion} for our proposal of what this definition should be.

\subsubsection{Tighter upper bounds}
Ideally, one would like to show the more fine-grained separation $\classCQC{d}\subsetneq \classCQC{d+1}$. Since the best known algorithm for solving YZ's $\codehashing$  uses polynomial depth, $\CH{d}$ inherits this limitation. We overcome this limitation and show the following.
\begin{thm} Relative to a random oracle, $\QNC_{2d+\calO(1)} \nsubseteq \classCQC{d}$ which implies $\classCQC{d} \subsetneq \classCQC{2d+\calO(1)}$.
\end{thm}

We obtain the above by instantiating our lifting procedure, $\recursive{d} [\cdot]$, with a variant of the proof of quantumness from~\cite{brakerski_simpler_2020}, which we refer to as $\collisionhashing$ (see \Tabref{ProblemsWeConsider}). It is straightforward to show that $\collisionhashing$ also satisfies classical query soundness by using the main argument in \cite{brakerski_simpler_2020} and the query lower bound for finding collisions proved in \cite{aaronson2004quantum}.

Let $g$ be a $2\to1$ function for which it is hard to find a collision. Then, the (slightly simplified) problem is to produce a pair $(y,r)$ such
that $r\cdot(x_{0}\oplus x_{1})\oplus H(x_{0})\oplus H(x_{1}) = 0$
where $\{x_{0},x_{1}\}\in g^{-1}(y)$. %
This problem can be solved in $\QNC_{{\cal O}(1)}$ (assuming that calls to $g$ take only depth 1) by preparing the superposition $\sum_{x}\left|g(x)\right\rangle \left|x\right\rangle $, measuring the second register in the standard basis, and the first in the Hadamard basis. 

We said simplified because in $\collisionhashing$, $g$ is in fact a uniformly random function $g$ (treated as an oracle) with a domain twice as large as the co-domain. Note that this is not a $2\to1$ function in general. However, with overwhelming probability, a constant fraction of the elements in the co-domain has exactly two pre-images. Then, we require a pair $(y,r)$ such that either $y$ has exactly two pre-images and $(y,r)$ satisfies the ``equation'', or $y$ does not have exactly two pre-images. The limitation of $\collisionhashing$ is that solutions to the problem are not verifiable, so the problem cannot be used to obtain a fine-grained proof of quantum depth.

\begin{table}
  \begin{centering}
    \begin{tabular}{c>{\centering}p{2cm}c>{\centering}p{2cm}>{\centering}p{2cm}c}
      \toprule
      Problem                         & Additional Assumption & Verifiable & Classical Query Soundness & Offline Soundness & Completeness\tabularnewline
      \midrule
      $\codehashing$ {[}YZ{]}         & None                  & Yes        & Yes                      & Yes               & $\BQP$\tabularnewline
      \midrule
      $\collisionhashing$             & None                  & No         & Yes                      & Yes               & $\QNC_{{\cal O}(1)}$\tabularnewline
      \bottomrule
    \end{tabular}
    \par\end{centering}
  \caption{\label{tab:ProblemsWeConsider}Problems in the random oracle model,
    which are intractable for $\protect\BPP$ and used as building blocks
    for establishing quantum depth separations.}

\end{table}

\subsection{Separations of hybrid quantum depth classes}

\subsubsection{Establishing $\protect\classCQ{{\cal O}(1)}\nsubseteq\protect\classQC d$. }
We describe our second lifting procedure, called $\serial d[\cdot]$. This takes any
problem ${\cal P}\notin \BPP$ (relative to a random oracle) that satisfies offline soundness, and produces a new problem $\serial d[{\cal P}]\notin\classQC d$
(see \Lemref{inf_dSerial}). %

Denote by $R_H$ the set of solutions to $\cal P$ (defined with respect to $H$). Then, the key idea is simple. The problem $\serial d[\cal P]$ is to return a tuple $(c_0, c_1, \ldots, c_d)$ such that: $c_0$ is a solution 
to ${\cal P}$, i.e. $c_{0}\in R_{H(\cdot)}$; $c_1$ is a solution to ${\cal P}$ but with respect to $H(c_{0}||\cdot)$, i.e. $c_{1}\in R_{H(c_{0}||\cdot)}$, and similarly until $c_d$, which should be such that $c_{d}\in R_{H(c_{0}\dots c_{d-1}||\cdot)}$.

To be a bit more concrete, take ${\cal P}$ to be $\collisionhashing$.
We know $\collisionhashing\in\QNC_{{\cal O}(1)}$. Clearly, $\serial d[\collisionhashing]\in\classCQ{{\cal O}(1)}$. This is because $\classCQ{{\cal O}(1)}$ allows one to run polynomially many $\QNC_{{\cal O}(1)}$ circuits. Consequently, one can use the first circuit
to obtain the classical output $c_{0}$, use the second circuit to
find $c_{1}$ and so on. On the other hand, intuitively, we expect that $\serial d[\collisionhashing]\notin\classQC d$. This is because to solve the $(i+1)$-th sub-problem, one seems to require the solution to all of the previous $i$ sub-problems. Since there are $d+1$ sub-problems in total, $\classQC d$ does not seem to suffice (here of course we are implicitly using the fact that $\cal P \notin \BPP$). Formally, the argument proceeds in a similar way as for the lifting map $\recursive{d}$ in \Subsecref{infQCcodehashing}, except for one subtlety which is handled by requiring that the problem $\cal P$ satisfies the extra property of offline soundness. %
We refer the reader to the main text for more details. We remark that offline soundness follows from classical query soundness and therefore both $\collisionhashing$ and $\codehashing$ satisfy it. %

The immediate consequence of the existence of the lifting map $\serial{d}[\cdot]$ is that $\protect\classCQ{{\cal O}(1)}\nsubseteq\protect\classQC d$ (first part of \Thmref{QCdiffCQ}). However, we can also leverage $\protect\serial d[\cdot]$, together with the separation from the next subsection, to show that $\classCQC{{\cal O}(1)}\nsubseteq\classCQ d\cup\classQC d$ (\Thmref{QCcupCQnotenough}). This is done as follows.

In Subsection~\ref{sec: andrea tech}, we introduce the problem $\hcollisionhashing d$ (which also satisfies offline soundness), and argue that it is in $\classQC{{\cal O}(1)}$, but not in $\classCQ d$. %
Now, applying the lifting map to it gives $\serial{d}[\hcollisionhashing d] \notin \classCQ{d} \cup \classQC{d}$. To obtain the containment, notice that $\serial{d}$ yields a problem that can be solved by solving $d+1$ many instances of the original problem. Thus, it follows that $\serial{d}[\hcollisionhashing d] \in \classCQC{{\cal O}(1)}$.

\subsubsection{Establishing $\classQC{{\cal O}(1)} \nsubseteq \classCQ d$} 
\label{sec: andrea tech}This is the more surprising of the two hybrid separations, and its proof is more involved. %
In this section, we fix $d \leq poly(\lambda)$. The problem that yields this separation is the following variation on $\collisionhashing$: given access to a 2-to-1 function $g$\,\footnote{Since we want our problem to be relative to a uniformly random oracle, in the formal description of the problem in the main text, we will not assume that $g$ is exactly 2-to-1. Rather we will take $g$ to be a uniformly random function with domain twice as large as the co-domain, and simply restrict our attention to $y$'s in the co-domain that have exactly two pre-images (this is a constant fraction of the elements of the co-domain with overwhelming probability).}, and to $H_0,\dots H_d$ (which specify $h$ as $h=H_d \circ \dots \circ H_0$), find a pair $(y,r)$ such that $$r \cdot (x_0 \oplus x_1) \oplus H(h(y) || x_0) \oplus H(h(y) || x_1) = 0\,,$$
where $\{x_0,x_1\} = g^{-1}(y)$. We refer to the new problem as $\hcollisionhashing d$.

Without relying on $h$ (that is, requiring that the equation to be satisfied is just $r \cdot (x_0 \oplus x_1) \oplus H(x_0) \oplus H(x_1) = 0$), this problem is the same as $\collisionhashing$. This can be solved in $\QNC_{\mathcal{O}(1)}$ as follows: 
\begin{itemize}
    \item[(i)] Evaluate $g$ on a uniform superposition of inputs, obtaining $\sum_x \ket{x} \ket{g(x)}$,
    \item[(ii)] Measure the image register obtaining some outcome $y$ and a state $(\ket{x_0} + \ket{x_1})\ket{y}$,
    \item[(iii)] Query a phase oracle for $H$ to obtain $((-1)^{H(x_0)} \ket{x_0} + (-1)^{H(x_1)} \ket{x_1})\ket{y}$,
    \item[(iv)] Make a Hadamard basis measurement of the first register, obtaining outcome $r$. 
\end{itemize} 

At a high level, in order to solve the new problem, which includes the evaluation of $h$ as an input to $H$, one needs the ability to perform a (classical) depth $d$ computation to evaluate $h(y)$ (since this requires the sequential evaluations of $H_0, \ldots,H_d$). Note that a $\qnc^{\bpp}$ algorithm can solve this problem: the only modification to the algorithm described above is that, at step (iii), the algorithm first computes $h(y)$ (using polynomial classical computation), and then queries the oracle $H$ on a superposition of $(h(y), x_0)$ and $(h(y), x_1)$. One can easily verify that this leads to a valid $y,r$ for the problem.

Next, we give a sketch of how one can argue that the problem cannot be solved in $\bpp^{\qnc}$. The key technical ingredient is a ``structure theorem'' that characterizes the structure of efficient quantum strategies that are successful at $\collisionhashing$. Our structure theorem applies equally to the proof of quantumness protocol from \cite{brakerski_simpler_2020} (recall that the latter is just a version of collision hashing where $g$ is replaced by a 2-to-1 trapdoor claw-free function).

\begin{thm}[informal] \label{thm:andreaMagic}
Let $P$ be any $\BQP$ prover that succeeds with $1-\negl$ probability at the proof of quantumness protocol from \cite{brakerski_simpler_2020}, by making $q$ queries to the oracle $H$. Then, with $1-\negl$ probability over pairs $(H, y)$, the following holds. Let $p_{y|H}$ be the probability that $P^H$ outputs $y$, and let $x_0$,$x_1$ be the pre-images of $y$. Then, for all $b \in \{0,1\}$, there exists $i \in [q]$ such that the state of the query register of $P^H$ right before the $i$-th query has weight $\frac12 p_{y|H}\cdot (1-\negl)$ on $x_b$.
\end{thm}

See Corollary~\ref{cor: bkvv} for a formal statement of this result. 
This is a crucial strengthening of a Theorem from \cite{coladangelo2022deniable}, and employs the compressed oracle technique \cite{zhandry2019record}. A slight adaptation of this to our problem asserts that a successful strategy must be querying the random oracle $H$ at a (close to) uniform superposition of $(h(y), x_0)$ and $(h(y), x_1)$. 

Now let $A$ be a $\bpp^{\qnc}$ algorithm that succeeds at $\hcollisionhashing d$ with high probability and let $q$ be the total number of queries to $h$ made by the algorithm.

Then, one can show that, since the $\qnc$ part of the algorithm does not have sufficient depth to evaluate $h$ (which is a sequential evaluation of $H_0, \ldots, H_d$), we can assume, %
without loss of generality, the $\qnc$ part of $A$ has no access to $h$. In other words, all of the queries to $h$ are classical.

Now, Theorem \ref{thm:andreaMagic} says essentially that, for any $y$, the only way to succeed with high probability (conditioned on that $y$ being the output) is to query (with as much weight as the probability of outputting $y$) a uniform superposition of $(h(y), x_0)$ and $(h(y), x_1)$. However, observe that, for any $y$, the only way for $A$ to query $H$ (with a high weight) at a uniform superposition of $(h(y), x_0)$ and $(h(y), x_1)$ is to correctly guess the value of $h(y)$. Since this value is uniformly random for any algorithm that has not queried $h$ at $y$, it follows that querying $H$ at the uniform superposition of $(h(y), x_0)$ and $(h(y), x_1)$ must necessarily happen \emph{after} the algorithm has already queried $h$ on $y$.

This implies that there must exist an $i^* \in [q]$ such that, with high probability, $A$ outputs $y,r$ such that $y$ is contained in the list of classical queries made to $h$ up to the $i^*$-th query. Denote such a list by $L_{i^*}$. Moreover, with high probability over $L_{i^*}$, the continuation of $A$ (from that point on) queries $H$ at a uniform superposition of $(h(y), x_0)$ and $(h(y), x_1)$ for some $y \in L_{i^*}$. We show that such an algorithm $A$ can be leveraged to extract a collision for $g$.

The key observation is that, since $A$ is a $\bpp^{\qnc}$ algorithm, and all of the queries to $h$ happen in the $\bpp$ portion of $A$, the ``state'' of algorithm $A$ right after the $i^*$-th query to $h$ is entirely \emph{classical}. Thus, one can take a ``snapshot'' of the state of $A$ at that point (i.e.\ copy it), and simply run \emph{two independent executions} of $A$ from that point on (with independent classical randomness). By what we argued earlier, with high probability, there exists $y \in L_{i^*}$, such that the execution of $A$ from that point on, queries $H$ at a uniform superposition of $(h(y), x_0)$ and $(h(y), x_1)$. Since the two executions are identical and independent, it follows that measuring the query registers of $H$ in both executions will yield distinct pre-images of $y$ with significant probability.

Finding collisions of $g$ is of course hard (for any query-bounded quantum algorithm) \cite{aaronson2004quantum}. Hence, this yields a contradiction.

\pagebreak

\section{Preliminaries} \label{sect:prelim}
\branchcolor{purple}{We state the preliminaries which are common to both parts in this section. Each part also has its own set of preliminary results.}
\subsection{Models of Computation\label{sec:Models-of-Computation}}

We first list the standard notation we use. PPT denotes a probabilistic polynomial time algorithm, QPT denotes a quantum polynomial time algorithm. As we primarily focus on search problems, to keep the presentation clean, we slightly abuse the notation and use decision class names to represent the corresponding search classes. For instance, we use $\BPP$ and $\BQP$ to denote the search classes $\FBPP$ and $\FBQP$ resp. which in turn are defined as follows.

\begin{defn}[$\FBPP, \FBQP$; paraphrased from \cite{forrelation_one,ScottRef-AskAndru_two}] Let $\FBPP$ be the set of relations $R \subseteq \bit^* \times \bit^*$ such that for each $R$, there is a PPT algorithm $\calA$ satisfying the following: for all input strings $x$, 
    $$\Pr[(x,y)\in R : y\leftarrow \calA(x)] \ge 1 - o(1)$$
$\FBQP$ is defined analogously (PPT is replaced with QPT).\footnote{NB: This, in particular, implies there is at least one $y$ for every $x$, s.t. $(x,y)\in R$.}
\end{defn}

\begin{figure}

    \begin{centering}
        \subfloat[$\QNC_{d}$ scheme; $U_{i}$ are single depth unitaries; the measurement at the end is performed in the computational basis.\label{fig:QNCd}]{\begin{centering}
                \hspace{1.75cm}\includegraphics[width=3.5cm]{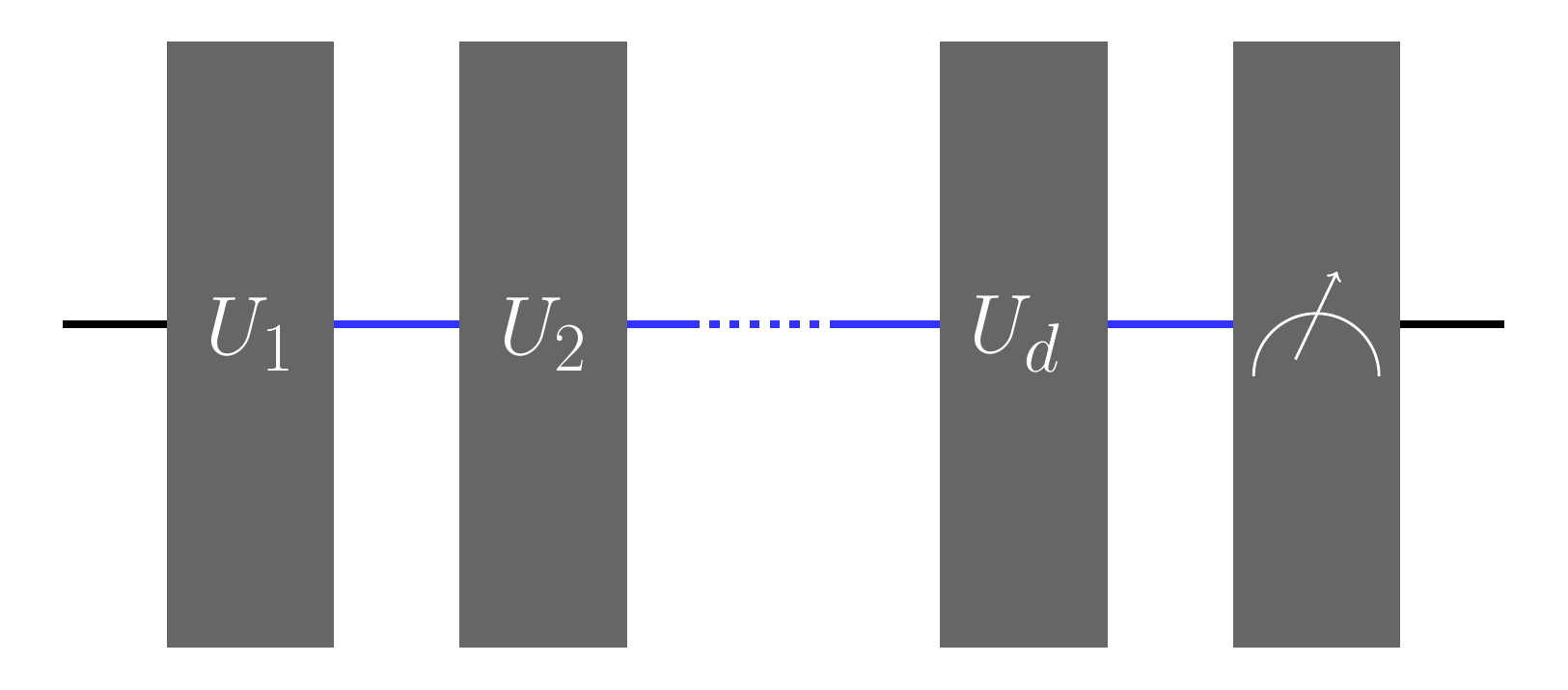}\hspace{1.75cm}
                \par\end{centering}
        }\enskip{}\subfloat[$\QC d$ circuit; $U_{i}$ are single layered unitaries, $\calA_{c,i}$
        are classical poly-sized circuits (in the figure, henceforth, we drop the subscript for
        $\calA_{c}$) and the measurements are in the
        computational basis. Dark lines denote qubits.\label{fig:QCd}]{\centering{}\includegraphics[width=8cm]{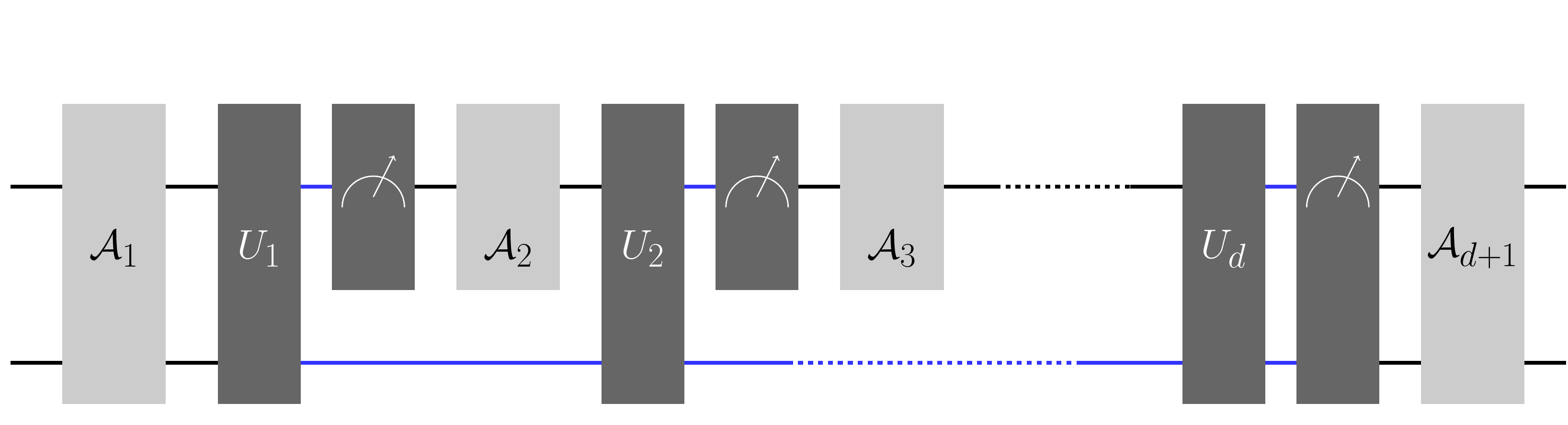}}
        \par\end{centering}

    \begin{centering}
        \subfloat[$\CQ d$ circuit; for clarity, we dropped the indices in $\calA_{c}$
        and the second indices in $U_{1,i},U_{2,i}\dots U_{d,i}$. \label{fig:CQd}]{\begin{centering}
                \includegraphics[width=8cm]{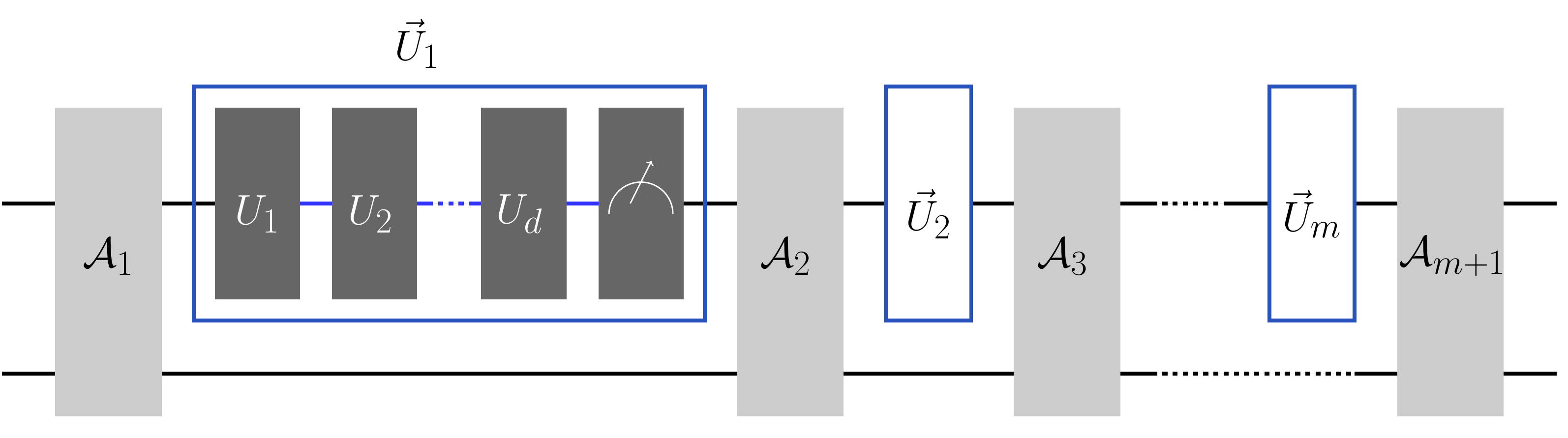}
                \par\end{centering}
        }\enskip{}\subfloat[$\CQC{d}$ circuit; $\mathcal{Q}_i$ denotes $i$th $\QC d$ circuit and $m=\mathrm{poly}(n)$. The measurements after the single layer unitaries are included in $U_{i,j}$ with $j=1,\cdots,d$. The final classical part is labelled $\calA_{m+1}$ instead of $\calA_{c,m+1,1}$ for simplicity. \label{fig:CQCd}]{\centering{}\includegraphics[width=8cm]{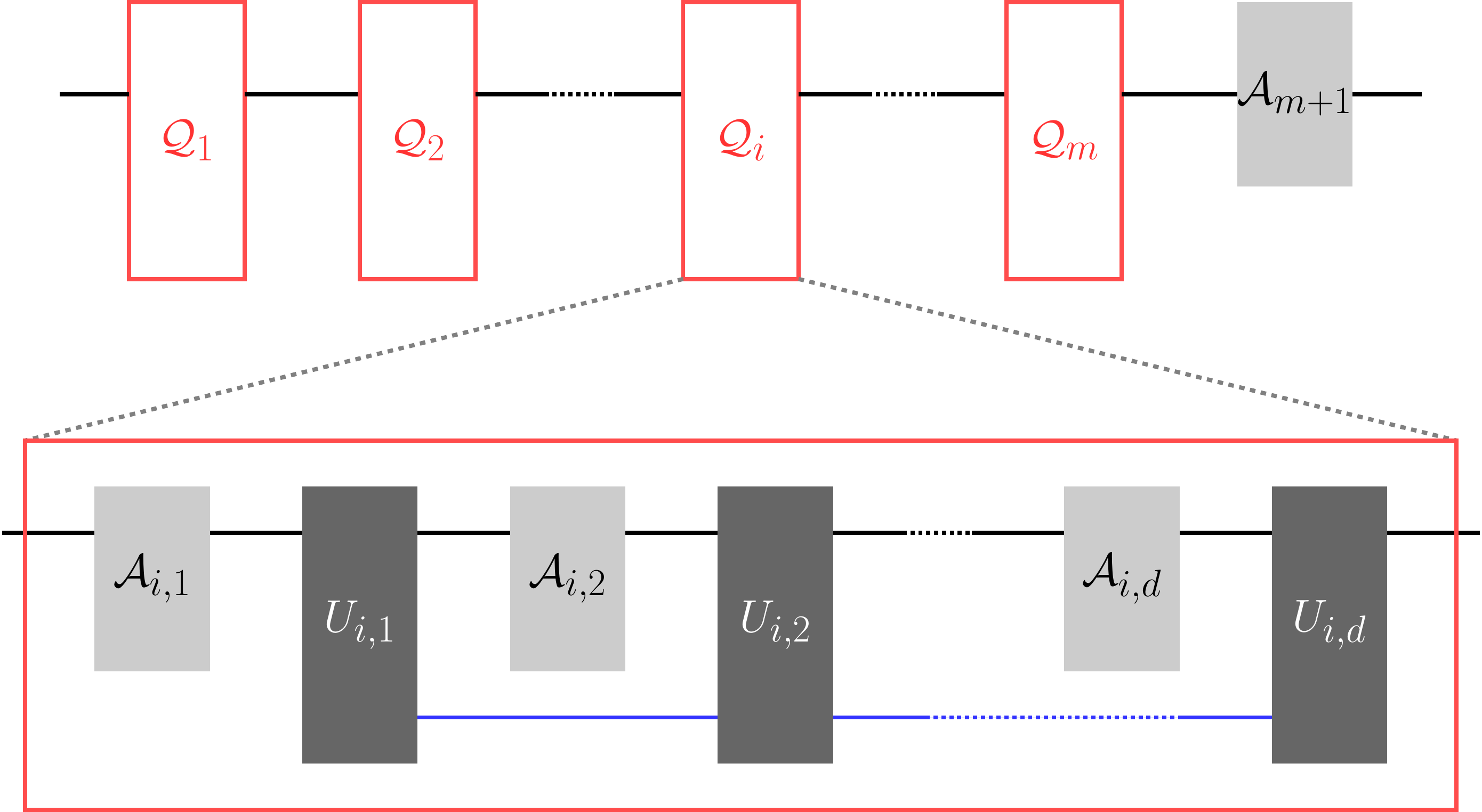}}
        \par\end{centering}
    \caption{The four circuit models we consider. We draw single wires to represent potentially polynomially many wires. Black lines and blue lines indicate wires carrying classical and quantum information, respectively. We implicitly follow this convention henceforth.}
\end{figure}

\branchcolor{purple}{Unlike the decision classes, it is unclear if changing the error from $o(1)$ to some constant (say $2/3$rds) preserves the class. For our purposes, $o(1)$ suffices. We now define circuit models and the associated classes, depending on their depth; we drop the ``F'' prefix entirely.}
\begin{notation}
    A \emph{single layer unitary}, is defined by a set of one and two-qubit
    gates which act on disjoint qubits (so that they can all act in parallel
    in a single step). The number of single layer unitaries in a circuit defines its \emph{depth}.
\end{notation}

\begin{defn}[${\QNC}_{d}$ circuits and ${\BQNC}_{d}$ relations]
    \label{def:QNCd} Denote by ${\QNC}_{d}$ the set of $d$-depth
    quantum circuits (see \Figref{QNCd}).

    Define ${\BQNC}_{d}$ to be the set of all relations $R\in \bit^*\times\bit^*$ which satisfy the following: for each relation $R\in{\BQNC}_{d}$,
    there is a circuit family $\{\calC_{n}:\calC_{n}\in{\QNC}_{d}\text{ and acts on }\poly\text{ qubits}\}$
    and for all strings 
    $x$, $$\Pr[(x,y)\in R\,:\, y\leftarrow \calC_{|x|}(x)] \ge 1- o(1).$$

\end{defn}

\begin{defn}[$\QC{d}$ circuits and ${\BQNC}_{d}^{{\BPP}}$ relations]
    \label{def:dQC} Denote by $\QC{d}$ the set of all circuits
    which, for each $n\in\mathbb{N}$, act on $\poly$ qubits and bits
    and can be specified by
    \begin{itemize}
        \item $d$ single layered unitaries, $U_{1},U_{2}\dots U_{d}$,
        \item $d+1$ $\poly$-sized classical circuits $\calA_{c,1}\dots\calA_{c,d},\calA_{c,d+1}$,
              and
        \item $d$ computational basis measurements
    \end{itemize}
    that are connected as in \Figref{QCd}.

    Define ${\BQNC}_{d}^{{\BPP}}$ analogously to $\BQNC_d$ relations, replacing $\QNC_d$ circuits with $\QC{d}$ circuits. When $d(n)=\polylog{n}$, denote the set of relations by ${\BQNC}^{{\BPP}}$.\\
\end{defn}

\begin{defn}[$\CQ d$ circuits and ${\BPP}^{\BQNC_{d}}$ relations]
    \label{def:dCQ} Denote by $\CQ d$ the set of all circuits
    which, for each $n\in\mathbb{N}$ and $m=\poly$, act on $\poly$
    qubits and bits and can be specified by
    \begin{itemize}
        \item $m$ tuples of $d$ single layered unitaries $(U_{1,i},U_{2,i}\dots U_{d,i})_{i=1}^{m}$,
        \item $m+1$, $\poly$ sized classical circuits $\calA_{c,1}\dots\calA_{c,m},\calA_{c,m+1}$,
              and
        \item $m$ computational basis measurements
    \end{itemize}
    that are connected as in \Figref{CQd}.

    Define, as above, ${\BPP}^{\BQNC_{d}}$ analogously to $\BQNC_d$ relations, replacing $\QNC_d$ circuits with $\CQ d$  circuits. When $d(n)=\polylog{n}$, denote the set of relations by ${\BPP}^{\BQNC}$.

\end{defn}

\begin{defn}[$\CQC d$ circuits and $\BPP^{\BQNC_{d}^{\BPP}}$ relations]
    \label{def:dCQC} Denote by $\CQC d$ the set of all circuits which,
    for each $n\in\mathbb{N}$ and $m=\poly$, which can be specified
    by $m$ $\QC d$ circuits acting on $\poly$ qubits and bits, that
    are connected as in \Figref{CQCd}.
  
    Define, as above, $\BPP^{\BQNC_{d}^{\BPP}}$ analogously to $\BQNC_d$ relations,
    replacing $\QNC_d$ circuits with $\CQC d$ circuits. With $d(n)=\polylog{n}$, denote the
    set of relations by $\BPP^{\BQNC^{\BPP}}$.
  \end{defn}

\begin{rem}
    Connection with the more standard notation: $\QNC_{d}$ has depth
    $d$ and $\QNC^{m}$ has depth $\log^{m}(n)$, i.e. $\QNC^{m}=\QNC_{\log^{m}(n)}$.
\end{rem}

\branchcolor{purple}{Later, it would be useful to symbolically represent these three circuit
    models but we mention them here for ease of reference.}
\begin{notation}
    \label{nota:CompositionNotation}We use the following notation convention.
    \begin{itemize}
        \item Probability: The probability of an event $E$ occurring, as a result of a process
              $P$, is denoted by $\Pr[E:P]$. In our context, the probability
              of a random variable $X$ taking the value $x$ when process $Y$
              takes place is denoted by $\Pr[x\leftarrow X:Y]$. When the process
              $Y$ is just a sampling of $X$, we drop the $Y$ and use $\Pr[x\leftarrow X]$.
        \item $\QNC_{d}$: We denote a $d$-depth quantum circuit (see \Defref{dQC}
              and \Figref{QNCd}) by $\calA=U_{d}\circ\dots\circ U_{1}$ and
              (by a slight abuse of notation) the probability that running the algorithm
              on all zero inputs yields $x$, by $\Pr[x\leftarrow\calA]$
              while that on some input state $\rho$ by $\Pr[x\leftarrow\calA(\rho)]$.
        \item $\QC d$: We denote a $\QC d$ circuit (see \Defref{dQC}
              and \Figref{QCd}) by $\calB=\calA_{c,d+1}\circ\calB_{d}\circ\calB_{d-1}\dots\circ\calB_{1}$
              where $\calB_{i}:=\Pi_{i}\circ U_{i}\circ\calA_{c,i}$
              and ``$\circ$'' implicitly denotes the composition as shown in
              \Figref{QCd}. As above, the probability of running the circuit $\calA$
              on all zero inputs and obtaining output $x$ is denoted by $\Pr[x\leftarrow\calB]$
              while that on some input state $\rho$ by $\Pr[x\leftarrow\calB(\rho)]$.
        \item $\CQ d$: We denote a $\CQ d$ circuit (see \Defref{dCQ}
              and \Figref{CQd}) by $\calC=\calA_{c,m+1}\circ\calC_{m}\circ\dots\circ\calC_{1}$
              where $\calC_{i}:=\Pi_{i}\circ U_{d,i}\circ\dots\circ U_{1,i}\circ\calA_{c,i}$
              and ``$\circ$'' implicitly denotes the composition as shown in
              \Figref{CQd}. Again, the probability of running the circuit $\calC$
              on all zero inputs and obtaining output $x$ is denoted by $\Pr[x\leftarrow\calC]$
              while that on some input state $\rho$ by $\Pr[x\leftarrow\calC(\rho)]$.
        \item $\CQC d$: We denote a $\CQC d$ circuit (see \Defref{dCQC} and \Figref{CQCd}) 
            by $\mathcal{D}=\mathcal{A}_{c,m+1,1}\circ\mathcal{D}_{m}\circ\dots\circ\mathcal{D}_{1}$
            where $\mathcal{D}_{i}=\mathcal{B}_{i,d}\circ\mathcal{B}_{i,d-1}\circ\dots\circ\mathcal{B}_{i,1}$
            is a\footnote{except we excluded the last classical circuit $\mathcal{A}_{c,i,d+1}$.
            This is without loss of generality because $\mathcal{A}_{c,i,d+1}$
            can be absorbed in the first classical circuit, $\mathcal{A}_{c,i+1,1}$,
            of $\mathcal{D}_{i+1}$. } $\QC d$ circuit with $\mathcal{B}_{i,j}:=\Pi_{i,j}\circ U_{i,j}\circ\mathcal{A}_{c,i,j}$
            for $i,j\in\{1,\dots d\}$ and ``$\circ$'' implicitly denotes the
            composition as shown in \Figref{CQCd}.    
    \end{itemize}
\end{notation}

\subsubsection{The Oracle Versions}

\branchcolor{purple}{We consider the standard Oracle/query model corresponding to functions---the
    oracle returns the value of the function when invoked classically
    and its action is extended by linearity when it is accessed quantumly. }
\begin{notation}
    An oracle $\calO_{f}$ corresponding to a function $f$ is given
    by its action on ``query'' and ``response'' registers as $\calO_{f}\left|x\right\rangle _{Q}\left|a\right\rangle _{R}=\left|x\right\rangle _{Q}\left|a\oplus f(x)\right\rangle _{R}$.
    An oracle $\calO_{(f_{i})_{i=1}^{k}}$ corresponding to multiple
    functions $f_{1},f_{2}\dots f_{k}$ is given by $\calO_{(f_{i})_{i=1}^{k}}\left|x_{1},x_{2}\dots x_{k}\right\rangle _{Q}\left|a_{1},a_{2}\dots a_{k}\right\rangle _{R}=\left|x_{1},x_{2}\dots x_{k}\right\rangle _{Q}\left|a_{1}\oplus f_{1}(x_{1}),a_{2}\oplus f_{2}(x_{2}),\dots a_{k}\oplus f_{k}(x_{k})\right\rangle _{R}$.
    \\
    When $\calO_{f}$ is accessed classically,
    we use $\calO_{f}(x)$ to mean it returns $f(x)$. \label{nota:functionToOracle}
\end{notation}

\begin{rem}[$\QNC_{d}^{\calO}$, $\QC d^{\calO}$, $\CQ d^{\calO}$]
    \label{rem:oracleVersionsQNC-CQ-QC} The oracle versions of $\QNC_{d}$,
    $\QC d$ and $\CQ d$ circuits are as shown in \Figref{QNCdOracle,dQC_oracle,dCQ_oracle}.
    We allow (polynomially many) parallel uses of the oracle even though in the figures we represent
    these using single oracles. We do make minor changes to the circuit
    models, following \cite{chia_need_2020-1} when we consider $\QNC_{d}$
    circuits and $\CQ d$ circuits---an extra single layered unitary
    is allowed to process the final oracle call.
\end{rem}

We end by explicitly augmenting \Notaref{CompositionNotation} to
include oracles.
\begin{notation}\label{nota:CompositionNotationOracles}
    When oracles are introduced, we use the following notation.
    \begin{itemize}
        \item $\QNC_{d}^{\calO}$: $\calA^{\calO}=U_{d+1}\circ\calO\circ U_{d}\circ\dots\calO\circ U_{1}$
              (see \Figref{QNCdOracle})
        \item $\QC d^{\calO}:$ $\calB^{\calO}=\ensuremath{\calA_{c,d+1}^{\calO}}\circ\text{\ensuremath{\calB_{d}^{\calO}}}\ensuremath{\circ}\dots\ensuremath{\calB_{1}^{\calO}}$
              where $\calB_{i}^{\calO}=\Pi_{i}\circ\calO\circ U_{i}\circ\calA_{c,i}^{\calO}$
              and $\calA_{c,i}^{\calO}$ can access $\calO$ classically
              (see \Figref{dQC_oracle}).
        \item $\CQ d^{\calO}$: $\calC^{\calO}=\calA_{m+1}^{\calO}\circ\calC_{m}^{\calO}\circ\dots\calC_{1}^{\calO}$
              where $\calC_{i}^{\calO}:=\Pi_{i}\circ U_{d+1,i}\circ\calO\circ U_{d,i}\circ\dots\circ\calO\circ U_{1,i}\circ\calA_{c,i}^{\calO}$
              where $\calA_{c,i}^{\calO}$ can access $\calO$
              classically (see \Figref{dCQ_oracle}).
        \item $\CQC d^{{\cal O}}$: ${\cal C}^{{\cal O}}={\cal A}_{c,m+1,1}\circ{\cal D}_{m}^{{\cal O}}\circ\dots{\cal D}_{1}^{{\cal O}}$
            where ${\cal D}_{i}={\cal B}_{i,d}\circ\dots{\cal B}_{i,1}$ with
            ${\cal B}_{i,j}^{{\cal O}}:=\Pi_{i,j}\circ{\cal O}\circ U_{i,j}\circ{\cal A}_{c,i,j}^{{\cal O}}$
            and ${\cal A}_{c,i,j}$ accesses ${\cal O}$ classically (see \Figref{CQCdOracle}
      \end{itemize}
      The classes $\left(\BQNC_{d}^{\BPP}\right)^{\mathcal{O}},\left(\BPP^{\BQNC_{d}}\right)^{\mathcal{O}}$
      and $\left(\BPP^{\BQNC_{d}^{\BPP}}\right)^{{\cal O}}$ are implicitly
      defined to be the query analogues of $\BQNC_{d}^{\BPP},\BPP^{\BQNC_{d}}$
      and $\BPP^{\QNC_{d}^{\BPP}}$ (resp.), i.e. class of relations solved
      by $\QC d^{\mathcal{O}},\CQ d^{\mathcal{O}}$ and $\CQC d^{\mathcal{O}}$
      circuits (resp.).

\end{notation}

\begin{figure}
    \begin{centering}
        \subfloat[A $\QNC_{d}$ circuit with access to oracle $\calO$. Following
            \cite{chia_need_2020-1}, in the oracle version of $\QNC_{d}$, we allow
            it to perform one extra single layered unitary to process the output.
            \label{fig:QNCdOracle}]{\begin{centering}
                \includegraphics[width=5cm]{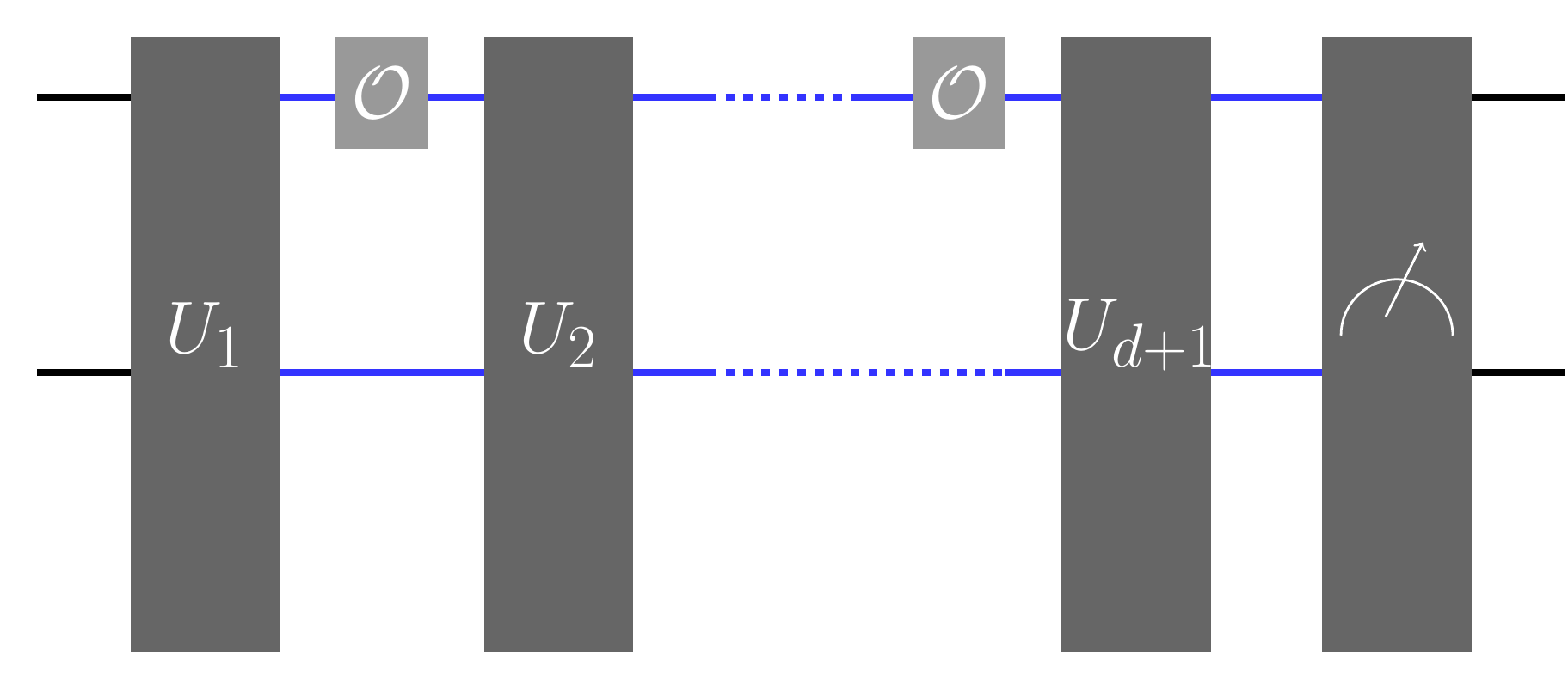}
                \par\end{centering}
        }
        \par\end{centering}
    \begin{centering}
        \subfloat[A $\QC d$ circuit with access to an oracle $\calO$.
            There is no ``extra'' single layered unitary in this model. \label{fig:dQC_oracle}]{\begin{centering}
                \includegraphics[width=10cm]{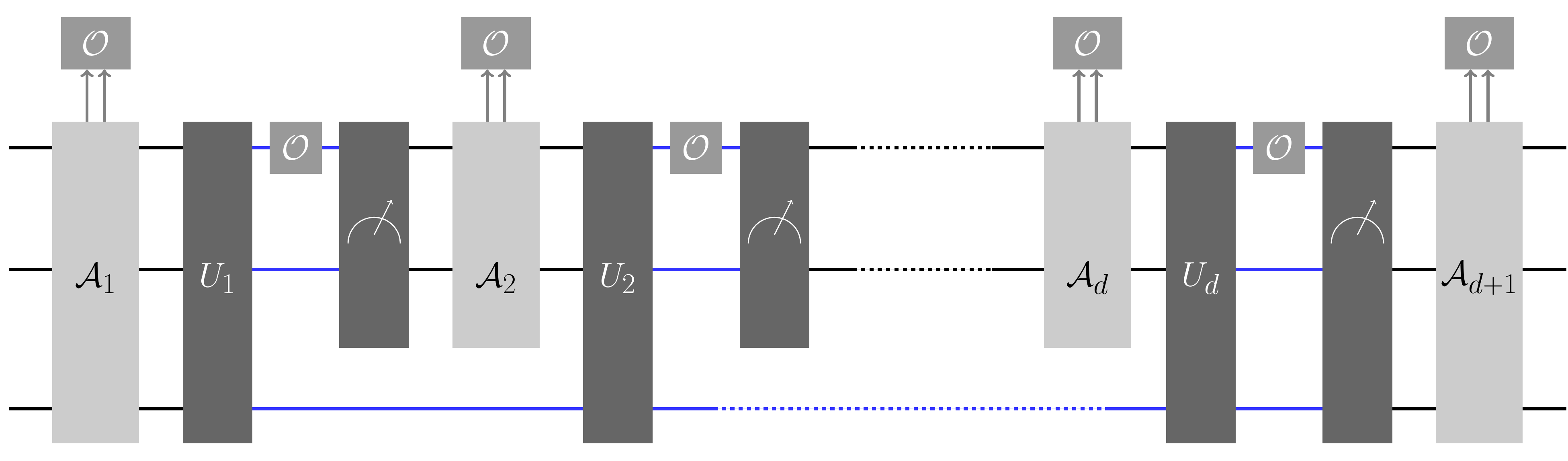}
                \par\end{centering}
        }\enskip{}\subfloat[A $\CQ d$ circuit with access to an oracle $\calO$.
            Again, following \parencite{chia_need_2020-1}, we allow an extra single layer
            unitary to process the result of the last oracle call.\label{fig:dCQ_oracle}]{\begin{centering}
                \includegraphics[width=8cm]{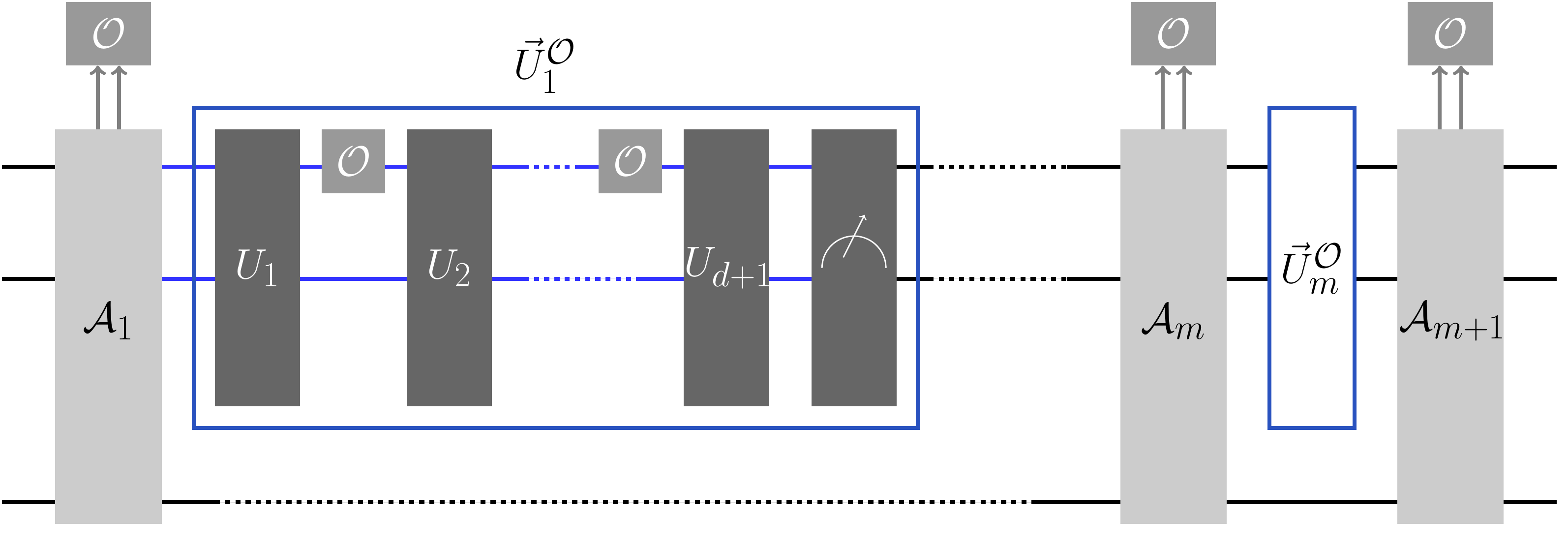}
                \par\end{centering}
        }
        \par\end{centering}
    \begin{centering}
        \subfloat[A $\CQC d$ circuit with access to oracle $\calO$. 
            \label{fig:CQCdOracle}]{\begin{centering}
                \includegraphics[width=8cm]{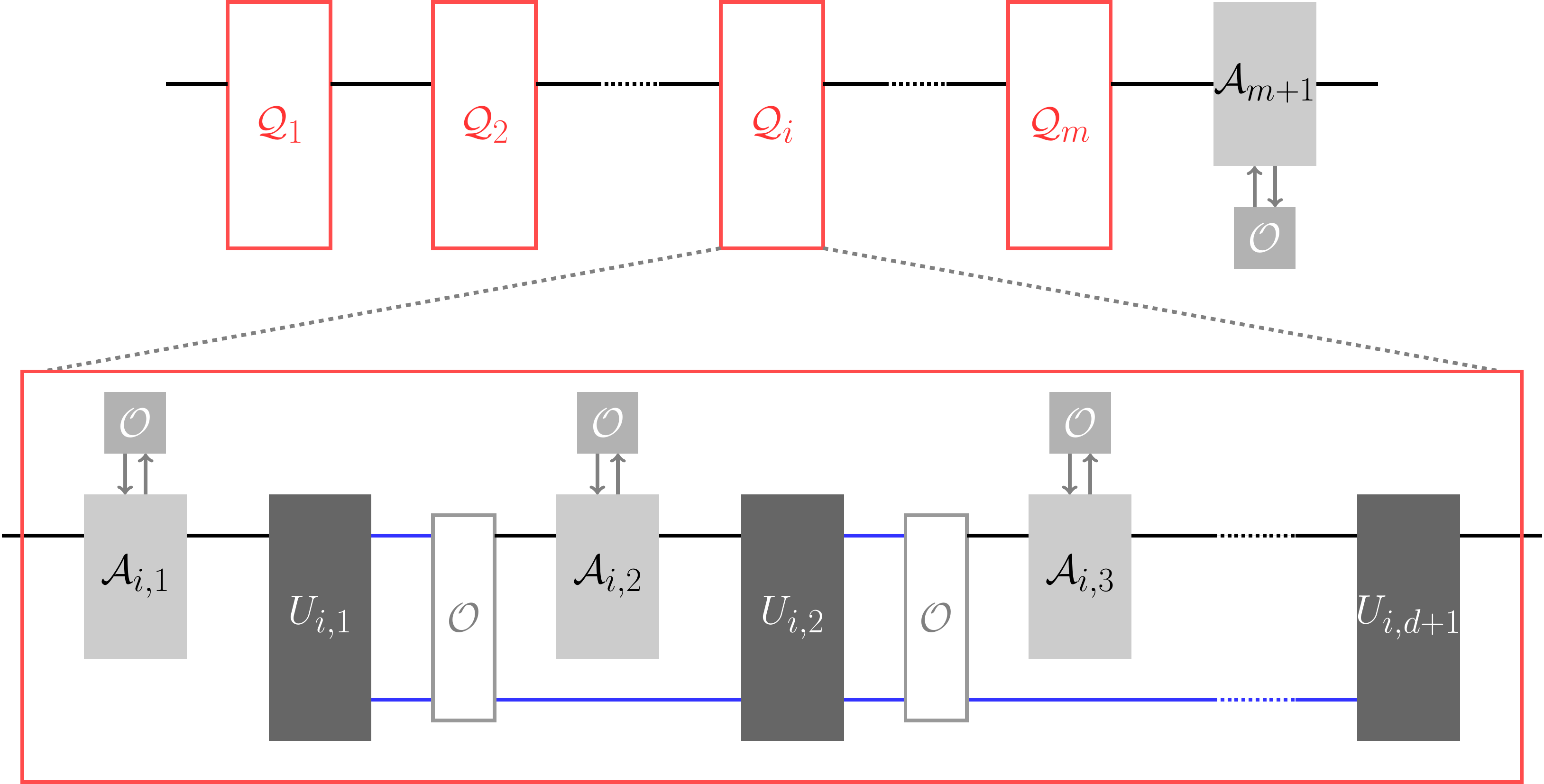} 
                \par\end{centering}
        }
        \par\end{centering}
        
    \caption{The same four circuit models, but with oracle access.} %
\end{figure}

\subsection{The Random Oracle Model\label{sec:ROM_model}}
In the random oracle model, all parties are given access to a random function $H$ which is defined from $\{0,1\}^* \to \{0,1\}$ s.t.\ it assigns $0$ or $1$ to each input $x$, independently with probability half. Quantum algorithms can access $H$ in superposition.

\subsubsection{Domain Splitting}
    Using domain splitting, one can efficiently construct expanding (compressing resp.) random functions, i.e.\ uniformly random functions $H'$ from $\{0,1\}^n \to \{0,1\}^m$ where $m\ge n$ ($m\le n$ resp.) using $H$. By efficiently (wrt $n$) we mean in time $\poly$ which translates to $m \le \poly$. More precisely, one can define $H'(x):=(H(x||0),H(x||1),\dots H(x||m))$ where $||$ denotes concatenation and the second part of the string has length at most $\log(m)$. Similarly, one can construct polynomially\footnote{In fact, exponentially many as we only need to polynomially many bits to index them.} many distinct random compressing/expanding functions from $H$. One can therefore use such random functions without loss of generality in the random oracle model.

\subsubsection{Oracle Independent, Uniform and Non-Uniform Adversaries\label{subsec:uniform_nonuniform}}
We consider three kinds of adversaries (circuit families $\{\mathcal{C}_n\}$) and their correlation with the random oracle.
    \begin{itemize}
        \item \emph{Oracle independent.} The circuit family $\{\mathcal{C}_n\}$ and $H$ are uncorrelated. First the circuit family is chosen, then $H$ is sampled.
        \item \emph{Uniformly oracle dependent.} First $H$ is chosen; then some fixed length string $a$ (possibly correlated with $H$) is given as advice to the circuit family $\{\mathcal{C}_n\}$. %
        \item \emph{Non-uniformly oracle dependent.} First $H$ is chosen; for each input length $n$, a potentially different string $a_n$ is chosen which is given to circuit $C_n$ as advice.
    \end{itemize}
In the cryptographic setting, security against the third type of adversary is desired. %
We will prove our results against oracle independent adversaries and invoke known results to lift the security to non-uniform oracle dependent adversaries for cryptographic applications. 

\subsection{Basic Quantum information results}

\branchcolor{purple}{We setup some notation for distances and recall some basic results. Here, all density matrices are defined on the same Hilbert space.}
\begin{defn}
    Let $\rho,\rho'$ be two mixed states. Then we define
    \begin{itemize}
        \item Fidelity: ${\rm F}(\rho,\rho'):=\tr(\sqrt{\sqrt{\rho}\rho'\sqrt{\rho}})$
        \item Trace Distance: ${\rm TD}(\rho,\rho'):=\frac{1}{2}\tr\left|\rho-\rho'\right|$
              and
        \item Bures Distance: ${\rm B}(\rho,\rho'):=\sqrt{2-2F(\rho,\rho')}$.
    \end{itemize}
\end{defn}

\begin{fact}
    For any set of strings $S$, any string $s$ and any two mixed states, $\rho$ and $\rho'$, and
    any quantum algorithm $\calA$ (which outputs a classical string), we have
    \[
        \left|\Pr[s\in S: s\leftarrow\calA(\rho)]-\Pr[s \in S: s\leftarrow\calA(\rho')]\right|\le \td(\rho,\rho') \le B(\rho,\rho').
    \]
\end{fact}

To see this, recall that $\td(\rho,\rho') = \max_{\mathbb{I} \ge P\ge 0} |P(\rho - \rho')| $ and therefore $|\Pr[s\in S:s\leftarrow \mathcal{A}(\rho)] - \Pr[s\in S: s\leftarrow \mathcal{A}(\rho')]| \le \td(\rho,\rho')$. Recalling also
${\rm TD}(\rho, \rho') \le \sqrt{1-F(\rho,\rho')} \le B(\rho,\rho')$ one obtains the asserted result.

We use the following basic properties repeatedly in our analysis. For any density matrices $\rho,\rho',\sigma$, it holds that: (1) $\td (\rho,\rho') \le \td (\rho,\sigma) + \td (\sigma,\rho')$, and (2) $\td(\Phi(\rho),\Phi(\sigma))\le \td(\rho,\sigma)$ for any completely positive trace non-increasing map $\Phi$ (see, e.g., \cite{traceNorm_PWPR_2006}).

\pagebreak

\part{Bounds on Quantum Depth\label{part:lowerbounds}}

    We first establish that $\classCQC{d} \subsetneq \BQP$ relative to a random oracle. Based on this result, we describe how to construct a proof of $d$ quantum depth which is insensitive to polynomial classical depth. Subsequently, we tighten the $\BQP$ upper bound to obtain a more fine-grained quantum depth separation, $\classCQC{d} \subsetneq \classCQC{2d+\calO(1)}$.

    More precisely, in \Secref{dRecursive}, we introduce a map which can be applied to any problem separating $\BQP$ and $\BPP$ which additionally specifies what we call \emph{classical query soundness}, to create a new problem which separates $\classCQC d$ and $\BQP$. For concreteness, in \Secref{dCodeHashingDefined} we apply this procedure to YZ's $\codehashing$ and in \Secref{lowerbounds} prove that the resulting problem is not in $\classCQC d$. We then, in \Secref{defn_PoQD}, formalise the notion of a \emph{proof of quantum depth} and construct a two-message proof of quantum depth protocol based on the previous result. Finally, in \Secref{ImprovedUpperBound}, we apply the map to a different problem and improve the upper bound to obtain the previously mentioned fine-grained quantum depth separation. Since this new problem is not efficiently verifiable, we do not obtain the associated fine-grained proof of quantum depth.

\section{\texorpdfstring{$d$-Recursive{[}$\mathcal{P}${]}}{d-Recursive[P]}\label{sec:dRecursive}}

Consider any problem $\cal P$ defined relative to a random oracle. We describe a map, which acts on $\cal P$ and creates a new problem $\recursive{d}[{\cal P}]$. If ${\cal P}$ can be solved quantumly but not classically (in the sense explained below), then $\recursive{d}[{\cal P}]$ can still be solved quantumly but cannot be solved with less than $d$ quantum depth, i.e. $\recursive{d}[{\cal P}]\in \BQP$ but $\recursive{d}[{\cal P}]\notin \classCQC {d}$. In fact, if $\cal P$ can be solved in depth $d'$ then one can tighten the upper bound on $\cal P$ from $\BQP$ to $\classCQC {d''}$ where $d''$ is a function of $d'$ (which we describe later).

\subsection{Definition of \texorpdfstring{$\mathcal{P}$}{P}}

\branchcolor{purple}{Any problem $\mathcal{P}$ which has the following two properties
can be lifted to a problem which is not in $\classCQC{d}$.
The first property, \emph{classical query soundness}, requires that no \emph{unbounded} algorithm can solve the problem by making only \emph{polynomially} many \emph{classical} queries to the oracles.
The second property, \emph{bounded oracle domain}, is also intuitively
quite simple. It requires that the problem only depends on a bounded
domain of the random oracle. We formalise this property by requiring
that the problem does not change if the random oracle is replaced
with an arbitrary function except that it behaves exactly like the
random oracle on the bounded oracle domain. We include this property
for technical reasons and it is possible that it is not really necessary.
However, for the problems we consider, both are easily satisfied.
Formally, we have the following.}
\begin{defn}[\emph{Classical query soundness} and \emph{bounded oracle domain}]
\label{def:semi-bounded}Denote by $H:\{0,1\}^{*}\to\{0,1\}$ a random
oracle. Define a problem $\mathcal{P}$ by a tuple $(\mathcal{X},R_{H})$
where $\mathcal{X}$ is a procedure which on input $1^{\lambda}$
generates a problem instance of size $\ply{\lambda}$ and $R_{H}=\{0,1\}^{*}\times\{0,1\}^{*}$
is a relation which depends on $H$. 
\begin{itemize}
\item We say $\mathcal{P}$ is \emph{classical query sound} if for
any unbounded algorithm $\mathcal{A}^{H}$ which makes at most $\ply{\lambda}$ classical queries to $H$, it holds that 
\[
\Pr_{H}\left[(x,y)\in R_{H}:\begin{array}{c}
x\leftarrow\mathcal{X}(1^{\lambda})\\
y\leftarrow\mathcal{A}^{H}(x)
\end{array}\right]\le\ngl{\lambda}
\]
for all sufficiently large $\lambda$. 
\item Let $R_{H}(x):=\{y:(x,y)\in R_{H}\}$. We say $\mathcal{P}$ has a \emph{bounded oracle domain} if there
is a set $\{0,1\}^{p(\lambda)}$ where $p(\lambda)$ is an integer
valued polynomial such that the following holds for each $\lambda$,
\[
R_{H}(x)=R_{H'}(x)
\]
for all $x\in\mathcal{X}(1^{\lambda})$ and all choices of functions
$H':\{0,1\}^{*}\to\{0,1\}$ such that $H'(z)=H(z)$ for all $z\in\{0,1\}^{p(\lambda)}$. 
\end{itemize}
\end{defn}

When we define YZ's $\codehashing$ problem, it would be evident that it satisfies both properties. %

\subsection{Definition of \texorpdfstring{$d$-Recursive{[}$\mathcal{P}${]}}{d-Recursive[P]}}

    Let $\mathcal{P}$ be any problem satisfying the properties in \Defref{semi-bounded}
and which is in $\classCQC{d'}$. %
We can now introduce $\recursive d[\mathcal{P}]$, a general construction
which, for any $0\le d\le\ply{\lambda}$, lifts $\mathcal{P}$, to
a problem which is not in $\classCQC{d}$ %
but is in $\classCQC{\ply{d,d'}}$. %
More precisely, the polynomial would be $(2d+1)\cdot d'$ because for
each oracle call in $\mathcal{P}$, $\recursive d[\mathcal{P}]$ would
need $2d+1$ oracle calls. To see why, we need to know how $\recursive d$ is defined. %

At a high level, instead of asking for $\cal P$ to be solved relative to the random oracle $H$, $\recursive{d} [\cal P]$ requires $\cal P$ to be solved relative to the random oracle $H$ composed with itself $d$ times. Clearly, $H$ cannot be composed with itself in general because the domain and co-domain may not match. Suppose $H:\Sigma \to \bit^n$. Then one natural choice to consider is  $\tilde H:= H_d \circ \dots \circ H_0$ where $H_d:\Sigma \to \bit^n $, $H_i:\Sigma \to \Sigma$ for $i\in \{1\dots d\}$. This has some issues, for instance the number of collisions in $\tilde H$ on an average would be larger than those in $H$. It turns out that for our analysis, enlarging the domain (as a function of $|\Sigma|$) is the appropriate choice, as explained below.

\begin{defn}[{$\recursive d[\mathcal{P}]$}]
\label{def:dRecursive} Let $\mathcal{P}=(\mathcal{X},R)$ be a problem
satisfying \Defref{semi-bounded}. We define $\recursive d[\mathcal{P}]$
as follows. On input $1^{\lambda}$, proceed as follows:
\begin{itemize}
\item Sample an instance of $\mathcal{P}$ as $x\leftarrow\mathcal{X}(1^{\lambda})$,
and
\item denote its bounded oracle domain by $\Sigma:=\{0,1\}^{p(\lambda)}$.
\item Define $\tilde{H}:=H_{d}\circ\dots\circ H_{1}\circ H_{0}$ where $H_{0}:\Sigma\to\Sigma^{d'}$,
for $\ell\in\{1,\dots d-1\}$, $H_{\ell}:\Sigma^{d'}\to\Sigma^{d'}$
and $H_{d}:\Sigma^{d'}\to\{0,1\}$ are independent random oracles
with $d'=2d+5$.
\end{itemize}
The ($\recursive d[\mathcal{P}]$) problem then is, given $x$, and
oracle access to $(H_{0},\dots H_{d})$, find a $y$ s.t. $(x,y)\in R_{\tilde{H}}$. 
\end{defn}

\subsubsection{Upper Bound}

\branchcolor{purple}{It may be the case that the algorithm which solves $\mathcal{P}$
makes only, say $1$, query to the random oracle while it still has
depth $d'$ which is some large constant. Clearly, in this case, one
can obtain a bound tighter than $(2d+1)\cdot d'$ on the depth of the
circuit which solves $\recursive d[\mathcal{P}]$. Indeed, we later consider a problem ($\collisionhashing$) which 
has this property and therefore we formally state this upper bound
as follows.}
\begin{lem}
\label{lem:dRecursive_upper}Suppose $\mathcal{P}$ is solved by a
family of circuits in $\QNC_{d'}$ with probability
$1-\epsilon(\lambda)$ and by making at most $t(\lambda)$ parallel
queries. Then there is a family of circuits
in $\QNC_{d''}$ which solves $\recursive d[\mathcal{P}]$
with probability $1-\epsilon(\lambda)$ where $d''\le\min[d'+(2d+1)\cdot t,(2d+1)\cdot d']$. The analagous statement holds for $\QC{}$ and $\CQ{}$ as well.%
\end{lem}

\branchcolor{blue}{\begin{proof}[Proof sketch.]
 Fix a $\lambda$, let $\mathcal{C}_{1}\in\QNC_{d'}$ be %
the circuit that solves $\mathcal{P}$ and let $\mathcal{C}_{2}\in\QNC_{d''}$%
be a circuit which we construct and assert that it solves $\recursive d[\mathcal{P}]$,
with the same $1-\epsilon(\lambda)$ probability. 

To obtain $d''\le d'+(2d+1)\cdot t$, suppose the circuits are identical,
except that for each of the $t$ set of parallel oracle calls that
$\mathcal{C}_{1}$ makes, $\mathcal{C}_{2}$ makes $(2d+1)\cdot t$
set of parallel oracle calls. This allows $\mathcal{C}_{2}$ to compute
$\tilde{H}$ and therefore proceed exactly like $\mathcal{C}_{1}$.
A simple upper bound on the quantum depth $d''$ of $\mathcal{C}_{2}$
then is $d'+(2d+1)\cdot t$. 

To obtain $d''\le(2d+1)\cdot d'$, suppose that the oracle is parallel
invoked (worst case) at each layer. Then, $\mathcal{C}_{2}$ is identical
to $\mathcal{C}_{1}$, except that for each of the $d'$ layers of
$\mathcal{C}_{1}$, $\mathcal{C}_{2}$ gets $(2d+1)$ layers and can
therefore evaluate $\tilde{H}$ exactly like $\mathcal{C}_{1}$. An
upper bound on the depth of $\mathcal{C}_{2}$ is then $d''\le(2d+1)\cdot d'$. 
\end{proof}
}

\subsubsection{Lower bound}

The main property of $\recursive{d}[\cdot]$ is the following which we establish in \Secref{lowerbounds}.

\begin{lem}[{$\recursive d[\mathcal{P}] \notin \classCQC{d}$}]
  \label{lem:dRecursive_CQCd}Every $\CQC d$ circuit succeeds at solving
  $\recursive d[\mathcal{P}]$ (see \Defref{dRecursive}) with probability
  at most $\ngl{\lambda}$ on input $1^{\lambda}$ for $d\le\poly$. 
\end{lem}

 \begin{table}
    \begin{centering}
      \resizebox{\textwidth}{!}{\begin{tabular*}{1.2\textwidth}{@{\extracolsep{\fill}}ccccc>{\centering}p{4cm}>{\centering}p{4cm}}
        \textbf{\footnotesize{}Problem} & \textbf{\footnotesize{}$\in$} &  & \textbf{\footnotesize{}$\notin$}  & \textbf{\footnotesize{}Verification} & \textbf{\footnotesize{}Interpretation} & \textbf{\footnotesize{}Remarks}\tabularnewline
        \midrule
        {\footnotesize{}$\mathsf{CodeHashing}$} & {\footnotesize{}$\BQP$} & {\footnotesize{}$\nsubseteq$} & {\footnotesize{}$\BPP$}  & {\footnotesize{}Public} & {\scriptsize{}Proof of Quantumness} & {\scriptsize{}\parencite{yamakawa_verifiable_2022}; \Defref{CodeHashing}}\tabularnewline
        \midrule
        {\footnotesize{}$\CH d$} & {\footnotesize{}$\BQP$} & {\footnotesize{}$\nsubseteq$} & {\footnotesize{}$\classCQC d$}  & {\footnotesize{}Public} & {\scriptsize{}$(d,\poly)$-Proof of Quantum Depth}\\
        {\scriptsize{}Refutes Jozsa's conjecture} & {\scriptsize{}See \Defref{dCodeHashingProblem}; Equivalent to $\recursive d[\mathsf{CodeHashing}]$}\tabularnewline
      \end{tabular*}}
      \par\end{centering}
    \caption{In \Partref{lowerbounds} we started with $\CodeHashing$ due to \cite{yamakawa_verifiable_2022} and created a new problem we termed $\CH{d}$ which showed $\classCQC{d}\subsetneq \BQP$, for any fixed $d\le \poly$. This refutes Jozsa's conjecture relative to the random oracle model. This problem also immediately serves as a Proof of $d$ Quantum Depth.  \label{tab:lowerbounds}}
  \end{table}

\section{Preliminaries}
Instead of working with $\recursive{d}[\cdot ]$ abstractly, we apply this map to the $\codehashing$ problem introduced by \cite{yamakawa_verifiable_2022}. To use and describe the seminal work of YZ, we need the following notions.

\subsection*{Error Correcting Codes}

\subsubsection*{Codes}

A code of length $n\in\mathbb{N}$ over an alphabet $\Sigma$ is a
subset $C\subseteq\Sigma^{n}$.

\begin{description}
    \item[Linear codes \cite{yamakawa_verifiable_2022}.]
    Let $\mathbb{F}_q$ be a finite field of order $q$ for some prime power $q$ and $\Sigma=\mathbb{F}_q$. A linear code $C$ of length $n\in \mathbb{N}$ over the alphabet $\Sigma$ is defined as a subset $C\subseteq \mathbb{F}_q^n$, which is also a linear subspace of $\mathbb{F}_q^n$. Further, we define the rank of a linear code $C$ as the dimension of the linear subspace $C\subseteq \mathbb{F}_q^n$. 

    \item[Folded linear codes \cite{Krachkovsky2003, Guruswami2008, yamakawa_verifiable_2022}.]
    Let $\mathbb{F}_q$ be a finite field of order $q$ for some prime power $q$ and $m$ be a positive integer. A code $C$ is said to be an $m$-folded linear code of length $n$ if its alphabet is $\Sigma=\mathbb{F}_q^m$ and $C\subseteq \Sigma^n$ is a linear subspace of $C\subseteq \Sigma^n$, where $C$ is embedded into $\mathbb{F}_{q}^{mn}$ in the canonical way.
    \end{description}

    It is clear that $1$-folded linear codes are just linear codes. In fact, for a positive integer $m$ that divides $n$ and a linear code $C\subseteq \mathbb{F}_q^n$, we can define its $m$-folded version $C^{(m)}$ as
    \begin{align*}
        C^{(m)}:=\Big\{\Big(\overbrace{(x_1,\cdots,x_m)}^{m},\cdots, \overbrace{(x_{n-m+1},\cdots,x_n)}^{m}\Big):(x_1,\cdots,x_n)\in C\Big\}.
    \end{align*}
    Conversely, any folded linear code can be written as $C^{(m)}$ for some linear code $C$ and a positive integer $m$.

    \begin{description}
    \item[Dual codes \cite{yamakawa_verifiable_2022}.]
    A dual code $C^{\perp}$ of a linear code $C$ of length $n$ and rank $k$ over the alphabet $\mathbb{F}_q$ is defined as the orthogonal complement of $C$ as a linear subspace over $\mathbb{F}_q^n$. That is, 
    \begin{align*}
        C^{\perp}:=\left\{\boldsymbol{x}\in \mathbb{F}_q^n: \boldsymbol{x}\cdot \boldsymbol{x'}=0~\forall~\boldsymbol{x'}\in C\right\}.
    \end{align*}
    $C^{\perp}$ is itself a linear code of length $n$ and rank $n-k$ over $\mathbb{F}_q$. Similarly, for an $m$-folded linear code $C\in \mathbb{F}_q^{mn}$ of length $n$ over the alphabet $\mathbb{F}_q^m$, its dual $C^{\perp}$ is defined as
    \begin{align*}
        C^{\perp}:=\left\{\boldsymbol{x}\in \mathbb{F}_q^{mn}: \boldsymbol{x}\cdot \boldsymbol{x'}=0~\forall~\boldsymbol{x'}\in C\right\}.
    \end{align*}
    Note that for any linear code $C$ of length $n$ and an integer $m$ that divides $n$, we have $\left(C^{\perp}\right)^{(m)}=\left(C^{(m)}\right)^{\perp}$.

    \item[List recovery \cite{yamakawa_verifiable_2022}.]Let $C\subseteq \Sigma^n$ be a code and ${\boldsymbol{x}}:=(x_1,\cdots, x_n)\in C$ be a codeword. For subsets $S_i\subseteq\Sigma$ such that $|S_i|\leq l$ for $i\in[n]$, define the index set $I_{\boldsymbol{x},\{S_i\},l}:=\left\{i\in[n]:x_i\in S_i\right\}$. Then, we say that $C\subseteq \Sigma^n$ is $(\zeta,l,L)$-list recoverable if for any subsets $S_i\subseteq\Sigma$ such that $|S_i|\leq l$ for $i\in[n]$, we have \footnote{List recovery usually requires efficient computation of all codewords $(x_1,\cdots, x_n)\in C$ that satisfy $\left|I_{\boldsymbol{x},\{S_i\},l}\right|\geq (1-\zeta)n$ using $\{S_i\}$, however, it is not relevant for our purposes, so we don't demand it here.}
    \begin{align*}
        \left|\left\{\boldsymbol{x}\in C: \left|I_{\boldsymbol{x},\{S_i\},l}\right|\geq (1-\zeta)n\right\}\right|\leq L.
    \end{align*}
\end{description}

\subsubsection*{Suitable Codes}

\branchcolor{purple}{YZ use folded codes which satisfy certain properties. They call these codes \emph{suitable codes}. They show that folded Reed-Solomon codes with appropriate parameters are suitable. We would not need these details for our result---the following suffices.}
\begin{lem}[Suitable Codes~\cite{yamakawa_verifiable_2022}]
    \label{lem:suitableCodes}For any constants $0<c<c'<1$, there exists
    an explicit family $\{C_{\lambda}\}_{\lambda\in\mathbb{N}}$ of folded
    linear codes over the alphabet $\Sigma=\mathbb{F}_{q}^{m}$ of length
    $n$ where $\left|\Sigma\right|=2^{{\lambda}^{\Theta(1)}}$, $n=\Theta(\lambda)$
    and $\left|C_{\lambda}\right|\ge2^{n+\lambda}$ that satisfies the
    following.\footnote{YZ point out that item 3 is not needed for proof of quantumness. It is used for showing one-way-functions. We inherit these in our construction of proof of depth and one-way functions resp.}
    \begin{enumerate}
        \item $C_{\lambda}$ is $(\zeta,\ell,L)$-list recoverable where $\zeta=\Omega(1),l=2^{\lambda^{c}}$
              and $L=2^{\tilde{O}(\lambda^{c'})}$.
        \item There is an efficient deterministic decoding algorithm $\mathsf{Decode}_{C^{\perp}}$
              for $C^{\perp}$ that satisfies the following. Let $\calD$
              be a distribution over $\Sigma$ that outputs $\mathbf{0}$ with probability
              $1/2$ and otherwise outputs an element of $\Sigma\backslash\{\mathbf{0}\}$
              uniformly at random. Then, it holds that
              \[
                  \Pr_{\mathbf{e}\leftarrow\calD^{n}}[\forall\mathbf{x}\in C^{\perp},\mathsf{Decode}_{C^{\perp}}(\mathbf{x}+\mathbf{e})=\mathbf{x}]=1-2^{-\Omega(\lambda)}.
              \]
        \item For all $j \in [n-1]$, $\Pr_{\mathbf{x} \leftarrow C_{\lambda}} \left[ \hw(\mathbf{x}) = n - j \right] \leq \left( \frac{n}{|\Sigma|} \right)^{j}$.
    \end{enumerate}

\end{lem}

\section{The \texorpdfstring{$\CH d$}{d-CodeHashing} Problem \label{sec:dCodeHashingDefined}}

This section introduces the problem we use for proving our main result. 

\subsection{Background | \texorpdfstring{$\mathsf{CodeHashing}$ Problem \cite{yamakawa_verifiable_2022}}{CodeHashing Problem [YZ22]}}

\branchcolor{purple}{To describe our problem, we first recall that YZ's proof of quantumness is based on the following problem.}

\begin{defn}[$\mathsf{CodeHashing}$ Problem; Paraphrased from \cite{yamakawa_verifiable_2022}]
 \label{def:CodeHashing}Let 
\begin{itemize}
\item $\{C_{\lambda}\}_{\lambda}$ be a family of codes over an alphabet
$\Sigma=\mathbb{F}_{q}^{m}$ that satisfies the requirements of \Lemref{suitableCodes}
with arbitrary $1<c<c'<1$,
\item $H:\Sigma\to\{0,1\}^{n}$ be a random oracle. %
\end{itemize}
Given the code family, and access to $H$, on input $1^{\lambda}$,
the $\mathsf{CodeHashing}$ problem is to find an $\mathbf{x}=(\mathbf{x}_{1},\mathbf{x}_{2},\dots\mathbf{x}_{n})\in C_{\lambda}$
such that for all $i\in \{1\dots n \}$, the $i$th bit of $H(\mathbf{x}_{i})$ equals $1$. 
\end{defn}

\branchcolor{purple}{Note that for suitable codes, $\lambda=\Theta(n)$. Also note that the oracle $H$ as described in the problem can be implemented using the standard random oracle from $\bit^*$ to $\bit$, as discussed in \Secref{ROM_model}. YZ showed that this problem is contained in $\NP$ and $\BQP$ but not in $\BPP$.}
\begin{thm}[Paraphrased from \cite{yamakawa_verifiable_2022}]
 \label{thm:YZ22}The following hold in the random
oracle model (for oracle-independent circuits). 

$\mathsf{CodeHashing}\in\BQP$: A QPT machine can solve the code hashing
problem with overwhelming probability, i.e. $1-\ngl{\lambda}$. %

$\mathsf{CodeHashing}\notin\BPP$: Every classical circuit which makes
at most $2^{\lambda^{c}}$ queries to the oracle solves the code hashing
problem with probability at most $2^{-\Omega(\lambda)}$. 
\end{thm}

\branchcolor{purple}{YZ not only show that $\mathsf{CodeHashing}\notin \BPP$, they show that even an unbounded machine would not succeed at solving $\mathsf{CodeHashing}$ with noticeable probability if it makes at most polynomially many (classical) queries to the random oracle. It is this \emph{classical query soundness} property that we use later in our proof. Also observe that a $\BPP$ machine can easily check if a solution to $\mathsf{CodeHashing}$ is valid.}

\subsection{The Problem | \texorpdfstring{$\protect\CH d$ Problem}{d-CodeHashing Problem}}

\branchcolor{purple}{We call our problem $\CH{d}$ which is basically\footnote{The only difference is in the range of $H$ but this is without loss of generality due to domain splitting.} To be conceret, we explicitly state it. $\recursive{d}[CH]$. %

}

\begin{defn}[$\CH d$ Problem]
    Let \label{def:dCodeHashingProblem}
   \begin{itemize}
   \item $\{C_{\lambda}\}_{\lambda}$ be a family of codes over an alphabet
   $\Sigma=\mathbb{F}_{q}^{m}$ that satisfies the requirements of \Lemref{suitableCodes}
   with arbitrary $1<c<c'<1$,
   \item $\tilde{H}:=H_{d}\circ\dots\circ H_{1}\circ H_{0}$ where $H_{0}:\Sigma\to\Sigma^{d'}$,
   for $\ell\in\{1,\dots,d-1\}$, $H_{\ell}:\Sigma^{d'}\to\Sigma^{d'}$
   and $H_{d}:\Sigma^{d'}\to\{0,1\}^{n}$ are independent random oracles
    with $d':=2d+5$,
   \item ${\rm Bit}_i[\tilde{H}]$ denote the $i$th bit of $\tilde{H}$,
   \end{itemize}
   Given the code family $\{C_{\lambda}\}_{\lambda}$, access to random
   oracles $H_{0},\dots H_{d}$, on input $1^{\lambda}$, the $\CH d$
   problem is to find an $\mathbf{x}=(\mathbf{x}_{1},\mathbf{x}_{2},\dots\mathbf{x}_{n})\in C_{\lambda}$
   such that for all $i$, the $i$th bit of $\tilde{H}(\mathbf{x}_{i})$
   is 1, i.e. ${\rm Bit}_i[\tilde{H}(\mathbf{x}_{i})]=1$. 
   \end{defn}

\section{Lower Bounds\label{sec:lowerbounds}}
\branchcolor{purple}{In this section, we establish the following key property of the $\CH{d}$ problem. The proof of \Lemref{dRecursive_CQCd} is also immediate from the proof of the following. %
}
\begin{lem}[$\CH{d}\notin \classCQC{d}$]
\label{lem:Main} Every $\CQC{d}$ circuit\footnote{We assume the circuits are ``oracle independent'' as described in \Subsecref{uniform_nonuniform}.} (which subsumes $\QCd$ and $\CQd$ circuits) with oracle access to $H_0,\dots H_d$, succeeds at solving
$\CH d$ with probability at most $\ngl{\lambda}$ on input $1^{\lambda}$. 
\end{lem}

\branchcolor{purple}{
    We prove \Lemref{Main} in three main steps. First, we establish $\QNCd$ hardness. We use this as a warm-up for introducing notations and concepts (in particular ``base sets'') which we build on for establishing $\QCd$ hardness. The basic tools we need are discussed next, in \Subsecref{known-tech-results-1}. $\CQd$ hardness requires more work (and more technical tools) and we defer that discussion to \Subsecref{CQdHardnessWarmup}. We then combine the ideas used in these three main steps to establish $\CQC{d}$ hardness. 
    Before delving into the proof of \Lemref{Main}, we look at one of its main consequences.
}

\subsection{Consequence: Jozsa's conjecture/Aaronson's challenge}

\branchcolor{purple}{Jozsa had conjectured that $\classCQ{} = \BQP$. \Lemref{Main} and \Thmref{YZ22}, however, immediately yield the following
theorem. Note that the classes stated below are the corresponding search variants (see \Secref{Models-of-Computation}).}
\begin{thm}[$\classCQC{} \subsetneq \BQP$]\label{thm:JozsaRefute} %
The following hold (unconditionally) in the random oracle model,
for $d=\lambda$ where $\lambda$ is the input size.

$\CH d\in\BQP$: A QPT machine can solve the code hashing problem

with overwhelming probability, i.e. $1-\ngl{\lambda}$, by making $\calO(\lambda)$ queries to the random oracle. %

$\CH d\notin \classCQC{}$: Every %
$\CQC{\log(\lambda)}$ circuit succeeds at solving $\CH d$ with probability
at most $\ngl{\lambda}$ on input $1^{\lambda}$. 
\end{thm}
\branchcolor{purple}{
    We emphasise that $\classQC{}\cup \classCQ{} \subseteq \classCQC{}$ and so \Thmref{JozsaRefute} shows that even a more liberal interpretation of Jozsa's conjecture, in the random oracle model, is false. One might wonder if $\classCQC{}$ is strictly larger than $\classQC{}\cup \classCQ{}$. Indeed, this is the case and we show it in \Partref{Distinguishing-between-types}.
}

\subsection{Known Results\label{subsec:known-tech-results-1}}

\branchcolor{purple}{
We first state a simplified version of the so called ``one-way to hiding'', or briefly, the O2H lemma (see \Subsecref{O2Hlemma}) due originally to \textcite{ambainis_quantum_2018}. Our presentation, however, is inspired by \cite{chia_need_2020-1} and \cite{arora_oracle_2022}. We then state a version tailored to our setup (see \Subsecref{O2HlemmaHere}) and end with some elementary results (see \Subsecref{Combinatorics}).
}

\subsubsection{The O2H lemma\label{subsec:O2Hlemma}}

\branchcolor{purple}{
    Informally, the O2H lemma says the
    following: suppose there are two oracles $\calO$ and $\calQ$
    which behave identically on all inputs except some subset $S$ of
    their input domain. Let $\calA^{\calO}$ and $\calA^{\calQ}$
    be identical quantum algorithms, except for their oracle access, which
    is to $\calO$ and $\calQ$ respectively. Then, the probability
    that the result of $\calA^{\calO}$ and $\calA^{\calQ}$
    will be distinct, is bounded by the probability of finding the set
    $S$. We suppress the details of the general finding procedure and
    only focus on the case of interest for us here.}

We begin by setting up some notation for this section (adapted from \cite{arora_oracle_2022}).
    We use the symbol $\calL $ for the oracle.\footnote{Instead of $\calQ$ as above to avoid confusion.}
    The workspace register is denoted by $W$ which is left untouched
    by the oracle. The query register is denoted by $Q$ and the response
    register by $R$. Suppose we make $m=\poly$ parallel queries to $\calL $.
    We use boldface to represent the associated quantities. In particular,
    the parallel queries $(q_{1},q_{2}\dots q_{m})$ are denoted by the
    tuple $\boldsymbol{q}$, the query registers $(Q_{1},Q_{2}\dots Q_{m})$
    which would hold these queries are denoted by $\boldsymbol{Q}$ and
    the corresponding response registers $(R_{1},R_{2}\dots R_{m})$ are
    denoted by $\boldsymbol{R}$.

\begin{defn}[$U^{\calL \backslash S}$]
    \label{def:ULS}Suppose $U$ acts on $\boldsymbol{Q}RW$, $\calL $
    is an oracle that acts on $\boldsymbol{QR}$ and $S$ is a subset
    of the query domain of $\calL $. We define
    \[
        U^{\calL \backslash S}\left|\psi\right\rangle _{\boldsymbol{QR}W}\left|0\right\rangle _{B}:=\calL U_{S}U\left|\psi\right\rangle _{\boldsymbol{QR}W}\left|0\right\rangle _{B}
    \]
    where $B$ is a qubit register, and $U_{S}$ flips qubit $B$ if any
    query is made inside the set $S$, i.e.
    \[
        U_{S}\left|\boldsymbol{q}\right\rangle _{\boldsymbol{Q}}\left|b\right\rangle _{B}:=\begin{cases}
            U_{S}\left|\boldsymbol{q}\right\rangle _{\boldsymbol{Q}}\left|b\right\rangle _{B}        & \text{if }\boldsymbol{q}\cap S=\emptyset \\
            U_{S}\left|\boldsymbol{q}\right\rangle _{\boldsymbol{Q}}\left|b\oplus1\right\rangle _{B} & \text{otherwise.}
        \end{cases}
    \]
    Here\footnote{i.e. the condition $\boldsymbol{q}\cap S=\emptyset$ reads there is
        no $i$ for which $q_{i}\notin S$.} we treat $\boldsymbol{q}$ as a set when we write $\boldsymbol{q}\cap S$.
\end{defn}

For notational simplicity, in the following, we drop the boldface
for the query and response registers as they do not play an active
role in the discussion.
\begin{defn}[{$\Pr[{\rm find}:U^{\calL \backslash S},\rho]$}; adapted from \cite{arora_oracle_2022}]
    \label{def:prFind}Let $U^{\calL \backslash S}$ be as above
    and suppose $\rho\in{\rm D}(QRWB)$. We define
    \[
        \Pr[{\rm find}:U^{\calL \backslash S},\rho]:={\rm tr}[\mathbb{I}_{QRW}\otimes\left|1\right\rangle \left\langle 1\right|_{B}U^{\calL \backslash S}\circ\rho].
    \]
    This will depend on $\calL $ and $S$. When $\calL $ and
    $S$ are random variables, we additionally take expectation over them.\footnote{i.e. $\Pr[{\rm find}:U^{\calL \backslash S},\rho]:=\mathbb{E}_{\calL ,S}\tr[\mathbb{I}_{QRW}\otimes\left|1\right\rangle \left\langle 1\right|_{B}U^{\calL \backslash S}\circ\rho].$}
\end{defn}

\ifthenelse{\boolean{keepOldProofs}}{
    \begin{rem}[adapted from \cite{arora_oracle_2022}]
        \label{rem:psiphi0phi1}Let $U^{\calL \backslash S}$ be as in
        \Defref{ULS} and let $\left|\psi\right\rangle \in QRW$. Note that
        we can always write
        \[
            \calL U\left|\psi\right\rangle _{QRW}=\left|\phi_{0}\right\rangle _{QRW}+\left|\phi_{1}\right\rangle _{QRW}
        \]
        where $\left|\phi_{0}\right\rangle $ and $\left|\phi_{1}\right\rangle $
        contains queries outside $S$ and inside $S$ respectively, i.e. $\left\langle \phi_{0}|\phi_{1}\right\rangle =0$.
        Further, we can write
        \[
            U^{\calL \backslash S}\left|\psi\right\rangle _{QRW}\left|0\right\rangle _{B}=\left|\phi_{0}\right\rangle _{QRW}\left|0\right\rangle _{B}+\left|\phi_{1}\right\rangle _{QRW}\left|1\right\rangle _{B}.
        \]
    \end{rem}
}{}

\branchcolor{purple}{
    The following is a special case of the O2H lemma introduced in \parencite{ambainis_quantum_2018}.
}
\begin{lem}[O2H lemma; as stated in \cite{arora_oracle_2022}]
    \label{lem:O2H}Let
    \begin{itemize}
        \item $\calL $ be an oracle which acts on $QR$ and $S$ be a subset
              of the query domain of $\calL $,
        \item $\calG$ be a shadow of $\calL $ with respect to $S$,
              i.e. $\calG$ and $\calL $ behave identically for all
              queries outside $S$,
        \item further, suppose that within $S$, $\calG$ responds with $\perp$
              while (again within $S$), $\calL $ does not respond with $\perp$.
              Finally, let $\Pi_{t}$ be a measurement in the computational basis,
              corresponding to the string $t$.
    \end{itemize}
    Then
    \begin{align*}
        \left|{\rm tr}[\Pi_{t}\calL \circ U\circ\rho]-{\rm tr}[\Pi_{t}\calG\circ U\circ\rho]\right| & \le B(\calL \circ U\circ\rho,\calG\circ U\circ\rho)       \\
                                                                                                    & \le\sqrt{2\Pr[\text{find }:U^{\calL \backslash S},\rho]}.
    \end{align*}
    If $\calL $ and $S$ are random variables with a joint distribution,
    we take the expectation over them in the RHS (see \Defref{prFind}).
\end{lem}

The right hand side in \Lemref{O2H} may be bounded using \Lemref{boundPfind} below. \Lemref{boundPfind} applies when the locations queried are independent of the set being hidden. 

\begin{lem}[{\parencite{chia_need_2020-1,ambainis_quantum_2018} Bounding $\Pr[{\rm find}:U^{\calL \backslash S},\rho]$}]
    \label{lem:boundPfind}Suppose $S$ is a random variable and $\Pr[x\in S]\le p$
    for some $p$. Further, assume that $U$ and $\rho$ are uncorrelated\footnote{i.e. the distribution from which $S$ is sampled is uncorrelated to
        the distribution from which $U$ and $\rho$ are sampled,} to $S$. Then, (see \Defref{ULS})
    \[
        \Pr[{\rm find}:U^{\calL \backslash S},\rho]\le\bar{q}\cdot p
    \]
    where $\bar{q}$ is the total number of queries $U$ makes to $\calL $.
\end{lem}

For completeness, we include the proofs in \Secref{O2Happendix} of the Appendix.

\subsubsection{O2H adapted to our analysis\label{subsec:O2HlemmaHere}}
\branchcolor{purple}{
    Recall that $\CH{d}$ (see \Defref{dCodeHashingProblem}) is defined using $d+1$ oracles, $\{H_i\}_{0,1,\dots d}$. Therefore, instead of considering a set $S$ where the oracles ($\calL$ and $\calG$) behave differently, we consider a sequence of sets. Let $S^{\rm out}$ denote a sequence of $d$ sets and similarly let $S^{\rm in}$ denote a sequence of $d$ sets contained in $S^{\rm out}$ (element-wise). Why we take $d$ and not $d+1$ should become evident later---briefly, it is because the domain of $H_0$ is known by construction but the domain of $H_1$ which is of interest, i.e. $H_0(\Sigma)$, is what we are trying to hide (and similarly for $H_2,\dots H_d$). Observe that in \Lemref{O2H}, the state $\rho$ was uncorrelated to the set $S$. However, in our application, the quantum state can potentially contain information about $\calL$ restricted to values outside $S^{\rm out}$. However, within $S^{\rm out}$, the values of $S^{\rm in}$ stay uncorrelated and we can apply \Lemref{O2H}. The following notation allows us to state this formally.

    \begin{figure}
        \begin{centering}
          \includegraphics[width=10cm]{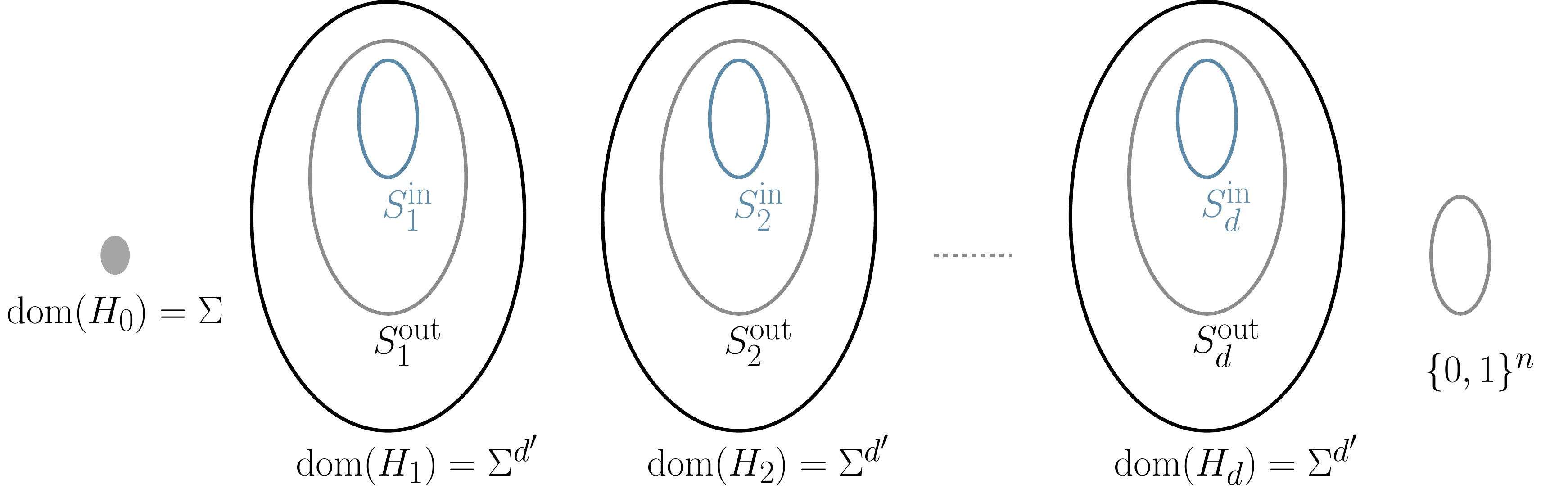}
          \par\end{centering}
        \caption{\label{fig:SinSout}}
      \end{figure}
    
}

\begin{notation}
    \label{nota:LhatAndsuch}Consider the following (see \Figref{SinSout}).
    \begin{itemize}
        \item Let $\calL ':=(H_{0}',H_{1}',\dots H_{d}')$ where the domain
              and range of $H'_{i}$ is the same as that of $H_{i}$ (as defined
              in \Defref{dCodeHashingProblem}).
              \begin{itemize}
                  \item These functions themselves may be sampled from an arbitrary distribution
                        (unlike $H_{i}$).
              \end{itemize}
        \item Let $S^{{\rm out}}:=(S_{1}^{{\rm out}},\dots S_{d}^{{\rm out}})$
              and $S^{{\rm in}}:=(S_{1}^{{\rm in}},\dots S_{d}^{{\rm in}})$ be
              a sequence of (random) subsets such that $S_{i}^{{\rm in}}\subseteq S_{i}^{{\rm out}}\subseteq{\rm dom}(H_{i}')$.
              \begin{itemize}
                  \item Note that $S^{{\rm out}}$ and $S^{{\rm in}}$ are random variables which may
                        be arbitrarily correlated with $\calL '$.
              \end{itemize}
        \item Let $\check{\calL }'$ refer to $\calL '$ outside of $S^{{\rm out}}$,
              i.e. $(\check{H}_{0}',\dots\check{H}_{d}')$ where $\check{H}_{i}':{\rm dom}(H_{i}')\backslash S^{{\rm out}}\to H_{i}'({\rm dom}(H_{i}')\backslash S^{{\rm out}})$
              and $\check{H}_{i}'(x):=H_{i}'(x)$ for all $x\in{\rm dom}(\check{H}_{i}')$.
        \item Let $\hat{\calL }'$ refer to $\calL '$ inside $S^{{\rm out}}$,
              i.e. $(\hat{H}_{0}',\dots\hat{H}'_{d})$ where $\hat{H}_{i}':S_{i}^{{\rm out}}\to H'_{i}(S_{i}^{{\rm out}})$.
    \end{itemize}
\end{notation}

\branchcolor{purple}{
    We used $\calL'$ instead of $\calL$ because, in our proofs, $\calL$ will be conditioned on various random variables and it is this conditioned $\calL$ we work with.
}

\begin{cor}
    \label{cor:Conditionals}Let $\calL ',S^{{\rm out}},S^{{\rm in}},\hat{\calL }',\check{\calL }'$
    be as in \Notaref{LhatAndsuch} above. Suppose a quantum state $\rho$
    and a unitary $U$ are drawn from a distribution which may be correlated
    with $\calL '$. Suppose, $\sigma:=\rho|\check{\calL }'$
    and $V:=U|\check{\calL }'$, are uncorrelated to $R:=S^{{\rm in}}|\check{\calL }'$.
    Let $\mathcal{N}:=\calL '|\check{\calL }'$. Given that
    $\Pr[x\in S_{i}^{{\rm in}}|\check{\calL }']\le p$ for all $x\in{\rm dom}(H'_{i})$
    and $i\in\{1\dots d\}$, it holds that
    \[
        \Pr[{\rm find}:V^{\mathcal{N}\backslash R},\sigma]\le d\cdot\bar{q}\cdot p
    \]
    where $\bar{q}$ is the total number of queries $V$ makes to the
    oracles $\calL '$.
\end{cor}

\branchcolor{blue}{\begin{proof}[Proof sketch.]
        We assume that the $\rho$ contains information about $\check{\calL }'$
        and therefore contains information about $S^{{\rm out}}$. At best,
        $U$ can query $\calL $ at $x$ such that $x\in S_{i}^{{\rm out}}$
        for some $i$. However, given $\hat{\calL }'$ (and therefore
        $S^{{\rm out}}$), $\Pr[x\in S_{i}^{{\rm in}}|\hat{\calL }]$
        is bounded by $p$ so, by argument used in the proof for \Lemref{boundPfind},
        together with a union bound, one obtains the asserted bound.
    \end{proof}
}

\branchcolor{purple}{When we apply the O2H lemma via the corrollary above, it would be helpful to consider shadows for a sequence of oracles---the analogue of $\calG$ in \Lemref{O2H}. Defining it formally helps the presentation. }
\begin{defn}[Shadow oracle wrt $\bar{S}'$]
    \label{def:shadowOracle} Let $\calL' :=(H_{0}',\dots H_{d}')$
    and $\Sigma$ be as in \Notaref{LhatAndsuch}. %
    Let $\bar{S}'=(S_{1}',S_{2}'\dots S_{d}')$
    be a tuple of $d$ sets where each set $S_{i}'\subseteq\Sigma^{d'}$
    for all $i\in\{1,\dots d\}$. The shadow oracle $\calM'$ of
    $\calL' $ wrt $\bar{S}'$ is defined as $\calM':=(M_{0}',\dots M_{d}')$
    where
    \[
        M_{i}'(\mathbf{l}):=\begin{cases}
            H_{i}'(\mathbf{l}) & \mathbf{l}\in\Sigma\backslash S_{i}' \\
            \perp              & \mathbf{l}\in S_{i}'.
        \end{cases}
    \]
\end{defn}

\subsubsection{Elementary results\label{subsec:Combinatorics}}
\branchcolor{purple}{The following elementary observations will be useful in computing probabilities which arise in our analysis.} We use the following convention: $\perm{a}{b}:=a!/(a-b)!$ and $\comb{a}{b}:=a!/(b!\cdot (a-b)!)$ for $a\ge b$.
\begin{fact}
    \label{fact:perms_and_combs}One has
    \[
      \frac{\perm ab}{\perm{a+1}{b+1}}=\frac{1}{a+1}\quad\text{and}\quad\frac{\comb ab}{\comb{a+1}{b+1}}=\frac{b+1}{a+1}.
    \]
\end{fact}

\begin{rem}
    \label{rem:Pr_x_in_X_or_t}Let $M\ge N$ be an integer and fix some
    element $x\in\{1,2\dots M\}$. Suppose $t$ is a \emph{tuple} of size
    $N$, sampled uniformly from the collection of all size $N$ \emph{tuples}
    containing distinct elements from $\{1,2\dots M\}$. Then
    \[
        \Pr(x\in t)=\frac{\perm{M-1}{N-1}\cdot N}{\perm MN}=\frac{N}{M}.
    \]
    Similarly, suppose $X$ is a \emph{set} of size $N$, sampled uniformly
    from the collection of all size $N$ \emph{subsets} of $\{1\dots M\}$.
    Then, again,
    \[
        \Pr(x\in X)=\frac{\comb{M-1}{N-1}}{\comb MN}=\frac{N}{M}.
    \]
\end{rem}

\branchcolor{purple}{The following elementary calculation was alluded to in the discussion following \Defref{dCodeHashingProblem}. It allows us to reduce our problem to permutations, without loss of generality.}
\begin{claim}
	\label{claim:permutationProb}Let $f:A\to A$ be a random function,
	i.e. for all \textbf{$a\in A$, }$f(a)$ is mapped to $a'\in A$ with
	probability $1/|A|$. Let $B\subsetneq A$ be an arbitrary set. Then
	the probability that $|f(B)|=|B|$ is at least $1-|B|^2/|A|$. Equivalently,
	the probability that $|f(B)|<|B|$ is at most $|B|^2/|A|$. 
\end{claim}

\branchcolor{blue}{\begin{proof}
		It suffices to show that $f$ is injective on $B$ with the same probability.
		We have 
		\begin{align*}
			\Pr[|f(B)|=|B|] & =\Pr(f\text{ has no collisions in }B)\\
			& =1-\Pr(f\text{ has at least one collision in }B)\\
			& \ge1-\epsilon
		\end{align*}
		if $\Pr(f\text{ has at least one collision in }B)\le\epsilon$. Since
		$f$ is random, the probability that a given $b$ collides with some
		$b'$ is simply the probability that $f(b')$ is assigned the value
		$f(b)$ by $f$ which is at most $|B|/|A|$, i.e. $\Pr(\exists\ \ b'\neq b\ \ {\rm s.t.}\ \ f(b)=f(b')) \le |B|/|A|$.
		Therefore,
		
		\begin{align*}
			\Pr(f\text{ has at least one collision in }B) & =\Pr(\lor_{b\in B}\quad b\text{ collides under }f)\\
			& \le\sum_{b\in B}\Pr(b\text{ collides under }f)\\
			& =\sum_{b\in B}\Pr(\exists\ \ b'\neq b\ \ {\rm s.t.}\ \ f(b)=f(b'))\\
			& \le|B| \cdot |B|/|A|.
		\end{align*}
	\end{proof}
}

\subsection{Warm-up | \texorpdfstring{$\protect\QNC_{d}$ exclusion}{QNC\_d exclusion}}

\branchcolor{purple}{
    We have now stated all the preliminaries we need to show our first lower bound. We do this in three stages. First, we define two algorithms which help us reduce to the case of permutations and allow us to perform ``domain hiding'' for each set of parallel calls. The latter is essentially the same as the ``russian nesting doll'' technique, as applied by \cite{chia_need_2020-1}, adapted to the random oracle setup. In the second stage, we prove that the first algorithm does indeed produce permutations with high probability and that the second algorithm satisfies the properties needed to apply \Corref{Conditionals}. In the third (final) stage, we combine these into a proof of $\QNCd$ hardness of $\CH{d}$. The primary purpose here is to setup the basic notation which is used to show $\QCd$ and later $\CQd$ hardness.
}

\subsubsection{Shadow oracles for \texorpdfstring{$\mathsf{QNC}_{d}$ hardness}{QNC\_d hardness}}
\begin{figure}%
    \centering
    \subfloat[Base Sets]{
        \includegraphics[width=10cm]{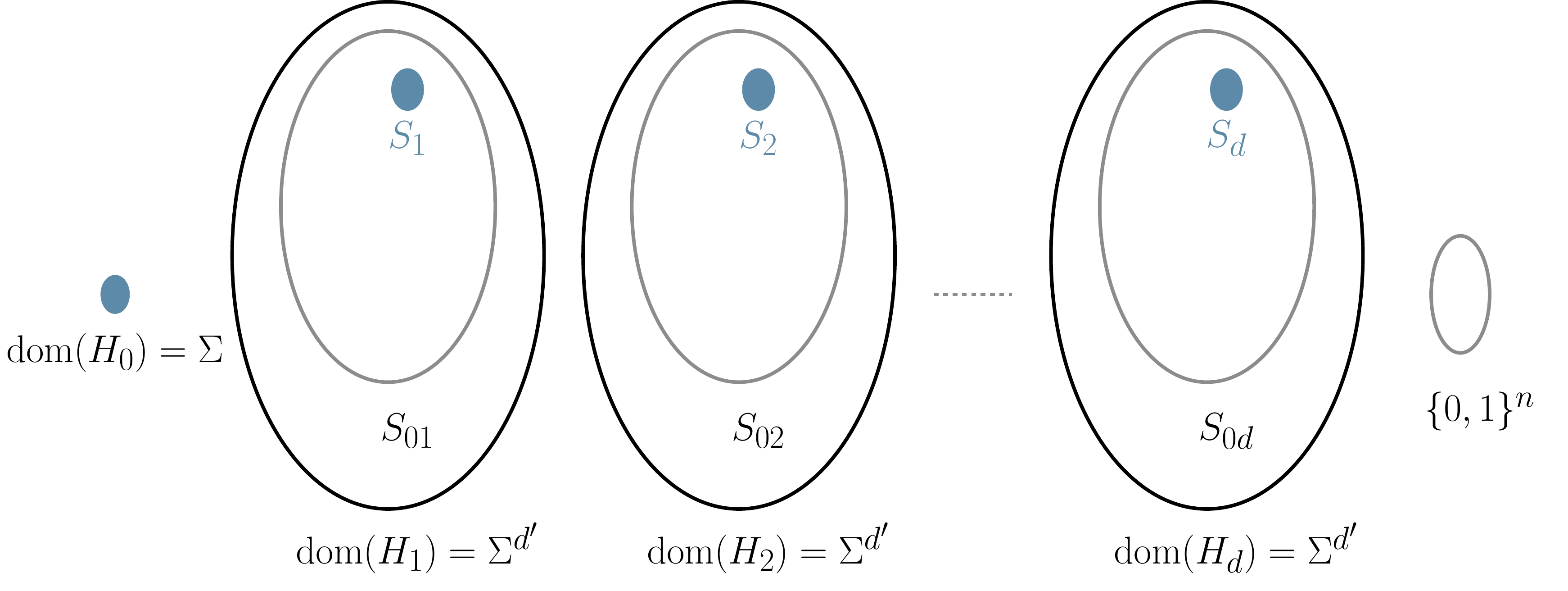}
        \label{fig:baseSets}
    }\\
    \subfloat[$S_{ij}$ inside Base Sets]{
        \includegraphics[width=10cm]{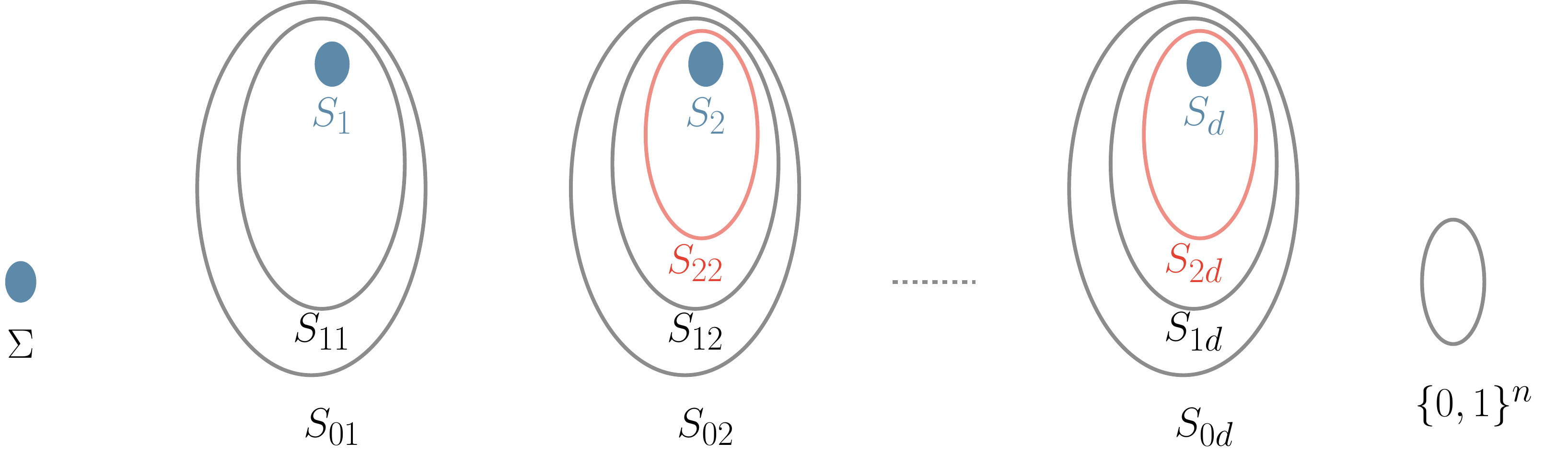}
        \label{fig:constructingSij}
    }\\
    \subfloat[Unified view: Base Sets and $S_{ij}$]{
        \includegraphics[width=10cm]{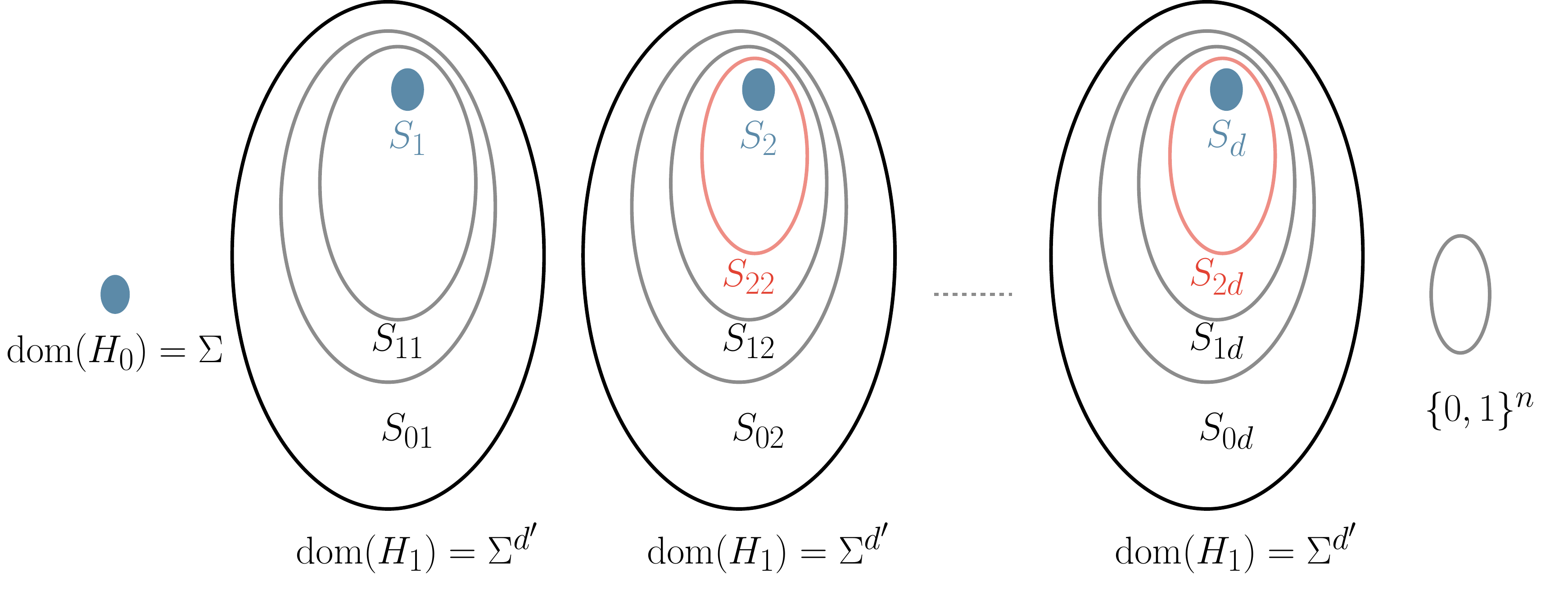}
        \label{fig:BaseSets_and_Sij}...
    }%
    \caption{Illustrating the sets produced by \Algref{baseSets} and \Algref{setMatrix}}
\end{figure}

\branchcolor{purple}{
    We begin with constructing ``base sets'' (see \Figref{baseSets}). We simply generate a random set $S_{01} \subseteq \dom(H_1)$ and propagate it through $\calL$. Ensuring this set is sufficiently small compared to $\Sigma^{d'}$, one can later show that $\calL$ restricted to the sequence of sets $(S_{01},H_1(S_{01}),H_2(H_1(S_{01})),\dots H_d(\dots H_1(S_{01})\dots))$  behaves as a permutation with high probability. 
}

\begin{lyxalgorithm}[Base sets]
	\label{alg:baseSets} Let $\mathcal{L}:=(H_{0},\dots H_{d})$, $d'$
	and $\Sigma$ be as in \Defref{dCodeHashingProblem}. Let $S_{i}:=H_{i-1}(\dots H_{0}(\Sigma)\dots)\subseteq\Sigma^{d'}$
	for $i\in\{1,\dots d\}$. 
	\begin{enumerate}
		\item Base Sets
		\begin{enumerate}
			\item Sample $S_{01}\subseteq\Sigma^{d'}$ uniformly at random,%
      s.t. $S_{1}\subseteq S_{01}$

			and $|S_{01}|^2/|\Sigma^{d'}|=1/|\Sigma|$ (i.e. $|S_{01}|=|\Sigma^{d+2}|$). 
			\item Define $S_{0,i+1}:=H_{i}(S_{0,i})$ for $i\in\{1,\dots d-1\}$.
		\end{enumerate}
		\item Abort if any of the following conditions are not met.
		\begin{enumerate}
			\item $|S_{0i}|=|\Sigma|^{d+2}$ for all $i\in\{1,\dots d\}$ (the $i=1$
			condition holds by construction).
			\item $|S_{1}|=|\Sigma|$ (which together with (a) implies $|S_{i}|=|\Sigma|$
			for all $i\in\{1,\dots d\}$). 
		\end{enumerate}
	\end{enumerate}
\end{lyxalgorithm}

\branchcolor{purple}{
    Conditions in item 2 are important because the random function may introduce collisions. The conditions ensure there are no collisions in the domains of interest.

    We now introduce the construction of the sets $S_{ij}$ (see \Figref{constructingSij}). These are perhaps best viewed as a matrix whose elements are subsets of $\Sigma^{d'}$,
    \[
		S_{ij}\doteq\left[\begin{array}{ccccc}
			S_{11} & H_{1}(S_{11}) & H_{2}(H_{1}(S_{11})) & \dots & H_{d}(\dots H_{1}(S_{11})\dots)\\
			\emptyset & S_{22} & H_{2}(S_{22}) & \dots & H_{d}(\dots H_{2}(S_{22})\dots)\\
			\emptyset & \emptyset & S_{33} & \dots & H_{d}(\dots H_{3}(S_{33})\dots)\\
			&  &  & \ddots\\
			\emptyset & \emptyset & \emptyset &  & S_{dd}
		\end{array}\right].
	\]
    The first row is, element-wise, a subset of $(S_{01},S_{02},\dots S_{0d})$. Similarly, each row is an element-wise subset of the previous row. With each row, the size of the set drops exponentially (in $n$, relative to the previous row). The diagonal sets are chosen uniformly at random, ensuring $S_i$ are contained within (just as we required for the ``base sets''). Formally, the procedure is defined as follows.
}

\begin{lyxalgorithm}[Procedure for constructing $S_{ij}$]
	\label{alg:setMatrix}Let $\mathcal{L}:=(H_{0},\dots H_{d})$, $\Sigma$
	and $S_{i}$ be as in \Algref{baseSets}. Suppose \Algref{baseSets}
	was executed. If \Algref{baseSets} aborts, define $S_{ij}=\emptyset$
	for all $i,j\in\{1,\dots d\}$. If \Algref{baseSets} does not abort
	then, for each $i\in\{1,\dots d\}$
	\begin{enumerate}
		\item Define $S_{ik}=\emptyset$ for $1\le k<i$. 
		\item Sample, uniformly at random, $S_{ii}\subseteq S_{i-1,i}$ such that
		$S_{i}\subseteq S_{ii}$ and $|S_{ii}|/|S_{i-1,i}|=1/\left|\Sigma\right|$. 
		\item Define $S_{ik}=H_{k-1}(\dots H_{i}(S_{ii})\dots)$ for $i<k\le d$. 
	\end{enumerate}
	In both cases, return $\bar{S}_{i}:=(S_{i1},S_{i2},\dots S_{id})$
	for each $i\in\{1,\dots d\}$.
\end{lyxalgorithm}

\branchcolor{purple}{
    Two short remarks---first, when \Algref{baseSets} fails, we simply abort and output $\emptyset$ as we don't care what happens in that case. This is because it fails with vanishing probability as we prove next. Second, it may help to note that $\bar{S}_i$, in the matrix representation above, is just the $i$th row of $S_{ij}$. 
}

\subsubsection{Properties of the shadow oracles}

\branchcolor{purple}{
Like we said, \Algref{baseSets} fails with vanishing probability.
}

\begin{claim}
	\Algref{baseSets} outputs abort with at most $(d+1)\cdot\ngl{\lambda}$
	probability where $\lambda$ is as in \Defref{dCodeHashingProblem}.
	
	\branchcolor{blue}{\begin{proof}
			We use \Claimref{permutationProb} and a union bound. For each $i\in[d]$,
			condition 2 (a) fails with probability at most $1/|\Sigma|$. To see
			this, in \Claimref{permutationProb}, set $f\leftarrow H_{1}$, $A\leftarrow\Sigma^{d'}$,
			$B\leftarrow S_{01}$ to conclude that the probability that $S_{02}=H_{1}(S_{01})$
			has size strictly less than $|S_{01}|$ is at most $|B|^2/|A|=1/|\Sigma|$.
			Proceeding similarly, set $f\leftarrow H_{i}$, $A\leftarrow\Sigma^{d'}$,
			$B\leftarrow S_{0i}$ to conclude that the probability that $S_{0,i+1}=H_{i}(S_{0i})$
			has size strictly less than $|S_{0i}|$, is at most $1/|\Sigma|$.
			By a union bound, condition 2 (a) fails with probability at most $d\cdot1/|\Sigma|$.
			
			Similarly, condition 2 (b) fails with probability at most $1/|\Sigma|^{d'-2}<1/|\Sigma|$
			by \Claimref{permutationProb} with $f\leftarrow H_{0}$, $A\leftarrow\Sigma^{d'}$
			and $B\leftarrow\Sigma$ (note that the claim is true even when $f:B\to A$).
			Therefore the probability of abort is at most $(d+1)\cdot1/|\Sigma|$
			where $|\Sigma|=2^{\Theta(\lambda)}$, yielding the asserted bound.
		\end{proof}
	}
\end{claim}

\branchcolor{purple}{To apply \Corref{Conditionals}, we would need a bound on $\Pr[x\in S_{ik}|S_{i,k-1}]$ conditioned on not aborting.
}

\begin{claim}
	\label{claim:x_in_S_QNC}Let $\mathcal{L}$ be as in \Defref{dCodeHashingProblem},
	run \Algref{baseSets} and let $E$ be the event that it does not
	abort. Obtain $S_{ij}$ by running \Algref{setMatrix}. Then, it holds
	that 
	\[
	\Pr[x\in S_{ik}|(S_{i-1,k},E)]\le1/\left|\Sigma\right|
	\]
	and 
	\[
	\Pr[x\in S_{ik}|(\check{\mathcal{L}},E)]\le1/|\Sigma|
	\]
	where $\check{\mathcal{L}}$ is $\mathcal{L}$ outside $(S_{i-1,1},\dots S_{i-1,d})$
	(see \Notaref{LhatAndsuch} with $S^{{\rm out}}\leftarrow(S_{i-1,j})_{j}$
	and $\mathcal{L}'\leftarrow\mathcal{L}$) for all $1\le i\le k\le d$
	where the probability is over $\mathcal{L}$, the randomness in \Algref{baseSets}
	and in \Algref{setMatrix}. 
	
	\branchcolor{blue}{\begin{proof}[Proof sketch.]
			Consider the $k=i$ case. Once $S_{i-1,i}$ is fixed, $S_{ii}$ is
			a (uniform) random subset of $S_{i-1,i}$ and therefore the probability
			that any $x\in S_{i,i}$ (assume $x\in S_{i-1,i}$ to get an upper
			bound), is at most $\left|S_{i,i}\right|/\left|S_{i-1,i}\right|=1/|\Sigma|$
			(see \Remref{Pr_x_in_X_or_t}, first observation). The result continues
			to hold if $\mathcal{L}$ (or in particular $\check{\mathcal{L}}$)
			is specified because $S_{ii}$ is sampled uniformly at random by \Algref{setMatrix}.
			For $k>i$, note that conditioned on $E$, each $H_{1},H_{2}\dots H_{d-1}$
			behaves as a random permutation on $S_{0,1},S_{0,2}\dots S_{0,d-1}$.
			In particular, conditioned on $E$, each $H_{1},\dots H_{d-1}$ behaves
			as a random permutation on $S_{i-1,1},S_{i-1,2}\dots S_{i-1,d-1}$
			(even if $\check{\mathcal{L}}$ is given since it does not determine
			the values within $(S_{i-1,j})_{j}$). From the first observation
			in \Remref{Pr_x_in_X_or_t}, it follows that $x\in S_{ik}$ conditioned
			on $E$ and $S_{i-1,k}$ for $k>i$, is also at most $1/|\Sigma|$.
			This is because $H_{i-1}$ maps every subset of $S_{i-1,k-1}$ of
			size $|S_{i-1,k-1}|/|\Sigma|$ to another set of the same size in
			$S_{i-1,k}$ (i.e. $H_{i-1}$ essentially behaves as a permutation)
			and \Remref{Pr_x_in_X_or_t} shows the probability that $x\in S_{i,k}|\left(S_{i-1,k-1},E\right)$
			and $x\in S_{i,k}|(\check{\mathcal{L}},E)$, are both bounded by $1/|\Sigma|$. 
		\end{proof}
	}
\end{claim}

\subsubsection{\texorpdfstring{$\protect\CH d$ is hard for $\protect\QNC_{d}$}{d-CodeHashing is hard for QNC\_d}}
\branchcolor{purple}{
    With all the intermediate results proven, we can stitch them together to establish the $\QNCd$ hardness of $\CH{d}$. 
}
\begin{lem}[$\CH{d}\notin\QNC_{d}$]
	\label{lem:QNC_d_hardness}Every $\QNC_{d}$ circuit succeeds at solving
	$\CH d$ (see \Defref{dCodeHashingProblem}) with probability at most
	$\ngl{\lambda}$ on input $1^{\lambda}$ for $d\le\poly$. 
\end{lem}

\branchcolor{blue}{\begin{proof}
		For clarity of presentation, we omit the input $1^{\lambda}$ when
		convenient. Let $\mathcal{L}:=(H_{0},\dots H_{d})$ and $\Sigma$
		be as in \Defref{dCodeHashingProblem}. Denote an arbitrary $\QNC_{d}$
		circuit, $\mathcal{A}^{\mathcal{L}}$ by 
		\[
		\mathcal{A}^{\mathcal{L}}:=\Pi_{\mathbf{x}} \circ U_{d+1}\circ\mathcal{L}\circ U_{d}\dots\mathcal{L}\circ U_{2}\circ\mathcal{L}\circ U_{1}
		\]
		where $\Pi_{\mathbf{x}}$ is a projector corresponding to output $\mathbf{x}$. Let $\Pivalid$ be a projector on the set $X_{\rm valid}=\{\mathbf{x}\}$ of all correct solutions to \Defref{dCodeHashingProblem}. $\Pivalid$ implicitly depends on $H$ and $\lambda$. We use $\Pivalid$ later. For now, run \Algref{baseSets}
		and let $E$ be the event that it does not abort. Note that\footnote{Using $\Pr[A]=\Pr[A|E]\Pr[E]+\Pr[A|\neg E]\Pr[\neg E]$, which yields
			$\Pr[A]-\Pr[A|E]\le\Pr[A|\neg E]\Pr[\neg E]$, and that $\Pr[\neg E]=\ngl{\lambda}$.} 
		\begin{equation}
			\left|\sum_{\mathbf{x}\in X_{\rm valid}}\Pr[\mathbf{x}\leftarrow\mathcal{A}^{\mathcal{L}}]-\sum_{\mathbf{x}\in X_{\rm valid}}\Pr[\mathbf{x}\leftarrow\mathcal{A}^{\mathcal{L}}|E]\right|\le\negl.\label{eq:NoAbort}
		\end{equation}
		Let $(\bar{S}_{i})_{i\in\{1,\dots d\}}$ be the output of \Algref{setMatrix}.
		Define 
		\[
		\mathcal{A}^{\mathcal{M}}:=\Pi_{\mathbf{x}} \circ U_{d+1}\circ\mathcal{M}_{d}\circ U_{d}\dots\mathcal{M}_{2}\circ U_{2}\circ\mathcal{M}_{1}\circ U_{1}
		\]
		where $\mathcal{M}_{i}$ is the shadow oracle of $\mathcal{L}$ wrt
		$\bar{S}_{i}$ (see \Defref{shadowOracle}). 
		
		\emph{$\mathcal{A}^{\mathcal{M}}|E$ cannot succeed with non-negligible
			probability:} In this paragraph, we condition on $E$ implicitly.
		Recall $\tilde{H}=H_{d}\circ\dots\circ H_{0}:\Sigma\to\{0,1\}^{n}$
		and $\tilde{H}_{i}$ is the $i$th bit of $\tilde{H}$ (see \Defref{dCodeHashingProblem}).
		Observe that if $\mathbf{x}=(\mathbf{x}_{1},\dots\mathbf{x}_{d})\in C_{\lambda}$
		is such that $\tilde{H}_{i}(\mathbf{x}_{i})=1$ for all $i$, then
		$\Pr[\mathbf{x}\leftarrow\mathcal{A}^{\mathcal{M}}]\le1/2^{n}$. This
		is because the oracles $\mathcal{M}_{1},\dots\mathcal{M}_{d}$ contain
		no information about $\tilde{H}_{i}(\mathbf{x}_{i})$ therefore $\mathbf{x}$
		cannot be correlated to the values the random oracle assigns to $\tilde{H}$.
		The probability that for any given $\mathbf{x}$, all $\tilde{H}_{i}(\mathbf{x}_{i})$
		output $1$ is at most $1/2^{n}$. 
		
		\emph{$\mathcal{A}^{\mathcal{M}}|E$ and $\mathcal{A}^{\mathcal{L}}|E$
			have practically the same behaviour:} Using a hybrid argument and
		the O2H lemma (see \Lemref{O2H}), one finds that the output distributions
		of $\mathcal{A}^{\mathcal{M}}|E$ and $\mathcal{A}^{\mathcal{L}}|E$
		cannot be noticeably different. We have (we dropped the $\circ$ symbol,
		the conditioning on $E$ and the subscript $\rm valid$ from $\Pivalid$
		for brevity/clarity)
		
		\begin{align}
			& \left|\sum_{\mathbf{x}\in X_{\rm valid}} \Pr[\mathbf{x}\leftarrow\mathcal{A}^{\mathcal{L}}]-\sum_{\mathbf{x}\in X_{\rm valid}} \Pr[\mathbf{x}\leftarrow\mathcal{A}^{\mathcal{M}}]\right|\nonumber \\
			= & \left|\tr[\Pivalid U_{d+1}\mathcal{L}U_{d}\dots\mathcal{L}U_{2}\mathcal{L}U_{1}\rho_{0}-\Pivalid U_{d+1}\mathcal{M}_{d}U_{d}\dots\mathcal{M}_{2}U_{2}\mathcal{M}_{1}U_{1}\rho_{0}]\right| & \text{monotonicity of TD}\nonumber \\
			& \le\left|\tr[\Pi U_{d+1}\mathcal{L}U_{d}\dots\mathcal{L}U_{2}\underbrace{\mathcal{L}U_{1}\rho_{0}}-\Pi U_{d+1}\mathcal{L}U_{d}\dots\mathcal{L}U_{2}\underbrace{\mathcal{M}_{1}U_{1}\rho_{0}}]\right|+ & \text{triangle inequality}\nonumber \\
			& \left|\tr[\Pi U_{d+1}\mathcal{L}U_{d}\dots U_{3}\underbrace{\mathcal{L}U_{2}\mathcal{M}_{1}U_{1}\rho_{0}}-\Pi U_{d+1}\mathcal{L}U_{d}\dots\mathcal{L}U_{3}\underbrace{\mathcal{M}_{2}U_{2}\mathcal{M}_{1}U_{1}\rho_{0}}]\right|+\nonumber \\
			& \vdots\nonumber \\
			& \left|\tr[\Pi U_{d+1}\underbrace{\mathcal{L}U_{d}\mathcal{M}_{d-1}U_{d-1}\dots U_{3}\mathcal{M}_{2}U_{2}\mathcal{M}_{1}U_{1}\rho_{0}}-\Pi U_{d+1}\underbrace{\mathcal{M}_{d}U_{d}\mathcal{M}_{d-1}\dots U_{3}\mathcal{M}_{2}U_{2}\mathcal{M}_{1}U_{1}\rho_{0}}]\right|\nonumber \\
			& \le{\rm B}(\mathcal{L}\circ U_{1}(\rho_{0}),\mathcal{M}_{1}\circ U_{1}(\rho_{0}))+ & \text{relation b/w TD and B}\nonumber \\
			& {\rm B}(\mathcal{L}\circ U_{2}(\rho_{1}),\mathcal{M}_{2}\circ U_{2}(\rho_{1}))+\nonumber \\
			& \vdots\nonumber \\
			& {\rm B}(\mathcal{L}\circ U_{d}(\rho_{d-1}),\mathcal{M}_{d}\circ U_{d}(\rho_{d-1}))\nonumber \\
			& \le\sum_{i=1}^{d}\sqrt{2\Pr[{\rm find}:U_{i}^{\mathcal{L}\backslash\bar{S}_{i}},\rho_{i-1}]} & \text{\prettyref{lem:O2H}}\label{eq:boundQNC}
		\end{align}
		where $\rho_{0}=\left|1^{\lambda},0\dots0\right\rangle \left\langle 1^{\lambda},0\dots0\right|$
		and $\rho_{i}=\mathcal{M}_{i}\circ U_{i}\circ\dots\mathcal{M}_{1}\circ U_{1}(\rho_{0})$
		for $i>0$. To bound the last expression, one can use \Lemref{boundPfind}
		via \Corref{Conditionals} (recall that everything is conditioned
		on $E$). Let $\check{\mathcal{L}}_{i}$ be $\mathcal{L}$ outside
		$(S_{i1},\dots S_{id})$ (see \Notaref{LhatAndsuch} with $\mathcal{L}'\leftarrow\mathcal{L}$,
		$S^{{\rm out}}\leftarrow(S_{ij})_{j}$ and define $\check{\mathcal{L}}_{i}:=\check{\mathcal{L}}'$)
		for each $i\in\{0,1\dots d\}$ (we include $0$ to include the base
		sets specified by \Algref{baseSets}). Similarly, let $\hat{\mathcal{L}}_{i}$
		be $\mathcal{L}$ inside $(S_{i1},\dots S_{id})$ (see \Notaref{LhatAndsuch}
		with $\mathcal{L}'\leftarrow\mathcal{L}$, $S^{{\rm out}}\leftarrow(S_{ij})_{j}$
		and define $\hat{\mathcal{L}}_{i}:=\hat{\mathcal{L}}'$). Note that
		the only information about $\mathcal{L}$ contained in $\mathcal{M}_{i}$,
		is $\check{\mathcal{L}}_{i}$, for each $i\in\{1,\dots d\}$. Consider
		$\Pr[{\rm find}:U_{i}^{\mathcal{L}\backslash\bar{S}_{i}},\rho_{i-1}]$
		and note that $\rho_{i-1}$ at most specifies $\check{\mathcal{L}}_{i-1}$
		(about $\mathcal{L}$). Let $\sigma_{i}:=\rho_{i-1}|\check{\mathcal{L}}_{i-1}$,
		$\bar{R}_{i}:=\bar{S}_{i}|\check{\mathcal{L}}_{i-1}$ and $\mathcal{N}_{i}:=\mathcal{L}|\check{\mathcal{L}}_{i-1}$.
		Observe that $\bar{R}_{i}$ is uncorrelated to $\sigma_{i}$ (because
		once $\check{\mathcal{L}}_{i-1}$ is fixed, $\sigma_{i}$ contains
		no information about how $\mathcal{L}$ behaves in $\bar{S}_{i-1}=(S_{i-1,1}\dots S_{i-1,d})$
		and $\bar{R}_{i}$ depends only on the randomness in \Algref{setMatrix}
		and on $\hat{\mathcal{L}}_{i-1}$). One can thus apply \Corref{Conditionals}
		with \Claimref{x_in_S_QNC} to obtain 
		\[
		\Pr[{\rm find}:V_{i}^{\mathcal{N}_{i}\backslash\bar{R}_{i}},\sigma_{i-1}]\le d\cdot\bar{q}\cdot\frac{1}{|\Sigma|}
		\]
		which entails 
		\[
		\Pr[{\rm find}:U_{i}^{\mathcal{L}\backslash\bar{S}_{i}},\rho_{i-1}]\le\ngl{\lambda}
		\]
		by using $\Pr[A]=\sum_{B=b}\Pr[A|B=b]\Pr[B=b]$ and the parameters
		$d,q\le{\rm poly}(\lambda)$, and $|\Sigma|=2^{\Theta(\lambda)}$. 
		
		Plugging these into the last expression above (\Eqref{boundQNC}),
		yields $\left|\Pr[\mathbf{x}\in X_{\rm valid}|E:\mathbf{x}\leftarrow\mathcal{A}^{\mathcal{L}}]-\Pr[\mathbf{x}\in X_{\rm valid} | E: \mathbf{x}\leftarrow\mathcal{A}^{\mathcal{M}}|E]\right|\le\ngl{\lambda}$
		where we now state $E$ explicitly. Using \Eqref{NoAbort} and the
		triangle inequality, we obtain the asserted result.
		
	\end{proof}
}

\subsection{\texorpdfstring{$\protect\classQC{d}$ exclusion}{QNC\_d\^{}BPP exclusion}} \label{sec:qcdhardness}

\branchcolor{purple}{Once the analysis of $\QNCd$ is clear, extending it to $\QCd$ is not too difficult. One needs to account for the actions of the intermediate classical circuits. The basic approach stays the same. We replace $\calL$ with shadow oracles successively. The difference is that after each set of parallel queries, we account for the polynomially many queries made by the corresponding intermediate classical algorithm by exposing those queries in the subsequent shadow oracles.}

\subsubsection{Shadow oracles for \texorpdfstring{$\protect\QC d$}{QC\_d} hardness}

\branchcolor{purple}{
    The procedure for constructing base sets stays unchanged. We need the analogue of \Algref{setMatrix}. However, unlike \Algref{setMatrix}, this time the procedure cannot directly produce $S_{ij}$ for all $i,j$, given the base sets. This is because the sets $S_{ij}$ now must also depend on the queries made by the classical algorithm at intermediate steps. 
    
    Before we present the algorithm, we make the following assumption (which only makes the impossibility result stronger): the classical algorithm makes ``path queries'',
    i.e. suppose when it queries $H_{i}$ at $t_{i}$, it learns all tuples
    $(t_{0},t_{1},t_{2}\dots,t_{i},\dots t_{d})$ such that $H_{j-1}(t_{j-1})=t_{j}$
    for all $j\in\{1,\dots d\}$. Since $H_0$ cannot span the domain of $H_1$, $t_0$ may not always exist,  corresponding to $(t_1,t_2 \dots t_d)$. More formally, we have the following.
    }

\begin{defn}[Path Queries]
	\label{def:pathQueries} Let $\mathcal{L}':=(H_{0}',\dots H_{d}')$
	be as in \Notaref{LhatAndsuch} and let $\bar{T}_{i}:=(T_{i0},T_{i1},\dots T_{id})$
	be a tuple of sets where for each $0\le j\le d$, $T_{ij}\subseteq\Sigma^{d'}$.
	We say $\bar{T}_{i}$ are \emph{path queries} if $T_{i1}\supseteq H_0(T_{i0})$, and $T_{ij}=H_{j-1}(T_{i,j-1})$
	for all $j\in\{2,\dots,d\}$. 
\end{defn}

\branchcolor{purple}{We can now define the algorithm. For context, it may help to recall that (see \Notaref{CompositionNotationOracles}) an arbitrary $\QC{d}$ circuit with oracle access to $\calL$ can be represented as 		\[
    \mathcal{B}^{\mathcal{L}}:=\Pi\circ\mathcal{A}_{c,d+1}^{\mathcal{L}}\circ\mathcal{B}_{d}^{\mathcal{L}}\circ\dots\mathcal{B}_{1}^{\mathcal{L}}\circ\rho_{0}
    \]
    where $\mathcal{B}_{i}^{\mathcal{L}}:=\Pi_{i}\circ\mathcal{L}\circ U_{i}\circ\mathcal{A}_{c,i}^{\mathcal{L}}$,
    $\rho_{0}$ is the initial state (in our case, encoding $1^{\lambda}$)
    and $\Pi$ is a measurement. Below, informally,\footnote{We say informally because the queries $\calA_{c,i}$ makes depends on the hybrid we are considering; these details appear later in the proof of \Lemref{QC_d_hardness}.} $\bar{T}_i$ corresponds to the set of queries made by the classical algorithm $\calA_{c,i}$ to $\calL$.}

\begin{lyxalgorithm}[Procedure for constructing $S_{ij}$, given $\bar{T}_{i}$s]
	\label{alg:setMatrixQC_d}Let $\mathcal{L}:=(H_{0},\dots H_{d})$,
	$\Sigma$ and $S_{i}$ be as in \Algref{baseSets} and suppose the
	\Algref{baseSets} was executed. 
	
	Input: 
	\begin{enumerate}
		\item The previous sequence of sets for creating the shadow oracle: $\bar{S}_{i-1}:=(S_{i-1,j})_{j\in\{1\dots d\}}$
		where $S_{i-1,j}\subseteq S_{0,j}$ for all $j\in\{1,\dots d\}$.
		\item The path queries made by the classical algorithm at step $i$: $\bar{T}_{i}:=(T_{i0},T_{i1},T_{i2}\dots T_{id})$
	\end{enumerate}
	If \Algref{baseSets} aborts, define $S_{ij}=\emptyset$ for all $i,j\in\{1,\dots d\}$.
	If \Algref{baseSets} does not abort then, for each $i\in\{1,\dots d\}$
	do the following.
	\begin{enumerate}
		\item Define $S_{ik}=\emptyset$ for $1\le k<i$.
		\item Sample, uniformly at random, $S_{ii}\subseteq S_{i-1,i}\backslash T_{ii}$
		such that $(S_{i} \cap S_{i-1,i}) \backslash T_{ii}\subseteq S_{ii}$ and $|S_{ii}|/|S_{i-1,i}|=1/|\Sigma|$. 
		\item Define $S_{ik}=H_{k-1}(\dots H_{i}(S_{ii})\dots)$ for $i<k\le d$. 
	\end{enumerate}
	In both cases, return $\bar{S}_{i}:=(S_{i1},S_{i2}\dots S_{id})$.
\end{lyxalgorithm}

\subsubsection{Properties of the shadow oracles}

\branchcolor{purple}{Points
	\begin{itemize}
		\item The following could potentially be more generally stated.
		\item We take the set $\bar{S}_{i-1}$ to be given (we only impose the bare
		requirements), and have $\bar{T}_{i}$ be arbitrary poly sized sets
		\item We show that given $\bar{S}_{i-1}$ and the sets $\bar{T}_{i}$, finding
		$x$ in $\bar{S}_{i}$ would happen with probability $\ply{\lambda}/|\Sigma|$
		at most.
	\end{itemize}
}
\begin{claim}
	\label{claim:x_in_S_QC_d}Let $\mathcal{L}$ be as in \Defref{dCodeHashingProblem},
	run \Algref{baseSets} and let $E$ be the event that it does not
	abort. Let $1\le i\le d$. Obtain $\bar{S}_{i}$ by running \Algref{setMatrixQC_d}
	with the following input: 
	\begin{enumerate}
		\item If $i=1$, use $\bar{S}_{0}$ generated by \Algref{baseSets}. \\
		Else, if $i>1$, let $\bar{S}_{i-1}:=(S_{i-1,1},S_{i-1,2}\dots S_{i-1,d})$
		be arbitrary sets such that
		\begin{itemize}
			\item for $j<i-1$, $S_{i-1,j}=\emptyset$, %
			\item for $j=i-1$, $S_{i-1,i-1}\subseteq S_{0,i-1} $ and $\left|S_{i-1,i-1}\right| = |\Sigma|^{d+2-(i-1)}| = |\Sigma^{d+1-i}|$ %
			and finally
			\item for $j>i-1$, $S_{i-1,j}\subseteq H_{j}(S_{i-1,j-1})=H_{j}(\dots H_{i-1}(S_{i-1,i-1})\dots)$.
		\end{itemize}
		\item $\bar{T}_{i}:=(T_{i0},\dots T_{id})$ be arbitrary path queries (see
		\Defref{pathQueries}) such that $\left|T_{ij}\right|\le\ply{\lambda}$
		for all $j\in\{0,\dots d\}$. 
	\end{enumerate}
	Then, it holds (for a large enough $\lambda$) that 
	\[
	\Pr[x\in S_{ik}|(S_{i-1,k},T_{i},E)]\le\ply{\lambda}/|\Sigma|
	\]
	and 
	\[
	\Pr[x\in S_{ik}|(\check{\mathcal{L}},E)]\le\ply{\lambda}/|\Sigma|
	\]
	where $\check{\mathcal{L}}$ is $\mathcal{L}$ outside $(S_{i-1,1}\backslash T_{i1},\dots S_{i-1,d}\backslash T_{id})$
	(see \Notaref{LhatAndsuch} with $S^{{\rm out}}\leftarrow(S_{i-1,j}\backslash T_{ij})_{j}$
	and $\mathcal{L}'\leftarrow\mathcal{L}$) for all $1\le i\le k\le d$
	where the probability is over $\mathcal{L}$, the randomness in \algref{setMatrixQC_d}. 
\end{claim}

\branchcolor{purple}{
    Before looking at the proof, we briefly comment on the claim. Item 1 is meant to enforce the form of the set $\bar S_{i-1}$ which would be produced by repeated applications of \Algref{setMatrixQC_d}. Therefore the first bullet ensures all sets before $i-1$ are empty, the second ensures the diagonal one has the right size (we start with $|\Sigma|^{d+2}$ for base sets and at each iteration, the size drops by $|\Sigma|$) and the last bullet ensures that the sets are no larger than if they were propogated through $\calL$. Item 2 allows one to specify the classical queries made at the $i$th step. The statement says that if these inputs are used in \Algref{setMatrixQC_d} to obtain the next sequence of sets, $\bar S_i$, then one can obtain a bound analogous to that of \Claimref{x_in_S_QNC}. The difference is that this time, both the previous sequence of sets $\bar S_{i-1}$ and the classical queries $\bar T_i$ are revealed. 
}

\branchcolor{blue}{\begin{proof}[Proof sketch.]
		The idea is the same as that we used in the proof of \Claimref{x_in_S_QNC}.
		The only difference is that instead of considering the sets $S_{i-1,j}$,
		one considers $S'_{i-1,j}:=S_{i-1,j}\backslash T_{i,j}$. Let $f(\lambda)$
		be such that $|T_{ij}|\le f(\lambda)$ and suppose $\lambda$ is large
		enough so that $|\Sigma|>f(\lambda)$. For the $k=i$ case, we get
		$x\in S_{i,i}$ is at most\footnote{We have 
			\begin{align*}
				|S_{ii}|/|S'_{i-1,i}| & =|S_{ii}|/(|S'_{i-1,i}|-f)\\
				& =|\Sigma|^{d'-i-1}/(|\Sigma|^{d'-1}-f)\\
				& =\frac{1}{|\Sigma|(1-f/|\Sigma|^{d'-i-1})}\\
				& \le\frac{\ply{\lambda}}{|\Sigma|}.
			\end{align*}
			using, $(1-x)^{-1}\le 1+x + \epsilon$ for small enough $x$, where $\epsilon>0$ is some constant. } $|S_{i,i}|/\left|S'_{i-1,i}\right|=\ply{\lambda}/|\Sigma|$. Similarly,
		for $k>i$, using \Remref{Pr_x_in_X_or_t} (first observation) with
		$N \leftarrow |S_{i,k}|=\left|\Sigma^{d+2-i}\right|$ and $M \leftarrow |S_{i-1,k}'|=\left(\left|\Sigma^{d+1-i}\right|-\ply{\lambda}\right)$,
		one obtains that $x\in S_{ik}$ (conditioned on knowing $T_{ik}$
		and $S_{i-1,k}$ and $E$) with probability at most $N/M\le\ply{\lambda}/|\Sigma|$. 
	\end{proof}
}

\subsubsection{\texorpdfstring{$\protect\CH d$ is hard for $\protect\QC d$}{d-CodeHashing is hard for QC\_d}}
\branchcolor{purple}{
    We can now establish $\QCd$ hardness of $\CH{d}$. 
}
\begin{lem}[$\CH{d}\notin \classQC{d}$] \label{lem:QC_d_hardness}
	Every $\QC d$ circuit succeeds at solving $\CH d$ (see \Defref{dCodeHashingProblem})
	with probability at most $\ngl{\lambda}$ on input $1^{\lambda}$
	for $d\le\poly$. 
\end{lem}

\branchcolor{blue}{\begin{proof}
		The proof is similar to that of \Lemref{QNC_d_hardness}. Again, we
		omit the input $1^{\lambda}$ when convenient. Let $\mathcal{L}:=(H_{0},\dots H_{d})$
		and $\Sigma$ be as in \Defref{dCodeHashingProblem}. Denote an arbitrary
		$\QC d$ circuit, $\mathcal{B}^{\mathcal{L}}$ by 
		\[
		\mathcal{B}^{\mathcal{L}}:=\Pi_{\mathbf{x}} \circ\mathcal{A}_{c,d+1}^{\mathcal{L}}\circ\mathcal{B}_{d}^{\mathcal{L}}\circ\dots\mathcal{B}_{1}^{\mathcal{L}}
		\]
		where $\mathcal{B}_{i}^{\mathcal{L}}:=\Pi_{i}\circ\mathcal{L}\circ U_{i}\circ\mathcal{A}_{c,i}^{\mathcal{L}}$
		and $\Pi_{\mathbf{x}}$ is a projector corresponding to output $\mathbf{x}$. Let $\Pivalid$ be a projector on the set $X_{\rm valid}=\{\mathbf{x}\}$ of all correct solutions to \Defref{dCodeHashingProblem}. %
    Run \Algref{baseSets}
		and let $E$ be the event that it does not abort. Note that 
		\begin{equation}
			\left|\Sumvalid \Pr[\mathbf{x}\leftarrow\mathcal{B}^{\mathcal{L}}]-\Sumvalid \Pr[\mathbf{x}\leftarrow\mathcal{B}^{\mathcal{L}}|E]\right|\le\negl.\label{eq:noAbortQC_d}
		\end{equation}
		Define 
		\[
		\mathcal{B}^{\mathcal{M}}:=\Pi_{\mathbf{x}} \circ\mathcal{A}_{c,d+1}^{\mathcal{L}}\circ\mathcal{B}_{d}^{\mathcal{M}}\circ\dots\circ\mathcal{B}_{1}^{\mathcal{M}}
		\]
		where $\mathcal{B}_{i}^{\mathcal{M}}:=\Pi_{i}\circ\mathcal{M}_{i}\circ U_{i}\circ\mathcal{A}_{c,i}^{\mathcal{L}}$
		and $\mathcal{M}_{i}$ is the shadow oracle of $\mathcal{L}$ wrt
		$\bar{S}_{i}$ (see \Defref{shadowOracle}). We are yet to define
		$\bar{S}_{i}$. Do the following for each $i\in(1,2\dots d)$. Suppose
		$\bar{S}_{1},\dots\bar{S}_{i-1}$ (and therefore $\mathcal{M}_{1},\dots\mathcal{M}_{i-1}$)
		have been defined and suppose $\mathcal{A}_{c,i}^{\mathcal{L}}$ makes
		path queries $\bar{T}_{i}=(T_{i0},T_{i1},\dots T_{id})$ to $\mathcal{L}$.
		Then, let $\bar{S}_{i}$ be the output of \Algref{setMatrixQC_d}
		with $\bar{S}_{i-1}$ and $\bar{T}_{i}$ as input. 
		
		\emph{$\mathcal{B}^{\mathcal{M}}|E$ cannot succeed with non-negligible
			probability:} We focus on the intermediate classical algorithms, $\{\mathcal{A}_{c,i}^{\mathcal{L}}\}_{i\in\{1,\dots,d+1\}}$
		because the quantum parts have no access to $\tilde{H}$ (other than
		that already exposed by classical queries). Consider the labelling
		in \Figref{QC_d_M_labelling} and suppose that the input to $\mathcal{A}_{c,i}^{\mathcal{L}}$
		is $c_{i-1}'$ and its output is $c_{i}$. Similarly, suppose the
		input to $\Pi_{i}\mathcal{M}_{i}U_{i}$ is $c_{i}$ (classical) and
		$q_{i-1}$ (quantum) and its output is $c'_{i}$ (classical) and $q_{i}$
		(quantum). Observe\footnote{To see this, observe that
			\begin{itemize}
				\item $c_{1}$ at most reveals $T_{10}$
				\item both $c_{1}'$ and $q_{1}$ reveal at most $T_{10}$
				\item $c_{2}$ at most reveals $T_{20}\cup T_{10}$
				\item both $c_{2}'$ and $q_{2}$ reveal at most $T_{20}\cup T_{10}$ 
				\item and so on...
			\end{itemize}
		} that, $c_{i},c_{i}',q_{i}$ at most reveal $\tilde{H}$ at $T_{i0}\cup T_{i-1,0}\dots\cup T_{1,0}$.
		Since $|T_{i0}\cup T_{i-1,0}\dots\cup T_{i,0}|$ is at most polynomial,
		from \Thmref{YZ22} (second part), we conclude that $\mathcal{A}_{c,d+1}^{\mathcal{L}}$
		succeeds at solving ${\rm CodeHashing}$ with probability at most
		negligible. Note in particular, that since the quantum part, $\Pi_{i}\mathcal{M}_{i}U_{i}$
		does not access $\tilde{H}$ outside $T_{i0}$, it can be classically
		simulated without making any calls to $\tilde{H}$. Consequently,
		one can treat the entire algorithm as a classical algorithm for applying
		\Thmref{YZ22} (second part) because the theorem statement only depends
		on the number of classical queries to $\tilde{H}$ (and not on the
		computational complexity of the circuit). 
		
		\begin{figure}
			\begin{centering}
				\includegraphics[width=10cm]{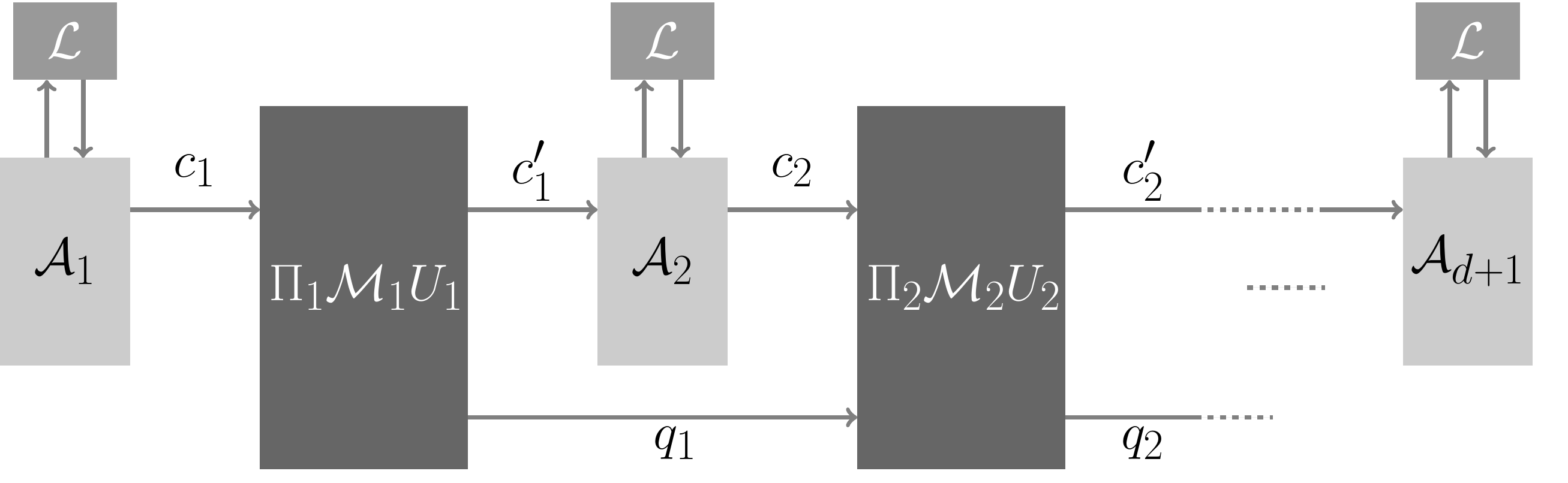}
				\par\end{centering}
			\caption{Illustration of the $\QCd$ circuit, $\calB^{\calM}$, where all oracles have been replaced by shadow oracles. Note that we dropped $\circ$ between the operators for brevity. \label{fig:QC_d_M_labelling}}%
		\end{figure}
		
		\emph{$\mathcal{B}^{\mathcal{M}}|E$ and $\mathcal{B}^{\mathcal{L}}|E$
			have practically the same behaviour:} We use a hybrid argument and
		the O2H lemma (see \Lemref{O2H}) to obtain the following (we dropped
		the $\circ$ symbol, the conditioning on $E$)
		
		\begin{align}
			& \left|\Sumvalid \Pr[\mathbf{x}\leftarrow\mathcal{B}^{\mathcal{L}}]-\Sumvalid \Pr[\mathbf{x}\leftarrow\mathcal{B}^{\mathcal{M}}]\right|\nonumber \\
			= & \left|\Pivalid\circ\mathcal{A}_{c,d+1}^{\mathcal{L}}\circ\mathcal{B}_{d}^{\mathcal{L}}\circ\dots\circ\mathcal{B}_{1}^{\mathcal{L}}\circ\rho_{0}-\Pivalid\circ\mathcal{A}_{c,d+1}^{\mathcal{L}}\circ\mathcal{B}_{d}^{\mathcal{M}}\circ\dots\circ\mathcal{B}_{1}^{\mathcal{M}}\circ\rho_{0}\right|\nonumber \\
			\le & \sum_{i=1}^{d}B(\mathcal{B}_{i}^{\mathcal{L}}(\rho_{i-1}),\mathcal{B}_{i}^{\mathcal{M}}(\rho_{i-1}))\le\sum_{i=1}^{d}\sqrt{2\Pr[{\rm find}:U_{i}^{\mathcal{L}\backslash\bar{S}_{i}},\mathcal{A}_{c,i}^{\mathcal{L}}\circ\rho_{i-1}]}\label{eq:boundQCd}
		\end{align}
		where for $i\in\{1,2\dots d-1\}$, $\rho_{i}:=\mathcal{B}_{i}^{\mathcal{M}}\circ\dots\mathcal{B}_{1}^{\mathcal{M}}\circ\rho_{0}$.
		To bound the last expression, one can use \Lemref{boundPfind} via
		\Corref{Conditionals} (recall that everything is conditioned on $E$).
		Let $\check{\mathcal{L}}_{i}$ be $\mathcal{L}$ outside $(S_{i1}\backslash T_{i+1,1},\dots S_{id}\backslash T_{i+1,d})$
		(see \Notaref{LhatAndsuch} with $\mathcal{L}'\leftarrow\mathcal{L}$,
		$S^{{\rm out}}\leftarrow(S_{ij}\backslash T_{i+1,j})_{j\in\{1\dots d\}}$
		and define $\check{\mathcal{L}}_{i}:=\check{\mathcal{L}}'$) for each
		$i\in\{0,1\dots d\}$ (we include $0$ to include the base sets specified
		by \Algref{baseSets}). Similarly, let $\hat{\mathcal{L}}_{i}$ be
		$\mathcal{L}$ inside $(S_{i1}\backslash T_{i+1,1},\dots S_{id}\backslash T_{i+1,d})$
		(see \Notaref{LhatAndsuch} with $\mathcal{L}'\leftarrow\mathcal{L}$,
		$S^{{\rm out}}\leftarrow(S_{ij}\backslash T_{i+1,j})_{j}$ and $\hat{\mathcal{L}}_{i}:=\hat{\mathcal{L}}'$).
		Note that the only information about $\mathcal{L}$ contained in $\mathcal{M}_{i}$,
		is at most $\check{\mathcal{L}}_{i}$, for each $i\in\{1,\dots d\}$
		(at most because $\check{\mathcal{L}}$ also contains information
		queried by $\mathcal{A}_{c,i+1}^{\mathcal{L}}$). Consider $\Pr[{\rm find}:U_{i}^{\mathcal{L}\backslash\bar{S}_{i}},\mathcal{A}_{c,i}^{\mathcal{L}}\circ\rho_{i-1}]$
		and note that $\mathcal{A}_{c,i}^{\mathcal{L}}\circ\rho_{i-1}$ at
		most specifies\footnote{For $i=1$, $\check{\mathcal{L}}_{i-1}=\check{\mathcal{L}}_{0}$ is
			$\mathcal{L}$ outside $\bar{S}_{0}$ (which is rather lenient because
			$\rho_{0}$ contains no information about $\mathcal{L}$; to be precise,
			one could have used $(\Sigma^{d'},\dots\Sigma^{d'})$ instead of $\bar{S}_{0}$). } $\check{\mathcal{L}}_{i-1}$ (about $\mathcal{L}$). Note also that
		the queries, $\bar{T}_{i}$, made by $\mathcal{A}_{c,i}^{\mathcal{L}}$
		have been exposed in $\check{\mathcal{L}}_{i-1}$ and, furthermore,
		by construction (of \Algref{setMatrixQC_d}) are excluded from $\bar{S}_{i}$.
		Let $\sigma_{i}:=\mathcal{A}_{c,i}^{\mathcal{L}}\circ\rho_{i-1}|\check{\mathcal{L}}_{i-1}$,
		$\bar{R}_{i}:=\bar{S}_{i}|\check{\mathcal{L}}_{i-1}$ and $\mathcal{N}_{i}:=\mathcal{L}|\check{\mathcal{L}}_{i-1}$.
		After conditioning, $\sigma_{i}$ is uncorrelated to $\bar{R}_{i}$
		(because once $\check{\mathcal{L}}_{i-1}$ is fixed (which also fixes
		$\bar{T}_{i}$), $\sigma_{i}$ contains no information about how $\mathcal{L}$
		behaves in $\bar{S}_{i-1}\backslash\bar{T}_{i}$ and $\bar{R}_{i}$
		depends only on the randomness in \Algref{setMatrixQC_d} and on $\hat{\mathcal{L}}_{i-1}$).
		One can thus apply \Corref{Conditionals} with \Claimref{x_in_S_QC_d}
		to obtain 
		\[
		\Pr[{\rm find}:V_{i}^{\mathcal{N}_{i}\backslash\bar{R}_{i}},\sigma_{i-1}]\le d\cdot\bar{q}\cdot\frac{\ply{\lambda}}{|\Sigma|}
		\]
		which entails 
		\[
		\Pr[{\rm find}:U_{i}^{\mathcal{L}\backslash\bar{S}_{i}},\mathcal{A}_{c,i}^{\mathcal{L}}\circ\rho_{i-1}]\le\ngl{\lambda}
		\]
		by using $\Pr[A]=\sum_{B=b}\Pr[A|B=b]\Pr[B=b]$ and the parameters
		$d,q\le\ply{\lambda}$ and $\left|\Sigma\right|=2^{{\lambda}^{\Theta(1)}}$. 
		
		Plugging these into \Eqref{boundQCd}, yields $\left|\Pr[\mathbf{x}\in X_{\rm valid}|E:\mathbf{x}\leftarrow\mathcal{B}^{\mathcal{L}}]-\Pr[\mathbf{x}\in X_{\rm valid} | E : \mathbf{x}\leftarrow\mathcal{B}^{\mathcal{M}}]\right|\le\ngl{\lambda}$ where
		we now state conditioning on $E$ explicitly. Using \Eqref{noAbortQC_d}
		and the triangle inequality, we obtain the asserted result. 
		
	\end{proof}
}

\subsection{\texorpdfstring{$\protect\classCQ d$}{BPP\^{}QNC\_d} exclusion | Warm up\label{subsec:CQdHardnessWarmup}}

\global\long\def\varstar{\llcorner\lrcorner}%

\branchcolor{purple}{Establishing $\CQ d$ hardness takes more work. We briefly outline
the approach first and formalise it in the following sections. We
take inspiration from \cite{chia_need_2020-1} and adapt the implementation/formalism
introduced in \cite{arora_oracle_2022}. Let $\mathcal{L}:=(H_{0},\dots H_{d})$
be as defined in \Defref{dCodeHashingProblem}.

Consider a $\CQ d$ circuit. To show that it cannot solve $\CH d$,
the first quantum part, can be analysed as we did the $\QNC_{d}$
part (using domain hiding). Let the output of this quantum part be
a string $s_{1}$ and suppose the ``paths'' queried by the subsequent
classical part be $Y_{1}$. To analyse the subsequent quantum part,
one could expose (in the shadow oracles) the paths uncovered by $Y_{1}$
(as we did in the analysis of $\QC d$ circuits, albeit there we had
to do it after every unitary layer). However, this is not enough because
the string $s_{1}$ is correlated to the oracle $\mathcal{L}$ and
it is unclear how our techniques would work with $\mathcal{L}|s_{1}$
instead of $\mathcal{L}$. It turns out that if the string $s_{1}$
appears with non-negligible probability, then $\mathcal{L}|s_{1}$
can be viewed as a ``convex combination'' of $\mathcal{L}$ with
a polynomial number of ``paths'' fixed. One can then proceed (almost)
as in the $\QNC_{d}$ case for the next second quantum part. This
procedure can be iterated polynomially many times to yield the desired
hardness.

Before we can make any of this precise, we need to introduce the sampling
argument. While the following overlaps with the informal discussion presented in the Technical Overview, there are more details and precise statements. 

}

\subsection{Technical Results II | The sampling argument}

\branchcolor{purple}{We first describe the sampling argument in its simplest form and subsequently
  show how to lift the result to our setting of interest.}

\subsubsection{Warm up | Sampling argument for Permutations\label{subsec:Sampling-argument-for-Permutations}}

\branchcolor{purple}{
We informally describe the prerequisites to state the sampling argument
for permutations, deferring formal definitions and proofs to \Secref{Appendix_Sampling-argument-for-Permutations}
in the Appendix. We are being slightly redundant below to aid readibility (we overlap slightly with \Secref{intro_tech_overview}).

Suppose $t$ is a permutation over $N$ elements labelled $\{0,\dots,N-1\}$.
This permutation $t$ is ordinarily viewed as a function, $t(x)$
which specifies how $x$ is mapped. However, one could equivalently view
$t$ as a collection of tuples $(x,y)$ such that $t(x)=y$. We call
such a tuple a ``path'' and any set of such ``paths'' a ``part''.

Now consider distributions over permutations. Let's begin with a uniform
distribution $\mathbb{F}$ over all permutations $u$. One may characterise
$\mathbb{F}$ as follows: for any $u\sim\mathbb{F}$, i.e. any $u$
sampled from $\mathbb{F}$, it holds that $\Pr[u(x)=y]=\Pr[(x,y)\in\paths(u)]=(N-1)!/N!$.
In fact, it also holds that $\Pr[S\subseteq\paths(u)]=(N-|S|)!/N!$
where $S$ is a collection of (non-colliding) paths. It turns out
that this way of viewing the uniform distribution helps us below.

We first state a basic version of the sampling argument. To this end,
we define a \emph{$(p,\delta)$ non-uniform distribution}, $\mathbb{F}^{(p,\delta)}$,
which is closely related to the uniform distribution $\mathbb{F}$.
At a high level, $\mathbb{F}^{(p,\delta)}$ is ``$\delta$ close
to'' $\mathbb{F}$ with at most $p$ many paths fixed. What does
``$\delta$ closeness'' mean? For any distribution $\mathbb{G}$
(over permutations), a distribution $\mathbb{G}^{\delta}$ is $\delta$
close to it if the following holds: when $t'\sim\mathbb{G}^{\delta}$
and $t\sim\mathbb{G}$, one has $\Pr[S\subseteq\paths(t')]\le2^{\delta|S|}\Pr[S\subseteq\paths(t)]$
for all parts $S$.

We are almost ready to state the basic sampling argument. We need
the notion of a ``convex combination'' of random variables. We say
a random variable (such as our permutation) $t$ is a convex combination
of random variables $t_{i}$, denoted by $t\equiv\sum_{i}\alpha_{i}t_{i}$
(where $\sum_{i}\alpha_{i}=1$ and $\alpha_{i}\ge0$), if the following
holds for all $t'$: $\Pr[t=t']=\sum_{i}\alpha_{i}\Pr[t_{i}=t']$.

Informally, the basic sampling argument is a statement about a uniform
permutation $u\sim\mathbb{F}$ and how the distribution $\mathbb{F}$
changes if we are given some ``advice'' about this permutation which
is simply a function $g(u)$. Roughly speaking, given that $g(u)$
evaluates to $r$ with probability at least $2^{-m}$, the distribution
$\mathbb{F}$ conditioned on $r$ is a convex combination\footnote{In the convex combination, there is a small component, of weight at
  most $2^{-m}$, of some arbitrary distribution.} of $\mathbb{F}^{(p,\delta)}$ distributions where the number of paths
fixed is at most $p=2m/\delta$. Here $\delta$ is a free parameter.
We slightly abuse the notation and write this basic sampling argument
as
\[
  \mathbb{F}|r\equiv{\rm conv}(\mathbb{F}^{(p,\delta)}).
\]
The formal statement is as follows.}
\begin{prop}[$\mathbb{F}|r\equiv{\rm conv}(\mathbb{F}^{(p,\delta)})$]
  \label{prop:sumOfDeltaNonUni_perm_simplified} Let $u\sim\mathbb{F}(N)$
  be a uniformly random permutation over $N=2^{n}$ elements and $g(u)$
  be an arbitrary function. Fix any $\delta>0$, $\gamma=2^{-m}>0$
  where $m=m(n)$ and suppose $\Pr[g(u)=r]\ge\gamma$. Then
  \[
    t\equiv\sum_{i}\alpha_{i}t_{i}+\gamma't'
  \]
  where $t=u|(g(u)=r)$, $t_{i}\sim\mathbb{F}_{i}^{(p,\delta)}$ and
  $\mathbb{F}_{i}^{(p,\delta)}$ is $(p,\delta)$ non-uniform with $p=\frac{2m}{\delta}$.
  The coefficients sum to $1$, i.e. $\sum_{i}\alpha_{i}+\gamma'=1$
  and the number of coefficients is finite. The permutation $t'$ is
  sampled from an arbitrary (but normalised) distribution over permutations
  and $\gamma'\le\gamma$.
\end{prop}

\branchcolor{purple}{If we view $g(u)$ as the output of the first quantum part of our
$\CQ d$ circuit, and $u$ as the oracle of interest (details are
in the next section), it is suggestive that $u|g(u)$ will be the
oracle for the second quantum part of $\CQ d$. We can use the sampling
argument above and re-use our analysis because $\mathbb{F}$ and $\mathbb{F}^{(p,\delta)}$
have very similar statistical properties. However, it is unclear how
to use the sampling argument thereafter as the basic sampling argument
seems to only apply to $\mathbb{F}$ (and not to $\mathbb{F}^{(p,\delta)}$).

To state the more general version of the sampling argument, we need
to define a \emph{$(p,\delta)$ non-$\beta$-uniform distribution}
$\mathbb{F}^{(p,\delta)|\beta}$. Just as we defined $\mathbb{F}^{(p,\delta)}$
using $\mathbb{F}$, we can define $\mathbb{F}^{(p,\delta)|\beta}$
using $\mathbb{F}^{|\beta}$, i.e. $\mathbb{F}^{(p,\delta)|\beta}$
is a distribution which is ``$\delta$ close to'' the $\beta$-uniform
distribution $\mathbb{F}^{|\beta}$, with at most $p$ many paths
fixed. It remains to define $\mathbb{F}^{|\beta}$. In this case,
$\beta:=\{(x_{i},y_{i})\}_{i}$ simply specifies an explicit set of
paths contained in the uniform distribution $\mathbb{F}$. Note that
these paths are distinct from those associated with $p$. Why do we
introduce $\beta$ when $p$ was already present? The parameter $p$
simply says there \emph{exist} at most $p$ paths which are fixed
while $\beta$ \emph{explicitly} fixes certain paths. This becomes
useful in stating the (general) sampling argument.

Suppose we start with $t\sim\mathbb{F}^{\delta'|\beta}$ (i.e. a distribution
which is ``$\delta'$ close to'' $\beta$-uniform) and are given
some advice $h(t)$ which happens to be $r$ with probability at least
$2^{-m}$. Then the distribution $\mathbb{F}^{\delta'|\beta}$ conditioned
on $r$ is, roughly speaking, a convex combination\footnote{Again, neglecting a component with weight at most $2^{-m}$.}
of $\mathbb{F}^{(p,\delta+\delta')|\beta}$ distributions where the
number of paths fixed is at most $p=2m/\delta$ and $\delta$ again
is a free parameter. Using the previous shorthand, we have
\[
  \mathbb{F}^{\delta'|\beta}|r\equiv{\rm conv}(\mathbb{F}^{(p,\delta+\delta')|\beta}).
\]
The formal statement is as follows.}

\begin{prop}[$\mathbb{F}^{\delta'|\beta}|r'={\rm conv}(\mathbb{F}^{(p,\delta+\delta')|\beta})$]
  \label{prop:composableP_Delta_non_beta_uniform}Let $t\sim\mathbb{F}^{\delta'|\beta}(N)$
  be sampled from a $\delta'$ non-$\beta$-uniform distribution with
  $N=2^{n}$. Fix any $\delta>0$ and let $\gamma=2^{-m}$ be some function
  of $n$. Let $s\sim\mathbb{F}^{\delta'|\beta}|r$, i.e. $s=t|(h(t)=r)$
  and suppose $\Pr[h(t)=r]\ge\gamma$ where $h$ is an arbitrary function
  and $r$ some string in its range. Then $s$ is ``$\gamma$-close''
  to a convex combination of finitely many $(p,\delta+\delta')$ non-$\beta$-uniform
  distributions, i.e.
  \[
    s\equiv\sum_{i}\alpha_{i}s_{i}+\gamma's'
  \]
  where $s_{i}\sim\mathbb{F}_{i}^{p,\delta+\delta'|\beta}$ with $p=2m/\delta$.
  The permutation $s'$ may have an arbitrary distribution (over $\Omega(2^{n})$)
  but $\gamma'\le\gamma$.
\end{prop}

\branchcolor{purple}{How does this solve the limitation of the basic sampling method---which
was, how do we apply the sampling argument to $\mathbb{F}^{(p',\delta')}$?
Using the observation that $\mathbb{F}^{(p',\delta')}=\mathbb{F}^{\delta'|\beta}$
for some $\beta$ which fixes at most $p'$ paths, it is not hard
to see that the sampling argument yields
\[
  \mathbb{F}^{(p',\delta')}|r\equiv{\rm conv}(\mathbb{F}^{(p+p',\delta'+\delta)}),
\]
and in particular, if the procedure is successively applied $\tilde{n}\le\poly$
times (starting with $\mathbb{F}$), the convex combination would
be over distributions of the form $\mathbb{F}^{(\tilde{n}p,\tilde{n}\delta)}$.
How the parameters are chosen is discussed later.

The proofs of these statements do not rely on any special property
of the distribution $\mathbb{F}$ nor do they depend on the fact that
we were considering permutations. Any object for which we can describe
a ``reasonable'' notion of ``parts'' admits such a sampling argument.
We don't attempt to formalise what we mean by ``reasonable''---we
simply construct such a notation for our oracle and inspect that the
properties required in the proof are satisfied.}

\subsubsection{Definitions and Notation | Sampling argument for Injective Shufflers\label{subsec:SamplingArgumentInjShuffler}}

\branchcolor{purple}{%

As we did for permutations, to describe the sampling argument, we
change our viewpoint and look at probabilities associated with ``paths''
in $\mathcal{L}=(H_{0},\dots H_{d})$ instead of looking at probabilities
associated with the individual outcomes of $H_{i}$s. By a ``path'',
we mean tuples of the form $(x_{0},x_{1}\dots)$ such that $x_{i}=H_{i-1}(x_{i-1})$
for all $i$. %

This viewpoint is inadequate for capturing the probabilistic behaviour
of $\mathcal{L}$ due to two reasons (which are not hard to rectify).
\emph{First}, since $H_{0}:\Sigma\to\Sigma^{d'}$, it is clear that
at least $\left|\Sigma^{d'-1}\right|$ many points will never be contained
in any ``path'' as described above. Therefore the behaviour of most
points in $H_{i}$ (for $i\in\{1\dots d\}$) will not be captured
by the ``paths'' viewpoint. \emph{Second}, even though $H_{i}$
maps $\Sigma^{d'}\to\Sigma^{d'}$ for $i\in\{1,\dots d-1\}$, $H_{i}$
may not be injective and therefore the paths might collide, which
again would mean the behaviour of many points would not be captured
by the ``paths'' viewpoint.

To rectify the \emph{second} issue, we can run \Algref{baseSets}
and condition on the event $E$, i.e. that the algorithm does not
abort. Since in our proofs, we only care about the behaviour of $\mathcal{L}$
on $\bar{S}_{0}=(S_{01},\dots S_{0d})$, it suffices to restrict our
attention to $\bar{S}_{0}$. By construction (of \Algref{baseSets}),
$\mathcal{L}|E$ behaves as a permutation on $\bar{S}_{0}$. Therefore
no ``path'' inside $\bar{S}_{0}$ collides. To rectify the \emph{first}
issue, we consider two kinds of paths---Type 0 paths and Type 1 paths.\footnote{The 0 and 1 represent where the first non-$\varstar$ component sits.}
A \emph{Type 0 path} is what we described earlier: a tuple of the
form $(x_{0},x_{1}\dots)$ such that $x_{i}=H_{i-1}(x_{i-1})$ for
all $i$. A \emph{Type 1 path} is a tuple of the form $(\varstar,x_{1},x_{2}\dots)$
such that $x_{1}\notin H_{0}(\Sigma)$ (i.e. $\nexists x_{0}$ st
$H_{0}(x_{0})=x_{1}$) and $x_{i}=H_{i-1}(x_{i-1})$ for all $i\in\{2,3\dots\}$.

Observe that, restricted to $\bar{S}_{0}$ and conditioned on $E$,
we have the following equivalence: given $\Pr[H_{i}(x)=x']$ for all
$i$, $x$ and $x'$, one can compute the probability associated with
both types of paths and conversely, given probabilities associated
with the paths, one can compute $\Pr[H_{i}(x)=x']$ for all $i$,
$x$ and $x'$.

To simplify the notation, we define the \emph{injective shuffler}.
Fix sets $S_{0i}\subseteq\Sigma^{d'}$ of size $|\Sigma^{d+2}|$ for
all $i\in\{1,\dots d\}$. Let $H_{0}':\Sigma\to S_{01}$, $H_{i}':S_{0i}\to S_{0,i+1}$
for all $i\in\{1,\dots d-1\}$ be injective functions and let $H_{d}':S_{0d}\to\{0,1\}^{n}\cup\{\perp\}$
(which may not be injective) such that $H'_{d}$ outputs $\perp$
for all paths originating from $\Sigma$ (and no other).\footnote{i.e. $H'_{d}(x_{d})=\perp$ iff $(x_{0},x_{1},\dots x_{d},x_{d+1})$
is a Type 0 path (therefore $x_{d+1}=\perp$)} We define the \emph{injective shuffler}, $\mathcal{K}$ as $(H_{0}',\dots H_{d}')$.
Think of $\mathcal{K}$ as a simpler way to denote the relevant object
associated $\mathcal{L}|E$ (with $\bar{S}_{0}$ being the output
of \Algref{baseSets}). What do we mean by the relevant object---as
we saw in the $\QNC_{d}$ and $\QC d$ analysis, it helps to use shadow
oracles in the analysis which never reveal any information\footnote{Except for polynomially possibly many paths exposed by classical queries;
  we handle these shortly.} about the values taken by $H_{d}(\dots(H_{0}(\ell))\dots)$ for any
$\ell\in\Sigma$. We capture this limitation in $\mathcal{K}$ by
ensuring $H'_{d}$ outputs $\perp$ for these queries.

To state the sampling argument for the injective shuffler, we define
$(p,\delta)$ non-$\beta$-uniform distributions for the injective
shuffler (analogous to the way we defined them for permutations).
However, this time we also give formal definitions (it may help to
look at the analogous formal definitions for permutations first, as
detailed in \Secref{Appendix_Sampling-argument-for-Permutations}
of the Appendix). We begin with the uniform distribution---it is
simply a distribution which assigns equal probabilities to all the
possible injective shufflers, given the sets $(S_{0i})_{i}$. As for
$\beta$-uniform distributions, we first need to define the ``paths'',
$\beta$. Here, $\beta$ will again be a set of ``non-colliding paths''
but formalising this requires some care (discussed later). Then a
$\beta$-uniform distribution is the same as the uniform distribution
except that the paths in $\beta$ are fixed.

}

\branchcolor{purple}{We first define ``base sets'' for convenience as they are repeatedly
used in this section. Using these, we define (valid) injective shuffler
wrt base sets. Then, one can trivially define $\mathbb{F}_{{\rm shuff}}$,
as the uniform distribution over injective shufflers. }
\begin{defn}[Base sets]
  \label{def:baseSets} Let $\Sigma$ and $d'$ be as in \Defref{dCodeHashingProblem}.
  For each $i\in\{1,\dots d\}$, suppose $S_{0i}\subseteq\Sigma^{d'}$
  are subsets of size $|\Sigma^{d+2}|$ then we call $\bar{S}_{0}:=(S_{01},\dots S_{0d})$
  \emph{base sets}.
\end{defn}

\begin{defn}[(valid) Injective Shuffler wrt base sets $\bar{S}_{0}$.]
  \label{def:InjectiveShufflerwrtS} Let $\Sigma,d'$ be as in \Defref{dCodeHashingProblem}
  and let $\bar{S}_{0}=(S_{01}\dots S_{0d})$ be a base set (see \Defref{baseSets}).
  Then a \emph{(valid) Injective Shuffler} wrt $\bar{S}_{0}$ is a sequence
  of functions $(H_{0}'\dots H'_{d})$ where $H'_{0}:\Sigma\to S_{01}$,
  $H'_{i}:S_{0i}\to S_{0,i+1}$ for all $i\in\{1\dots d-1\}$ are injective
  functions and $H'_{d}:S_{0d}\to\{0,1\}^{n}\cup\{\perp\}$ is an arbitrary
  function satisfying the following constraint:
  \[
    H'_{d}(x)\in\begin{cases}
      \{\perp\}   & x\in H_{d-1}'(\dots H_{0}'(\Sigma)\dots)                   \\
      \{0,1\}^{n} & x\in S_{0d}\backslash H_{d-1}'(\dots H_{0}'(\Sigma)\dots).
    \end{cases}
  \]
\end{defn}

\branchcolor{purple}{The conditions on $H'_{0},\dots H'_{d-1}$ are straightforward. The
conditions on $H'_{d}$ ensures that all paths originating from $\Sigma$
(i.e. Type 0 paths) output $\perp$ which, as we remarked earlier,
ensures our definition can be used with shadow oracles.}
\begin{defn}[$\mathbb{F}_{{\rm inj}}$---Uniform Distribution over Injective Shufflers]
  Let $\bar{S}_{0}$ be base sets (see \Defref{baseSets}) for $\CH d$
  (see \Defref{dCodeHashingProblem}). Then $\mathbb{F}_{{\rm inj}}$
  is the uniform distribution over all injective shufflers wrt $\bar{S}_{0}$
  (see \Defref{InjectiveShufflerwrtS}).
\end{defn}

\branchcolor{purple}{So far everything was intuitive. To proceed, we would need to condition
  these injective shufflers. The conditioning will be in terms of existence
  of certain non-colliding paths, $\beta$, in the injective shuffler.
  There are two subtleties when we do this, as we alluded to. The \emph{first}
  is that there are Type 0 and Type 1 paths and thus one must be careful
  in how collisions are defined. The \emph{second} is that an injective
  shuffler is defined to yield $\perp$ on paths originating from $\Sigma$
  (i.e. on Type 0 paths) and yet, (as we shall see) we would like to
  be able to condition on polynomially many paths originating from $\Sigma$
  which yield non-$\perp$ responses. This corresponds to (excluding
  from the shadow oracles) the paths queried by the classical algorithm
  because the classical algorithm will have access to $\mathcal{L}$
  (and not its shadow). These concerns are addressed in the following
  definition.}
\begin{defn}[(valid) paths, $\beta$, wrt $\bar{S}_{0}$. $X_{i}(\beta)$]
  \label{def:valid_beta_paths} Let $\bar{S}_{0}$ be base sets (see
  \Defref{baseSets}) for $\CH d$ (see \Defref{dCodeHashingProblem})
  and let $\beta=\{(x_{j,0},\dots x_{j,d+1})\}_{j\in\{1\dots|\beta|\}}$
  be a set of tuples with $d+2$ elements. We say $\beta$ specifies
  (valid) paths wrt $\bar{S}_{0}$ if it satisfies the following:
  \begin{enumerate}
    \item (domain validation) For each $j\in\{1\dots|\beta|\}$, it holds that
          (a) $x_{j,0}\in\Sigma\cup\{\varstar\}$; (b) for all $i\in\{1\dots d\}$,
          $x_{j,i}\in S_{0i}$ and (c) $x_{j,d+1}\in\{0,1\}^{n}$ (but cannot
          output $\perp$),
    \item (no collisions) for each distinct pair $j,j'\in\{1\dots|\beta|\}$,
          $x_{j,i}\neq x_{j',i}$, for all $i\in\{1\dots d\}$ and
    \item (handling Type 0 paths) for any distinct pair $j,j'\in\{1\dots|\beta|\}$,
          $x_{j,0}=x_{j',0}\iff x_{j,0}=x_{j',0}=\varstar$.
  \end{enumerate}
  Notation: For (valid) paths $\beta$, define (for $i\in\{0,\dots d+1\}$
  \begin{itemize}
    \item $X_{i}(\beta):=\{x_{j,i}\}_{j\in\{1\dots|\beta|\}}$, using this, define $X_{i:i'}(\beta)=(X_{i}(\beta),\dots X_{j}(\beta))$ for $i\le i'$ and let $X(\beta)=X_{1:d}(\beta)$,
    \item $X_{i}^{(0)}(\beta):=\{x_{j,i}\}_{j:x_{j0}\neq\varstar}$, and
    \item $X_{i}^{(1)}(\beta):=\{x_{j,i}\}_{j:x_{j0}=\varstar}$.
  \end{itemize}
\end{defn}

\branchcolor{purple}{The first condition simply requires that the paths are inside $\bar{S}_{0}$.
  The second condition ensures that the paths don't collide but excluding
  the first component. The third condition ensures that the only way
  the first component can ``collide'' is if the path is Type 1; Type
  0 paths cannot have the same first component. With (valid) paths $\beta$
  defined, we can define a (valid) injective shuffler conditioned on
  $\beta$ and the associated uniform distribution.}

\begin{defn}[(valid) Injective Shuffler conditioned on $\beta$ wrt base sets $\bar{S}_{0}$]
  \label{def:validInjectiveShufflerConditionedOnBeta}Let $\Sigma,d'$
  be as in \Defref{dCodeHashingProblem}, let $\bar{S}_{0}=(S_{01}\dots S_{0d})$
  be base sets (see \Defref{baseSets}) and let $\beta=:\{(x_{j,0},\dots x_{j,d+1})\}_{j\in\{1\dots|\beta|\}}$
  denote (valid) paths wrt $\bar{S}_{0}$ (see \Defref{valid_beta_paths}).
  Then, a (valid) Injective Shuffler conditioned on $\beta$ wrt $\bar{S}_{0}$
  is a sequence of functions $(H'_{0},\dots H'_{d})$ where $H'_{0}:\Sigma\to S_{01}$,
  $H'_{i}:S_{0i}\to S_{0,i+1}$ for all $i\in\{1,\dots d-1\}$ are injective
  functions and $H'_{d}:S_{0d}\to\{0,1\}^{n}\cup\{\perp\}$ is an arbitrary
  function which satisfy the following constraints:
  \begin{itemize}
    \item $H'_{0}$: it holds that $H'_{0}(x_{j0})=x_{j1}$ for all $j\in\{1\dots|\beta|\}$
          such that $x_{j0}\neq\varstar$ and $H'_{0}(\Sigma)\cap X_{1}^{(1)}(\beta)=\emptyset$
          (see \Defref{valid_beta_paths})
    \item $H'_{i}$: it holds that $H'_{i}(x_{j,i})=x_{j,i+1}$ for all $i\in\{1\dots d-1\}$
          and $j\in\{1\dots|\beta|\}$
    \item $H'_{d}$: it holds that
          \begin{enumerate}
            \item $H'_{d}(x_{j,d})=x_{j,d+1}$ for all $j\in\{1\dots|\beta|\}$
            \item $H'_{d}(x)=\perp$ for all $x\in H_{d-1}'(\dots H_{0}'(\Sigma)\dots)\backslash X_{d}(\beta)$
            \item $H'_{d}(x)\in\{0,1\}^{n}$ otherwise, i.e. for all $x\in S_{0d}\backslash\left(H_{d-1}'(\dots H_{0}'(\Sigma)\dots)\cup X_{d}(\beta)\right)$.
          \end{enumerate}
  \end{itemize}
\end{defn}

\branchcolor{purple}{The requirements on $H'_{1},\dots H'_{d-1}$ are quite clear. On $H'_{0}$,
the first condition is enforcing consistency with Type 0 paths and
the second one is enforcing that none of the Type 1 paths could possibly
have originated from\footnote{The reason is that that this avoids double counting;
  otherwise a Type 1 path could be treated as a partially specified
  Type 0 path and our sampling argument is not a priori robust to these.} $\Sigma$. For $H'_{d}$, we enforce that it is consistent with the

paths in $\beta$ and that it outputs $\perp$ for all remaining paths
originating in $\Sigma$ (Type 1 paths) while for all other paths,
it outputs non-$\perp$. We can finally define the uniform distribution
over injective shufflers conditioned on $\beta$. }

\begin{defn}[$\mathbb{F}_{{\rm inj}}^{|\beta}$---$\beta$-uniform distribution
    over injective shufflers]
  Let $\bar{S}_{0}$ be base sets (see \Defref{baseSets}) for $\CH d$
  (see \Defref{dCodeHashingProblem}), and let $\beta$ denote a (valid)
  set of paths wrt $\bar{S}_{0}$ (see \Defref{valid_beta_paths}).
  Then, $\mathbb{F}_{{\rm inj}}^{|\beta}$ is the uniform distribution
  over all (valid) injective shufflers conditioned on $\beta$ wrt $\bar{S}_{0}$
  (see \Defref{validInjectiveShufflerConditionedOnBeta}).

  \branchcolor{purple}{We can now introduce some notation for describing paths of injective
    shufflers. These paths are slightly different from (valid) paths $\beta$
    wrt $\bar{S}_{0}$ (see \Defref{valid_beta_paths})---these paths
    must assign $\perp$ to paths originating from $\Sigma$ (Type 0 paths)
    to any injective shuffler.\footnote{If the injective shuffler is conditioned on $\beta$, then the statement
      holds excluding the Type 1 paths specified by $\beta$.} This is required to stay consistent with the definition of injective
    shufflers.

    We use these paths to define the parts notation explicitly. These
    in turn, would allow us to easily obtain the analogue of \Propref{composableP_Delta_non_beta_uniform}
    for injective shufflers.\footnote{Notation: We are using both $\mathcal{K}$ and $\Xi$ to refer to injective shufflers.}

    }

  \global\long\def\func#1{\mathsf{func}_{#1}}%
  \global\long\def\cfunc#1{\mathsf{cfunc}_{#1}}%
\end{defn}

\begin{notation}
  \label{nota:cfuncs_funcs_paths_forTheInjectiveShuffler}Let $\Xi$
  be an injective shuffler (possibly conditioned on paths $\beta$)
  wrt base sets $\bar{S}_{0}$ (see \Defref{validInjectiveShufflerConditionedOnBeta}).
  Denote by
  \begin{itemize}
    \item $\func{i,\Xi}$ the function $H_{i}'$ where $(H'_{0},\dots H'_{d}):=\Xi$
    \item $\cfunc{i:j,\Xi}$ the function $H'_{j}(\dots H'_{i}(\cdot)\dots)$
          where $H'_{i}$ is as above for $i,j\in\{0\dots d\}$ satisfying $i\le j$.
    \item $\paths(\Xi)$ the set of all tuples $(x_{0},x_{1}\dots x_{d},x_{d+1})$
          where $x_{0}\in\Sigma\cup\{\varstar\}$, $x_{1}\in S_{01},\dots,x_{d}\in S_{0d}$
          and $x_{d+1}\in\{0,1\}^{n}\cup\{\perp\}$ satisfy the following
          \begin{itemize}
            \item for $i\in\{1,\dots d\}$, it holds that $x_{i+1}=\func{i,\Xi}(x_{i})$
            \item if $x_{0}=\varstar$, it holds that $x_{1}\notin\func{0,\Xi}(\Sigma)$
            \item if $x_{0}\in\Sigma$, it holds that $x_{1}=\func{0,\Xi}(x_{0})$
          \end{itemize}
  \end{itemize}
\end{notation}

\branchcolor{purple}{As stated, we now describe the parts notation for injective shufflers.}
\begin{notation}
  Suppose $\beta$ is a (valid) path wrt base sets $\bar{S}_{0}$. Let
  $\Xi$ be an arbitrary injective shuffler conditioned on $\beta$
  wrt base sets $\bar{S}_{0}$ (see \Defref{validInjectiveShufflerConditionedOnBeta}).
  \begin{itemize}
    \item \emph{Parts:} Any set $S$ is a \emph{part} if it holds that $S\subseteq\paths(\Xi)$
          for some $\Xi$.
          \begin{itemize}
            \item Denote by $\Omega_{\parts}^{\beta}$ the set of all such ``parts''.
            \item Call two parts $S,S'\in\Omega_{\parts}^{\beta}$ \emph{distinct} if
                  $S\cap S'=\emptyset$ and $S\cup S'\subseteq\paths(\Xi)$ for some
                  $\Xi$.
            \item Denote by $\Omega_{\parts}^{\beta}(S)$ the set of all parts $S'\in\Omega_{\parts}^{\beta}$
                  distinct from $S$.
          \end{itemize}
    \item Suppose $\Xi$ is a random variable.
          \begin{itemize}
            \item \emph{Probability of a part $S$:} The probability that $\Xi$ maps
                  paths as described in $S$ is denoted by $\Pr[S\subseteq\paths(\Xi)]$.
            \item \emph{Conditioning $\Xi$ on a part:} We use the notation $\Xi_{S}$
                  to denote the random variable $\Xi$ conditioned on the event $S\subseteq\paths(\Xi)$.
          \end{itemize}
  \end{itemize}
  \branchcolor{purple}{Before we use these definitions for stating and proving the sampling
    argument for injective shufflers, we use them to define $(p,\delta)$
    non-$\beta$-uniform distributions for injective shufflers. }
\end{notation}

\begin{defn}[$\mathbb{G}^{(p,\delta)|\beta}$---a $(p,\delta)$ non-$\mathbb{G}^{|\beta}$
    distribution]
  \label{def:p_delta_G} Suppose $\beta$ is a valid path wrt base
  sets $\bar{S}_{0}$ (see \Defref{valid_beta_paths}). Let $\Xi\sim\mathbb{G}^{|\beta}$
  be a an injective shuffler conditioned on $\beta$ wrt $\bar{S}_{0}$
  (see \Defref{validInjectiveShufflerConditionedOnBeta}), sampled from
  some arbitrary distribution $\mathbb{G}^{|\beta}$. Let $p,\delta\ge0$.
  Then, we say $\Xi'\sim\mathbb{G}^{(p,\delta)|\beta}$ is sampled from
  a $(p,\delta)$ non-$\mathbb{G}^{|\beta}$ distribution\footnote{(over injective shuffler conditioned on $\beta$)}
  if for all parts $S\in\Omega_{\parts}^{\beta}(S')$ it holds that
  \[
    \Pr[S\subseteq\paths(\Xi')|S'\subseteq\paths(\Xi')]\le2^{|S|\delta}\Pr[S\subseteq\paths(\Xi)|S'\subseteq\paths(\Xi)]
  \]
  for some part $S'\in\Omega_{\parts}^{\beta}$ of size $|S'|\le p$.
\end{defn}

\branchcolor{purple}{Using different distributions in place of $\mathbb{G}^{|\beta}$,
  one can obtain the following which will be relevant to the sampling
  argument.}
\begin{notation}
  \label{nota:F_inj_with_some_qualifier_like_p_delta_etc}The distribution
  specified in \Defref{p_delta_G}
  \begin{itemize}
    \item with $\mathbb{F}_{{\rm inj}}^{|\beta}\leftarrow\mathbb{G}^{|\beta}$,
          $0\leftarrow p$ is termed $\mathbb{F}_{{\rm inj}}^{\delta|\beta}$,
    \item with $\mathbb{F}_{{\rm inj}}^{|\beta}\leftarrow\mathbb{G}^{|\beta}$
          is termed $\mathbb{F}_{{\rm inj}}^{(p,\delta)|\beta}$, and
    \item with $\mathbb{F}_{{\rm inj}}^{\delta'|\beta}\leftarrow\mathbb{G}^{|\beta}$
          is termed $\mathbb{F}_{{\rm inj}}^{(p,\delta+\delta')|\beta}$.
  \end{itemize}
  We call $\mathbb{F}_{{\rm inj}}^{(p,\delta)|\beta}$ a $(p,\delta)$
  non-$\beta$-uniform distribution.
\end{notation}

\subsubsection{Statement | Sampling argument for Injective Shufflers}

\branchcolor{purple}{We now state the sampling argument and prove its basic variant to
  convey the idea, deferring the general proof to the appendix. }
\begin{prop}[$\mathbb{F}_{{\rm inj}}^{\delta'|\beta}|r'\equiv{\rm conv}(\mathbb{F}_{{\rm inj}}^{(p,\delta+\delta')|\beta})$]
  \label{prop:main_delta_non_uniform_inj_shuffler}Suppose $\beta$
  is a valid path wrt base sets $\bar{S}_{0}$ (see \Defref{valid_beta_paths}).
  Let $\Xi^{t}\sim\mathbb{F}_{{\rm inj}}^{\delta'|\beta}$ be sampled
  from a $\delta'$ non-$\beta$-uniform distribution. Fix any $\delta>0$
  and let $\gamma=2^{-m}$ be some function of $n$ (where $n$ is as
  in \Defref{dCodeHashingProblem}). Let $\Xi^{s}\sim\mathbb{F}_{{\rm inj}}^{\delta'|\beta}|r$
  , i.e. $\Xi^{s}:=\Xi^{t}|(h(\Xi^{t})=r)$ and suppose that $\Pr[h(\Xi^{t})=r]\ge\gamma$
  where $h$ is an arbitrary function and $r$ some string in its range.
  Then $\Xi^{s}$ is ``$\gamma$-close'' to a convex combination of
  finitely many injective shufflers sampled from $(p,\delta+\delta')$
  non-$\beta$-uniform distributions, i.e.
  \[
    \Xi^{s}\equiv\sum_{i}\alpha_{i}\Xi_{i}^{s}+\gamma'\Xi^{s\prime}
  \]
  where there are finitely many $\alpha_{i}$, $\sum_{i}\alpha_{i}+\gamma'=1$,
  $\Xi_{i}^{s}\sim\mathbb{F}_{{\rm inj},i}^{(p,\delta+\delta')|\beta}$
  with\footnote{(the $i$ in $\mathbb{F}_{{\rm inj},i}^{(p,\delta+\delta')|\beta}$,
  indicates that each $\Xi_{i}^{s}$ can come from a different distribution
  which is still $(p,\delta+\delta')$ non-$\beta$-uniform; e.g. they
  may be fixing different paths but there are at most $p$ such paths)} $p=2m/\delta$. The injective shuffler (conditioned on $\beta$),
  $\Xi^{s\prime}$, may have arisen from an arbitrary distribution,
  however, $\gamma'\le\gamma$.
\end{prop}

\subsubsection{Properties of the \texorpdfstring{$\delta$ non-$\beta$-uniform injective shuffler}{\textbackslash delta non-\textbackslash beta-uniform injective shuffler}}

\branchcolor{purple}{Suppose $\mathcal{L}'\sim\mathbb{F}_{{\rm inj}}$ and $\Xi\sim\mathbb{F}_{{\rm inj}}^{\delta}$.
It would be useful to go back from the paths perspective to functions
and see how their behaviour is related, i.e. we relate the behaviour
of $\func{i,\mathcal{L}'}$ to that of $\func{i,\Xi}$ (or more generally,
$\cfunc{i:j,\mathcal{L}'}$ to that of $\cfunc{i:j,\Xi}$).}
\begin{claim}
  \label{claim:cfunc_relation_with_delta}Suppose $\Xi\sim\mathbb{F}_{{\rm inj}}^{|\beta}$
  and $\mathcal{K}\sim\mathbb{F}_{{\rm inj}}^{\delta|\beta}$ be injective
  shufflers wrt base sets $\bar{S}_{0}$ (see \Defref{InjectiveShufflerwrtS}).
  Then, for all $x_{i}\in S_{0,i}$ and $x_{i+1}\in S_{0,i+1}$, it
  holds that
  \[
    \Pr[\cfunc{i:j,\mathcal{K}}(x_{i})=x_{j+1}]\le2^{\delta}\Pr[\cfunc{i:j,\Xi}(x_{i})=x_{j+1}]
  \]
  which in particular entails
  \[
    \Pr[\func{i,\mathcal{K}}(x_{i})=x_{i+1}]\le2^{\delta}\Pr[\func{i,\Xi}(x_{i})=x_{i+1}].
  \]
\end{claim}

\branchcolor{blue}{\begin{proof}
    First consider $\beta=\emptyset$ and $i=j$. In the paths notation
    $\cfunc{i:i}$ corresponds to $\func i$, so we have
    \begin{align*}
      \Pr[\func{i,\mathcal{K}}(x_{i}')=x_{i+1}'] & =\sum_{\substack{x_{j}\in S_{0,j},\ j\in\{1\dots d\}\backslash\{i,i+1\}                                                                                                    \\
      x_{i}=x_{i}'                                                                                                                                                                                                            \\
      x_{i+1}=x_{i+1}'                                                                                                                                                                                                        \\
      x_{0}\in\Sigma\cup\{\varstar\}
      }
      }\Pr[(x_{0},x_{1},\dots x_{d})\in\paths(\mathcal{K})]                                                                                                                                                                   \\
                                                 & \le\sum_{\dots}2^{\delta}\Pr[(x_{0},\dots x_{d})\in\paths(\Xi)]         & \text{\prettyref{nota:F_inj_with_some_qualifier_like_p_delta_etc} and \prettyref{def:p_delta_G}} \\
                                                 & =2^{\delta}\Pr[\func{i,\Xi}(x_{i}')=x_{i+1}']
    \end{align*}
    where the second sum is over the same variables as the first sum.
    For $\beta\neq\emptyset$, the same calculation goes through---some
    of the paths might be assigned zero probability (e.g. if they conflict
    with the values assigned by paths in $\beta$). Similarly for $j>i$.
  \end{proof}
}

\subsection{\texorpdfstring{$\protect\classCQ{d}$ exclusion}{BPP\^{}QNC\_d exclusion}}

\branchcolor{purple}{

The analysis would be very similar to the $\QNC_{d}$ case, once we use the sampling argument is invoked.
We would construct shadows for $\mathcal{L}$ directly as before, except that certain paths $\beta$ would be fixed
The injective shuffler will show up at two places. 
\begin{itemize}
          \item We will state the probability of finding in terms of a distribution
                over injective shufflers.
          \item When we apply the sampling argument, we would only focus on how the
                distribution restricted to $\bar{S}_{0}$ changes, i.e. over injective
                shufflers. %
        \end{itemize}
  How many times are the algorithms for generating $\bar{S}$ called?
        \begin{itemize}
          \item The base algorithm, for generating $\bar{S}_{0}$, is called once,
                at the very beginning of the analysis.
          \item The other algorithm, for generating $(\bar{S}_{j})_{j\in\{1\dots d\}}$,
                is called after each $\mathcal{C}_{i}$ is applied
        \end{itemize}

}

\subsubsection{Shadow oracles for \texorpdfstring{$\protect\CQ d$}{CQ\_d} hardness}

\branchcolor{purple}{Here, we can state everything in terms of $\mathcal{L}$ and we simply need to add a condition for the event $E$ happening which is meant to denote that \Algref{baseSets} succeeding. %
}

\begin{lyxalgorithm}[Procedure for generating $S_{ij}$, given $\beta$]
  \label{alg:S_ij_given_beta}Let $\mathcal{L}'=(H'_{0}\dots H'_{d})$ and 
  $\Sigma$ be as in \Notaref{LhatAndsuch} and \Defref{dCodeHashingProblem}. Let $S_{i}=H_{i-1}(\dots H{}_{0}(\Sigma)\dots)$
  (as defined in \Algref{baseSets}). \\
  Input:
  \begin{enumerate}
    \item Base sets $\bar{S}_{0}$ (see \Defref{baseSets})
    \item (valid) paths $\beta$ wrt $\bar{S}_{0}$ (see \Defref{valid_beta_paths})
    \item Whether or not event $E$ happened.
  \end{enumerate}
  Output:\\
  If $E$ did not happen, set $S_{ik}=\emptyset$ for all $i,k\in\{1\dots d\}$.
  \\
  Otherwise, for each $i\in\{1\dots d\}$, do the following:
  \begin{enumerate}
    \item Define $S_{ik}:=\emptyset$ for $1\le k < i$. 
    \item Sample, uniformly at random, $S_{ii}\subseteq S_{i-1,i}\backslash X_{i}(\beta)$
          such that $S_{i} \backslash X_i(\beta) \subseteq S_{ii}$ and $|S_{ii}|/|S_{i-1,i}|=1/|\Sigma|$. 
    \item Define $S_{ik}:=H'_{k-1}(\dots H'_{i}(S_{ii})\dots)$ for $i<k\le d$.
  \end{enumerate}
  In both cases, return $\bar{S}_{i}:=(S_{i1},\dots S_{id})$ for each
  $i\in\{1\dots d\}$.

\end{lyxalgorithm}

\subsubsection{Properties of the shadow oracles}

\branchcolor{purple}{
    We would need the analogue of \Claimref{x_in_S_QNC} and \Claimref{x_in_S_QC_d} which in this case turns out to be the following. Note that the probability of interest can be computed by looking at the injective shuffler associated with the oracle.
}
\begin{claim}
  \label{claim:x_in_S_CQ_d}Let $\mathcal{K}\sim\mathbb{F}_{{\rm inj}}^{\delta|\beta}$
  be an injective shuffler conditioned on $\beta$ wrt base sets $\bar{S}_{0}$,
  sampled from a $\delta$ non-$\beta$-uniform distribution (see \Defref{validInjectiveShufflerConditionedOnBeta}
  and \Notaref{F_inj_with_some_qualifier_like_p_delta_etc}) where $|\beta|\le\poly$.
  Suppose $\mathcal{L}'$ is $\mathcal{L}$ conditioned on some variable
  $\tau$ such that $\mathcal{L}'$ restricted to $\bar{S}_{0}$ is
  exactly $\mathcal{K}$. Suppose \Algref{S_ij_given_beta} is run with
  inputs $\bar{S}_{0}$, $\beta$ and the assertion that $E$ happened
  and let its output be $S_{ij}$ for $i,j\in\{1\dots d\}$. Then,
  \[
    \Pr\left[x\in S_{ij}|\check{\mathcal{L}}'\right]\le(2^{\delta}+c)\cdot\poly\cdot\negl
  \]
  where $c$ is some constant (independent of $\delta$, $d$ etc.)
  $\check{\mathcal{L}}'$ is $\mathcal{L}'$ outside $\bar{S}_{i-1}$
  (see \Notaref{LhatAndsuch} with $S^{{\rm out}}\leftarrow\bar{S}_{i-1}$
  and $\mathcal{L}'\leftarrow\mathcal{L}'$) for all $1\le i\le j\le d$
  where the probability is over the randomness in $\mathcal{K}$ (i.e.
  from $\mathbb{F}_{{\rm inj}}^{\delta|\beta}$) and the randomness
  in \Algref{S_ij_given_beta}.

  \branchcolor{blue}{\begin{proof}[Proof sketch.]

      \begin{figure}%
        \centering
        \subfloat[For analysing $\Pr(x \in S_{ii})$.]{
            \includegraphics[width=5cm]{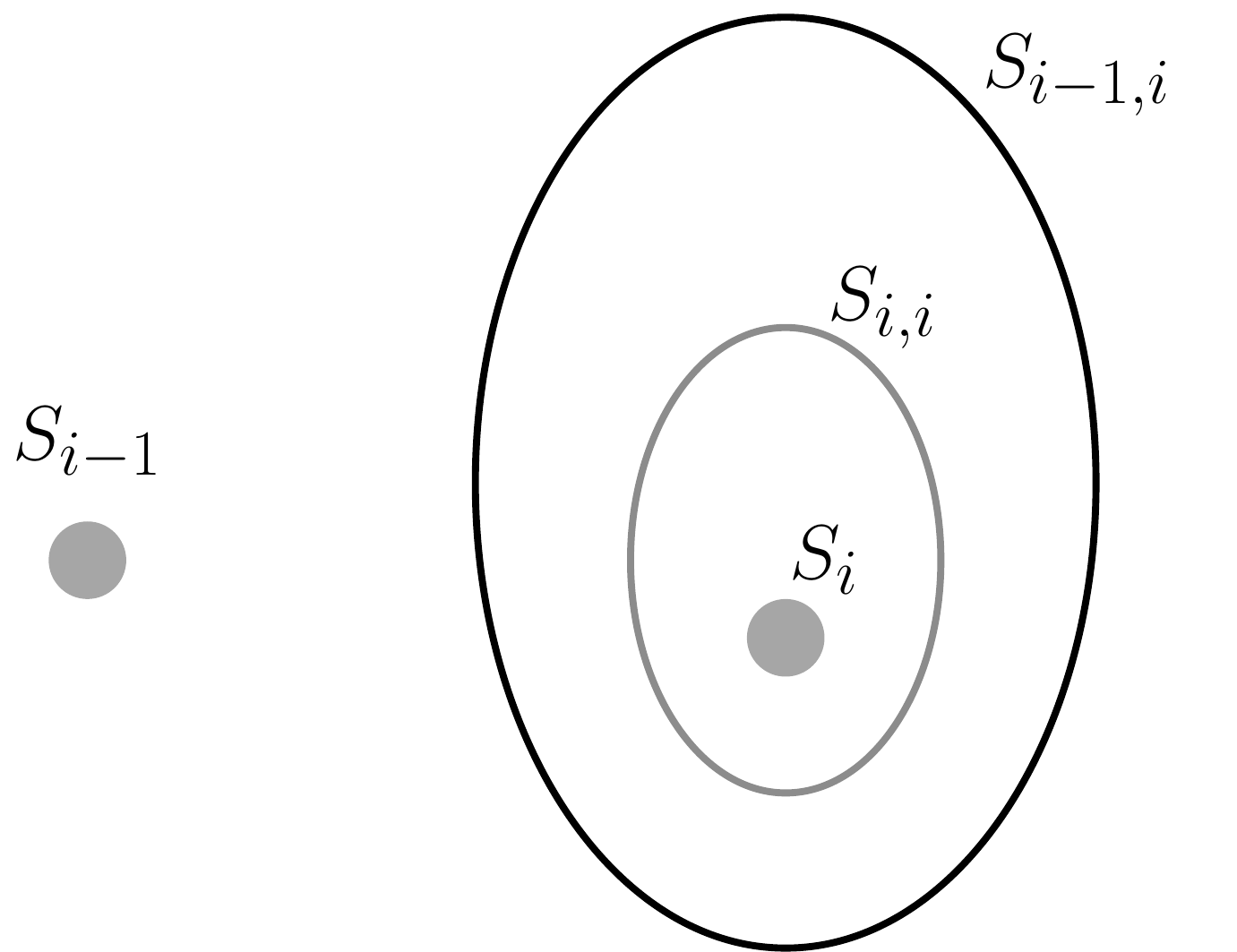}
            \label{fig:paths_delta_1}
        }\qquad\qquad
        \subfloat[For analysing $\Pr(x \in S_{ij})$ for $j>i$]{  
            \includegraphics[width=5cm]{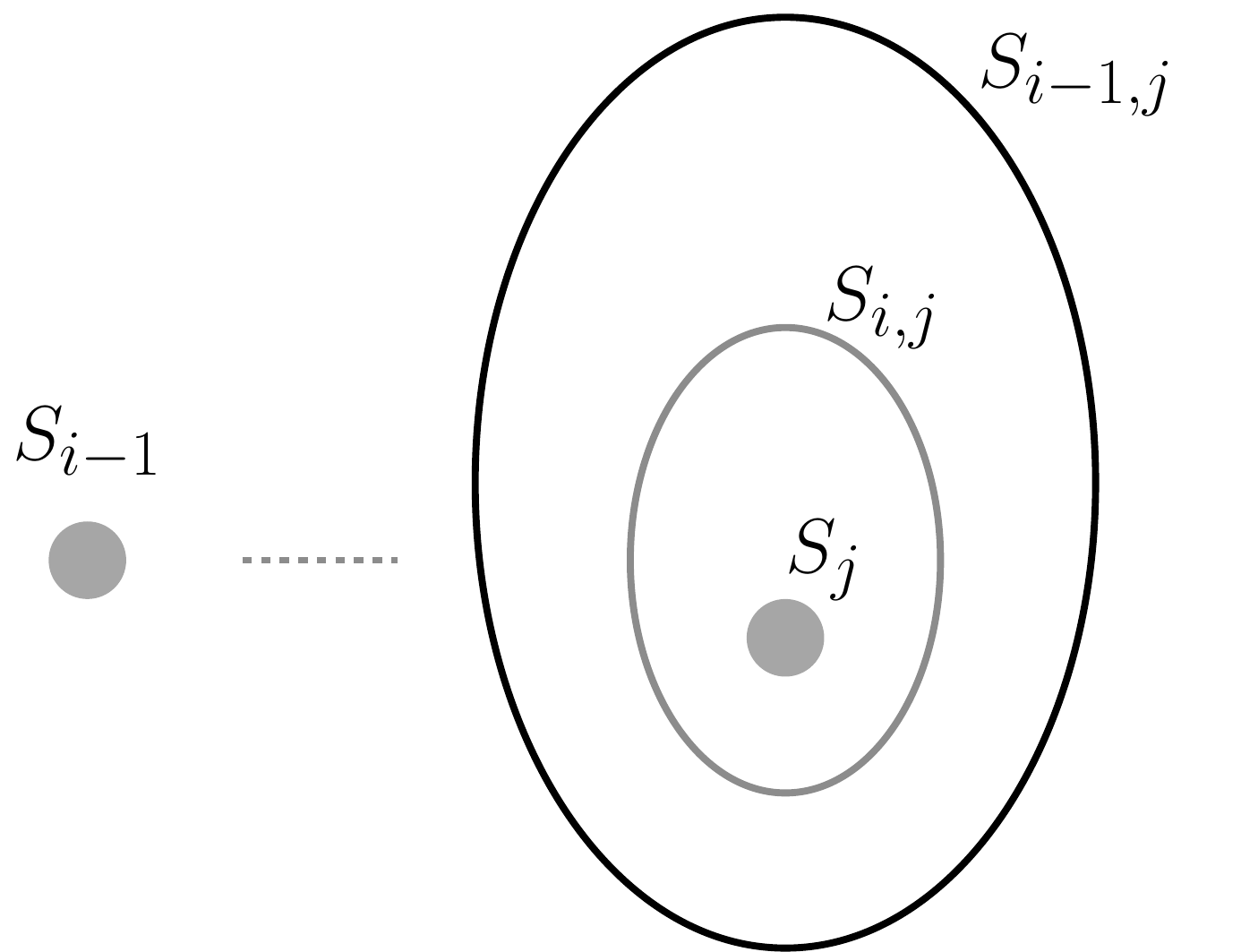}
            \label{fig:paths_delta_2}
        }
        \caption{Visual aid for analysing $\Pr[x\in S_{ij}]$.}
        \label{fig:paths_delta}
      \end{figure}

      Our strategy is to reduce the analysis to the case where the injective
      shuffler is uniformly distributed. We show this for $\beta=\emptyset$
      (the $\beta\neq\emptyset$ case follows by reasoning as we did for
      the proof of \Claimref{x_in_S_QC_d}).

      Consider the \textbf{$k=i$ case} (see \Figref{paths_delta} left).
      Let $S_{i}$ be as in \Algref{S_ij_given_beta}, i.e. $S_{i+1}=\cfunc{0:i,\mathcal{K}}(\Sigma)$
      using \Notaref{cfuncs_funcs_paths_forTheInjectiveShuffler} for $i\in\{0,\dots d\}$.
      We have
      \begin{align*}
        \Pr[x\in S_{ii}|\underbrace{S_{i-1};\bar{S}_{i-1}}_{\mathcal{I}}] & =\Pr[x\in S_{ii}|x\in S_{i},\mathcal{I}]\Pr[x\in S_{i}|\mathcal{I}]+\Pr[x\in S_{ii}|x\notin S_{i},\mathcal{I}]\Pr[x\notin S_{i}|\mathcal{I}] \\
                                                                          & \le\underbrace{\Pr[x\in S_{i}|\mathcal{I}]}_{\mathsf{I}}+\underbrace{\Pr[x\in S_{ii}|x\notin S_{i},\mathcal{I}]}_{\mathsf{II}}
      \end{align*}
      where observe that $\Pr[x\in S_{ii}|x\in S_{i},\mathcal{I}]=1$ because
      by construction, $S_{i}\subseteq S_{ii}$ and where we use the trivial
      bound $\Pr[x\notin S_{i}|\mathcal{I}]\le1$ (which as we shall see
      is almost saturated). \\
      \emph{Bounding Term I.} The first term may be bounded as
      \begin{align}
        \Pr[x\in S_{i}|\mathcal{I}] & \le\sum_{x'\in S_{i-1}}\Pr[\func{i-1,\mathcal{K}}(x')=x]\nonumber                   \\
                                    & \le\sum_{x'\in S_{i-1}}2^{\delta}\Pr[\func{i-1,\Xi}(x')=x]\label{eq:2_delta_pr_i-1} \\
                                    & =2^{\delta}\frac{|S_{i-1}|}{|S_{i-1,i}|}\nonumber
      \end{align}
      where the first inequality is just a union bound, the second follows
      from \Defref{InjectiveShufflerwrtS} where $\Xi\sim\mathbb{F}_{{\rm inj}}$,
      and the third is computed by proceeding as follows.

      For $i>1$, $\func{i-1,\Xi}$ is just a uniformly random permutation
      from $S_{i-1,i-1}\to S_{i-1,i}$ (which have the same size) and we
      are asking for the probability that one of the elements is mapped
      as we like. This is readily computed to be
      \[
        \Pr[\func{i-1,\Xi}(x')=x]=\frac{(Z-1)!}{Z!}=\frac{1}{Z}
      \]
      where $Z=|S_{i-1,i}|$. One can therefore bound \Eqref{2_delta_pr_i-1}
      as $2^{\delta}\sum_{x'\in S_{i-1}}\frac{1}{Z}=2^{\delta}|S_{i-1}|/|S_{i-1,i}|$
      as asserted.

      For $i=1$, the function $\func{i-1,\Xi}$ is a uniformly random injective
      function from $\Sigma\to S_{i-1,0}$. The probability that one element
      maps as we like is given by
      \[
        \Pr[\func{i-1,\Xi}(x')=x]=\frac{\perm{Z-1}{|\Sigma|-1}}{\perm Z{|\Sigma|}}=\frac{1}{Z}
      \]
      as we observed in \Factref{perms_and_combs}. Proceeding as before,
      one can again bound \Eqref{2_delta_pr_i-1} as asserted.\\
      \emph{Bounding Term II.} The second term may be bounded as
      \[
        \Pr[x\in S_{ii}|x\notin S_{i},\mathcal{I}]\le\frac{\left|S_{ii}\right|-|S_{i}|}{|S_{i-1,i}|-|S_{i}|}
      \]
      where we use \Remref{Pr_x_in_X_or_t} (second observation) with $M=|S_{i-1,i}\backslash S_{i}|$
      and $N=|S_{ii}\backslash S_{i}|$. Note that the randomness used in
      this bound comes from that of \Algref{S_ij_given_beta} while in the
      previous step (for ``Bounding Term 1''), we had to use the fact
      that $\mathcal{K}$ is sampled from $\mathbb{F}_{{\rm inj}}^{\delta}$
      and \Defref{InjectiveShufflerwrtS}.

      We therefore get
      \begin{align}
        \Pr[x\in S_{ii}|S_{i-1};\bar{S}_{i-1}] & \le2^{\delta}\frac{\left|S_{i-1}\right|}{\left|S_{i-1,i}\right|}+\frac{|S_{ii}|-|S_{i}|}{|S_{i-1,i}|-|S_{i}|}\nonumber \\
                                               & \le(2^{\delta}+c)\negl\label{eq:x_in_S_ii_2_delta}
      \end{align}
      where $c$ is some constant (independent of $\delta$ etc; see \Subsecref{sillySteps_in_claim_x_in_S_CQ_d}
      for details).

      We now consider the \textbf{$j>i$ case} (see \Figref{paths_delta}
      right). We proceed analogously to the $i=j$ case and see that almost
      nothing changes. In particular, one has
      \begin{align*}
        \Pr[x\in S_{ij}|\underbrace{S_{i-1};\bar{S}_{i-1}}_{\mathcal{I}}] & =\Pr[x\in S_{ij}|x\in S_{j};\mathcal{I}]\Pr[x\in S_{j}|\mathcal{I}]+\Pr[x\in S_{ij}|x\notin S_{j};\mathcal{I}]\Pr[x\notin S_{j};\mathcal{I}] \\
                                                                          & \le\underbrace{\Pr[x\in S_{j}|\mathcal{I}]}_{\mathsf{I}}+\underbrace{\Pr[x\in S_{ij}|x\in S_{j},\mathcal{I}]}_{\mathsf{II}}.
      \end{align*}
      \emph{Bounding Term I.} One can write
      \begin{align*}
        \Pr[x\in S_{j}|\mathcal{I}] & \le\sum_{x'\in S_{i-1}}\Pr[\cfunc{i-1:j-1,\mathcal{K}}(x')=x]   \\
                                    & \le\sum_{x'\in S_{i-1}}2^{\delta}\Pr[\cfunc{i-1:j-1,\Xi}(x')=x] \\
                                    & =2^{\delta}\frac{\left|S_{i-1}\right|}{\left|S_{i-1,j}\right|}
      \end{align*}
      where the second inequality follows from \Defref{InjectiveShufflerwrtS}
      where $\Xi\sim\mathbb{F}_{{\rm inj}}$, and the third is computed
      as in the $j=i$ case. More precisely, for the $i=1$ sub-case (within
      this $j>i$ case), $\cfunc{i-1:j-1}$ is a concatenation of uniformly
      random injective functions, where the first goes from $\Sigma\to S_{i-1,1}$
      and the subsequent ones from $S_{i-1,l}\to S_{i-1,l+1}$. This concatenation
      may be treated as a uniformly random injective function from $\Sigma\to S_{i-1,j}$
      and one can then proceed as in the $i=1$ sub-case (within the $j=i$
      case). As for the $i>1$ sub-case (within this $j>i$ case), $\cfunc{i-1:j-1}$
      is a concatenation of uniform permutations from $S_{i-1,l}\to S_{i-1,l+1}$
      which may be viewed as a single uniform permutation from $S_{i-1,i-1}\to S_{i-1,j-1}$.
      Therefore, again, one can proceed as in the $i>0$ sub-case (within
      the $j=i$ case). \\
      \emph{Bounding Term II.} One can write
      \[
        \Pr[x\in S_{ij}|x\notin S_{j};\mathcal{I}]\le\frac{\left|S_{ij}\right|-|S_{j}|}{|S_{i-1,j}|-|S_{j}|}
      \]
      where we can use \Remref{Pr_x_in_X_or_t} (second observation) with
      $M=|S_{i-1,j}\backslash S_{j}|$ and $N=|S_{ij}\backslash S_{j}|$.
      This is because the set $S_{ii}$ was chosen uniformly at random (excluding
      the $S_{i}$ part which we have anyway accounted for) and therefore
      $S_{ij}=\cfunc{i,j-1,\mathcal{K}}(S_{ii})$ may also be viewed as
      set which is chosen uniformly at random (excluding the $S_{j}$ part).
      This is because $\cfunc{i,j-1,\mathcal{K}}$ is just a permutation
      (its distribution does not matter as long as $S_{ii}$ is chosen uniformly
      at random\footnote{Just as for any $x,r\in\{0,1\}$, $r\oplus x$ is uniformly random
        if $r$ is uniformly random.}). Thereafter, one can proceed as in the $j=i$ case.

      This yields the analogue of \Eqref{x_in_S_ii_2_delta}, i.e.
      \[
        \Pr[x\in S_{ij}|S_{i-1};\bar{S}_{i-1}]\le(2^{\delta}+c)\cdot\negl.
      \]
      The result directly generalises for the $\beta\neq\emptyset$ case
      which at most adds a $\poly$ factor in the final calculations with
      uniform distributions by changing the sizes of the exponential sized
      sets by a polynomial factor.
    \end{proof}
  }
\end{claim}

\subsubsection{\texorpdfstring{$\protect\CH d$ is hard for $\protect\CQ d$}{d-CodeHashing is hard for CQ\_d}}
\branchcolor{purple}{
    We now state the main lemma of this subsection.
}
\begin{lem}[$\CH{d}\notin\classCQ{d}$]\label{lem:CQd_hardness}
  Every $\CQ d$ circuit succeeds at solving $\CH d$ (see \Defref{dCodeHashingProblem})
  with probability at most $\ngl{\lambda}$ on input $1^{\lambda}$
  for $d\le\poly$.
\end{lem}

We begin with setting up the notation we use in the proof. It helps
to recall that\footnote{See the definition of suitable codes (see \Lemref{suitableCodes})
  and $\CH d$ (see \Defref{dCodeHashingProblem}).} $n=\Theta(\lambda)$.
\begin{itemize}
  \item Denote by $\sigma_{0}$ the initial state (containing the input $1^{\lambda}$
        and ancillae initialised to zero)
  \item From \Notaref{CompositionNotation}, recall that $\CQ d$ circuits
        can be represented as $\mathcal{C}=\mathcal{C}_{\tilde{n}}\circ\dots\mathcal{C}_{2}\circ\mathcal{C}_{1}$
        where $\tilde{n}\le\poly$. We write $\mathcal{C}_{i}:=\vec{U}_{i}\circ\mathcal{A}_{c,i}$
        where $\vec{U}_{i}$ denotes $d$ layers of unitaries, followed by
        a measurement. For brevity, we drop the subscript ``$c$'' from
        $\mathcal{A}_{c,i}$ and even ``$\circ$'' to aid readibility.
  \item Let $\mathcal{L}=(H_{0},\dots H_{d+1})$, $d'$ and $\Sigma$ be as
        in \Defref{dCodeHashingProblem}.
  \item Denote by $\mathcal{C}^{\mathcal{L}}:=\mathcal{A}_{\tilde{n}+1}^{\mathcal{L}}\vec{U}_{\tilde{n}}^{\mathcal{L}}\mathcal{A}_{\tilde{n}}^{\mathcal{L}}\dots\vec{U}_{1}^{\mathcal{L}}\mathcal{A}_{1}^{\mathcal{L}}(\sigma_{0})$,
        i.e. a $\CQ d$ circuit with oracle access to $\mathcal{L}$.\\
        We make the following assumptions which only makes the result stronger
        (compare \Figref{CQd_1} with \Figref{CQd_2}; also see \Exaref{learningPaths_strongerResult} below)
        \begin{itemize}
          \item $\mathcal{A}_{i}$ ensures that its input is forwarded with its output
          \item $\vec{U}_{i}$ forwards all classical information it receives as output
          \item For $i>1$, $\mathcal{A}_{i}$ receives an extra random variable (a
                set of paths, details appear later), correlated with $\mathcal{L}$
                as input, labelled $\beta^{*}(s_{i-1})$.
          \item both $\vec{U}_{i}$ and $\mathcal{A}_{i}$ (implicitly) receive the
                transcript (classical input/output messages) until they are invoked.
        \end{itemize}
  \item In the analysis below, we consider $\tilde{n}$ sequences of shadow
        oracles. Each sequence is denoted by $\vec{\mathcal{M}}_{i}=(\mathcal{M}_{i,1},\mathcal{M}_{i,2}\dots\mathcal{M}_{i,d})$,
        one set for each $\mathcal{C}_{i}$.
        \begin{itemize}
          \item We use $\vec{U}_{i}^{\vec{\mathcal{M}}_{i}}$ to denote $\Pi_{i}\circ U_{i,d+1}\circ\mathcal{M}_{i,d}\circ U_{i,d}\circ\dots\mathcal{M}_{i,1}\circ U_{i,1}$.
          \item $(\mathcal{M}_{i,j})_{j}$ are shadows of $\mathcal{L}$ using the
                sets outputted by \Algref{S_ij_given_beta} (and are conditioned on
                \Algref{baseSets} succeeding). The input to the algorithm is described
                later.
        \end{itemize}
  \item Denote by $\mathcal{C}^{\mathcal{M}}:=\mathcal{A}_{\tilde{n}+1}^{\mathcal{L}}\vec{U}_{\tilde{n}}^{\vec{\mathcal{M}}_{\tilde{n}}}\mathcal{A}_{\tilde{n}}^{\mathcal{L}}\dots\vec{U}_{1}^{\vec{\mathcal{M}}_{1}}\mathcal{A}_{1}^{\mathcal{L}}(\sigma_{0})$,
        i.e. a $\CQ d$ circuit with access to only shadow oracles.
  \item After each circuit $\mathcal{C}_{i}$, the state is classical and
        this allows us to consider ``transcripts'' which we denote by $T$
        (the details appear later).
  \item Parameters for the sampling argument: Use $\delta=\Delta/\tilde{n}$,
        $\gamma=2^{-m}$ where $\Delta>0$ is an arbitrary, small constant
        and $m$ is such that $m-\tilde{m}\ge\Omega(n)$ where $\tilde{m}$
        is the length of the ``advice'', i.e. the number of bits $\mathcal{A}_{i}$
        sends to $\vec{U}_{i}$.
  \item Shorthand for the $\Pr[{\rm find}:\dots]$ notation: Suppose $\mathcal{L}$
        is an oracle, $\bar{S}$ is a sequence of sets, $\rho$ is a quantum
        state and $T$ is some variable. We use $\Pr[{\rm find}:U^{\mathcal{L}\backslash\bar{S}},\rho|T]$
        to denote the expression $\Pr[{\rm find}:V^{\mathcal{N}\backslash\bar{R}},\mathcal{\sigma}]$
        where $V=U|T$, $\mathcal{N}=\mathcal{L}|T$, $\sigma=\rho|T$.
\end{itemize}
\branchcolor{purple}{Before we begin with the proof, we briefly illustrate how giving additional
  information to the classical algorithm (and conditioning at the same
  time) only strengthens our result.}
\begin{example}
  \label{exa:learningPaths_strongerResult}Suppose $\mathcal{O}$ is
  an oracle for Simon's Problem, encoding the period $s$. Let $\mathcal{A}$
  denote an algorithm which takes no input and $\mathcal{B}$ denote
  an algorithm which takes an input $S$ which is some variable correlated
  to $\mathcal{O}$. Then\footnote{For concreteness, if $\mathcal{A}$ and $\mathcal{B}$ are classical
    algorithms, and $S$ is the period encoded in $\mathcal{O}$, then
    clearly the upper bound becomes $1$ but it is not achievable; illustrating
    that this procedure can only strengthen the hardness result.}
  \begin{align*}
    \max_{\mathcal{A}}\Pr[s\leftarrow\mathcal{A}^{\mathcal{O}}] & =\max_{\mathcal{A}}\sum_{S}\Pr[s\leftarrow\mathcal{A}^{\mathcal{O}|S}]\Pr[S]         \\
                                                                & \le\max_{\mathcal{B}}\sum_{S}\Pr[s\leftarrow\mathcal{B}{}^{\mathcal{O}|S}(S)]\Pr[S].
  \end{align*}
\end{example}

\branchcolor{blue}
{
  \begin{proof}
  For the overall template, we follow the proof of $\QNC_{d}$ hardness
  (see \Lemref{QNC_d_hardness}). Run \algref{baseSets} on $\mathcal{L}$
  and let $E$ be the event that it does not abort. Observe that
  \begin{equation}
    \left|\sum_{\mathbf{x}\in X_{\rm valid}}\Pr[\mathbf{x}\leftarrow\mathcal{C}^{\mathcal{L}}]-\sum_{\mathbf{x}\in X_{\rm valid}}\Pr[\mathbf{x}\leftarrow\mathcal{C}^{\mathcal{L}}|E]\right|\le\negl\label{eq:conditionOnE}
  \end{equation}
  as was the case before (recall $X_{\rm valid}$ was the set of valid solutions to $\CH{d}$). We will show in \emph{step one}, that $\mathcal{C}^{\mathcal{L}}|E$
  and $\mathcal{C}^{\mathcal{M}}|E$ have essentially the same behaviour,
  i.e.
  \begin{equation}
    \left|\sum_{\mathbf{x}\in X_{\rm valid}}\Pr[\mathbf{x}\leftarrow\mathcal{C}^{\mathcal{L}}|E]-\sum_{\mathbf{x}\in X_{\rm valid}}\Pr[\mathbf{x}\leftarrow\mathcal{C}^{\mathcal{M}}|E]\right|\le\negl\label{eq:shadowOnlyUseless_CQ_d}
  \end{equation}
  and then in \emph{step two}, that $\mathcal{C}^{\mathcal{M}}|E$ succeeds
  with at most negligible probability at solving $\CH d$ (see \Defref{dCodeHashingProblem}).
  These two steps, together with \Eqref{conditionOnE}, entail that
  $\mathcal{C}^{\mathcal{L}}$ solves $\CH d$ with at most $\negl$
  probability.

  In the rest of this proof, we \emph{implicitly condition everything
    on the event $E$} and do not explicitly state this, for notational
  convenience. Let $\bar{S}_{0}$ be the output of \Algref{baseSets}.
  \Figref{CQ_d_basic_idea} may help in conveying the overarching idea.

  \begin{figure}%
    \centering
    \subfloat[Initial $\CQd$ circuit]{
        \includegraphics[width=10cm]{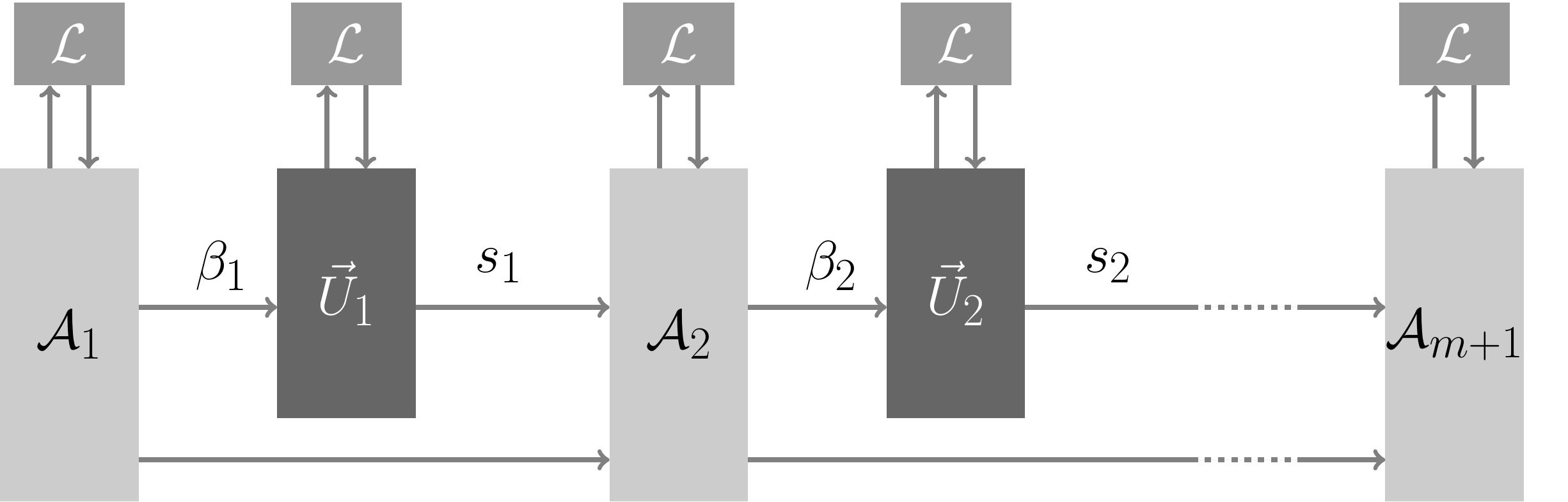}
        \label{fig:CQd_1}
    }\\
    \subfloat[$\CQd$ circuit with $\beta^*$ from the sampling argument]{
        \includegraphics[width=10cm]{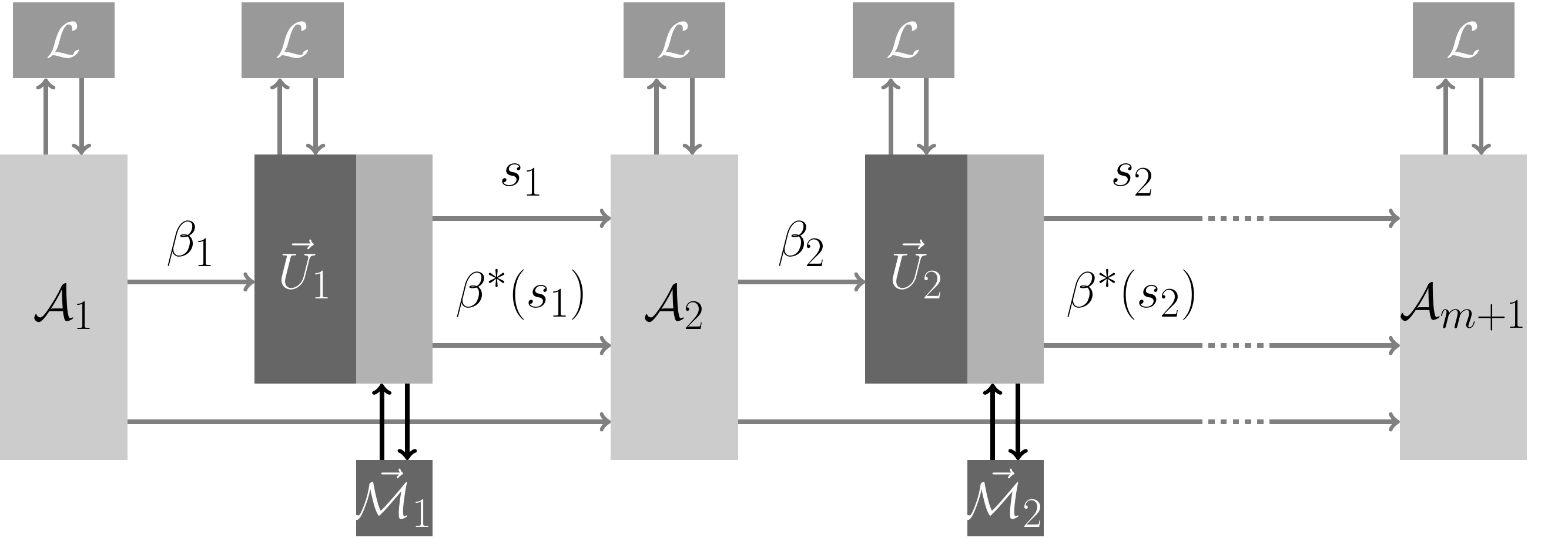}
        \label{fig:CQd_2}
    }\\
    \subfloat[$\CQd$ circuit with $\beta^*$ where all oracles replaced with shadow oracles.]{
        \includegraphics[width=10cm]{insta_depth_images/hybrid_cqd_circuit_with_advice.pdf}
        \label{fig:CQd_3}
    }%
    \caption{Variants of the $\CQd$ circuit which arise in establishing hardness of solving $\CH{d}$. Observe that \Figref{CQd_2} can simulate \Figref{CQd_1}. We analyse the latter and show its behaviour is essentially the same as that of \Figref{CQd_3}. }
    \label{fig:CQ_d_basic_idea}
  \end{figure}

  \textbf{Step One.} $\mathcal{C}^{\mathcal{L}}|E$ and $\mathcal{C}^{\mathcal{M}}|E$
  have essentially the same behaviour,\\
  Using a hybrid argument, one can bound the LHS of \Eqref{shadowOnlyUseless_CQ_d}
  by bounding %
  \[
    \TD[\mathcal{C}^{\mathcal{L}},\mathcal{C}^{\mathcal{M}}]=\TD\left[\mathcal{A}_{\tilde{n}+1}^{\mathcal{L}}\vec{U}_{\tilde{n}}^{\mathcal{L}}\mathcal{A}_{\tilde{n}}^{\mathcal{L}}\dots\vec{U}_{1}^{\mathcal{L}}\mathcal{A}_{1}^{\mathcal{L}}(\sigma_{0}),\quad\mathcal{A}_{\tilde{n}+1}^{\mathcal{L}}\vec{U}_{\tilde{n}}^{\vec{\mathcal{M}}_{\tilde{n}}}\mathcal{A}_{\tilde{n}}^{\mathcal{L}}\dots\vec{U}_{1}^{\vec{\mathcal{M}}_{1}}\mathcal{A}_{1}^{\mathcal{L}}(\sigma_{0})\right]
  \]
  with
  \begin{align}
    \le & \sum_{i=1}^{\tilde{n}}\TD \Big[\mathcal{A}_{\tilde{n}+1}^{\mathcal{L}}\vec{U}_{\tilde{n}}^{\mathcal{L}}\mathcal{A}_{\tilde{n}}^{\mathcal{L}}\dots\vec{U}_{i+1}^{\mathcal{L}}\mathcal{A}_{i+1}^{\mathcal{L}}\ \ \ \vec{U}_{i}^{\mathcal{L}}\mathcal{A}_{i}^{\mathcal{L}}\dots\vec{U}_{1}^{\mathcal{L}}\mathcal{A}_{1}^{\mathcal{L}}(\sigma_{0}),\nonumber                                                                                                                    \\
        & \quad\quad\quad\mathcal{A}_{\tilde{n}+1}^{\mathcal{L}}\vec{U}_{\tilde{n}}^{\mathcal{L}}\mathcal{A}_{\tilde{n}}^{\mathcal{L}}\dots\vec{U}_{i+1}^{\mathcal{L}}\mathcal{A}_{i+1}^{\mathcal{L}}\ \ \ \vec{U}_{i}^{\vec{\mathcal{M}}_{i}}\mathcal{A}_{i}^{\mathcal{L}}\dots\vec{U}_{1}^{\vec{\mathcal{M}}_{1}}\mathcal{A}_{1}^{\mathcal{L}}(\sigma_{0})\Big]\nonumber                                                                                                         \\
    \le & \sum_{i=1}^{\tilde{n}}\TD\left[\vec{U}_{i}^{\mathcal{L}}\mathcal{A}_{i}^{\mathcal{L}}\vec{U}_{i-1}^{\vec{\mathcal{M}}_{i-1}}\dots\vec{U}_{1}^{\vec{\mathcal{M}}_{1}}\mathcal{A}_{1}^{\mathcal{L}}(\sigma_{0}),\quad\vec{U}_{i}^{\vec{\mathcal{M}}_{i}}\mathcal{A}_{i}^{\mathcal{L}}\vec{U}_{i-1}^{\vec{\mathcal{M}}_{i-1}}\mathcal{A}_{i-1}^{\mathcal{L}}\dots\vec{U}_{1}^{\vec{\mathcal{M}}_{1}}\mathcal{A}_{1}^{\mathcal{L}}(\sigma_{0})\right].\label{eq:sum_of_Bs_CQd}
  \end{align}

  \textbf{The $i=1$ case}\\
  Begin with $i=1$. Let $\mathcal{A}_{1}^{\mathcal{L}}(\sigma_{0})=:\sigma_{1}$.
  One can write

  \begin{align*}
    \TD[\vec{U}_{1}^{\mathcal{L}}(\sigma_{1}),\vec{U}_{1}^{\vec{\mathcal{M}}_{1}}(\sigma_{1})] & =\TD[\mathcal{L}U_{1,d}\dots\mathcal{L}U_{1,1}(\sigma_{1}),\quad\mathcal{M}_{1,d}U_{1,d}\dots\mathcal{M}_{1,1}U_{1,1}(\sigma_{1})]                                         \\
                                                                                             & \le\sum_{j=1}^{d}\TD[\mathcal{L}U_{1,j}\ \ \ \underbrace{\mathcal{M}_{1,j-1}U_{1,j-1}\dots\mathcal{M}_{1,1}U_{1,1}(\sigma_{1})}_{:=\rho_{1,j-1}}, & \text{hybrid argument} \\
                                                                                             & \quad\quad\quad\mathcal{M}_{1,j}U_{1,j}\ \ \ \overbrace{\mathcal{M}_{1,j-1}U_{1,j-1}\dots\mathcal{M}_{1,1}U_{1,1}(\sigma_{1})}].
  \end{align*}
  Our goal is to bound each term in the sum using the O2H lemma (\Lemref{O2H})
  as
  \begin{equation}
    B(\mathcal{L}U_{1,j}\rho_{1,j-1},\mathcal{M}_{1,j}U_{1,j}\rho_{1,j-1})\le\sqrt{\Pr[{\rm find}:U_{1,j}^{\mathcal{L}\backslash\bar{S}_{1,j}},\rho_{1,j-1}]}\label{eq:prFind_i1_CQ_d}
  \end{equation}
  where $\mathcal{M}_{1,j}$ is the shadow of $\mathcal{L}$ wrt $\bar{S}_{1,j}$
  and $\bar{S}_{1,j}$ is defined as follows.

  Denote by $\beta_{1}$ the set of all paths queried $\mathcal{A}_{1}^{\mathcal{L}}$.
  Denote by $\beta_{1}'\subseteq\beta_{1}$ the subset of paths queried
  by $\mathcal{A}_{1}^{\mathcal{L}}$ wrt $\bar{S}_{0}$ (see \Defref{valid_beta_paths}),
  i.e. let $\beta_{1}'$ denote the set of path queries made by the
  first classical part of the $\CQ d$ circuit within the base sets
  $\bar{S}_{0}$. Run \Algref{S_ij_given_beta} with $\beta\leftarrow\beta_{1}'$,
  $\bar{S}_{0}\leftarrow\bar{S}_{0}$ as inputs and define $\bar{S}_{1,j}\leftarrow\bar{S}_{j}$
  where $\bar{S}_{j}$ is the output of the algorithm for $j\in\{1\dots d\}$.
  Let $\hat{\mathcal{L}}_{1,j}$ (resp. $\check{\mathcal{L}}_{1,j}$)
  be $\mathcal{L}$ inside (resp. outside) $\bar{S}_{1,j}$ (see \Notaref{LhatAndsuch}).

  To apply \Corref{Conditionals} we condition the RHS of \Eqref{prFind_i1_CQ_d}
  on $\check{\mathcal{L}}_{1,j}$ to write $\Pr[{\rm find}:U_{1,j}^{\mathcal{L}\backslash\bar{S}_{1,j}},\rho_{1,j-1}|\check{\mathcal{L}}_{1,j-1}]$.
  The conditioning ensures that $\rho_{1,j-1}|\check{\mathcal{L}}_{1,j-1}$
  is uncorrelated\footnote{This is because, given $\mathcal{\check{L}}_{1,j-1}$ (which, in particular,
    specifies $\bar{S}_{1,j-1}$ but not the values of $\mathcal{L}$
    inside $\bar{S}_{1,j-1}$), $\bar{S}_{1,j}|\check{\mathcal{L}}_{1,j-1}$
    is, by construction, a (component-wise) subset of $\bar{S}_{1,j-1}$
    but within $\bar{S}_{1,j-1}$, the distribution of $\bar{S}_{1,j}|\check{\mathcal{L}}_{1,j-1}$
    is determined by the randomness in \Algref{S_ij_given_beta} and by
    the distribution of $\hat{\mathcal{L}}_{1,j-1}$. The randomness of
    the algorithm \Algref{S_ij_given_beta} is independent of $\mathcal{L}$
    and $\rho_{1,j-1}$ contains at most as much information about $\mathcal{L}$
    as is present in $\check{\mathcal{L}}_{1,j-1}$ (that is because $\rho_{1,j-1}$
    only has access to $\mathcal{M}_{1,1}\dots\mathcal{M}_{1,j-1}$ which
    block all information about $\mathcal{L}$ inside $\bar{S}_{j-1}$).} with $\bar{S}_{1,j}|\check{\mathcal{L}}_{1,j-1}$. Using \Claimref{x_in_S_CQ_d},
  with $\delta\leftarrow0$, $\beta\leftarrow\beta_{1}'$ and $\bar{S}_{0}\leftarrow\bar{S}_{0}$,
  one can apply \Corref{Conditionals} to obtain
  \[
    \Pr[{\rm find}:U_{1,j}^{\mathcal{L}\backslash\bar{S}_{1,j}},\rho_{1,j-1}|\check{\mathcal{L}}_{1,j-1}]\le\negl
  \]
  which in turn bounds \Eqref{prFind_i1_CQ_d} by $\negl$.
  \end{proof}
}

\branchcolor{purple}{Before moving to the $i=2$ and then the general case, we describe
an intuitive picture to keep in mind. Observe that the shadow oracles
$\{\mathcal{M}_{1,j}\}_{j}$ were determined, in particular, by the
set of paths $\beta_{1}$ queried by the classical algorithm.

In step $2$, the shadow oracles $\{\mathcal{M}_{2,j}\}_{j}$ would
be determined by (in addition to $\beta_{1}$) both, the set of paths
$\beta_{2}$ queried by the classical algorithm $\mathcal{A}_{2}^{\mathcal{L}}$
and by the set of paths $\beta(s_{1})$ the classical algorithm $\mathcal{A}_{2}^{\mathcal{L}}$
receives as an extra input.\footnote{Strictly, (as explained later) it receives $\beta^{*}(s_{1})$ which
may have some paths with $\perp$ as the last coordinate but since
$\mathcal{A}_{2}^{\mathcal{L}}$ can access $\mathcal{L}$, we assume
it learns the last coordinates of all paths in $\beta^{*}(s_{1})$
and denote the complete paths as $\beta(s_{1})$.} We have not yet defined how the paths $\beta(s_{1})$ are specified.
They are specified by the sampling argument (\Propref{main_delta_non_uniform_inj_shuffler}).
As illustrated in \Figref{CQ_d_basic_idea}, treat the sampling argument
as an algorithm which interacts with $\mathcal{M}_{1,d}$ and produces
$\beta^{*}(s_{1})$ as output (the star indicates that the last coordinate
of some of the paths may be $\perp$; in $\beta(s_{1})$ all coordinates
are non-$\perp$ as $\mathcal{A}_{2}$ has access to $\mathcal{L}$).
Using the notation in \Propref{main_delta_non_uniform_inj_shuffler},
it outputs the $p$-many paths (as $\beta^{*}(s_{1})$), which are
present in $\mathcal{G}_{1}$ with probability $\alpha_{k}$, for
each $k$.

In step $i$, proceeding analogously, the shadow oracles $\{\mathcal{M}_{i,j}\}_{j}$
would be determined by $\beta_{1}\cup\beta_{2}\dots\cup\beta_{i}$
and by $\beta(s_{1})\cup\dots\cup\beta(s_{i-1})$, i.e. by the ``transcript''
encoding the paths exposed so far. We would then condition on these
paths and use the fact that after conditioning, these distributions
stay $(ip,i\delta)$ uniform, which in turn allows us to argue that
the analogue of \Eqref{prFind_i1_CQ_d} (i.e. $\Pr[{\rm find}:U_{i,j}^{\mathcal{L}\backslash\bar{S}_{i,j}},\rho_{i,j-1}]$
below) stays negligible.

To apply the sampling argument, it would be useful to restrict to
the distribution over the base sets which is facilitated by the following
notation.

}

\global\long\def\inj#1{\mathsf{inj}[#1]}%

\begin{notation}[$\inj{\mathcal{L}'|\beta}$ wrt to $\bar{S}_{0}$]
  \label{nota:inj_from_L} Suppose $\mathcal{L}'=(H'_{0},\dots H'_{d})$
  (as in \Notaref{LhatAndsuch}) is a random variable sampled from some
  arbitrary distribution. Let $\beta$ be a set of paths in $\mathcal{L}'$.
  Then $\inj{\mathcal{L}'|\beta}=:\Xi$ wrt $\bar{S}_{0}$ denotes a
  (random) injective shuffler conditioned on $\beta$ wrt base sets
  $\bar{S}_{0}$ (see \Defref{validInjectiveShufflerConditionedOnBeta})
  such that for all $x_{0}\in\Sigma$, $x_{1}\in S_{01},\dots x_{d}\in S_{0d},x_{d+1}\in\{0,1\}^{n}\cup\{\perp\}$
  \[
    \Pr[(x_{0},\dots x_{d+1})\in\paths(\Xi)]=\begin{cases}
      \Pr[H_{0}(x_{0})=x_{1}\land\dots H_{d-1}(x_{d-1})=x_{d}] & \text{for }x_{0}\in\Sigma\backslash X_{0}(\beta) \\
      \Pr[H_{0}(x_{0})=x_{1}\land\dots H_{d}(x_{d})=x_{d+1}]   & \text{otherwise}.
    \end{cases}
  \]
\end{notation}

\branchcolor{purple}{We now resume with the proof.}
\branchcolor{blue}{
\begin{proof}[Proof (cont.)]
  \textbf{The $i=2$ case}\\
  Let $\mathcal{A}_{2}^{\mathcal{L}}\vec{U}_{1}^{\vec{\mathcal{M}}_{1}}\mathcal{A}_{1}^{\mathcal{L}}(\sigma_{0})=:\sigma_{2}$.
  One can write the $i=2$ term in the RHS of \Eqref{sum_of_Bs_CQd}
  as
  \[
    \TD[\vec{U}_{2}^{\mathcal{L}}(\sigma_{2}),\vec{U}_{2}^{\vec{\mathcal{M}}_{2}}(\sigma_{2})]\le\sum_{j=1}^{d}\TD[\mathcal{L}U_{2,j}\rho_{2,j-1},\ \ \mathcal{M}_{2,j}U_{2,j}\rho_{2,j-1}]
  \]
  where $\rho_{2,j-1}:=\mathcal{M}_{2,j-1}U_{2,j-1}\dots\mathcal{M}_{2,1}U_{2,1}(\sigma_{2})$.
  Using \Lemref{O2H}, one can write
  \begin{equation}
    B[\mathcal{L}U_{2,j}\rho_{2,j-1},\ \ \mathcal{M}_{2,j}U_{2,j}\rho_{2,j-1}]\le\sqrt{\Pr[{\rm find}:U_{2,j}^{\mathcal{L}\backslash\bar{S}_{2,j}},\rho_{2,j-1}]}\label{eq:prFind_i2_CQd}
  \end{equation}
  where $\mathcal{M}_{2,j}$ is the shadow of $\mathcal{L}$ wrt $\bar{S}_{2,j}$
  and $\bar{S}_{2,j}$ is defined as follows.

  Recall that $\beta_{1}$ denoted the paths queried by $\mathcal{A}_{1}^{\mathcal{L}}$.
  Note that $\inj{\mathcal{L}|\beta_{1}}$ (see \Notaref{inj_from_L})
  is distributed as $\mathbb{F}_{{\rm inj}}^{|\beta_{1}}$. Let the
  output of $\vec{U}_{1}^{\vec{\mathcal{M}}_{1}}$ be\footnote{which in particular, contains $\beta_{1}$}
  $s_{1}$. Given that $\Pr[s_{1}|\beta_{1}]\ge\gamma$, $\inj{\mathcal{L}|\beta_{1}s_{1}}$
  which is distributed as $\mathbb{F}_{{\rm inj}}^{|\beta_{1}}$ may
  be expressed as a convex combination (as described in \Propref{main_delta_non_uniform_inj_shuffler})
  $\inj{\mathcal{L}|s_{1}\beta_{1}\beta^{*}(s_{1})}$ distributed as
  $\mathbb{F}_{{\rm inj}}^{(p,\delta)|\beta_{1}}$ where $|\beta^{*}(s_{1})|\le p\le2m/\delta$
  whenever the convex coefficient is larger than $\gamma$. When $\Pr[s_{1}|\beta_{1}]<\gamma$,
  let $\beta^{*}(s_{1})=\emptyset$. Note that this implicitly defines
  the random variable $\beta^{*}(s_{1})$ which we had initially left
  unspecified. $\mathcal{A}_{2}^{\mathcal{L}}$ takes as input $s_{1}$
  and $\beta^{*}(s_{1})$. $\mathcal{A}_{2}^{\mathcal{L}}$ learns $\beta(s_{1})$
  which is $\beta^{*}(s_{1})$ with $\perp$s replaced by the value
  $\mathcal{L}$ takes in the last coordinate. Let $\beta_{2}$ denote
  the addition paths queried by $\mathcal{A}_{2}^{\mathcal{L}}$.

  We are now ready to define $\bar{S}_{2,j}$. Let $\beta'_{2}\subseteq\beta_{2}$
  be the subset of paths in $\beta_{2}$ which are within the base sets
  $\bar{S}_{0}$. Run \Algref{S_ij_given_beta} with $\beta\leftarrow\beta_{2}'\cup\beta(s_{1})\cup\beta_{1}'$,
  $\bar{S}_{0}\leftarrow\bar{S}_{0}$ as inputs and define $\bar{S}_{2,j}\leftarrow\bar{S}_{j}$
  where $\bar{S}_{j}$ is the output of the algorithm, for $j\in\{1,\dots d\}$.
  Let $\hat{\mathcal{L}}_{2,j}$ (resp. $\check{\mathcal{L}}_{2,j}$)
  be $\mathcal{L}$ inside (resp. outside) $\bar{S}_{2,j}$ (see \Notaref{LhatAndsuch}).

  To apply \Corref{Conditionals} we condition the RHS of \Eqref{prFind_i2_CQd}
  on $\check{\mathcal{L}}_{2,j-1}$ to write $\Pr[{\rm find}:U_{2,j}^{\mathcal{L}\backslash\bar{S}_{2,j}},\rho_{2,j-1}|\check{\mathcal{L}}_{2,j-1}]$.
  The conditioning, as before, ensures that $\rho_{2,j-1}|\check{\mathcal{L}}_{2,j-1}$
  is uncorrelated with $\bar{S}_{2,j}|\check{\mathcal{L}}_{2,j-1}$
  (for exactly the same reason as the $i=1$ case). However, to apply
  \Claimref{x_in_S_CQ_d} we condition on the ``transcript'' until
  the output of $\mathcal{A}_{2}^{\mathcal{L}}$, i.e. $T(\sigma_{2}):=(\beta_{1},s_{1},\beta(s_{1}),\beta_{2})$,
  by writing $\Pr[{\rm find}:U_{2,j}^{\mathcal{L}\backslash\bar{S}_{2,j}},\rho_{2,j-1}|\check{\mathcal{L}}_{2,j-1}]$

  \begin{align*}
     & =\sum_{s_{1},\beta_{1},\beta(s_{1}),\beta_{2}}\Pr[T(\sigma_{2})]\cdot\Pr[{\rm find}:U_{2,j}^{\mathcal{L}\backslash\bar{S}_{2,j}},\rho_{2,j-1}|\check{\mathcal{L}}_{2,j-1}\ T(\sigma_{2})]                                                                                                   \\
     & \le\sum_{\substack{s_{1}:\Pr[s_{1}|\beta_{1}]\ge2^{-m}                                                                                                                                                                                                                                      \\
        \beta_{1},\beta(s_{1}),\beta_{2}
      }
    }\underbrace{\Pr[\beta_{2}|s_{1}\beta_{1}\beta(s_{1})]\Pr[\beta(s_{1})|s_{1}\beta_{1}]\Pr[s_{1}|\beta_{1}]\Pr[\beta_{1}]}_{\Pr[T(\sigma_{2})]}\cdot\Pr[{\rm find}:U_{2,j}^{\mathcal{L}\backslash\bar{S}_{2,j}},\rho_{2,j-1}|\check{\mathcal{L}}_{2,j-1}\ T(\sigma_{2})]+2^{-(m-\tilde{m})}     \\
     & \le\sum_{\substack{s_{1}:\Pr[s_{1}|\beta_{1}]\ge2^{-m}                                                                                                                                                                                                                                      \\
    \beta(s_{1}):\Pr[\beta(s_{1})|s_{1}\beta_{1}]\ge2^{-m}                                                                                                                                                                                                                                         \\
    \beta_{1},\beta_{2}
    }
    }\Pr[\beta_{2}|s_{1}\beta_{1}\beta(s_{1})]\Pr[\beta(s_{1})|s_{1}\beta_{1}]\Pr[s_{1}|\beta_{1}]\Pr[\beta_{1}]\cdot\underbrace{\Pr[{\rm find}:U_{2,j}^{\mathcal{L}\backslash\bar{S}_{2,j}},\rho_{2,j-1}|\check{\mathcal{L}}_{2,j-1}\ T(\sigma_{2})]}_{\mathsf{Term\ I}}+2\cdot2^{-(m-\tilde{m})} \\
     & \le\negl
  \end{align*}
  where to obtain the first inequality, we note that for each $s_{1}:\Pr[s_{1}|\beta_{1}]\ge2^{-m}$,
  one can use \Propref{main_delta_non_uniform_inj_shuffler} and one can account for all
  $s_{1}:\Pr[s_{1}|\beta_{1}]<2^{-m}$, by simply upper bounding
  the sum by $2^{-(m-\tilde{m})}$ because $s_{1}$ is of length $\tilde{m}$.
  In the second inequality, we use the fact that either the convex weight
  (i.e. $\Pr[\beta(s_{1})|s_{1}\beta_{1}]$) as specified in \Propref{main_delta_non_uniform_inj_shuffler}
  is less than $2^{-m}$ (for at most each $s_{1}$, therefore it contributes
  at most $2^{-(m-\tilde{m})}$ to the sum) or it is greater than $2^{-m}$.
  In the latter case, the injective shuffler is $\mathbb{F}_{{\rm inj}}^{(p,\delta)|\beta}$
  distributed and therefore one can apply \Corref{Conditionals} together
  with \Claimref{x_in_S_CQ_d} with $\delta\leftarrow \delta $, $\beta\leftarrow\beta_{2}'\cup\beta(s_{1})\cup\beta'_{1}$
  and $\bar{S}_{0}\leftarrow\bar{S}_{0}$ to obtain $\mathsf{Term\ I}\le2^{\delta}\cdot\poly\cdot\negl$.

  \textbf{The general $i\in\{1\dots\tilde{n}\}$ case.}\\
  This is a straightforward generalisation of the $i=2$ case and hence
  we only outline the key steps. Let $\sigma_{i}:=\mathcal{A}_{i}^{\mathcal{L}}\vec{U}_{i-1}^{\vec{\mathcal{M}}_{i-1}}\dots\mathcal{A}_{2}^{\mathcal{L}}\vec{U}_{1}^{\vec{\mathcal{M}}_{1}}\mathcal{A}_{1}^{\mathcal{L}}(\sigma_{0})$
  where $\mathcal{M}_{i-1,j}$ is the shadow of $\mathcal{L}$ wrt $\bar{S}_{i-1,j}$,
  let $T(\sigma_{i}):=(\beta_{1},s_{1},\beta(s_{1}),\dots\beta_{i-1},s_{i-1},\beta(s_{i-1}),\beta_{i})$
  where $\beta_{i}$ denotes the paths queried by $\mathcal{A}_{i}^{\mathcal{L}}$,
  $s_{i-1}$ denotes the output of $\vec{U}_{i-1}^{\vec{\mathcal{M}}_{i-1}}$,
  $\beta^{*}(s_{i-1})$ be the paths as in \Propref{main_delta_non_uniform_inj_shuffler}
  when $\Pr[s_{i-1}|\beta_{i-1}\dots\beta(s_{1})s_{1}\beta_{1}]\ge\gamma$,
  $\inj{\mathcal{L}|\beta_{1}s_{1}\beta(s_{1})\dots\beta_{i-1},s_{i-1}}$
  is distributed as $\mathbb{F}_{{\rm inj}}^{(i-2)\delta|\beta_{i-1}\cup\beta(s_{i-2})\cup\beta_{i-2}\dots\cup\beta_{1}}$
  so that\footnote{Note that $i\ge2$ when this reasoning is applied because $s_{i-1}$ is $s_{1}$ for $i=2$.} $\inj{\mathcal{L}|\beta_{1}s_{1}\beta(s_{1})\dots\beta_{i-1}s_{i-1}\beta^{*}(s_{i-1})}$
  is distributed as $\mathbb{F}_{{\rm inj}}^{(p,(i-2)\delta)|\beta_{i-1}\cup\dots\cup\beta_{1}}$
  whenever the convex coefficient is larger than $\gamma=2^{-m}$. $\beta(s_{i-1})$
  is $\beta^{*}(s_{i-1})$ with $\perp$s replaced be the values taken
  by $\mathcal{L}$ at those coordinates. Let $\beta'_{i}\subseteq\beta_{i}$
  be the subset of paths in $\bar{S}_{0}$. Run \Algref{S_ij_given_beta}
  with $\beta\leftarrow\beta'_{i}\cup\beta(s_{i-1})\cup\dots\cup\beta_{1}'$,
  $\bar{S}_{0}\leftarrow\bar{S}_{0}$ as inputs and define $\bar{S}_{i,j}$
  to be $\bar{S}_{j}$ which is the output of the algorithm for $j\in\{1\dots d\}$.
  Let $\hat{\mathcal{L}}_{i,j}$ (resp. $\check{\mathcal{L}}_{i,j}$)
  be $\mathcal{L}$ inside (resp. outside) $\bar{S}_{i,j}$ (see \Notaref{LhatAndsuch}).
  Let $\mathcal{M}_{i,j}$ be the shadow of $\mathcal{L}$ wrt $\bar{S}_{i,j}$.
  The $i$th term in \Eqref{sum_of_Bs_CQd} can then be expressed as
  \[
    \TD[\vec{U}_{i}^{\mathcal{L}}(\sigma_{i}),\vec{U}_{i}^{\vec{\mathcal{M}}_{i}}(\sigma_{i})]\le\sum_{j=1}^{d}\TD[\mathcal{L}U_{i,j}\rho_{i,j-1},\ \ \mathcal{M}_{i,j}U_{i,j}\rho_{i,j-1}]
  \]
  where $\rho_{i,j-1}:=\mathcal{M}_{i,j-1}U_{i,j-1}\dots\mathcal{M}_{i,1}U_{i,1}(\sigma_{i})$.
  The square of the $j$th term, can then be bounded (using \Lemref{O2H})
  by $\Pr[{\rm find}:U_{i,j}^{\mathcal{L}\backslash\bar{S}_{i,j}},\rho_{i,j-1}]$
  which is
  \begin{align*}
     & \le\sum_{\substack{s_{1}:\Pr[s_{1}|\beta_{1}]\ge2^{-m},\dots s_{i-1}:\Pr[s_{i-1}|\beta_{1}\dots]\ge2^{-m}                                                                                                       \\
    \beta(s_{1}):\Pr[\beta(s_{1})|s_{1}\beta_{1}]\ge2^{-m},\dots\beta(s_{i-1}):\Pr[\beta(s_{i-1})|s_{1}\beta_{1}\dots]\ge2^{-m}                                                                                        \\
    \beta_{1},\beta_{2}\dots\beta_{i}
    }
    }\alpha(T(\sigma_{i}))\cdot\underbrace{\Pr[{\rm find}:U_{i,j}^{\mathcal{L}\backslash\bar{S}_{i,j}},\rho_{i,j-1}|\check{\mathcal{L}}_{i,j-1}\ T(\sigma_{i})]}_{\mathsf{Term\ I}}+2\cdot(i-1)\cdot2^{-(m-\tilde{m})} \\
     & \le2^{\Delta}\cdot\poly\cdot\negl+2\cdot(i-1)\cdot2^{-(m-\tilde{m})}\le\negl
  \end{align*}
  where $\alpha$ is the probability coefficient (bounded by $1$),
  and the distribution of the injective shuffler in $\mathsf{Term\ I}$
  is $\mathbb{F}^{i\cdot\delta|\beta'_{i}\cup\beta(s_{i-1})\cup\dots\cup\beta'_{1}}$.
  This is obtained by repeatedly applying \Propref{main_delta_non_uniform_inj_shuffler}
  (for the $k$th application, $\delta'\leftarrow(k-1)\delta$, $\beta\leftarrow\beta'_{k}\cup\beta(s_{k-1})\dots\cup\beta'_{1}$
  and $\bar{S}_{0}\leftarrow\bar{S}_{0}$) and arguing as before to
  collect terms for which the distribution over the injective shuffler
  is unknown (but which occur with probability at most $2^{-m}$). Independence
  of $\bar{S}_{i,j}$ from $\rho_{i,j-1}$ can be argued as before once
  it is conditioned on $\check{\mathcal{L}}_{i,j-1}$ and one can apply
  \Corref{Conditionals} together with \Claimref{x_in_S_CQ_d} (with
  $\delta\leftarrow i\cdot\delta$, $\beta\leftarrow\beta'_{i}\cup\beta(s_{i-1})\dots\cup\beta'_{1}$
  and $\bar{S}_{0}\leftarrow\bar{S}_{0}$) to obtain the stated bound
  on $\mathsf{Term\ I}$ (recall $\gamma=2^{-m}$ and $\delta=\Delta/\tilde{n}$).

  \textbf{Step Two.} $\mathcal{C}^{\mathcal{M}}|E$ succeeds at solving
  $\CH d$ with at most negligible probability. This is analogous to
  how we argued in the proof of \Lemref{QC_d_hardness}. The quantum
  part never has any information about $\tilde{H}$ (recall $\tilde{H}(\cdot)=H_{d}\circ\dots H_{0}(\cdot)$)
  which the classical algorithm before it does not already have. Therefore
  the success probability of $\mathcal{C}^{\mathcal{M}}|E$ is limited
  by the number of classical queries it makes. Since this is polynomial,
  from \Thmref{YZ22} (second part), it follows that $\mathcal{C}^{\mathcal{M}}|E$
  succeeds with negligible probability.

\end{proof}
}

\subsection{\texorpdfstring{$\protect\classCQC{d}$ exclusion}{BPP\^{}\{QNC_d\}\^{}BPP exclusion}}

\branchcolor{purple}{The proof of $\CQC d$ hardness of $\CH d$ is, \emph{conceptually},
  a straightforward combination of $\QC d$ hardness and of $\CQ d$
  hardness. In the proof of $\CQ d$ hardness, we analysed each $\QNC_{d}$
  circuit by following the ideas behind the $\QNC_{d}$ hardness proof.
  The difference was that instead of using the random oracles in $\mathcal{L}$
  directly (see \Defref{dCodeHashingProblem}), we used the conditioned
  oracle $\mathcal{L}|s$ and then relied on the behaviour of the injective
  shuffler $\inj{\mathcal{L}|s}$ to argue indistinguishability from
  the appropriate shadow oracles.

  We now proceed almost exactly as in the $\CQ d$ hardness case and
  analyse each $\QC d$ circuit by following the ideas behind the $\QC d$
  hardness proof. As before, the difference would be that we would use
  properties of $\inj{\mathcal{L}|s}$ (see \Notaref{inj_from_L}) instead
  of $\mathcal{L}$. Recall that in the analysis of $\QC d$ we had
  to introduce the notion of ``query paths'' (see \Defref{pathQueries}).
  However, we already introduced ``paths'' more carefully for analysing
  $\CQ d$ hardness and this makes it easier to analyse $\CQC d$ hardness
  (see \Defref{valid_beta_paths}). Compared to $\CQ d$, at a high
  level, the difference would just be that we expose additional paths
  $\beta$ after each layer of unitaries. This slightly changes the
  way shadows are defined, i.e. we need to adapt \Algref{S_ij_given_beta}
  to our setting (currently it is closer to the $\QNC_{d}$ case, \Algref{setMatrix},
  and we want it to be more like the $\QC d$ case, \Algref{setMatrixQC_d}).
  However, one can still use \Claimref{x_in_S_QC_d} which was used
  to argue that shadow oracles are hard to distinguish from the originals
  as this already accounts for paths $\beta$ exposed. }

\subsubsection{Shadow oracles for \texorpdfstring{$\protect\CQC d$}{CQC\_d} hardness and their properties}

\branchcolor{purple}{The procedure for generating the sets $S_{i,j}$ in this case, is
essentially the same as \Algref{S_ij_given_beta} with only one difference:
the procedure is applied to each index $i\in\{1\dots d\}$ because
the paths exposed by the classical algorithm are only determined after
each layer of unitary is applied. To contrast, in the $\CQ d$ case,
these paths (queried by the classical algorithm) were determined before
the quantum part of the circuit was executed and one could therefore
construct the all $\{S_{i,j}\}_{i,j}$ at once. }
\begin{lyxalgorithm}[Procedure for generating $S_{i,j}$, given $\beta$, and the previous
    sets $\bar{S}_{i-1}$]
  \label{alg:vec_S_i_given_beta_previous_vecSs} Let $\calL'=(H_{0}',\dots H_{d}')$
  be $\mathcal{L}$ conditioned on some variable, as in \Notaref{LhatAndsuch},
  $\Sigma$ be as in \Defref{dCodeHashingProblem} and $S_{i}=H'_{i-1}(\dots H_{0}'(\Sigma)\dots)$
  be as in \Algref{baseSets}. \\
  Input:
  \begin{enumerate}
    \item Index: $i\in\{1\dots d\}$
    \item Base sets $\bar{S}_{0}=(S_{0,j})_{j\in\{1\dots d\}}$ (see \Defref{baseSets})
    \item The set of paths queried: (valid) paths $\beta$ wrt $\bar{S}_{0}$
          (see \Defref{valid_beta_paths})
    \item The previous sequence of sets for creating the shadow oracle: If $i>0$,
          then $\bar{S}_{i-1}:=(S_{i-1,j})_{j\in\{1\dots d\}}$ where $S_{i-1,j}\subseteq S_{0,j}$
          for all $j\in\{1\dots d\}$.
    \item Whether or not event $E$ happened.
  \end{enumerate}
  Output:\\
  If $E$ did not happen, set $S_{ik}=\emptyset$ for all $i,k\in\{1,\dots d\}$.
  Otherwise, for each $i\in\{1\dots d\}$ do the following.
  \begin{enumerate}
    \item Define $S_{ik}=\emptyset$ for $1\le k<i$.
    \item Sample, uniformly at random, $S_{ii}\subseteq S_{i-1,i}\backslash X_{i}(\beta)$
          such $S_{i}\backslash X_{i}(\beta)\subseteq S_{ii}$, and $|S_{ii}|/|S_{i-1,i}|=1/|\Sigma|$
    \item Define $S_{ik}:=H'_{k-1}(\dots H'_{i}(S_{ii})\dots)$ for $i<k\le d$.
  \end{enumerate}
  In both cases, return $\bar{S}_{i}:=(S_{i1},\dots S_{id})$.
\end{lyxalgorithm}

\branchcolor{purple}{The key property satisfied by \Algref{S_ij_given_beta} was \Claimref{x_in_S_CQ_d}.
  The analogous property for \Algref{vec_S_i_given_beta_previous_vecSs}
  is the following which is almost identical to \Claimref{x_in_S_CQ_d}
  except that one specifies some conditions to ensure $\bar{S}_{i-1}$
  is appropriately defined. When we apply the algorithm, as in the $\QC d$
  case, we would begin with $\bar{S}_{0}$ and successively apply \Algref{vec_S_i_given_beta_previous_vecSs}
  to produce $\bar{S}_{1},\dots\bar{S}_{d}$ and the stated conditions
  would automatically hold.}
\begin{claim}
  \label{claim:x_in_S_CQC_d}Let $\mathcal{K}\sim\mathbb{F}_{{\rm inj}}^{\delta|\beta}$
  be an injective shuffler conditioned on $\beta$ wrt base sets $\bar{S}_{0}$,
  sampled form a $\delta$ non-$\beta$-uniform distribution (see \Defref{validInjectiveShufflerConditionedOnBeta}
  and \Notaref{F_inj_with_some_qualifier_like_p_delta_etc}) where $|\beta|\le\poly$.
  Suppose $\mathcal{L}'$ is $\mathcal{L}$ conditioned on some variable
  $\tau$ such that $\inj{\mathcal{L}'}$ wrt $\bar{S}_{0}$ (see \Notaref{inj_from_L})
  is exactly $\mathcal{K}$. Suppose \Algref{vec_S_i_given_beta_previous_vecSs}
  is run with the following inputs: an index $i\in\{1\dots d\}$, the
  base sets $\bar{S}_{0}$, valid paths $\beta$, a sequence of sets
  $\bar{S}_{i-1}$ (defined next) and the assertion that $E$ happened
  and let its output be $S_{ij}$ for $j\in\{1,\dots d\}$. If $i>1$,
  $\bar{S}_{i-1}:=(S_{i-1,1},S_{i-1,2},\dots S_{i-1,d})$ are arbitrary
  sets such that
  \begin{itemize}
    \item for $j<i-1$, $S_{i-1,j}=\emptyset$
    \item for $j=i-1$, $S_{i-1,i-1}\subseteq S_{0,i-1}$, $S_{i-1}\backslash X_{i-1}(\beta)\subseteq S_{i-1,i-1}$
          (where $S_{i}$ is as in \Algref{baseSets}) and $\left|S_{i-1,i-1}\right|=|\Sigma|^{d+2-(i-1)}=|\Sigma|^{d+1-i}$
    \item for $j>i-1$, $S_{i-1,j}=H_{j}(S_{i-1,j-1})=H_{j}(\dots H_{i-1}(S_{i-1,i-1})\dots)$.
  \end{itemize}
  Then,
  \[
    \Pr[x\in S_{ij}|\check{\mathcal{L}}']\le(2^{\delta}+c)\cdot\poly\cdot\negl
  \]
  where $c$ is some constant (independent of $\delta$, $d$ etc.)
  $\check{\mathcal{L}}'$ is $\mathcal{L}'$ outside $\bar{S}_{i-1}$
  (see \Notaref{LhatAndsuch} with $S^{{\rm out}}\leftarrow\bar{S}_{i-1}$
  and $\mathcal{L}'\leftarrow\mathcal{L}'$) for all $1\le i\le j\le d$
  where the probability is over the randomness in $\mathcal{K}$ (i.e.
  from $\mathbb{F}_{{\rm inj}}^{\delta|\beta}$) and the randomness
  in \Algref{vec_S_i_given_beta_previous_vecSs}.

  \branchcolor{blue}{\begin{proof}
      The same as that of \Claimref{x_in_S_CQ_d} except that there, the
      proof worked for all $i,j\in\{1\dots d\}$. Here, the same arguments
      apply for a fixed $i$ and $\beta$ over all values of $j\in\{1\dots d\}$.
    \end{proof}
  }

  \branchcolor{purple}{It might not be clear why it suffices to consider only one path, $\beta$
    in the claim if different paths $\beta_{i}$ are specified for different
    $i$s when $\bar{S}_{i}$ are created using \Algref{vec_S_i_given_beta_previous_vecSs}.}
\end{claim}

\begin{rem}
  Suppose $\bar{S}_{i}$ are created successively using \Algref{vec_S_i_given_beta_previous_vecSs}
  with $\beta\leftarrow\cup_{i'\in\{1\dots,i-1\}}\beta_{i'}=:\beta_{1:i-1}$
  and $\bar{S}_{i-1}$ as inputs for index $i$. Then, the condition
  $S_{i-1}\backslash X_{i-1}(\beta_{1:i-1})\subseteq S_{i-1,i-1}$ holds
  by construction, and it trivially holds that $S_{i-1}\backslash X_{i-1}(\beta_{1:i})\subseteq S_{i-1}\backslash X_{i-1}(\beta_{1:i-1})$
  because $\beta_{1:i}\supseteq\beta_{1:i-1}$. If \Claimref{x_in_S_CQC_d}
  is invoked with $\bar{S}_{i-1}$ and $\beta\leftarrow\beta_{1:i}$,
  then the condition $S_{i-1}\backslash X_{i-1}(\beta)\subseteq S_{i-1,i-1}$
  is satisfied as required.
\end{rem}

\subsubsection{\texorpdfstring{$\protect\CH d$ is hard for $\protect\CQC d$}{d-CodeHashing is hard for CQC\_d}}
\begin{lem}[$\CH{d}\notin\classCQC{d}$]
  \label{lem:CQCd_hardness}Every $\CQC d$ circuit succeeds at solving
  $\CH d$ (see \Defref{dCodeHashingProblem}) with probability at most
  $\ngl{\lambda}$ on input $1^{\lambda}$ for $d\le\ply{\lambda}$.
\end{lem}

Following the previous proof, we begin with setting up the notation
(recall $n=\Theta(\lambda)$).
\begin{itemize}
  \item Denote by $\sigma_{0}$ the initial state (containing the input $1^{\lambda}$
        and ancillae initialised to zero).
  \item From \Notaref{CompositionNotation}, recall that $\CQC d$ circuits
        can be represented as\footnote{We dropped the preceding $\mathcal{A}_{c,m+1,1}$ classical circuit.
        This is without loss of generality because it can be accounted for
        by adding a $\mathcal{D}_{\tilde{n}+1}$; but $\tilde{n}$ is just
        an arbitrary polynomial of $n$. } $\mathcal{D}=\mathcal{D}_{\tilde{n}}\circ\dots\circ\mathcal{D}_{1}$
        where $\mathcal{D}_{i}=\mathcal{B}_{i,d}\circ\mathcal{B}_{i,d-1}\circ\dots\circ\mathcal{B}_{i,1}$
        is a $\QC d$ circuit with $\mathcal{B}_{i,j}:=\Pi_{i,j}\circ U_{i,j}\circ\mathcal{A}_{c,i,j}$.
        Here $U_{i,j}$ is a single layer unitary and $\mathcal{A}_{c,i,j}$
        is a poly sized classical circuit. We drop the subscript ``$c$''
        from $\mathcal{A}_{c,i,j}$ for brevity.
  \item Let $\mathcal{L}=(H_{0},\dots H_{d+1})$, $d'$ and $\Sigma$ be as
        in \defref{dCodeHashingProblem}.
  \item Denote by $\mathcal{D}^{\mathcal{L}}:=\mathcal{D}_{\tilde{n}}^{\mathcal{L}}\dots\mathcal{D}_{1}^{\mathcal{L}}$
        where $\mathcal{D}_{i}^{\mathcal{L}}=\mathcal{B}_{i,d}^{\mathcal{L}}\mathcal{B}_{i,d-1}^{\mathcal{L}}\dots\mathcal{B}_{i,1}^{\mathcal{L}}$
        and $\mathcal{B}_{i,j}^{\mathcal{L}}=\Pi_{i,j}\circ\mathcal{L}\circ U_{i,j}\circ\mathcal{A}_{i,j}^{\mathcal{L}}$.
        \\
        We make the following assumptions which only makes the result stronger
        (as explained in the $\CQ d$ case)
        \begin{itemize}
          \item Classical information entering $U_{i,j}$ and $\mathcal{A}_{i,j}$
                is forwarded with their output (for all $i,j$ in their domain).
          \item For $i>1$, $\mathcal{A}_{i,1}$ receives an extra random variable
                (a set of paths) correlated with $\mathcal{L}$ as input, labelled
                $\beta^{*}(s_{i-1})$.
        \end{itemize}
  \item In the analysis below, we consider $\tilde{n}$ sequences of shadow
        oracles. Each sequence is denoted by $\vec{\mathcal{M}}_{i}=(\mathcal{M}_{i,1},\mathcal{M}_{i,2}\dots\mathcal{M}_{i,d})$,
        one for each $\mathcal{D}_{i}$.
        \begin{itemize}
          \item We use $\mathcal{D}_{i}^{\vec{\mathcal{M}}_{i}}$ to denote $\mathcal{B}_{i,d}^{\vec{\mathcal{M}}_{i}}\mathcal{B}_{i,d-1}^{\vec{\mathcal{M}}_{i}}\dots\mathcal{B}_{i,1}^{\vec{\mathcal{M}}_{i}}$
                where\footnote{This is a slight abuse of notation because $\vec{\calM}_{i}$ is just
                a tuple $({\cal M}_{i,1}\dots{\cal M}_{i,d})$ but ${\cal B}_{i,j}^{\vec{{\cal M}}_{i}}$
                has ${\cal A}_{i,j}^{{\cal L}}$ which depends on ${\cal L}$ explicitly. } $\mathcal{B}_{i,j}^{\vec{\mathcal{M}}_{i}}=\Pi_{i,j}\circ\mathcal{M}_{i,j}\circ U_{i,j}\circ\mathcal{A}_{i,j}^{\mathcal{L}}$.
          \item $\vec{\mathcal{M}}_{i}=(\mathcal{M}_{i,j})_{j}$ are a sequence of
                shadows of $\mathcal{L}$ created using the sets outputted by \Algref{vec_S_i_given_beta_previous_vecSs}
                (and are conditioned on \Algref{baseSets} succeeding). The input
                to the algorithm is described later.
        \end{itemize}
  \item Denote by $\mathcal{D}^{\mathcal{M}}:=\mathcal{D}_{\tilde{n}}^{\vec{\mathcal{M}}_{\tilde{n}}}\dots\mathcal{D}_{1}^{\vec{\mathcal{M}}_{1}}$,
        i.e. a $\CQC d$ circuit with access to only shadow oracles.
\end{itemize}
The following are essentially unchanged from the $\CQ d$ case.
\begin{itemize}
  \item After each circuit $\mathcal{D}_{i}$, the state is classical and
        this allows us to consider ``transcripts'' which we denote by $T$
        (the details appear later).
  \item Parameters for the sampling argument: Same as the $\CQ d$ case (the
        advice is now the number of bits sent by $\mathcal{A}_{i,1}$ to $U_{i,1}$).
  \item Shorthand for $\Pr[{\rm find}:\dots]$ notation: Same as the $\CQ d$
        case.
\end{itemize}

\branchcolor{blue}{
\begin{proof}
  Proceeding as in the $\CQ d$ case, we run \algref{baseSets} on $\mathcal{L}$
  and let $E$ be the event that it does not abort. Observe that
  \begin{equation}
    \left|\Sumvalid \Pr[\mathbf{x}\leftarrow\mathcal{D}^{\mathcal{L}}]-\Sumvalid \Pr[\mathbf{x}\leftarrow\mathcal{D}^{\mathcal{L}}|E]\right|\le\negl\label{eq:conditionedOnE_CQC_d}
  \end{equation}
  as was the case before (recall $X_{\rm valid}$ was the set of valid solutions to $\CH{d}$). We will show in step one, that $\mathcal{D}^{\mathcal{L}}|E$
  and $\mathcal{D}^{\mathcal{M}}|E$ have essentially the same behaviour,
  i.e.
  \begin{equation}
    \left|\Pr[\mathbf{x}\leftarrow\mathcal{D}^{\mathcal{L}}|E]-\Pr[\mathbf{x}\leftarrow\mathcal{D}^{\mathcal{M}}|E]\right|\le\negl\label{eq:shadowOnlyUseless_CQC_d}
  \end{equation}
  and then in step two, that $\mathcal{D}^{\mathcal{M}}|E$ succeeds
  with at most negligible probability at solving $\CH d$ (see \Defref{dCodeHashingProblem}).
  These two steps, together with \Eqref{conditionedOnE_CQC_d}, entail
  that $\mathcal{D}^{\mathcal{L}}$ solves $\CH d$ with at most $\negl$
  probability.

  \begin{figure}
    \begin{centering}
      \includegraphics[width=8cm]{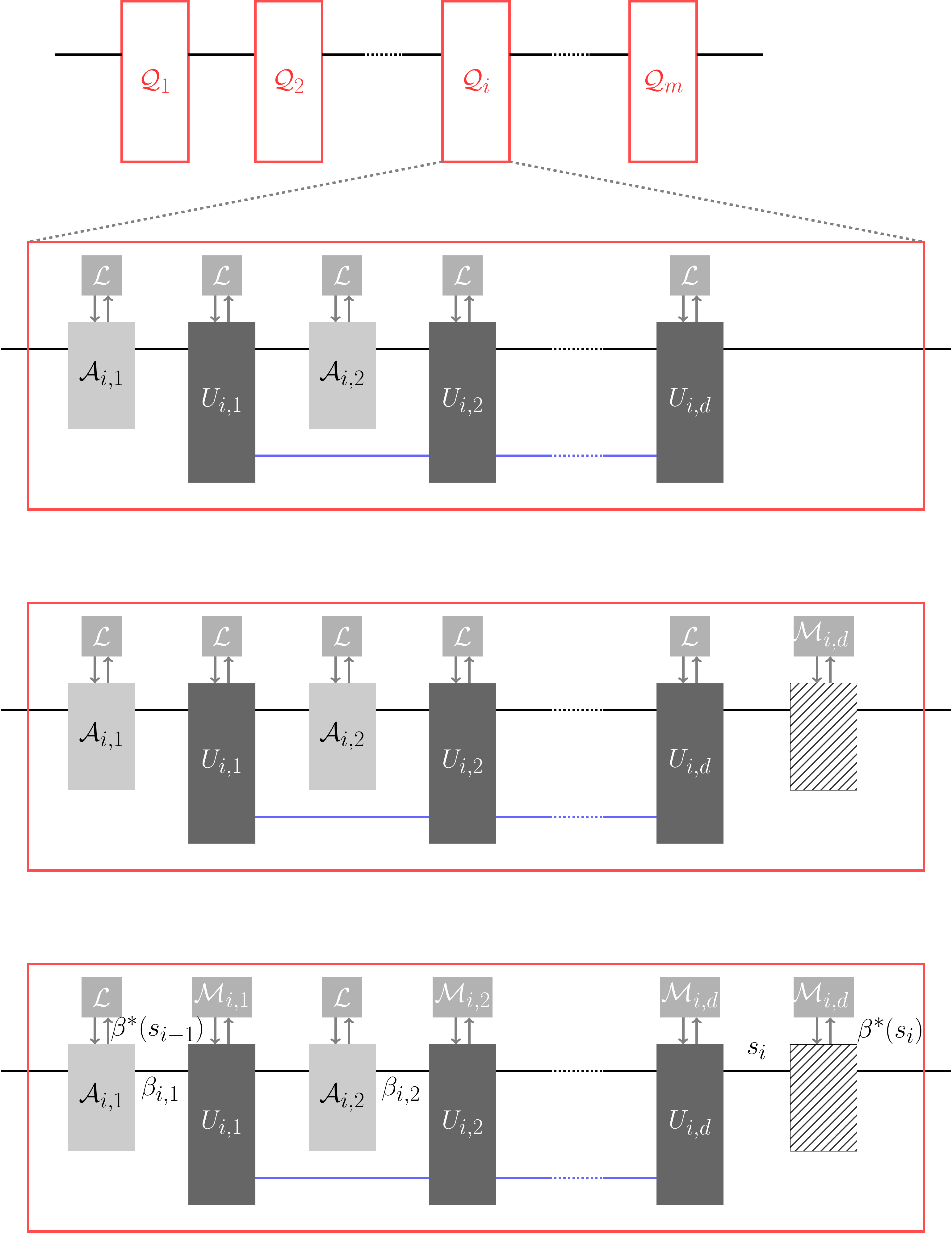}
      \par\end{centering}
    \caption{Black lines indicate wires carrying classical information. It is assumed
      that all circuits append their classical inputs into their classical
      outputs. The blue wires represent wires carrying quantum information.
      The figure is meant to illustrate three types of $\protect\CQC d$
      circuit, obtained by replacing the red blocks with the three respective
      circuits enclosed in red boxes. The second circuit is at least as
      powerful as the first, which in turn is the circuit we wish to study.
      The shaded circuit in the second red box represents the implementation
      of the sampling argument. \protect \\
      Notation in the third circuit: For each classical wire, only new information
      contained in that wire is labelled. For $i\in\{1\dots\tilde{n}\}$
      and $j\in\{1\dots\tilde{n}\}$, $\beta_{i,j}$ represents the paths
      in $\mathcal{L}$ queried by $\mathcal{A}_{i,j}$, $s_{i}$ denotes
      the measurement outcome after $U_{i,d}$, and $\beta^{*}(s_{i})$
      denotes the paths exposed by the sampling argument (the $^{*}$ indicates
      that the last coordinate may not be known). For $j=1$, $\mathcal{A}_{i,j}$
      also outputs $\beta(s_{i-1})$ which specifies the path $\beta^{*}(s_{i-1})$
      with the last coordinate also revealed. \protect \\
      We show that the behaviour of the second and third circuits is essentially
      the same where the third has its quantum parts only connected to shadow
      oracles. It is not hard to establish that the third circuit only solves
      $\protect\CH d$ with at most negligible probability. Together, these
      prove $\protect\CH d$ is hard for $\protect\CQC d$. \label{fig:CQCd_proof_idea}}

  \end{figure}

  In the rest of the proof, we \emph{implicitly condition everything
    on the event $E$}. Let $\bar{S}_{0}$ be the output of \Algref{baseSets}.
  \Figref{CQCd_proof_idea} may aid in visualising the overarching idea.
  We also make the simplifying assumption that all classical algorithms
  only query paths inside the base set $\bar{S}_{0}$. The general case
  changes almost nothing, but makes the notation more involved (especially
  since in this case we have classical algorithms after every layer
  of unitaries) and can be handled as in the proof of $\CQ d$ hardness.

  \textbf{Step One.} $\mathcal{D}^{\mathcal{L}}|E$ and $\mathcal{D}^{\mathcal{M}}|E$
  have essentially the same behaviour.\\
  Using a hybrid argument, one can bound the LHS of \Eqref{shadowOnlyUseless_CQC_d}
  by bounding
  \[
    \TD[\mathcal{D}^{\mathcal{L}},\mathcal{D}^{\mathcal{M}}]=\td[\mathcal{D}_{\tilde{n}}^{\mathcal{L}}\dots\mathcal{D}_{1}^{\mathcal{L}}(\sigma_{0}),\quad\calD_{\tilde{n}}^{\vec{{\cal M}}_{\tilde{n}}}\dots\calD_{1}^{\vec{\calM}_{1}}(\sigma_{0})]
  \]
  with
  \begin{align}
     & \le\sum_{i=1}^{\tilde{n}}\TD[\calD_{\tilde{n}}^{{\cal L}}\dots\calD_{i+1}^{{\cal L}}\ \ \ \calD_{i}^{{\cal L}}\dots\calD_{1}^{{\cal L}},\quad\calD_{\tilde{n}}^{{\cal L}}\dots\calD_{i+1}^{{\cal L}}\quad\calD_{i}^{\vec{{\cal M}}_{i}}\dots\calD_{1}^{\vec{\calM}_{1}}]\nonumber \\
     & \le\sum_{i=1}^{\tilde{n}}\TD[\calD_{i}^{{\cal L}}\dots\calD_{1}^{{\cal L}},\quad\calD_{i}^{{\cal \vec{M}}_{i}}\dots\calD_{1}^{\vec{\calM}_{1}}].\label{eq:sum_of_Ds_CQCd}
  \end{align}
  \textbf{The $i=1$ case:} \\
  This $i=1$ case may be seen as an adaptation of the $\QC d$ hardness
  proof, using a slightly more general notation suited for our analysis
  here. Our goal is to bound $\TD[\calD_{1}^{\calL},{\cal D}_{1}^{\vec{{\cal M}}_{1}}]$
  but we have not completely specified $\vec{{\cal M}}_{1}$. To this
  end, consider $\mathcal{D}_{1}^{\vec{{\cal M}}_{1}}={\cal B}_{1,d}^{\vec{{\cal M}}_{1}}\dots{\cal B}_{1,1}^{\vec{{\cal M}}_{1}}$
  where recall ${\cal B}_{1,j}^{\vec{{\cal M}}_{1}}=\Pi_{1,j}\circ{\cal M}_{1,j}\circ U_{1,j}\circ{\cal A}_{1,j}^{\calL}$.
  Let $\beta_{1,j}$ denote the set of paths (wrt $\bar{S}_{0}$; see
  \Defref{valid_beta_paths}) queried by $\mathcal{A}_{1,j}^{\mathcal{L}}$
  when $\mathcal{D}_{1}^{\vec{\mathcal{M}}_{1}}$ is executed. For $j\in\{1\dots d\}$,
  let $\bar{S}_{1,j}$ be the output of \Algref{vec_S_i_given_beta_previous_vecSs}
  with the index $i\leftarrow j$, base sets $\bar{S}_{0}\leftarrow\bar{S}_{0}$,
  the paths $\beta\leftarrow\cup_{j'\in\{1\dots j\}}\beta_{1,j'}$ and
  the previous sequence of sets $\bar{S}_{i-1}\leftarrow\bar{S}_{1,j-1}$
  as inputs.\footnote{And the assertion that $E$ happened. We don't explicitly state this
    any more.} When $j=1$, instead of $\bar{S}_{1,j-1}$ use $\bar{S}_{i-1}\leftarrow\bar{S}_{0}$.
  Finally, define $\calM_{1,j}$ as the shadow of $\mathcal{L}$ wrt
  $\bar{S}_{1,j}$.

  Returning to the bound, one can write
  \begin{align}
    \td[\calD_{1}^{\mathcal{L}},{\cal D}_{1}^{\vec{\mathcal{M}}_{1}}] & =\TD[{\cal B}_{1,d}^{{\cal L}}\dots{\cal B}_{1,1}^{\calL}(\sigma_{0}),\quad{\cal B}_{1,d}^{\vec{\calM}_{1}}\dots{\cal B}_{1,1}^{\vec{{\cal M}}_{1}}(\sigma_{0})]\nonumber \\
                                                                      & \le\sum_{j=1}^{d}\TD[{\cal B}_{1,j}^{{\cal L}}(\rho_{1,j-1}),{\cal B}_{1,j}^{\vec{{\cal M}}_{1}}(\rho_{1,j-1})]\nonumber                                                  \\
                                                                      & \le\sum_{j=1}^{d}\sqrt{\Pr[{\rm find}:U_{1,j}^{{\cal L}\backslash\bar{S}_{1,j}},{\cal A}_{1,j}^{{\cal L}}(\rho_{1,j-1})]}\label{eq:prFind_i1_CQC_d}
  \end{align}
  where for $j\in\{1\dots d-1\}$, $\rho_{1,j}:={\cal B}_{1,j}^{\vec{{\cal M}}_{1}}\dots{\cal B}_{1,1}^{\vec{{\cal M}}_{1}}(\sigma_{0})$
  and we used \Lemref{O2H} (and the relation between $\td$ and ${\rm B}$)
  to obtain the last inequality. To bound the RHS of \Eqref{prFind_i1_CQC_d},
  one can apply \Corref{Conditionals}. Use $\bar{S}_{1,j}\backslash X(\beta_{1,j+1})$
  (see \Defref{valid_beta_paths}) to denote the sequence of sets $(S_{1,j,k}\backslash X_{k}(\beta_{1,j+1}))_{k\in\{1\dots d\}}$.
  Let $\check{\mathcal{L}}_{1,j}$ be $\mathcal{L}$ outside \textbf{$\bar{S}_{1,j}\backslash X(\beta_{1,j+1})$}
  (see \Notaref{LhatAndsuch}). Observe that $\rho_{1,j-1}|\check{\mathcal{L}}_{1,j-1}$
  is uncorrelated with $\bar{S}_{1,j}|\check{\mathcal{L}}_{1,j-1}$
  because $\mathcal{A}_{1,j}^{\mathcal{L}}(\rho_{1,j-1})$ at most specifies
  $\check{\mathcal{L}}_{1,j-1}$; the queries made by $\mathcal{A}_{1,j}$
  have been exposed in $\check{\mathcal{L}}_{1,j-1}$ and are by construction
  of \Algref{vec_S_i_given_beta_previous_vecSs}, excluded from $\bar{S}_{1,j}$.
  Using the notation for conditioning $\Pr[{\rm find}:\dots]$, one
  can apply, for each $j\in\{1\dots d\}$, \Corref{Conditionals} together
  with \Claimref{x_in_S_CQC_d} (where $\delta\leftarrow0$, $\beta\leftarrow\cup_{j'\in\{1,\dots j\}}\beta_{1,j'}$,
  $\bar{S}_{0}\leftarrow\bar{S}_{0}$ and $\bar{S}_{i-1}\leftarrow\bar{S}_{1,j-1}$)
  to obtain
  \[
    \Pr[{\rm find}:U_{1,j}^{\mathcal{L}\backslash\bar{S}_{1,j}},\mathcal{A}_{1,j}^{\mathcal{L}}(\rho_{1,j-1})|\check{\mathcal{L}}_{1,j-1}]\le\negl.
  \]
  This in turn bounds the RHS of \Eqref{prFind_i1_CQC_d} by $\negl$.

  \textbf{The $i=2$ case}:\\
  Since we analysed the $i=1$ case using a more general notation (than
  both the $\CQ d$ case and the $\QC d$ case), we proceed as in that
  case but additionally, apply the sampling argument to account for
  the output of ${\cal D}_{1}^{\vec{{\cal M}}_{1}}$. Let $\sigma_{1}:=\calD_{1}^{\vec{\calM}_{1}}(\sigma_{0})$.
  Our goal is to bound $\TD[{\cal D}_{2}^{\calL}(\sigma_{1}),\calD_{2}^{\vec{\calM}_{2}}(\sigma_{1})]$
  but we have not yet specified $\vec{\calM}_{2}$.

  To this end, we apply the sampling argument. Let\footnote{We slightly abuse the notation. By $\cup_{j=1}^{d}\beta_{1,j}$ we
    mean component-wise union as each $\beta_{1,j}$ is a sequence of
    sets.} $\beta_{1}:=\cup_{j=1}^{d}\beta_{1,j}$ be the set of paths queried
  by the classical algorithms in $\calD_{1}^{\vec{\calM}_{1}}$. Note
  that $\inj{{\cal L}|\beta_{1}}$ (see \Notaref{inj_from_L}) is distributed
  as $\mathbb{F}_{{\rm inj}}^{|\beta_{1}}$. Let the string $s_{1}$
  denote the output of ${\cal D}_{1}^{\vec{\calM}_{1}}$. Given that
  $\Pr[s_{1}|\beta_{1}]\ge\gamma$, $\inj{\calL|\beta_{1}s_{1}}$ may
  be expressed as a convex combination (as described in \Propref{main_delta_non_uniform_inj_shuffler})
  over $\inj{\calL|s_{1}\beta_{1}\beta^{*}(s_{1})}$ which are distributed
  as $\mathbb{F}_{{\rm inj}}^{(p,\delta)|\beta_{1}}$ where $\left|\beta^{*}(s_{1})\right|\le p\le2m/\delta$
  whenever the convex coefficient is larger than $\gamma$. When $\Pr[s_{1}|\beta_{1}]<\gamma$,
  let $\beta^{*}(s_{1})=\emptyset$. This implicitly defines the random
  variable $\beta^{*}(s_{1})$ which was initially left unspecified.
  The first (classical) circuit of ${\cal D}_{2}^{\vec{{\cal M}_{2}}}$,
  i.e. ${\cal A}_{2,1}^{{\cal L}}$, takes as input $s_{1}$ and $\beta^{*}(s_{1})$.
  We assume (without loss of generality) ${\cal A}_{2,1}^{{\cal L}}$
  learns $\beta(s_{1})$ which is $\beta^{*}(s_{1})$ with $\perp$s
  replaced by the value ${\cal L}$ takes in the last coordinate.

  We now proceed as in the $i=1$ case and consider $\calD_{2}^{\vec{{\cal M}_{2}}}={\cal B}_{2,d}^{\vec{{\cal M}_{2}}}\dots{\cal B}_{2,1}^{\vec{{\cal M}}_{2}}$
  acting on $\sigma_{1}$ where ${\cal B}_{2,j}^{\vec{{\cal M}}_{2}}=\Pi_{2,j}\circ{\cal M}_{2,j}\circ U_{2,j}\circ{\cal A}_{2,j}^{{\cal L}}$.
  Let $\beta_{1,j}$ denote the set of paths queried by ${\cal A}_{2,j}^{{\cal L}}$
  when ${\cal D}_{2}^{\vec{{\cal M}}_{2}}$ is executed (for $j=1$,
  $\beta_{1,j}$ counts paths distinct from $\beta(s_{1})$). For $j\in\{1\dots d\}$,
  let $\bar{S}_{2,j}$ be the output of \Algref{vec_S_i_given_beta_previous_vecSs}
  with index $i\leftarrow j$, base sets $\bar{S}_{0}\leftarrow\bar{S}_{0}$,
  the paths $\beta\leftarrow\cup_{j'\in\{1\dots j\}}\beta_{2,j'}\cup\beta_{1}\cup\beta(s_{1})$,
  and the previous sequence of sets $\bar{S}_{i-1}\leftarrow\bar{S}_{2,j-1}$
  as inputs. When $j=1$, instead use $\bar{S}_{i-1}\leftarrow\bar{S}_{0}$.
  Finally, define ${\cal M}_{2,j}$ as the shadow of ${\cal L}$ wrt
  $\bar{S}_{2,j}$.

  Returning to the bound, one can write
  \begin{align}
    \td[{\cal D}_{2}^{{\cal L}}(\sigma_{1}),{\cal D}_{2}^{\vec{{\cal M}}_{2}}(\sigma_{1})] & \le\sum_{j=1}^{d}\td[{\cal B}_{2,j}^{{\cal L}}(\rho_{2,j-1}),{\cal B}_{2,j}^{\vec{{\cal M}}_{2}}(\rho_{2,j-1})]\nonumber                           \\
                                                                                           & \le\sum_{j=1}^{d}\sqrt{\Pr[{\rm find}:U_{2,j}^{{\cal L}\backslash\bar{S}_{2,j}},{\cal A}_{2,j}^{{\cal L}}(\rho_{2,j-1})}\label{eq:prFind_i2_CQC_d}
  \end{align}
  where for $j\in\{1\dots d\}$, $\rho_{2,j}:={\cal B}_{2,j}^{\vec{{\cal M}}_{2}}\dots{\cal B}_{2,1}^{\vec{{\cal M}}_{2}}(\sigma_{1})$
  and we used \Lemref{O2H} to get the last inequality. To bound the
  RHS \Eqref{prFind_i2_CQC_d}, one can apply \Corref{Conditionals}.
  Let $\check{{\cal L}}_{2,j}$ be ${\cal L}$ outside $\bar{S}_{2,j}\backslash X(\beta_{2,j+1})$
  (see \Notaref{LhatAndsuch}). Observe that $\rho_{2,j-1}|\check{{\cal L}}_{2,j-1}$
  is uncorrelated with $\bar{S}_{2,j}|\check{{\cal L}}_{2,j-1}$ (for
  the same reason as the $i=1$ case). However, to apply \Claimref{x_in_S_CQC_d}
  we condition on the ``transcript'' until the output of ${\cal A}_{2,1}$,
  i.e. $T_{2}:=:T({\cal A}_{2,1}(\sigma_{1})):=(\beta_{1},s_{1},\beta(s_{1}),\beta_{2,1})$,
  by writing $\Pr[{\rm find}:U_{2,j}^{{\cal L}\backslash\bar{S}_{2,j}},{\cal A}_{2,j}^{{\cal L}}(\rho_{2,j-1})|\check{{\cal L}}_{2,j-1}]$
  \begin{align*}
     & =\sum_{s_{1},\beta_{1},\beta(s_{1}),\beta_{2,1}}\Pr[T_{2}]\cdot\Pr[{\rm find}:U_{2,j}^{{\cal L}\backslash\bar{S}_{2,j}},{\cal A}_{2,j}^{{\cal L}}(\rho_{2,j-1})|\check{{\cal L}}_{2,j-1}T_{2}]                                                                                                           \\
     & \le\sum_{\substack{s_{1}:\Pr[s_{1}|\beta_{1}]\ge2^{-m}                                                                                                                                                                                                                                                   \\
        \beta_{1},\beta(s_{1}),\beta_{2,1}
      }
    }\underbrace{\Pr[\beta_{2,1}|s_{1}\beta_{1}\beta(s_{1})]\Pr[\beta(s_{1})|s_{1}\beta_{1}]\Pr[s_{1}|\beta_{1}]\Pr[\beta_{1}]}_{\Pr[T_{2}]}\cdot\Pr[{\rm find}:U_{2,j}^{{\cal L}\backslash\bar{S}_{2,j}},{\cal A}_{2,j}^{{\cal L}}(\rho_{2,j-1})|\check{{\cal L}}_{2,j-1}T_{2}]+2^{-(m-\tilde{m})}             \\
     & \le\sum_{\substack{s_{1}:\Pr[s_{1}|\beta_{1}]\ge2^{-m}                                                                                                                                                                                                                                                   \\
    \beta(s_{1}):\Pr[\beta(s_{1})|s_{1}\beta_{1}]\ge2^{-m}                                                                                                                                                                                                                                                      \\
    \beta_{1},\beta_{2,1}
    }
    }\Pr[\beta_{2,1}|s_{1}\beta_{1}\beta(s_{1})]\Pr[\beta(s_{1})|s_{1}\beta_{1}]\Pr[s_{1}|\beta_{1}]\Pr[\beta_{1}]\cdot\underbrace{\Pr[{\rm find}:U_{2,j}^{{\cal L}\backslash\bar{S}_{2,j}},{\cal A}_{2,j}^{{\cal L}}(\rho_{2,j-1})|\check{{\cal L}}_{2,j-1}T_{2}]}_{\mathsf{Term\ I}}+2\cdot2^{-(m-\tilde{m})} \\
     & \le\negl
  \end{align*}
  where to obtain the first inequality (proceeding almost exactly as
  in the $\CQ d$ case), we note that for each $s_{1}:\Pr[s_{1}|\beta_{1}]\ge2^{-m}$,
  one can use \Propref{main_delta_non_uniform_inj_shuffler} and one
  can account for all $s_{1}:\Pr[s_{1}|\beta_{1}]<2^{-m}$, by simply
  upper bounding the sum by $2^{-(m-\tilde{m})}$ because $s_{1}$ is
  of length $\tilde{m}$. In the second inequality, we use the fact
  that either the convex weight (i.e. $\Pr[\beta(s_{1})|s_{1}\beta_{1}]$)
  as specified in \Propref{main_delta_non_uniform_inj_shuffler} is
  less than $2^{-m}$ (for at most each $s_{1}$, there it contributes
  at most $2^{-(m-\tilde{m})}$ to the sum) or it is greater than $2^{-m}$.
  In the latter case, the injective shuffler is distributed as $\mathbb{F}_{{\rm inj}}^{\delta|\beta}$
  where $\beta\leftarrow\beta_{1}\cup\beta(s_{1})\cup\beta_{2,1}$.
  Therefore, one can apply \Corref{Conditionals} together with \Claimref{x_in_S_CQC_d}
  (where $\delta\leftarrow\delta$, $\beta\leftarrow\cup_{j'\in\{1,\dots j\}}\beta_{2,j'}\cup\beta_{1}\cup\beta(s_{1})$,
  $\bar{S}_{0}\leftarrow\bar{S}_{0}$ and $\bar{S}_{i-1}\leftarrow\bar{S}_{2,j-1}$)
  to obtain $\mathsf{Term\ 1}\le2^{\delta}\cdot\poly\cdot\negl$.

  \textbf{The general $i\in\{1\dots\tilde{n}\}$ case:}\\
  This is a straightforward generalisation of the $i=2$ case and hence
  we only outline the key steps. Let $\sigma_{i}:={\cal D}_{i}^{\vec{{\cal M}}_{i}}\dots{\cal D}_{1}^{\vec{{\cal M}}_{1}}(\sigma_{0})$
  where $\calM_{i-1,j}$ is the shadow of ${\cal L}$ wrt $\bar{S}_{i-1,j}$
  which in turn are defined below. It may help to keep the last circuit
  of \Figref{CQCd_proof_idea} in mind. Consider ${\cal D}_{i}^{\vec{{\cal M}}_{i}}={\cal B}_{i,d}^{\vec{{\cal M}}_{i}}\dots{\cal B}_{i,1}^{\vec{{\cal M}}_{i}}$
  acting on $\sigma_{i-1}$ where recall ${\cal B}_{i,j}^{\vec{{\cal M}}_{2}}=\Pi_{i,j}\circ{\cal M}_{i,j}\circ U_{i,j}\circ{\cal A}_{i,j}^{\calL}$.
  Let $\beta_{i,j}$ denote the set of paths queried by ${\cal A}_{i,j}^{\calL}$
  when ${\cal D}_{i}^{\vec{{\cal M}}_{i}}$ is executed. Let $\beta_{i}:=\cup_{j\in\{1,\dots d\}}\beta_{i,j}$.
  Let $s_{i}$ denote the string output by ${\cal D}_{i}^{\vec{{\cal M}}_{i}}(\sigma_{i-1})$.

  Now, we apply the sampling argument to $\inj{{\cal L}|\beta_{1}s_{1}\beta(s_{1})\dots\beta_{i-1}}$.
  Let $\beta^{*}(s_{i-1})$ be the paths as in \Propref{main_delta_non_uniform_inj_shuffler}
  such that when $\Pr[s_{i-1}|\beta_{1}s_{1}\beta(s_{1})\dots\beta_{i-1}]\ge\gamma$,
  $\inj{{\cal L}|\beta_{1}s_{1}\beta(s_{1})\dots\beta_{i-1}s_{i-1}}$
  is distributed as $\mathbb{F}_{{\rm inj}}^{(i-2)\delta|\beta_{1}\cup\dots\cup\beta_{i-1}\cup\beta(s_{1})\cup\dots\beta(s_{i-1})}$
  so that\footnote{Note that $i\ge2$ here because $s_{i-1}$ is $s_{1}$ when $i=2$. }
  $\inj{\mathcal{L}|\beta_{1}s_{1}\beta(s_{1})\dots\beta_{i-1}s_{i-1}\beta^{*}(s_{i-1})}$
  is distributed as $\mathbb{F}_{{\rm inj}}^{(p,(i-2)\delta)|\beta_{1}\cup\dots\cup\beta_{i-1}\cup\beta(s_{1})\cup\dots\cup\beta(s_{i-1})}$
  whenever the convex coefficient (i.e. probability associated with
  $\beta^{*}(s_{i-1})$) is larger than $\gamma=2^{-m}$. $\beta(s_{i-1})$
  is $\beta^{*}(s_{i-1})$ with $\perp$s replaced by the values taken
  by ${\cal L}$ at those coordinates.

  Returning to $\bar{S}_{i,j}$, define it to be the output of \Algref{vec_S_i_given_beta_previous_vecSs}
  with index $i\leftarrow j$, base sets $\bar{S}_{0}\leftarrow\bar{S}_{0}$,
  the paths $\beta\leftarrow\left(\cup_{j'\in\{1\dots j\}}\beta_{i,j'}\right)\cup\left(\beta_{i-1}\cup\dots\cup\beta_{1}\right)\cup\left(\beta(s_{i-1})\cup\dots\cup\beta(s_{1})\right)$,
  and the previous sequence of sets $\bar{S}_{i-1}\leftarrow\bar{S}_{i,j-1}$
  as inputs. When $j=1$, use $\bar{S}_{i-1}\leftarrow\bar{S}_{0}$
  instead.

  To obtain the bound, we need two more definitions. Let the ``transcript''
  be denoted by $T_{i}:=:T({\cal A}_{i,1}(\sigma_{i-1})):=(\beta_{1},s_{1},\beta(s_{1}),\dots\beta_{i-1},s_{i-1},\beta(s_{i-1}),\beta_{i,1})$.
  Let $\check{{\cal L}}_{i,j}$ be ${\cal L}$ outside $\bar{S}_{i,j}\backslash X(\beta_{i,j+1})$
  (see \Notaref{LhatAndsuch}).

  We bound the $i$th term in \Eqref{sum_of_Ds_CQCd} by expressing
  it as
  \[
    \TD[{\cal D}_{i}^{{\cal L}}(\sigma_{i-1}),{\cal D}_{i}^{\vec{{\cal M}}_{i}}(\sigma_{i-1})]\le\sum_{j=1}^{d}\TD[{\cal B}_{i,j}^{{\cal L}}(\rho_{i,j-1}),{\cal B}_{i,j}^{\vec{{\cal M}}_{i}}(\rho_{i,j-1})]
  \]
  where $\rho_{i,j}:={\cal B}_{i,j}^{\vec{{\cal M}}_{i}}\dots{\cal B}_{i,j}^{\vec{{\cal M}}_{i}}(\sigma_{i-1})$.
  The square of the $j$th term can then be bounded (using \Lemref{O2H})
  by $\Pr[{\rm find}:U_{i,j}^{{\cal L}\backslash\bar{S}_{i,j}},{\cal A}_{i,j}^{{\cal L}}(\rho_{i,j-1})]$
  which is

  \begin{align*}
     & \le\sum_{\substack{s_{1}:\Pr[s_{1}|\beta_{1}]\ge2^{-m},\dots,s_{i-1}:\Pr[\beta_{1}\dots]\ge2^{-m}                                                                                                                \\
    \beta(s_{1}):\Pr[\beta(s_{1})|s_{1}\beta_{1}]\ge2^{-m},\dots\beta(s_{i-1}):\Pr[\beta(s_{i-1})|s_{1}\beta_{1}\dots]\ge2^{-m}                                                                                         \\
    \beta_{1},\beta_{2}\dots\beta_{i-1},\ \beta_{i,1}
    }
    }\Pr[T_{i}]\cdot\underbrace{\Pr[{\rm find}:U_{i,j}^{{\cal L}\backslash\bar{S}_{i,j}},{\cal A}_{i,j}^{{\cal L}}(\rho_{i,j-1})|\check{{\cal L}}_{i,j-1},T_{i}]}_{\mathsf{Term\ I}}+2\cdot(i-1)\cdot2^{-(m-\tilde{m})} \\
     & \le2^{\Delta}\cdot\poly\cdot\negl+2\cdot(i-1)\cdot2^{-(m-\tilde{m})}\le\negl
  \end{align*}
  where the distribution of the injective shuffler in $\mathsf{Term\ I}$
  is $\mathbb{F}^{i\cdot\delta|\left(\beta(s_{i-1})\cup\dots\beta(s_{1})\right)\cup\left(\beta_{i-1}\cup\dots\beta_{1}\right)\cup\beta_{i,1}}$.
  This is obtained by repeatedly applying \Propref{main_delta_non_uniform_inj_shuffler}
  (for the $k$th application, $\delta'\leftarrow(k-1)\delta,\ \beta\leftarrow\beta_{k,1}\cup\left(\beta(s_{k-1})\cup\dots\beta(s_{1})\right)\cup\left(\beta_{k-1}\cup\dots\beta_{1}\right)$
  and $\bar{S}_{0}\leftarrow\bar{S}_{0}$) and arguing as before to
  collect terms for which the distribution over the injective shuffler
  is unknown (but which occur with probability at most $2^{-m}$). Independence
  of $\bar{S}_{i,j}$ from $\rho_{i,j-1}$ can be argued as before once
  it is conditioned on $\check{{\cal L}}_{i,j-1}$ and one can apply
  \Corref{Conditionals} together with \Claimref{x_in_S_CQC_d} (with
  $\delta\leftarrow i\cdot\delta,\ \beta\leftarrow\left(\cup_{j'\in\{1,\dots j\}}\beta_{i,j'}\right)\cup\left(\beta_{i-1}\cup\dots\beta_{1}\right)\cup\left(\beta(s_{i-1})\cup\dots\beta(s_{1})\right)$,
  $\bar{S}_{0}\leftarrow\bar{S}_{0}$ and $\bar{S}_{i-1}\leftarrow\bar{S}_{i,j-1}$)
  to obtain the stated bound on $\mathsf{Term\ I}$ (recall $\gamma=2^{-m}$
  and $\delta=\Delta/\tilde{n}$).

  \textbf{Step Two.} ${\cal D}^{{\cal M}}|E$ succeeds at solving $\CH d$
  with at most negligible probability. The argument for the $\CQ d$
  case go through with the only change that there are more classical
  algorithms to account for but this does not affect the conclusion.
\end{proof}
}

\section{Proof of Quantum Depth\label{sec:defn_PoQD}}
Since YZ's $\codehashing$ can be efficiently verified (i.e. it is $\NP$), it is evident that $\CH{d}$ can also be efficiently verified. Therefore $\CH{d}$ also serves a proof of quantum depth. However, in the cryptographic context, one would ideally like security against oracle dependent adversaries (in our proofs so far, we assumed the adversary is oracle independent). Fortunately, this issue can be resolved generically and to this end, we first formalise what we mean by a proof of quantum depth. YZ also followed a similar approach for their proof of quantumness which is based on $\codehashing$. 
\subsection{The Definition}
It may help to recall the definitions of uniform and non-uniform oracle dependent adversaries (see \Subsecref{uniform_nonuniform}). %

\begin{defn}[Proof of $d$ Quantum Depth in the Random Oracle Model\label{def:PoD}] 
    Consider three algorithms, $(\gen,\ver^H)$ and $\prov^H$.
    \begin{description}
        \item[$\gen(1^{\lambda})$.] A PPT algorithm which returns $(\sk,\pk)$. %
        \item[$\ver^H(\sk,\pk,\pi)$.] A PPT algorithm that makes at most $\ply{\lambda}$ queries to $H$ and outputs $0$ or $1$.
        \item[$\prov^H(\pk)$.] Consider an oracle independent quantum circuit family $\{\calC_n\}_n$. %
        $\prov^H(\pk)$ executes $\calC_{|\pk|}$ with input $\pk$.
    \end{description}
    The algorithms $(\gen,\ver^H)$ and $\prov^H$ constitute a Proof of $d$ Quantum Depth in the Random Oracle Model, if the following holds for every sufficiently large security parameter $\lambda$.
    \begin{itemize}
        \item \emph{Completeness.} There is an honest prover which applies a poly-sized quantum circuit, i.e. $\{\calC_n\} \in \QPT$ for all $n$, with the following property. Let $\prov^{H}(\pk)$ be $\calC_{|\pk|}$ with input $\pk$. Then, the verifier interacts with the prover and accepts with overwhelming probability, i.e. 
            \begin{equation*}
                \Pr_{H} \left[ \ver^H(\sk,\pk,\pi)=1 : \substack{(\sk,\pk)\leftarrow \gen(1^{\lambda}) \\ \pi \leftarrow \prov^H(\pk)} \right] \ge 1 - \ngl{\lambda}.
            \end{equation*}
        \item \emph{Soundness.} Consider any arbitrary prover which applies a $\CQC{d}$ circuit, i.e. $\{\calC_n\}_n$ where each $\calC_n \in \CQC{d}$. %
        Let $\prov^{H}(\pk)$ be $\calC_{|\pk|}$ with input $\pk$. Then, the verifier interacting with any such prover accepts with negligible probability, i.e.
            \begin{equation*}
                \Pr_{H} \left[ \ver^H(\sk,\pk,\pi^*)=1 : \substack{(\sk,\pk)\leftarrow \gen(1^{\lambda}) \\ \pi^* \leftarrow \prov^{H}(\pk)} \right] \le  \ngl{\lambda}
            \end{equation*}
        for all $\prov^{H}$.
    \end{itemize}
    Soundness against uniformly and non-uniformly oracle dependent provers is defined analogously. When $|\pk|=0$, the prover is given $1^{\lambda}$ as input.
\end{defn}

Observe that the protocol above is a two-message protocol (the verifier sends $\pk$ and the prover sends $\pi$). \branchcolor{purple}{In fact, observe that any two-message protocol (where the verifier is classical and sends the first message) can be cast in the aforementioned form by splitting the verification algorithm into two $(\gen,\verify)$ and have all the information from $\gen$ passed to $\verify$ and some information from $\gen$ passed to $\prove$ as the first message.} In addition to being two-message, the protocol above may also have the following properties if the appropriate conditions are satisfied. %
\begin{description}
    \item[Publicly verifiable:] If $|\sk|=0$ the proof can be publicly verified by looking at the transcript.\footnote{It is standard practice to assume that the algorithms themselves are public knowledge.} If, in addition, $|\pk|>0$, then we call the proof of quantum depth \emph{keyed}.
    \item[Non-interactive (or \emph{keyless}):] If $|\pk|=0$ the verifier does not need to send any information to the prover. Note that soundness in this case cannot hold against non-uniform adversaries.\footnote{The proof can be hardcoded into the prover's advice.}%
\end{description}

We conclude by noting that in our definition of proof of quantum depth, we allowed the completeness to be $\BQP$ which may not be practical. This is analogous to the definition of proof of quantumness where the soundness is against $\BPP$ and completeness is again $\BQP$. In both cases, it is desirable to have low depth circuits\footnote{Ideally, $\QNC_{\calO(1)}$ for quantumness and $\QNC_{\calO(d+1)}$ for quantum depth} suffice for establishing completeness. Nonetheless, they are meaningful formalisations because they do certify the respective notions of quantum depth and quantumness.

\subsection{Salting and oracle dependent adversaries}

\newcommand{\samp}{\mathsf{Samp}}
\newcommand{\query}{\mathsf{Query}}
\newcommand{\ch}{\mathsf{ch}}
\newcommand{\ans}{\mathsf{ans}}
\newcommand{\Gmi}[1]{G^{\otimes #1}}
\newcommand{\Cmi}[1]{C^{\otimes #1}}
\newcommand{\eps}{\varepsilon}
\newcommand{\crh}{{\sf CRH}}

\newcommand{\As}{\mathcal{A}}

For non-interactive proofs of quantum depth (as for non-interactive proofs of quantumness~\cite{yamakawa_verifiable_2022}) in the %
random oracle model, security holds only against oracle-independent adversaries, i.e.\ adversaries that are fixed before the random oracle is chosen, but not against non-uniform oracle-dependent adversaries, i.e.\ adversaries that receive advice strings after the random oracle has been chosen. To see this, observe that the advice can be arbitrarily correlated with the chosen random oracle. For instance, in our setting, the advice could simply be the codeword that hashes as required. Then, an adversary which simply outputs the advice it receives can already break the security of the proof of quantum depth protocol. %

To also achieve security against non-uniform oracle-dependent adversaries, we rely on a result of Chung et al.~\cite{FOCS:ChungGLQ20}. In this work, it was shown (among other results) that salting, i.e. appending a random string to the query, can be used to render the oracle-dependent advice useless. This is a quantum adaptation of the results of \cite{EC:CorettiDGS18} and can be used in our setting to turn a non-interactive proof of quantum depth secure against oracle-independent adversaries into an interactive (two message) proof of quantum depth secure against non-uniform oracle-dependent adversaries, in which the first message (sent by the verifier) only consists simply of a random string. More formally, the following holds.

\begin{thm}
    \label{thm:salt-poqd}
    Let $\poqd=(\prove^H,\verify^H(\pi))$ be a keyless proof of quantum depth secure against oracle-independent adversaries, then $\poqd'=(\gen'(1^\secpar),\prove'^H(\pk),\verify'^H(\pk,\pi))$ is a keyed proof of quantum depth secure against oracle-dependent non-uniform adversaries, where $\gen(1^\secpar)$ simply outputs a random $\pk$, i.e.\ $\pk\gets\{0,1\}^\secpar$.
\end{thm}

This theorem is an immediate consequence of~\cite[Theorem 7.4]{FOCS:ChungGLQ20} where a proof of quantum depth is viewed as a publicly verifiable security game~\cite[Definition 3.3]{FOCS:ChungGLQ20}.
Yamakawa and Zhandry used the same result to lift their oracle-independent security to non-uniform security for the case of proofs of quantumness, one-way functions and collision resistant functions \cite[Theorem 3.7 \& 3.8]{yamakawa_verifiable_2022}.

\ignore{
Yamakawa and Zhandry already applied the results of~\cite{FOCS:ChungGLQ20} in the same way to lift oracle-independent security to non-uniform security for the case of one-way functions and collision resistant functions \cite[Theorem 3.7 \& 3.8]{yamakawa_verifiable_2022}. They also applied those results to proofs of quantumness \cite[Theorem 6.1 \& Corollary 6.2]{yamakawa_verifiable_2022} \hendrik{In Theorem 3.8 they actually also mention the fact that a proof of quantumness can be salted to be secure against non-uniform adversaries}

\begin{thm}[Informal]
  For any publicly verifiable security game $G$, adversary time bound $T$, and any $g > 0$, consider its multi-instance game $\Gmi g$, which requires the adversary to break $g$ independent challenges \emph{sequentially}, and each instance is given time $T$. If the best winning probability for $\Gmi g$ is $\delta^g$ in the QROM, meaning that the best adversary can only win the game with probability at most $\delta^g$, then for adversaries with $S = g$ \hendrik{I guess this should enforce the polynomial bound, i.e.\ the adversary can't have more than poly bits of advice.} (qu)bits of advice, $G$ is $\negl$-secure in QAI-QROM if $\delta = \negl$.
\end{thm}

\begin{thm}
 Let $G$ be a publicly-verifiable security game with security $\delta'(T)$ in the QROM. Let $T_\ver = \poly(\log M, \log N)$ be the upper bound on the number of $\tilde{H}$ queries for computing $\ver^{\tilde{H}}(\ch, \cdot)$, i.e., the number of queries required to do the public verification. 
  Let $G_S$ be the salted security game with parameter $K$.
  Let $\delta_0 = 1/\poly(N, M)$ be a non-negligible lower bound on the winning probability of an adversary that outputs a random answer in the $G_S$ without advice or making any query. 
  Then $G_S$'s security $\delta = \delta(S, T)$ against $(S, T)$ adversaries in the QAI-QROM satisfying 
 \begin{align*}
     \delta \leq  \tilde{O}\left( \frac{S T}{K \delta} + \delta'\left(  \tilde O(T) /\delta \right)  \right),
 \end{align*}
 where $\tilde O$ absorbs $\poly(\log N, \log M, \log S)$ factors.
\end{thm}

\begin{defn}[Proof of Quantum Depth Security Game]
    Security game $G_\poqd=(C_\poqd)$ is specified by three procedures $(\samp,\query,\verify)$, where
    \begin{enumerate}
        \item $\samp^H(\bot)=\bot$
        \item $\query^H(\bot,\cdot)=H(\cdot)$ provides query access to $H$.
        \item $\verify^H(\bot,\pi)=\poqd.\verify^H(1^\secpar,\pi)$
    \end{enumerate}
\end{defn}

We can directly apply this onto our setting of proofs of quantum depth to obtain the following theorem.
}

\ignore{
\begin{defn}[$(d,d')$-Proof of Quantum Depth in the Random Oracle Model\label{def:PoD}] 
    Consider three algorithms, $(\gen,\ver^H)$ and $\prov^H$.
    \begin{description}
        \item[$\gen(1^{\lambda})$.] A PPT algorithm which returns $(\sk,\pk)$.\footnote{\hendrik{I am not sure if this is the best way to define it because we have the two subcases where gen doesn't output anything ($|\pk|=|\sk|=0$) and where it only outputs a random string ($|\sk|=0$). But I also don't really know a better way to define it. I had a definition with just an interactive prove procedure but then the soundness definition is not as nice as if you can abstract it with a gen algorithm. One possibility is also to have two definition (non-interactive and publicly verifiable and interactive) but I also do not think that this is optimal} \atul{I think one can just say that we are writing a general 2 message protocol and in every such protocol, one can imagine the classical verifier as having two parts---one we call ``gen'' and the other we call ``verify'' without loss of generality. The non-interactive thing, I think you're right; we can emphasise that explicitly in the ``property'' of being non-interactive. I have added some explanation below. Let me know if that helps.}}
        \item[$\ver^H(\sk,\pk,\pi)$.] A PPT algorithm that makes at most $\ply{\lambda}$ queries to $H$ and outputs $0$ or $1$.
        \item[$\prov^H(\pk)$.] Consider an oracle independent quantum circuit family $\{\calC_n\}_n$ where each $\calC_n \in \QC{d'} \cup \CQ{d'}$. $\prov^H(\pk)$ is $\calC_{|\pk|}$ with input $\pk$.
    \end{description}
    The algorithms $(\gen,\ver^H)$ and $\prov^H$ constitute a $(d,d')$ Proof of Quantum Depth in the Random Oracle Model, if the following holds for every sufficiently large security parameter $\lambda$.
    \begin{itemize}
        \item \emph{Completeness.} The verifier generates $(\sk,\pk)\leftarrow \gen(1^{\lambda})$, sends $\pk$ to the honest prover $\prov^H$ which returns $\pi$. The verifier runs $\ver^H(\sk,\pk,\pi)$ which outputs $1$ (accept) with overwhelming probability, i.e.
            \begin{equation*}
                \Pr_{H} \left[ \ver^H(\sk,\pk,\pi)=1 : \substack{(\sk,\pk)\leftarrow \gen(1^{\lambda}) \\ \pi \leftarrow \prov^H(\pk)} \right] \ge 1 - \ngl{\lambda}.
            \end{equation*}
        \item \emph{Soundness.} Consider an arbitrary oracle independent quantum circuit family $\{\calC_n\}_n$ where each $\calC_n \in \QC{d}\cup\CQ{d}$. Let $\prov^{*H}(\pk)$ be $\calC_{|\pk|}$ with input $\pk$. Then, the verifier interacting with any such prover accepts with negligible probabibility, i.e.
            \begin{equation*}
                \Pr_{H} \left[ \ver^H(\sk,\pk,\pi^*)=1 : \substack{(\sk,\pk)\leftarrow \gen(1^{\lambda}) \\ \pi^* \leftarrow \prov^{*H}(\pk)} \right] \le  \ngl{\lambda}.
            \end{equation*}
        for all $\prov^{*H}$.
    \end{itemize}
    Soundness against uniformly and non-uniformly oracle dependent provers is defined analogously. When $|\pk|=0$, the prover is given $1^{\lambda}$ as input.\footnote{\uttam{Shouldn't one say that $d<d'$?} \atul{Added a line above.}}
\end{defn}
}

\ignore{
    \hendrik{This is example has been taken from the paper. Our proof should be more or less the same as theirs.}
    
    \subsection{Collision Hashing Example}
    
    \begin{defn}[Collision-Resistant Hash Security Game]
        Security game $G_\crh = (C_\crh)$ is specified by three procedures $(\samp, \query, \ver)$, where:
        \begin{enumerate}
            \item $\samp^H(\bot) = \bot$.
            \item $\query^H(\bot, \cdot) = H(\cdot)$ provides query access to $H$.
            \item $\ver^H(\bot, (x_1, x_2))$ outputs $1$ if and only if $x_1 \neq x_2$ and $H(x_1) = H(x_2)$. 
        \end{enumerate}
    \end{defn}
    It is easy to see that $S = 2\log M$ suffices to break this security game with optimal probability.
    However, without advice, the problem is hard even against quantum computations.
    
    \begin{prop}[{\cite[Corollary 2]{zhandry2019record}}]
        $G_\crh$ has security
        $\Theta(T^3/M)$ in QROM.
    \end{prop}
    
    It is easy to verify that this game is publicly-verifiable, as $\ver^H(\bot, (x_1, x_2)) = \ver^{\query^H}(\bot, (x_1, x_2))$.
    Denote the salted collision-resistant hash game as $G_{\crh, S}$.
    Combining this lemma with \thmref{generalsalting-aiqrom}, and the fact that random guessing has winning probability $1/M$, we obtain the following corollary.
    
    \begin{cor}
        $G_{\crh,S}$ is $O(T^3/M + (S + \log M) T / K)$-secure in the AI-QROM, and $\tilde O(T^3/M + ST/K)^{1/4}$-secure in the QAI-QROM, where $\tilde O$ absorbs $\poly(\log N, \log M, \log K, \log S)$ factors.
    \end{cor}
}
\subsection{\texorpdfstring{A Proof of $d$ Quantum Depth}{A d-Proof of Quantum Depth}}
We give a \emph{non-interactive} Proof of $d$ Quantum Depth protocol, sound against oracle independent adversaries (see \Defref{PoD}). In the following, let $\tilde{H}$, and $\calC$ be as in \Defref{dCodeHashingProblem}.

\begin{description}
    \item[$\ver^H(\onelambda,\pi)$.] $\ver^H$ parses $\pi$ as $\mathbf{x}=(\mathbf{x}_1,\dots \mathbf{x}_n)$ and checks if (a) $\mathbf{x}\in \mathcal{C}$ and (b) $\tilde{H}(\mathbf{x})=1$. If both conditions are satisfied, it outputs 1, otherwise it outputs 0.
    \item[$\prov^H(\onelambda)$.] It runs the QPT machine in \Thmref{YZ22} with $\tilde H$ as the random oracle and returns the output $\mathsf{x}$ as $\pi$.
\end{description}

Completeness is immediate from \Thmref{YZ22}. Soundness against oracle independent $d$ depth circuits (i.e. circuits in $\CQC{d}$) follows directly from \Lemref{Main}. As discussed in \Secref{defn_PoQD} above, using known results, we obtain the following.

\begin{thm}[$d$-Proof of Quantum Depth] There is a publicly verifiable Proof of $d$ Quantum Depth (see \Defref{PoD}) sound against non-uniform oracle dependent adversaries.
\end{thm}

\section{Improved Upper Bound \label{sec:ImprovedUpperBound}}

To obtain the fine-grained separation $\classCQC {d} \subsetneq \classCQC {2d+\calO(1)}$, we introduce $\collisionhashing$.

\subsection{CollisionHashing}

$\collisionhashing$ is essentially the same problem used by \parencite{brakerski_simpler_2020} to obtain a proof of quantumness protocol except that instead of using claw-free function, we use a random function. $\collisionhashing$ shows $\QNC_{\calO(1)} \nsubseteq \BPP$, relative to a random oracle and satisfies classical query soundness. The main limitation of $\collisionhashing$ is that it cannot be efficiently verified, unlike YZ's $\codehashing$. 

The following elementary result about the probability of producing a superposition of two pre-images relative to a random oracle, would be useful in analysing $\collisionhashing$.

\begin{claim}
  \label{claim:collisionHashing}Let $g:A\to B$ be a random function
  where $A$ and $B$ are finite sets with $|A|\ge|B|$ and $\log\left|A\right|,\log|B|\le\poly$.
  Then there is a $\QNC_{2}$ circuit with oracle access to $g$ which
  produces the state
  \begin{equation}
    \frac{\left|a_{0}\right\rangle +\left|a_{1}\right\rangle }{\sqrt{2}}\label{eq:a_0a_1}
  \end{equation}
  with probability at least $c$ where $\{a_{0},a_{1}\}=g^{-1}(b)$
  for some $b\in B$ and $0<c<1$ depends only on $|A|$ and $|B|$.
  Further, $\lim_{|A|\to\infty}c\ge k^{2}/2(e^{k}-1)$ when $|A|=k|B|$
  for $k\in\mathbb{N}$.
\end{claim}

\branchcolor{blue}{\begin{proof}
    Producing $\sum_{a\in A}\left|a\right\rangle \left|g(a)\right\rangle /\sqrt{|A|}$
    takes one layer of Hadamards and a call to the oracle for $g$ and
    therefore $\QNC_{2}$ can prepare this state. If the second register
    is measured, the probability that the first register holds (up to
    normalisation) $\left|a_{0}\right\rangle +\left|a_{1}\right\rangle $
    is then $\Pr[|g^{-1}(b)|=2|b\in g(A)]$ where the probability is over
    $g\overset{\$}{\leftarrow}{\rm Functions}(A\to B)$. That in turn,
    for any fixed $b$, may be computed as follows\footnote{Note that the aforementioned probability is the same for every fixed
      $b$ and the probability (over $g$) of getting any fixed $b$ upon
      measurement of the second register is also the same by symmetry.}
    \[
      c(|A|,|B|)=\frac{\left|\{g:|g^{-1}(b)|=2\}\right|}{\left|\{g:b\in g(A)\}\right|}=\frac{\frac{|A|\cdot(|A|-1)}{2!}\cdot(|B|-1)^{|A|-2}}{(|B|^{|A|}-(|B|-1)^{|A|})}
    \]
    where
    to obtain the numerator, we count the number of ways of choosing exactly
    two points in $A$ (which are mapped to $b$) and the number of ways
    of assigning non-$b$ values to the remaining $|A|-2$ points. To obtain
    the denominator, we count the number of functions from $A$ to $B$
    and subtract from it all functions which do not map to $b$, i.e.
    none of the $|A|$ points are assigned the value $b$. Using $|A|=k|B|$,
    $\lim_{|A|\to\infty}\left(1-k/|A|\right)^{|A|}=e^{-k}$ and with some simplification,
    one obtains $\lim_{|A|\to\infty}c\ge k^{2}/2(e^{k}-1)$.
  \end{proof}
}

\branchcolor{purple}{We now state the $\collisionhashing$ problem as follows.}

\begin{defn}[$\collisionhashing$]
  \label{def:collisionHashing}The $\collisionhashing$ problem is
  defined by $(\gen,R_{H})$ where $\gen(1^{\lambda})=1^{\lambda}$
  and the relation $R_{H}$ is specified as follows: Let $g:\{0,1\}^{\lambda+1}\to\{0,1\}^{\lambda}$
  be a random function, let $H':\{0,1\}^{*}\to\{0,1\}$ be another random
  function (both generated using $H$ in some canonical way) and let
  $c$ be as in \Claimref{collisionHashing}. We say $(1^{\lambda},((y_{i},m_{i},r_{i})_{i\in\{1\dots\lambda\}})\in R_{H}$
  if the following hold
  \begin{enumerate}
    \item all $y_{i}$ are distinct
    \item $\frac{|I|}{\lambda}>\frac{3c}{4}$ where $I\subseteq[1\dots\lambda]$ is the subset
          of indices satisfying $\left|g^{-1}(y_{i})\right|=2$ for all $i\in I$.
    \item $\frac{\rm count}{|I|}>\frac{3}{4}$ where ${\rm count}=\sum_{i\in I}{\rm valid}(i)$
          and \\
          ${\rm valid}(i)$ returns $1$ if the following holds, otherwise it
          returns $0$:\\
          $m_{i}=r_{i}\cdot(z_{i0}\oplus z_{i1})\oplus H'(z_{i0})\oplus H'(z_{i1})$
          where $\{z_{i0},z_{i1}\}=g^{-1}(y_{i})$.
  \end{enumerate}
  
\end{defn}

\branchcolor{purple}{$\collisionhashing$ satisfies the following properties.}

\begin{lem}
  \label{lem:collisionHashing_mainProperties}Let $\collisionhashing$
  be as stated in \Claimref{collisionHashing}. It satisfies the following
  properties
  \begin{itemize}
    \item Completeness: $\QNC_{10}$ can solve $\collisionhashing$ with probability
          $1-\ngl{\lambda}$
    \item Soundness: $\collisionhashing$ satisfies classical query soundness.
    \item Bounded Oracle Domain: $\collisionhashing$ has a bounded oracle domain of size at most $2^{3\lambda}$.
  \end{itemize}
\end{lem}

\branchcolor{blue}{\begin{proof}[Proof sketch.]
    \emph{Completeness:} \parencite{brakerski_simpler_2020} showed that
    if one is given $\ply{\lambda}$ copies of the state \Eqref{a_0a_1},
    then using at most 7 layers of quantum operations, one can solve $\collisionhashing$
    with probability $1-\ngl{\lambda}$. The aforesaid state can be generated with probability $c$ (in at most $2$
    quantum layers) therefore the probability of generating $0.75c\lambda$
    states is $1-\negl$ (using Chernoff). Thus, $\QNC_{10}$
    can solve the problem with $1-\negl$ probability.

    \emph{Classical query soundness:} \parencite{brakerski_simpler_2020}
    showed that every PPT machine solves $\collisionhashing$ with probability
    at most\footnote{They show it for their problem but the results carry over unchanged.} $\ngl{\lambda}$. Their argument is more general. Their proof showed that succeeding with non-negligible probability implies one
    can find collisions which is assumed to be hard. More precisely, they
    neither require the machine to be PPT (only that access to the oracle
    $H'$ is classical), nor that \emph{PPT} machines cannot find collisions
    in the function (which is $g$ in this case) but only assert that
    collisions can be extracted. For establishing classical query soundness, it suffices to show that with only polynomially many classical queries to $H$, no (potentially unbounded) machine can solve $\collisionhashing$. 
    It is known that finding collisions in $g$ (a random function) with
    non-negligible probability requires at least $\Omega(2^{n/O(1)})$
    (quantum) queries. Using \parencite{brakerski_simpler_2020}'s argument
    (see their Section 3.2) on elements in the set $I\subseteq\{1\dots\lambda\}$,
    we deduce that solving $\collisionhashing$ with non-negligible probability
    implies there is an algorithm that finds collisions in $g$ by making only polynomially many (classical) queries to $g$ which in turn violates the previous statement. Thus, we conclude $\collisionhashing$ satisfies classical query soundness.

    \emph{Bounded Oracle Domain:} By inspection, it is clear that $H'$
    is only queried on a domain of size $2^{2\lambda}$ and $g$ is only
    queried on a domain of size $2^{2\lambda}$. Since both are generated
    using $H$, we take $2^{3\lambda}$ as a loose upper bound on the
    oracle domain.
  \end{proof}
}

\begin{table}[H]
  \begin{centering}
    \resizebox{\textwidth}{!}{\begin{tabular*}{1.2\textwidth}{@{\extracolsep{\fill}}cccccc>{\centering}p{2cm}>{\centering}p{2cm}}
      \textbf{\footnotesize{}Problem} & \textbf{\footnotesize{}$\in$} &  & \textbf{\footnotesize{}$\notin$} & \textbf{\footnotesize{}Assumption} & \textbf{\footnotesize{}Verification} & \textbf{\footnotesize{}Interpretation} & \textbf{\footnotesize{}Remarks}\tabularnewline
      \hline
      {\footnotesize{}$\collisionhashing$} & {\footnotesize{}$\QNC_{{\cal O}(1)}$} & {\footnotesize{}$\nsubseteq$} & {\footnotesize{}$\BPP$} & {\footnotesize{}RO} & {\footnotesize{}No} & {\scriptsize{}Even the simplest constant quantum depth is hard to
          simulate} & {\scriptsize{}\Defref{collisionHashing}}\tabularnewline
      \hline
      {\footnotesize{}$\recursive d[\collisionhashing]$} & {\footnotesize{}$\BQNC_{2d+{\cal O}(1)}\subseteq\classCQC{2d+{\cal O}(1)}$} & {\footnotesize{}$\nsubseteq$} & {\footnotesize{}$\classCQC d$} & {\footnotesize{}RO} & {\footnotesize{}No} & {\scriptsize{}Finer refutation of Jozsa's conjecture in ROM} & {\tiny{}\Thmref{finerJozsasConjectureFalse}}\tabularnewline
    \end{tabular*}}
    \par\end{centering}
  \caption{We tighten the quantum depth bounds to $\classCQC{d}\subsetneq \classCQC{2d+\const}$ relative to the random oracle. \label{tab:TightBounds}} %
\end{table}
\subsection{Jozsa's conjecture/Aaronson's challenge}

Using \Lemref{dRecursive_upper}, observe that $\recursive d[\collisionhashing]$
can be solved in $\BQNC_{2d+{\cal O}(1)}$. From \Lemref{dRecursive_CQCd},
observe that $\recursive d[\collisionhashing]$ cannot be solved in
$\CQC d$. We therefore have the following.
\begin{thm}[Stronger refutation of Jozsa's conjecture.]
  \label{thm:finerJozsasConjectureFalse}With respect to a random oracle,
  the following hold: $\BQNC_{2d+{\cal O}(1)}\nsubseteq\BPP^{\BQNC_d^{\BPP}}$,
  which implies $\BPP^{\BQNC_d^{\BPP}}\subsetneq\BPP^{\BQNC_{2d+{\cal O}(1)}^{\BPP}}$.

\end{thm}

\pagebreak{}

\part{Separations of Hybrid Quantum Depth\label{part:Distinguishing-between-types}}

\branchcolor{purple}{In the previous discussions, we studied the relation of $\classCQC d$
  with $\classCQC{d'}$. In particular, we showed that relative a random oracle,
  $\BQP\not\subseteq\classCQC d$ and we tightened
  this result to $\classCQC d\nsubseteq\classCQC{d+\mathcal{O}(1)}$. We now study
  the relation between $\classCQC d$, $\classQC d$ and $\classCQ d$. We define a
  problem and prove that it can be solved by $\classQC{\mathcal{O}(1)}$
  but not by $\classCQ d$ and conversely, a problem that can be solved by
  $\classCQ{\mathcal{O}(1)}$ but not by $\classQC d$. %
  The former shows that having constant quantum depth with adaptive control cannot be
  simulated by repeating constant quantum depth machines without adaptive
  control. The latter does not seem to have as clear an interpretation.
  However, we can combine these ideas to construct another problem which
  also shows $\classCQC{{\cal O}(1)}\nsubseteq\classQC d\cup\classCQ d$. %
  giving further evidence that it is important to show soundness against $\classCQC d$ when considering quantum depth because even with constant quantum depth, $\classCQC{{\cal O}(1)}$ already contains
  problems which are neither in $\classCQ d$ nor in $\classQC d$. Therefore,
  it is crucial establishing that every proof of quantum depth is sound
  against $\classCQC d$.

  The results in this section are summarised in \Tabref{Distinguishing-between-types}. Establishing
  $\classCQ d\nsubseteq\classQC d$ is the most involved and requires the use
  of the compressed oracle simulation technique. We defer it to the
  end and instead first establish a general lifting theorem which takes
  almost any proof of quantumness and excludes it from $\classQC d$. We
  apply it to $\collisionhashing$ %
  to establish that $\classCQ{{\cal O}(1)}\nsubseteq\classQC d$ in the random oracle
  model. }

  \begin{table}
    \begin{centering}
      \resizebox{\textwidth}{!}{\begin{tabular*}{1.2\textwidth}{@{\extracolsep{\fill}}cccccc>{\centering}p{2cm}>{\centering}p{2cm}}
        \textbf{\footnotesize{}Problem} & \textbf{\footnotesize{}$\in$} &  & \textbf{\footnotesize{}$\notin$} & \textbf{\footnotesize{}Assumption} & \textbf{\footnotesize{}Verification} & \textbf{\footnotesize{}Interpretation} & \textbf{\footnotesize{}Remarks}\tabularnewline
        \hline
        {\footnotesize{}$\hcollisionhashing d$} & {\footnotesize{}$\classQC{{\cal O}(1)}$} & {\footnotesize{}$\nsubseteq$} & {\footnotesize{}$\classCQ d$} & {\footnotesize{}RO} & {\footnotesize{}No} & {\tiny{}Even the simplest constant depth adaptive quantum control
            cannot be simulated by running a $d$ depth quantum circuit poly many
            times.} & {\tiny{} \Subsecref{hclawhcoll}} \tabularnewline
        \hline
        {\footnotesize{}$\serial d[\collisionhashing]$} & {\footnotesize{}$\classCQ{{\cal O}(1)}$} & {\footnotesize{}$\nsubseteq$} & {\footnotesize{}$\classQC d$} & {\footnotesize{}RO} & {\footnotesize{}No} & {\tiny{}(Perhaps unsurprisingly) repeating a constant depth quantum
            circuit cannot be simulated with running a $d$ adaptive quantum depth
            circuit once.} & {\tiny{} \Subsecref{consequences_dSer}}\tabularnewline
        \hline
        {\footnotesize{}$\serial d[\hcollisionhashing d]$} & {\footnotesize{}$\classCQC{{\cal O}(1)}$} & {\footnotesize{}$\nsubseteq$} & {\footnotesize{}$\classCQ d\cup\classQC d$} & {\footnotesize{}RO} & {\footnotesize{}No} & {\tiny{}Evidence that $\CQC d$ is the right notion of depth.} & {\tiny{} \Subsecref{consequences_hcollhclaw}}\tabularnewline
      \end{tabular*}}
      \par\end{centering}
    \caption{\label{tab:Distinguishing-between-types}A summary of the relations between $\classCQC{d}$, $\classQC{d}$ and $\classCQ{d}$.%
    }
  \end{table}

\section{\texorpdfstring{$\protect\BPP^{\protect\QNC_{{\cal O}(1)}}\nsubseteq\protect\QNC_{d}^{\protect\BPP}$}{BPP\^{}\{QNC\_O(1)\} is not a subset of \{QNC\_d\}\^{}BPP}}

\branchcolor{purple}{The idea behind $\serial d[\mathcal{P}]$ is quite intuitive. Suppose
$\mathcal{P}$ is a problem which is specified by the relation $R_{H}$
where $H$ is the random oracle. For simplicity, suppose the input
to the problem is $1^{\lambda}$ and $(1^{\lambda},c)\in R_{H}$ means
that $c$ is a solution. Then $\serial d[{\cal P}]$ is a relation
$R'_{H}$ where $(1^{\lambda},(c_{0}\dots c_{d}))\in R'_{H}$ if $(1^{\lambda},c_{0})\in R_{H}$,
$(1^{\lambda},c_{1})\in R_{H(c_{0}||\cdot)}$, $(1^{\lambda},c_{2})\in R_{H(c_{0},c_{1}||\cdot)}$
and so on. The rationale is that until the first problem is solved,
the subsequent problems are not even specified. The problems must
therefore be solved serially---they cannot be solved in parallel.
So far, we have not constrained the model of computation. We want
$\serial d[{\cal P}]$ to be hard for $\QC d$ whenever ${\cal P}$
is hard for $\BPP$ (but can be solved by adding quantumness, e.g.
in $\QNC_{0}$). Recall that for $\recursive d[{\cal P}]$ we wanted
${\cal P}$ to satisfy classical query soundness. In this case,
we require ${\cal P}$ to satisfy a different property which we call
\emph{offline soundness}. Intuitively, suppose after running a classical
algorithm to solve ${\cal P}$, access to $H$ is revoked and thereafter
unbounded computation is allowed. Offline soundness requires that
even in this case, ${\cal P}$ cannot be solved with non-negligible
probability.

The main difference between $\serial d$ and $\recursive d$ is that
in $\serial d$ one need not maintain ``coherence'' across all the
problems (which use different oracles); it suffices to individually
solve the problems. In $\recursive d$, even to access the oracle
$\tilde{H}=H_{d}\circ\dots\circ H_{0}$, one had to maintain coherence
across $d$ layers.}

\subsection{Offline Soundness}

\branchcolor{purple}{We state offline soundness formally first.}
\begin{defn}[Offline Soundness]
  \label{def:offlineSoundness}As in \Defref{semi-bounded}, let $H:\{0,1\}^{*}\to\{0,1\}$
  be a random oracle. Define a problem ${\cal P}$ by a tuple $({\cal X},R_{H})$
  where ${\cal X}$ is a procedure which on input $1^{\lambda}$ generates
  a problem instance of size $\ply{\lambda}$ and $R_{H}=\{0,1\}^{*}\times\{0,1\}^{*}$
  is a relation which depends on $H$. We define \emph{offline soundness}
  as follows.

  Let ${\cal A}^{H}$ be a PPT algorithm with access to $H$. Let $\tau[{\cal A}^{H}(x)]$
  be the tableaux (or the computational transcript) obtained by running
  ${\cal A}^{H}$ on input $x\in{\cal X}$. Let ${\cal B}$ be an unbounded
  machine with no access to $H$ which takes $\tau$ as input. We say
  ${\cal P}$ satisfies \emph{offline soundness} if
  \[
    \Pr_{H}\left[(x,y)\in R_{H}:\substack{(x,y)\leftarrow{\cal B}(\tau)\\
        \tau=\tau[{\cal A}^{H}(x)]\\
        x\leftarrow{\cal X}(1^{\lambda})
      }
      \right]\le\ngl{\lambda}
  \]
  for all ${\cal B}$ and ${\cal A}^{H}$.
\end{defn}

\branchcolor{purple}{Offline soundness is clearly a special case of classical query soundness and therefore both $\collisionhashing$ and $\codehashing$ satisfy it. 
  }
\begin{lem}
  \label{lem:everythingSatisfiesOfflineSoundness}$\codehashing$ and $\collisionhashing$ satisfy offline soundness.
\end{lem}

It appears reasonable to expect offline soundness to be a strictly weaker requirement than classical query soundness. Indeed, this is true because there are problems which satisfy offline soundness but not classical query soundness, e.g. the problem considered by~\cite{brakerski_simpler_2020}.

\subsection{The \texorpdfstring{$d\text{-}\mathsf{Ser}[\mathcal{P}]$}{d-Ser[P]} Problem}

\branchcolor{purple}{With offline soundness in place, we can define $\serial d[{\cal P}]$
  as follows.}
\begin{defn}[{$\serial d[{\cal P}]$}]
  \label{def:dSerial}Let ${\cal P}=({\cal X},R)$ be a problem (see
  \Defref{semi-bounded}) defined with respect to a random oracle $H:\{0,1\}^{*}\to\{0,1\}$,
  having a bounded oracle domain (as specified in \Defref{semi-bounded})
  and satisfying offline soundness (as defined in \Defref{offlineSoundness}).

  Define $\serial d[{\cal P}]$ as follows: On input $1^{\lambda}$,
  sample $d+1$ independent instances of ${\cal P}$ as $(x_{0},\dots x_{d})$
  where $x_{i}\leftarrow{\cal X}(1^{\lambda})$ for each $i\in\{0\dots d\}$.
  Accept $(y_{0},\dots y_{d})$ if for each $i\in\{1\dots d\}$, $(x_{i},y_{i})\in R_{H((x_{0},y_{0},\dots x_{i-1},y_{i-1})||\cdot)}$
  and for $i=0$, $(x_{0},y_{0})\in R_{H(\cdot)}$.
\end{defn}

\subsection{Lower-bounds}

\branchcolor{purple}{In this section, we analyse everything for a fixed $\lambda$ and
  introduce some notation to that end. Since ${\cal P}$ has bounded
  oracle domain, one can consider $d+1$ oracles with bounded domains
  which in turn make the analysis easier.}
\begin{notation}
  \label{nota:dSerial_oracles}Fix a $\lambda$. Let ${\cal P}=({\cal X},R)$
  be as in \Defref{dSerial} where the bounded oracle domain of ${\cal P}$
  is ${\cal D}:=\{0,1\}^{p(\lambda)}$. Fix an input instance $(x_{0}\dots x_{d})$
  of $\serial d[{\cal P}]$. With respect to this, let ${\cal S}_{i,H}:=\{(x_{i},y_{i}):(x_{i},y_{i})\in R_{H}\}$
  denote all pairs $(x_{i},y_{i})$ in $R_{H}$. Let ${\cal S}:={\cal X}\times{\cal Y}$
  where ${\cal Y}$ is the set of all $y$s. It would be useful to
  consider $d+1$ oracles with bounded domains instead of considering
  $H:\{0,1\}^{*}\to\{0,1\}$. More precisely, let $H_{0}:{\cal D}\to\{0,1\}$,
  $H_{1}:{\cal S}\times{\cal D}\to\{0,1\},\dots H_{d}:{\cal S}^{d}\times{\cal D}\to\{0,1\}$.
  Let ${\cal L}:=(H_{0},\dots H_{d})$ denote the sequence of oracles
  $H_{i}$s.
\end{notation}

\branchcolor{purple}{It would be helpful to define the analogue of \Defref{shadowOracle}
  for our sequence of oracles.}
\begin{defn}[{Shadow Oracles wrt $\bar{S}$ for $\serial d[{\cal P}]$}]
  \label{def:RandomShadowOracleSequence} Let ${\cal L}=(H_{0}\dots H_{d})$,
  $p$ and ${\cal S}$ be as in \Notaref{dSerial_oracles}. Let $\bar{S}=(S_{1},\dots S_{d})$
  be a tuple of $d$ sets where each set $S_{i}\subseteq{\cal S}^{i}\times\{0,1\}^{p}$.
  The random shadow oracle ${\cal M}$ of ${\cal L}$ wrt $\bar{S}$
  is defined as ${\cal M}:=(M_{0},\dots M_{d})$ where for each $i\in\{1\dots d\}$,
  $M_{i}$ is the shadow oracle of $H_{i}$ wrt $S_{i}$ (as in \Defref{shadowOracle})
  and $M_{0}=H_{0}$.
\end{defn}

\subsubsection{\texorpdfstring{Exclusion from $\protect\QNC_{d}$}{Exclusion from QNC\_d}}

\branchcolor{purple}{We first show that $\serial d[{\cal P}]$ is hard for $\QNC_{d}$
and then extend the analysis to $\QC d$. To this end, we first introduce
the shadow oracles by describing the sets we hide. Let ${\cal D}=\{0,1\}^{p}$
denote the oracle domain of ${\cal P}$.
\begin{itemize}
  \item Define the hidden sets for $i,j\in\{1\dots d\}$ as $S_{ij}\doteq$

\resizebox{\textwidth}{!}{
$\left[\begin{array}{ccccc}
  {\cal S}_{H_{0}}\times{\cal D} & {\cal S}_{H_{0}}\times{\cal S}\times{\cal D}                                                                        & {\cal S}_{H_{0}}\times{\cal S}^{2}\times{\cal D}                                                                                                   & \dots  & {\cal S}_{H_{0}}\times{\cal S}^{d-1}\times{\cal D}                                                                                                                     \\
  \emptyset                      & \underset{(x_{0}y_{0})\in{\cal S}_{0}}{\bigcup}(x_{0}y_{0})\times{\cal S}_{H_{1}(x_{0},y_{0}||\cdot)}\times{\cal D} & \underset{(x_{0}y_{0})\in{\cal S}_{0}}{\bigcup}(x_{0}y_{0})\times{\cal S}_{H_{1}(x_{0},y_{0}||\cdot)}\times{\cal S}\times{\cal D}                  & \dots  & \underset{(x_{0}y_{0})\in{\cal S}_{0}}{\bigcup}(x_{0}y_{0})\times{\cal S}_{H_{1}(x_{0},y_{0}||\cdot)}\times{\cal S}^{d-2}\times{\cal D}                                \\
  \emptyset                      & \emptyset                                                                                                           & \underset{(x_{0}y_{0}x_{1}y_{1})\in{\cal S}_{0:1}}{\bigcup}(x_{0}y_{0}x_{1}y_{1})\times{\cal S}_{H_{2}(x_{0}y_{0}x_{1}y_{1}||\cdot)}\times{\cal D} & \dots  & \underset{(x_{0}y_{0}x_{1}y_{1})\in{\cal S}_{0:1}}{\bigcup}(x_{0}y_{0}x_{1}y_{1})\times{\cal S}_{H_{2}(x_{0}y_{0}x_{1}y_{1}||\cdot)}\times{\cal S}^{d-3}\times{\cal D} \\
                                 &                                                                                                                     &                                                                                                                                                    & \ddots
\end{array}\right]$
}

\begin{equation}=\left[\begin{array}{ccccc}
                  {\cal S}_{0:0}\times{\cal D} & {\cal S}_{0:0}\times{\cal S}\times{\cal D} & {\cal S}_{0:0}\times{\cal S}^{2}\times{\cal D} & \dots  & {\cal S}_{0:0}\times{\cal S}^{d-1}\times{\cal D} \\
                  \emptyset                    & {\cal S}_{0:1}\times{\cal D}               & {\cal S}_{0:1}\times{\cal S}\times{\cal D}     & \dots  & {\cal S}_{0:1}\times{\cal S}^{d-2}\times{\cal D} \\
                  \emptyset                    & \emptyset                                  & {\cal S}_{0:2}\times{\cal D}                   & \dots  & {\cal S}_{0:2}\times{\cal S}^{d-3}\times{\cal D} \\
                                               &                                            &                                                & \ddots
                \end{array}\right] 
                \label{eq:S_ij_dSerial_QNC_d}
              \end{equation}

       where the union in the first matrix is over ``correct solutions'',
        i.e. (a) ${\cal S}_{0:0}:=:{\cal S}_{0}:={\cal S}_{H_{0}}={\cal S}_{H}$,
        denotes the set of solutions (corresponding to $\lambda$) to ${\cal P}$
        wrt $H_{0}$, (b) ${\cal S}_{0:1}=\cup_{(x_{0},y_{0})\in{\cal S}_{0}}(x_{0},y_{0})\times{\cal S}_{H_{1}(x_{0}y_{0}||\cdot)}$,
        denotes the set of solutions to ${\cal P}$ wrt $H_{0}$ (in the first
        two coordinates) and corresponding to each solution, the set of solutions
        to ${\cal P}$ wrt $H_{1}(x_{0}y_{0}||\cdot)$ (in the last two coordinates)
        and (c) in general ${\cal S}_{0:i}=\cup_{(x_{0},y_{0}\dots x_{i-1}y_{i-1})\in{\cal S}_{0:i-1}}(x_{0}y_{0}\dots x_{i-1}y_{i-1})\times{\cal S}_{H_{i}(x_{0}\dots y_{i-1}||\cdot)}$.
        By ${\cal S}_{H_{i}(s||\cdot)}$ we mean ${\cal S}_{i,H_{i}(s||\cdot)}$
        where $s$ is some string.
  \item We now try to justify this definition. The main structure of the proof
        is similar to the proof of $\QNC_{d}$ hardness of $\CH d$, i.e.
        \Lemref{QNC_d_hardness}. Let $\bar{S}_{i}=(S_{i1},\dots S_{id})$
        denote the $i$th row of the matrix above. Let ${\cal M}_{1}$ denote
        the shadow of ${\cal L}$ wrt $\bar{S}_{1}$. As in the proof of \Lemref{QNC_d_hardness},
        we want to ensure that the information contained in ${\cal M}_{1}$
        is not enough to guess $\bar{S}_{2}$ which will be used to define
        ${\cal M}_{2}$. This would allow us to apply \Lemref{boundPfind}
        as before. Once this is clear, the remaining steps are straightforward.
        Observe that ${\cal M}_{1}$ specifies $H_{0}$ and therefore (information
        theoretically) specifies ${\cal S}_{H_{0}}$. It also specifies $H_{1}$
        partially---it does not specify $H_{1}$ on ${\cal S}_{H_{0}}\times{\cal D}$.
        Note that, in particular, this means that ${\cal M}_{1}$ contains
        no information about ${\cal S}_{H_{1}(x_{0}y_{0}||\cdot)}$ for $(x_{0}y_{0})\in{\cal S}_{H_{0}}={\cal S}_{0:0}$.
        That, in turn, means that none of the sets in $\bar{S}_{2}|H_{0}$
        are correlated with ${\cal M}_{1}|H_{0}$. %
  \item Let us look at the next case as well, as it would help with the general
        argument in the proof. Suppose ${\cal M}_{2}$ is the shadow of ${\cal L}$
        wrt $\bar{S}_{2}$. We want to argue that even knowing ${\cal M}_{2}$
        it is hard to find $\bar{S}_{3}$. Observe that ${\cal M}_{2}$ specifies
        $H_{0}$ and $H_{1}$. It also specifies $H_{2}$ partially---it
        does not specify $H_{2}$ at $\underset{(x_{0}y_{0})\in{\cal S}_{0}}{\bigcup}(x_{0}y_{0})\times{\cal S}_{H_{1}(x_{0}y_{0}||\cdot)}\times{\cal D}$.
        This, in particular, means that ${\cal M}_{2}$ contains no information
        about ${\cal S}_{H_{2}(x_{0}y_{0}x_{1}y_{1}||\cdot)}$ for $(x_{0}y_{0}x_{1}y_{1})\in{\cal S}_{0:1}$.
        That in turn means that none of the sets in $\bar{S}_{3}$ are correlated
        with ${\cal M}_{2}$, given $H_{0},H_{1}$.
  \item Intuitively, ${\cal M}_{i-1}$ completely specifies $H_{0},\dots H_{i-2}$
        but it does not specify $H_{i-1}$ completely and $\bar{S}_{i}$ (conditioned
        on $H_{0}\dots H_{i-2}$) depends only on this unspecified part of
        $H_{i-1}$.
\end{itemize}
}

\begin{lyxalgorithm}
  \label{alg:sets_dSerialQNCd}Let ${\cal L}=(H_{0}\dots H_{d})$, ${\cal S}$,
  ${\cal S}_{i,H}$ and $p$ be as in \Notaref{dSerial_oracles}. Assume
  $\lambda$ and the input instance $(x_{0},\dots x_{d})$ have been
  implicitly specified. We use ${\cal S}_{H_{i}(s||\cdot)}$ to denote
  ${\cal S}_{i,H_{i}(s||\cdot)}$. Define, for each $i\in\{1,\dots d\}$,
  $S_{ij}$ as follows.
  \begin{enumerate}
    \item If $i=1$, define $S_{1j}:={\cal S}_{H_{0}}\times{\cal S}^{j-1}\times{\cal D}$
    \item If $i>1$,
          \begin{enumerate}
            \item Define $S_{ij}:=\emptyset$ for $1\le j<i$ and
            \item otherwise,
                  \begin{align*}
                    S_{ij} & :=\underset{(x_{0}y_{0}\dots x_{i-2}y_{i-2})\in{\cal S}_{0:i-2}}{\bigcup}(x_{0}y_{0}\dots x_{i-2}y_{i-2})\times{\cal S}_{H_{i-1}(x_{0}y_{0}\dots x_{i-2}y_{i-2}||\cdot)}\times{\cal S}^{i-j}\times{\cal D} \\
                           & ={\cal S}_{0:i-1}\times{\cal S}^{i-j}\times{\cal D}
                  \end{align*}
                  where
                  \[
                    {\cal S}_{0:i}=\begin{cases}
                      {\cal S}_{0}:={\cal S}_{H_{0}}                                                                                                                            & i=0  \\
                      \underset{(x_{0},y_{0}\dots x_{i-1}y_{i-1})\in{\cal S}_{0:i-1}}{\bigcup}(x_{0}y_{0}\dots x_{i-1}y_{i-1})\times{\cal S}_{H_{i}(x_{0}\dots y_{i-1}||\cdot)} & i>0.
                    \end{cases}
                  \]
          \end{enumerate}
  \end{enumerate}
  Let $\bar{S}_{i}:=(S_{i1},\dots S_{id})$. Return $(\bar{S}_{1},\dots\bar{S}_{d})$.

\end{lyxalgorithm}

\begin{lem}[{$\serial d[{\cal P}] \notin \QNC_d$}]
  Every $\QNC_{d}$ circuit succeeds at solving $\serial d[{\cal P}]$
  (see \Defref{dSerial}) with probability at most $\ngl{\lambda}$
  on input $1^{\lambda}$ for $d\le\ply{\lambda}$.
\end{lem}

\branchcolor{blue}{\begin{proof}
    Fix a $\lambda$. Let ${\cal L}=(H_{0},\dots H_{d})$ be as in \Notaref{dSerial_oracles}.
    Suppose the problem instance $\serial d[{\cal P}]$ is specified by
    $(x_{0},\dots x_{d})$ and let $\rho_{0}$ be the initial state, containing
    this input. Denote by ${\cal A}^{{\cal L}}$ an arbitrary $\QNC_{d}$
    circuit
    \[
      {\cal A}^{{\cal L}}(\rho_{0}):=\Pivalid\circ U_{d+1}\circ{\cal L}\circ U_{d}\circ\dots{\cal L}\circ U_{2}\circ{\cal L}\circ U_{1}\circ\rho_{0}
    \]
    where $\Pivalid$ corresponds to projection on all output strings
    which solve $\serial d[{\cal P}]$ (for a fixed $\lambda$). let $(\bar{S}_{1},\dots\bar{S}_{d})$
    be the output of \Algref{sets_dSerialQNCd}. Define
    \[
      {\cal A}^{{\cal M}}(\rho_{0}):=\Pivalid\circ U_{d+1}\circ{\cal M}_{d}\circ U_{d}\dots{\cal M}_{2}\circ U_{2}\circ{\cal M}_{1}\circ U_{1}\circ\rho_{0}
    \]
    where ${\cal M}_{i}$ is the random shadow oracle of ${\cal L}$ wrt
    $\bar{S}_{i}$ (see \Defref{RandomShadowOracleSequence}). We proceed
    in two steps.

    \textbf{Step 1:}\emph{ ${\cal A}^{{\cal L}}$ and ${\cal A}^{{\cal M}}$
    behave the same.} We show that the probability that ${\cal A}^{{\cal L}}$
    and ${\cal A}^{{\cal M}}$ produce a valid output is negligibly close,
    i.e. we bound
    \begin{align*}
          & \left|\tr[\Pivalid\circ U_{d+1}\circ{\cal L}\circ U_{d}\circ\dots{\cal L}\circ U_{2}\circ{\cal L}\circ U_{1}\circ\rho_{0}]-\tr[\Pivalid\circ U_{d+1}\circ{\cal M}_{d}\circ U_{d}\dots{\cal M}_{2}\circ U_{2}\circ{\cal M}_{1}\circ U_{1}\circ\rho_{0}]\right| \\
      \le & \sum_{i=1}^{d}B[{\cal L}\circ U_{i}(\rho_{i-1}),{\cal M}_{i}\circ U_{i}(\rho_{i-1})]\le\sum_{i=1}^{d}\sqrt{2\Pr[{\rm find}:U_{i}^{{\cal L}\backslash\bar{S}_{i}},\rho_{i-1}]}
    \end{align*}
    where $\rho_{i}={\cal M}_{i}\circ U_{i}\circ\dots{\cal M}_{1}\circ U_{1}\circ\rho_{0}$
    for $i>0$, we used (as in the proof of \Lemref{QNC_d_hardness})
    the triangle inequality, monotonicity of the trace distance, the relation
    between trace distance and Bures distance and finally applied \Lemref{O2H}.
    To bound the RHS above, one can use \Lemref{boundPfind} if it holds
    that $\rho_{i-1}$ is uncorrelated with $\bar{S}_{i}$. It suffices
    to show that ${\cal M}_{i-1}$ is uncorrelated with $\bar{S}_{i}$,
    given $H_{0},\dots H_{i-2}$. We argued the $i=1,2$ case above. In
    general, for $i>2$ (for notational ease), observe that ${\cal M}_{i-1}$
    specifies $H_{0},\dots H_{i-2}$ completely and specifies $H_{i-1}$
    at all points except at $\bigcup_{(x_{0}y_{0}\dots x_{i-3}y_{i-3})\in{\cal S}_{0:i-3}}(x_{0}y_{0}\dots x_{i-3}y_{i-3})\times{\cal S}_{H_{i-2}(x_{0}y_{0}\dots x_{i-3}y_{i-3}||\cdot)}\times{\cal D}={\cal S}_{0:i-2}\times{\cal D}$.
    This in particular means that ${\cal M}_{i-1}$ contains no information
    about ${\cal S}_{H_{i-1}(x_{0}y_{0}\dots x_{i-2}y_{i-2}||\cdot)}$
    for $(x_{0}y_{0}\dots x_{i-2}y_{i-2})\in{\cal S}_{0:i-2}$ where $\{{\cal S}_{0:i}\}_{i}$
    are as defined in \Algref{sets_dSerialQNCd}. This, in turn, entails
    that $\bar{S}_{i}$ is uncorrelated with ${\cal M}_{i-1}$, given
    $H_{0},\dots H_{i-2}$ as asserted. From offline soundness of ${\cal P}$,
    and the aforesaid, it follows that
    \[
      \Pr\left[(x_{i},y_{i})\in R_{H_{i}(x_{0}y_{0}\dots x_{i-1}y_{i-1}||\cdot)}\Big|{\cal M}_{i-1}:\substack{H\overset{\$}{\leftarrow}{\rm Funcs[}\{0,1\}^{*}\to\{0,1\}]\\
      (x_{0},\dots x_{i})\leftarrow{\cal X}(1^{\lambda})
      }
      \right]\le\ngl{\lambda}
    \]
    for all $(x_{0}y_{0}\dots x_{i-1}y_{i-1})\in{\cal S}_{0:i-1}$. This
    entails that, for $i\le k\le d$, $\Pr[(x_{0}y_{0}\dots x_{k}y_{k})\in S_{ik}|{\cal M}_{i-1}]\le\ngl{\lambda}$
    which means, via \Lemref{boundPfind}, $\Pr[{\rm find}:U_{i}^{{\cal L}\backslash\bar{S}_{i}},\rho_{i-1}|{\cal M}_{i-1}]\le\ngl{\lambda}$
    (the conditioning notation for $\Pr[{\rm find}:\dots]$ is the same
    as the last bullet after the $\CQ d$ \Lemref{CQd_hardness}; all
    variables involved ($\rho_{i-1},\bar{S}_{i},{\cal L}$) are conditioned
    on ${\cal M}_{i-1}$).

    \textbf{Step 2:}\emph{ ${\cal A}^{{\cal M}}$ cannot succeed with
    non-negligible probability.} Note that by construction, ${\cal M}_{d}$
    contains all the information in ${\cal M}_{1}\dots{\cal M}_{d-1}$.
    Further, observe that ${\cal M}_{d}$ does not contain any information
    about ${\cal S}_{H_{d}(x_{0}y_{0}\dots x_{d-1}y_{d-1}||\cdot)}$ for
    $(x_{0}y_{0}\dots x_{d-1}y_{d-1})\in{\cal S}_{0:d-1}$ (which includes
    the set of valid answers until $d-1$). From offline soundness of
    ${\cal P}$, it follows that ${\cal A}^{{\cal M}}$ cannot find $(x_{d},y_{d})\in{\cal S}_{H_{d}(x_{0}\dots y_{d-1}||\cdot)}$
    with probability greater than $\ngl{\lambda}$ which upper bounds
    the success probability of ${\cal A}^{{\cal M}}$.
  \end{proof}
}

\subsubsection{\texorpdfstring{Exclusion from $\protect\classQC{d}$}{Exclusion from QNC\_d\^{}BPP}}

\branchcolor{purple}{Recall that $\QC d$ circuits are represented as ${\cal B}^{{\cal L}}={\cal A}_{c,d+1}^{{\cal L}}\circ{\cal B}_{d}^{{\cal L}}\circ\dots{\cal B}_{1}^{{\cal L}}\circ\rho_{0}$
where ${\cal B}_{i}^{{\cal L}}:=\Pi_{i}\circ{\cal L}\circ U_{i}\circ{\cal A}_{c,i}^{{\cal L}}$.
Here ${\cal A}_{c,i}^{{\cal L}}$ represent classical algorithms and
we drop ``c'' in this section. Since the oracles $H_{0},\dots H_{d}$
have different domains, we make the following assumption about the
classical algorithms in the $\QC d$ circuit. This simplifies our
analysis and only makes our impossibility result stronger. }

Assumption: if $H_{k}$ is queried at $(x_{0}y_{0}\dots x_{k}y_{k})$,
then for all $i\in\{0\dots k-1\}$ $H_{i}$ is also queried at $(x_{0}y_{0}\dots x_{i}y_{i})$.

It would also be helpful to setup some notation for describing the
classical queries. Since ${\cal A}$ makes queries on different domains,
the set of queries is simply a collection of strings with varying
number of ``coordinates''. For example, if $H_{k}$ is queried at
$(x_{0}y_{0}\dots x_{k}y_{k})$, by the $j$th coordinate we would
mean $(x_{j}y_{j})$.

Suppose $T$ abstractly denotes all the queries made by a classical
algorithm ${\cal A}$. We use $XY_{i:k}(T)$ to denote all the tuples
$(x_{i}y_{i}\dots x_{k}y_{k})$ queried by ${\cal A}_{c}$ from the
$i$th to $k$th coordinate. We use $XY_{i}(T)$ to denote $XY_{i:i}(T)$,
i.e. pairs $(x_{i}y_{i})$ queried by $T$ at the $i$th coordinate.
In the following, when not explicitly stated, we assume the security
parameter is fixed to be $\lambda$ and assume that the problem instance
for $\serial d[{\cal P}]$ is specified by $(x_{0},\dots x_{d})$.

\branchcolor{purple}{As before, we use \Notaref{dSerial_oracles} below.
  \begin{itemize}
    \item The following simple but crucial observation will be used repeatedly
          in our analysis. It adapts offline soundness to our setting. %
          \begin{itemize}
            \item Let $x\leftarrow{\cal X}(1^{\lambda})$. Let ${\cal A}^{{\cal L}}$
                  be a PPT algorithm (trying to solve ${\cal P}$, i.e. finding an $(x,y)\in{\cal S}_{H_{0}}$).
                  Run ${\cal A}^{H}(x)$ and denote its query transcript by $T$.
            \item Let $E$ denote the event that $XY_{0}(T)\cap S_{H_{0}}=\emptyset$.
            \item Fix any $y\in{\cal Y}$. The assertion is that for any fixed $(x,y)$,
                  it holds that $\Pr[(x,y)\in{\cal S}_{H_{0}}|T\land E]\le\ngl{\lambda}$
                  when ${\cal A}^{H}$ is executed.
                  \begin{itemize}
                    \item Suppose the assertion is false. Then for some $(x,y)$ it holds that
                          $\Pr[(x,y)\in{\cal S}_{H_{0}}|T\land E]\ge\mu(\lambda)$ where $\mu$
                          is noticeable (i.e. non-negligible).
                    \item If $E$ does not happen, then $T$ already contains some $(x,y')\in{\cal S}_{H_{0}}$.
                          If $E$ does happen, then $(x,y)\in{\cal S}_{H_{0}}$ with non-vanishing
                          probability (as stated above). In both cases, an element in ${\cal S}_{H_{0}}$
                          is found with non-negligible probability. However, this violates offline
                          soundness of ${\cal P}$.
                  \end{itemize}
          \end{itemize}
    \item Let us build some intuition by starting with a simple circuit of the
          form ${\cal L}\circ U_{1}\circ{\cal A}_{1}^{{\cal L}}$ (on input
          $(x_{0}\dots x_{d})$) and comparing it to ${\cal M}_{1}\circ U_{1}\circ{\cal A}_{1}^{{\cal L}}$
          where ${\cal M}_{1}$ is going to be the shadow of ${\cal L}$ wrt
          some sequence of sets $\bar{S}_{1}$. Take $\bar{S}_{1}$ to be as
          in the $\QNC_{d}$ case, i.e. $\bar{S}_{1}=({\cal S}_{H_{0}}\times{\cal D},{\cal S}_{H_{0}}\times{\cal S}\times{\cal D},\dots)$.
          Let $T_{1}$ denote the queries made by ${\cal A}_{1}^{{\cal L}}$.
          If $XY_{0}(T_{1})$ contains any pair $(x_{0},y_{0})\in{\cal S}_{H_{0}}$,
          then ${\cal L}$ and ${\cal M}_{1}$ can be distinguished (i.e. ${\cal L}\circ U_{1}\circ{\cal A}_{1}^{{\cal L}}$
          and ${\cal M}_{1}\circ U_{1}\circ{\cal A}_{1}^{{\cal L}}$ can behave
          differently) because they can be queried at $\bar{S}_{1}$ (which
          is precisely where ${\cal L}$ and ${\cal M}_{1}$ behave differently).
          Denote by $E_{1}$ the event that $XY_{0}(T_{1})\cap{\cal S}_{H_{0}}=XY_{0}(T_{1})\cap{\cal S}_{0:0}=\emptyset$.
          Using the fact that ${\cal P}$ satisfies offline soundness (in fact
          just from soundness against PPT machines), one has that $\Pr[\neg E_{1}]\le\ngl{\lambda}$.
          From the discussion above, it is also clear that offline soundness
          ensures
          \begin{equation}
            \Pr[(x_{0},y_{0})\in{\cal S}_{H_{0}}|T_{1}E_{1}]\le\ngl{\lambda}\label{eq:explanation_dSer_QCd}
          \end{equation}
          for all $y_{0}\in{\cal Y}$. To apply \Lemref{boundPfind} we would
          need to ensure that the state received by $U_{1}$ is independent
          of $\bar{S}_{1}$. Conditioned on $T_{1}$ and $E_{1}$, this is clearly
          the case (conditioning only reduces polynomially many possible values
          of ${\cal S}_{H_{0}}$). Further, the probability of finding an element
          in $\bar{S}_{1}$ is negligible because of \Eqref{explanation_dSer_QCd}.
    \item This argument can also be applied to ${\cal L}\circ U_{2}\circ{\cal A}_{2}^{{\cal L}}\circ\rho_{1}$
          and ${\cal M}_{2}\circ U_{2}\circ{\cal A}_{2}^{{\cal L}}\circ\rho_{1}$
          where $\rho_{1}={\cal M}_{1}\circ U_{1}\circ{\cal A}_{1}^{{\cal L}}(\rho_{0})$,
          once we appropriately condition the variables. It may help to look
          at the first matrix \Eqref{S_ij_dSerial_QNC_d}. ${\cal M}_{1}$ corresponds
          to the first row. We argue that ${\cal M}_{1}$ does not specify $H_{1}$
          at ${\cal S}_{H_{0}}\times{\cal D}$ and since the second row, i.e.
          $\bar{S}_{2}$, depends on precisely the values of $H_{1}$ on ${\cal S}_{H_{0}}\times{\cal D}$,
          knowing $\rho_{1}$ does not help in determining $\bar{S}_{1}$. This
          is the same as the $\QNC_{d}$ case. To account for the classical
          algorithm, we simply condition ${\cal A}_{1}^{{\cal L}}$ on not querying
          $H_{0}$ inside ${\cal S}_{H_{0}}$ and ${\cal A}_{2}^{{\cal L}}$
          on not querying $H_{1}$ inside ${\cal S}_{H_{1}}$ (or more precisely,
          inside ${\cal S}_{0:1}$). The previous argument then goes through
          unchanged. \\
          We now make this reasoning more precise. Let us condition on the event
          $E_{1}$. Then, it is clear that $\rho_{1}$ contains no information
          about $H_{1}(x_{0},y_{0}||\cdot)$ for any $(x_{0},y_{0})\in{\cal S}_{H_{0}}$.
          This is because, by definition of $E_{1}$, the classical algorithm
          ${\cal A}_{1}^{{\cal L}}$ never accessed $H_{1}$ on the said domain,
          and ${\cal M}_{1}$ contains no information about $H_{1}$ on that
          domain (by definition of $\bar{S}_{1}$). (Note that information about
          $H_{0}$ was present in ${\cal M}_{1}$ and therefore, information
          theoretically, ${\cal S}_{H_{0}}$ could have been determined.) Let
          $T_{2}$ denote the queries made by ${\cal A}_{2}^{{\cal L}}$ in
          the circuits above and denote by $E_{2}$ the event that $XY_{0:1}(T_{2})\cap{\cal S}_{0:1}=\emptyset$,
          i.e. the query transcript so far does not contain a solution to ${\cal P}$
          corresponding to $H_{1}(x_{0}y_{0}||\cdot)$ for any $(x_{0},y_{0})\in{\cal S}_{H_{0}}$.
          Again, from (offline) soundness of ${\cal P}$, it follows that $\Pr[\neg E_{2}|E_{1}]\le\ngl{\lambda}$,
          i.e. ${\cal A}_{2}^{{\cal L}}$ solves ${\cal P}$ corresponding to
          $H_{1}(x_{0}y_{0}||\cdot)$ for $(x_{0}y_{0})\in{\cal S}_{0}$ given
          $E_{1}$, because ${\cal A}_{2}^{{\cal L}}$ does not learn anything
          about $H_{1}(x_{0}y_{0}||\cdot)$ for $(x_{0}y_{0})\in{\cal S}_{0}$
          and $E_{1}$ guarantees ${\cal A}_{1}^{{\cal L}}$ did not even query
          at $(x_{0}y_{0})\in{\cal S}_{0}$. Given $T_{2},E_{2},E_{1},{\cal M}_{1}$,
          from offline soundness of ${\cal P}$ corresponding to $H_{1}(x_{0}y_{0}||\cdot)$,
          it holds that
          \begin{equation}
            \Pr[(x_{1},y_{1})\in{\cal S}_{H_{1}(x_{0}y_{0}||\cdot)}|T_{2}T_{1}E_{2}E_{1}{\cal M}_{1}]\le\ngl{\lambda}\label{eq:explanation_dSer_QCd_2}
          \end{equation}
          for all $y_{1}\in{\cal Y}$ and $(x_{0},y_{0})\in{\cal S}_{0:0}={\cal S}_{H_{0}}$
          follow. To apply \Lemref{boundPfind} one needs to ensure that ${\cal A}_{2}^{{\cal L}}\circ\rho_{1}$
          is independent of $\bar{S}_{2}$. Conditioning on $v_{2}:=(T_{2}T_{1}E_{2}E_{1}{\cal M}_{1})$,
          it is clear that ${\cal A}_{2}^{{\cal L}}\circ\rho_{1}$ contains
          no information about $H_{1}(x_{0}y_{0}||\cdot)$ for $(x_{0}y_{0})\in{\cal S}_{0:0}$.
          Further, $\bar{S}_{2}$, conditioned on $v_{2}$ only depends on
          $H_{1}(x_{0}y_{0}||\cdot)$ for $(x_{0}y_{0})\in{\cal S}_{0:0}$ (again,
          excluding the values in $T_{2}$). Therefore, $\bar{S}_{2}|v_{2}$
          and ${\cal A}_{2}^{{\cal L}}\circ\rho_{1}|v_{2}$ are uncorrelated.
          Finally, the probability of finding an element in $\bar{S}_{2}|v_{2}$
          is negligible due to \Eqref{explanation_dSer_QCd_2}.
  \end{itemize}
  This readily generalises and accounting for these arguments in the
  $\QNC_{d}$ case yields the following.}

\begin{lem}[{$\serial d[{\cal P}] \notin \classQC{d}$}]
  \label{lem:dSer_QC_d_hard}Every $\QC d$ circuit succeeds at solving
  $\serial d[{\cal P}]$ (see \Defref{dSerial}) with probability at
  most $\ngl{\lambda}$ on input $1^{\lambda}$ for $d\le\poly$.
\end{lem}

\branchcolor{blue}{\begin{proof}
    Fix a $\lambda$. Let ${\cal L}=(H_{0},\dots H_{d})$ be as in \Notaref{dSerial_oracles}.
    Suppose the problem instance $\serial d[{\cal P}]$ is specified by
    $(x_{0}\dots x_{d})$ and let $\rho_{0}$ be the initial state. Denote
    by ${\cal B}^{{\cal L}}$ an arbitrary $\QC d$ circuit
    \[
      {\cal B}^{{\cal L}}(\rho_{0}):=\Pivalid\circ{\cal A}_{d+1}^{{\cal L}}\circ{\cal B}_{d}^{{\cal L}}\circ\dots{\cal B}_{1}^{{\cal L}}\circ\rho_{0}
    \]
    where ${\cal B}_{i}^{{\cal L}}:=\Pi_{i}\circ{\cal L}\circ U_{i}\circ{\cal A}_{i}^{{\cal L}}$
    and $\Pivalid$ corresponds to projection on all output strings which
    solve $\serial d[{\cal P}]$. Let $(\bar{S}_{1}\dots\bar{S}_{d})$
    be the output of \Algref{sets_dSerialQNCd}. Define
    \begin{equation}
      {\cal B}^{{\cal M}}(\rho_{0}):=\Pivalid\circ{\cal A}_{d+1}^{{\cal L}}\circ{\cal B}_{d}^{{\cal M}}\circ\dots\circ{\cal B}_{1}^{{\cal M}}\circ\rho_{0}\label{eq:calBM_dSerial_QC_d}
    \end{equation}
    where ${\cal B}_{i}^{{\cal M}}:=\Pi_{i}\circ{\cal M}_{i}\circ U_{i}\circ{\cal A}_{i}^{{\cal L}}$
    and ${\cal M}_{i}$ is the shadow oracle of ${\cal L}$ wrt $\bar{S}_{i}$
    (see \Defref{RandomShadowOracleSequence}). We proceed in two steps.

    \textbf{Step 1:} \emph{${\cal B}^{{\cal L}}$ and ${\cal A}^{{\cal M}}$
    behave the same.} We show that the probability that ${\cal A}^{{\cal L}}$
    and ${\cal A}^{{\cal M}}$ produce a valid output is negligibly close,
    i.e. we bound
    \begin{align*}
          & \left|\tr[\Pivalid\circ{\cal A}_{d+1}^{{\cal L}}\circ{\cal B}_{d+1}^{{\cal L}}\circ\dots{\cal B}_{1}^{{\cal L}}\circ\rho_{0}-\Pivalid\circ{\cal A}_{d+1}^{{\cal L}}\circ{\cal B}_{d+1}^{{\cal M}}\circ\dots{\cal B}_{1}^{{\cal M}}\circ\rho_{0}]\right| \\
      \le & \sum_{i=1}^{d}B({\cal B}_{i}^{{\cal L}}(\rho_{i-1}),{\cal B}_{i}^{{\cal M}}(\rho_{i-1}))\le\sum_{i=1}^{d}\sqrt{2\Pr[{\rm find}:U_{i}^{{\cal L}\backslash\bar{S}_{i}},{\cal A}_{i}^{{\cal L}}\circ\rho_{i-1}]}
    \end{align*}
    where for $i\in\{1,2\dots d-1\}$, $\rho_{i}:={\cal B}_{i}^{{\cal M}}\circ\dots{\cal B}_{1}^{{\cal M}}\circ\rho_{0}$,
    proceeding as in the $\QNC_{d}$ case. To bound the RHS above, one
    can use \Lemref{boundPfind} if it holds that ${\cal A}_{i}^{{\cal L}}\circ\rho_{i-1}$
    is uncorrelated with $\bar{S}_{i}$, upon appropriate conditioning.
    Let $v_{i}:=(T_{i}\dots T_{1}E_{i}\dots E_{1}{\cal M}_{i-1})$ denote
    the random variables we condition on, where $T_{i}$ is the transcript
    of queries made by ${\cal A}_{i}^{{\cal L}}$, $E_{i}$ is the event
    that $XY_{0:i-1}(T_{i})\cap{\cal S}_{0:i-1}=\emptyset$, i.e. the
    transcript does not contain a solution to ${\cal P}$ corresponding
    to $H_{i-1}(x_{0}y_{0}\dots x_{i-2}y_{i-2}||\cdot)$ for any $(x_{0}y_{0}\dots x_{i-2}y_{i-2})\in{\cal S}_{0:i-2}$
    (${\cal S}_{i:j}$ are as in \Algref{sets_dSerialQNCd}) and ${\cal M}_{i-1}$
    is the shadow oracle wrt $\bar{S}_{i-1}$ and contains all the information
    in ${\cal M}_{1}\dots{\cal M}_{i-2}$. As argued above, it is the
    case that ${\cal A}_{i}^{{\cal L}}\circ\rho_{i-1}|v_{i}$ is uncorrelated
    with $\bar{S}_{i}|v_{i}$ because no classical query has been made
    to $H_{i-1}(x_{0}y_{0}\dots x_{i-2}y_{i-2}||\cdot)$ for $(x_{0}y_{0}\dots x_{i-2}y_{i-2})\in{\cal S}_{0:i-2}$
    and all previous shadow oracles, ${\cal M}_{1}\dots{\cal M}_{i-1}$
    output $\perp$ on the aforesaid domain of $H_{i-1}$ while $\bar{S}_{i}$
    conditioned on $v_{i}$ depends only on $H_{i-1}$ at the aforementioned
    domain. It remains to bound the probability of finding an element
    in $\bar{S}_{i}|v_{i}$. To this end, note that given $v_{i}$, from
    the offline soundness of ${\cal P}$ corresponding to $H_{i-1}(x_{0}y_{0}\dots x_{i-2}y_{i-2}||\cdot)$,
    it follows that
    \[
      \Pr[(x_{i},y_{i})\in{\cal S}_{H_{i}(x_{0}\dots y_{i-2}||\cdot)}|v_{i}]\le\ngl{\lambda}
    \]
    for all $y_{i}\in{\cal Y}$ and $(x_{0}\dots y_{i-2})\in{\cal S}_{0:i-2}$.
    This entails that for $i\le k\le d$, $\Pr[(x_{0}y_{0}\dots x_{k}y_{k})\in S_{ik}|v_{i}]\le\ngl{\lambda}$.
    (Offline) soundness of ${\cal P}$ also implies that $\Pr[\neg E_{i}|E_{1}\dots E_{i-1}]\le\ngl{\lambda}$.
    Together, these yield $\Pr[{\rm find}:U_{i}^{{\cal L}\backslash\bar{S}_{i}},{\cal A}_{i}^{{\cal L}}\circ\rho_{i-1}]\le\ngl{\lambda}$.

    \textbf{Step 2:} \emph{${\cal B}^{{\cal M}}$ cannot succeed with
    non-negligible probability.} Consider ${\cal B}^{{\cal M}}$ as in
    \Eqref{calBM_dSerial_QC_d} and let $E_{i}$ and $v_{i}$ be as defined
    above. Since $\Pr[\neg E_{i}|E_{i-1}\dots E_{1}]\le\ngl{\lambda}$,
    it holds that $\Pr[E_{1}\dots E_{d}]\ge1-\ngl{\lambda}$. Conditioned
    on $E_{1}\dots E_{d}$, note that ${\cal M}_{1}\dots{\cal M}_{d}$
    do not specify $H_{d}(x_{0}\dots y_{d-1}||\cdot)$ for $(x_{0}\dots y_{d-1})\in{\cal S}_{0:d-1}$.
    Therefore, $\rho_{d}$ also does not specify $H_{d}(x_{0}\dots y_{d}||\cdot)$
    in the aforesaid domain. From (offline) soundness of ${\cal P}$,
    it follows that ${\cal A}_{d+1}^{{\cal L}}(\rho_{d})|v_{d}$ outputs
    a solution to ${\cal P}$ corresponding to $H_{d}(x_{0}\dots y_{d}||\cdot)$
    is negligible. Together, these yield $\Pr[s\in{\cal S}_{0:d}:s\leftarrow{\cal B}^{{\cal M}}]\le\ngl{\lambda}$
    proving the assertion.

  \end{proof}
}

\subsection{Upper-bounds}

\branchcolor{purple}{If ${\cal P}$ can be solved using $\QNC_{d'}$, then it is evident
  that for any $d\le\ply{\lambda}$, $\serial d[{\cal P}]$ can be solved
  in $\CQ{d'}$. One simply solves ${\cal P}$ corresponding to $H_{0}$
  using the first $\QNC_{d'}$ circuit in $\CQ{d'}$, then using this
  result, solves ${\cal P}$ corresponding to $H_{1}$ and so son. Since
  $d\le\ply{\lambda}$, it follows that $\CQ{d'}$ is sufficient to
  solve the problem. Similarly, if ${\cal P}$ can be solved in $\QC{d'}$,
  then $\serial d[{\cal P}]$ can be solved in $\CQC{d'}$. This yields
  the following.}
\begin{lem}[{${\cal P} \in \QNC_{d'} \implies \serial d[{\cal P}] \in \classCQ{d'}$} and {${\cal P} \in \classQC{d'} \implies \serial d[{\cal P}] \in \classCQC{d'}$}]
  \label{lem:d-Serial_upper}Let ${\cal P}$ be a problem (see \Defref{semi-bounded})
  which can be solved in $\QNC_{d'}$ (resp. $\QC{d'}$). Then, for
  any $d\le\ply{\lambda}$, it holds that $\serial d[{\cal P}]$ (see
  \Defref{dSerial}) can be solved in $\CQ{d'}$ (resp. $\CQC{d'}$).
\end{lem}

\subsection{Consequences\label{subsec:consequences_dSer}}
\begin{thm}
  Fix any $d\le\poly$. Then, with respect to a random oracle, it holds
  that $\BPP^{\QNC_{{\cal O}(1)}}\nsubseteq\QNC_{d}^{\BPP}$.
\end{thm}

\begin{proof}
  Recall $\collisionhashing$ from \Defref{collisionHashing}. One has
  that (using \Defref{dSerial}) $\serial d[\collisionhashing]\in\BPP^{\QNC_{{\cal O}(1)}}$
  using the fact that $\collisionhashing\in\QNC_{{\cal O}(1)}$ (see
  \Lemref{collisionHashing_mainProperties}) and \Lemref{d-Serial_upper}
  with $d'={\cal O}(1)$. One also has that $\serial d[\collisionhashing]\notin\QNC_{d}^{\BPP}$
  because $\collisionhashing$ satisfies all properties required of
  ${\cal P}$ in the definition of $\serial d[{\cal P}]$ (see \Defref{dSerial},
  \Lemref{collisionHashing_mainProperties}, \Lemref{everythingSatisfiesOfflineSoundness})
  and therefore \Lemref{dSer_QC_d_hard} applies, yielding the asserted
  exclusion.
\end{proof}
\branchcolor{purple}{The rest of this article is dedicated to establishing $\hcollisionhashing d$
  is not in $\BPP^{\QNC_{d}}$. Using $\serial d[\hcollisionhashing d]$
  one immediately obtains the separation,
  $\BPP^{\QNC_{d}}\cup\QNC_{d}^{\BPP}\nsubseteq\BPP^{\QNC_{d}^{\BPP}}$. }

\section{\texorpdfstring{$\protect\QNC_{{\cal O}(1)}^{\protect\BPP}\nsubseteq\protect\BPP^{\protect\QNC_{d}}$}{QNC\_O(1)\^{}BPP is not a subset of BPP\^{}\{QNC\_d\}}}

In this section, we define the problem $\hcollisionhashing d$, which is a variant of $\collisionhashing$. This problem shows that $\QNC_{{\cal O}(1)}^{\BPP}\nsubseteq\BPP^{\QNC_{d}}$, relative to a random oracle. %

\subsection{The Problem \label{subsec:hclawhcoll}}

Some notation before we proceed: for $d \in \mathbb{N}$ and $\Sigma\subset \bit^*$ define $h:=H_{d}\circ\dots\circ H_{1}\circ H_{0}$
where $H_{0}:\Sigma\to\Sigma^{d'}$, for $j\in\{1,\dots d-1\}$,
$H_{j}:\Sigma^{d'}\to\Sigma^{d'}$ and $H_{d}:\Sigma^{d'}\to\Sigma$
are independent random oracles with $d'=2d+5$.

\begin{defn}[$\hcollisionhashing d$ or simply $\problem$]\label{def:hcollisionhashing_} Let $d: \mathbb{N} \rightarrow \mathbb{N}$, and\footnote{Obtained by setting $C:=c/4$ where $c$ is as in $\Claimref{collisionHashing}$  with $|A|=2^{\lambda +1}$ and $|B|=2^{\lambda}$ in the limit $\lambda\to \infty$; the $1/4$ factor relates the $G_0,G_1$ based construction to the $g$ based construction. One can treat $G_0,G_1$ as special cases of $g$ with the first input bit $0$ or $1$.} $C=1/(2(e^2-1))$. The $\hcollisionhashing d$ problem is defined as follows. Let $\lambda$ denote be a security parameter for the problem. Consider the following oracles.
    \begin{itemize}
        \item $\G_0, \G_1 : \{0,1\}^{{\lambda}} \rightarrow \{0,1\}^{{\lambda}}$ is a random oracle with domain twice as large as co-domain. 
        \item $h: \{0,1\}^{\lambda} \rightarrow \{0,1\}^{\lambda} $ is a composition of $d(\lambda)+1$ random oracles (as described above with $\Sigma = \bit^{\lambda}$). 
        \item $H: \{0,1\}^{{\lambda}} \times \{0,1\}^{\lambda} \rightarrow \{0,1\}$ is a random oracle with one-bit output. 
    \end{itemize}
Let $\twoone(\G_0, \G_1) := \{y \in \{0,1\}^{\lambda} : |\G_0^{-1}(y)| = |\G_1^{-1}(y)| =1 \}$. Then, the $\hcollisionhashing d$ problem (later referred to simply as $\problem$) is, given access to the oracles $G_0, G_1, H, H_0,\dots H_d$ (but not to $h$ directly) return $(y_1, \ldots, y_{\lambda}$), ($r_1, \ldots, r_{\lambda})$, and $(m_1, \ldots, m_{\lambda})$ such that the following conditions are satisfied.
    \begin{itemize}
        \item All $y_i$'s are distinct.
        \item Let $\mathcal{I} = \{i: y_i \in \twoone(\G_0, \G_1)\}$. Then, $|\mathcal{I}| \geq \frac34 C \lambda$.
        \item Let 
        $\mathcal{I}_{\mathsf{win}} = \{i: y_i \in \twoone(\G_0, \G_1) \textnormal{ and } r_i \cdot (x_0^{y_i} \oplus x_1^{y_i}) \oplus H(x_0^{y_i}, h(y_i)) \oplus H(x_1^{y_i}, h(y_i)) = m_i \}$, where $x_0^{y_i}$ and $x_1^{y_i}$ 
        are the pre-images of $y_i$. Then $|\mathcal{I}_{\mathsf{win}}| \geq 3|\mathcal{I}|/4$.
    \end{itemize} 
\end{defn}

It is helpful to also consider a ``single-copy'' version of $\hcollisionhashing d$, that we refer to as $\subproblem$ and define as follows %
Given the same oracles as in $\hcollisionhashing d$, output $(y,r,m)$ %
such that, $y \in \twoone(\G)$ and $r \cdot (x_0^y \oplus x_1^y) \oplus H(x_0^y, h(y)) \oplus H(x_1^y, h(y)) = m$, where $x_0^y$ and $x_1^y$ are the pre-images of $y$ under $G_0$ and $G_1$ respectively. We call such a $(y,r,m)$ a ``valid equation''.

From \Lemref{collisionHashing_mainProperties} %
it is clear that $\hcollisionhashing d\in\classQC{{\cal O}(1)}$. %
The main result of this
section is the following.
\begin{lem}\label{lem:QC_1_not_in_CQ_d}
Fix any function $d \leq \polyA$. Relative to a random oracle, $\hcollisionhashing d \notin \classCQ d$. %
\end{lem}

\subsection{Consequences\label{subsec:consequences_hcollhclaw}}
Before we get into the proof of \Lemref{QC_1_not_in_CQ_d}, we concisely state its consequences.
\begin{thm}
  Fix any function $d\le\polyA$. Then, relative to a random oracle,
  it holds that $\classQC{{\cal O}(1)}\nsubseteq\classCQ d$.
\end{thm}

Note that $\hcollisionhashing d$ %
satisfies offline soundness because $\collisionhashing$ %
satisfies offline soundness. Therefore, using $\serial d[\hcollisionhashing d]$
and \Lemref{dSer_QC_d_hard}, we conclude the following.
\begin{thm}
  Fix any function $d\le\polyA$. Then, relative to a random oracle,
  it holds that $\classCQC{{\cal O}(1)}\nsubseteq\classQC d\cup\classCQ d$.
\end{thm}

The rest of this section is dedicated to proving \Lemref{QC_1_not_in_CQ_d}. %
Since our proof makes use of the compressed oracle technique, we start by introducing it below.

\subsection{The compressed oracle technique\label{subsec:comporacles}}

\subsubsection{An informal overview}
In this subsection, we give an informal exposition of Zhandry's compressed oracle technique. This subsection is taken almost verbatim from \cite{coladangelo2022deniable}. A reader who is familiar with the technique should feel free to skip this subsection. 

Let $H: \{0,1\}^n \rightarrow \{0,1\}$ be a fixed function. For simplicity, in this overview we restrict ourselves to considering boolean functions (since this is also the relevant case for our scheme).

While classically it is always possible to record the queries of the algorithm, in a way that is undetectable to the algorithm itself, this is not possible in general in the quantum case. The issue arises because the quantum algorithm can \emph{query in superposition}. We illustrate this with an example.

Consider an algorithm that prepares the state $\frac{1}{\sqrt{2}}(\ket{x_0} + \ket{x_1})\ket{y}$, and then makes an oracle query to $H$. The state after the query is
\begin{equation}
    \frac{1}{\sqrt{2}}\ket{x_0}\ket{y\oplus H(x_0)} + \frac{1}{\sqrt{2}}\ket{x_1}\ket{y\oplus H(x_1)}.
\end{equation}
Suppose we additionally ``record'' the query made, i.e. we copy the queried input into a third register. Then the state becomes:
\begin{equation}
    \frac{1}{\sqrt{2}}\ket{x_0}\ket{y\oplus H(x_0)}\ket{x_0} + \frac{1}{\sqrt{2}}\ket{x_1}\ket{y\oplus H(x_1)}\ket{x_1}
\end{equation}

Now, suppose that $H(x_0) = H(x_1)$, then it is easy to see that, in the case where we didn't record queries, the state of the first register after the query is exactly $\frac{1}{\sqrt{2}}(\ket{x_0} + \ket{x_1})$. On the other hand, if we recorded the query, then the third register is now entangled with the first, and as a result the state of the first register is no longer $\frac{1}{\sqrt{2}}(\ket{x_0} + \ket{x_1})$ (it is instead a mixed state). Thus, recording queries is not possible in general without disturbing the state of the oracle algorithm.

Does this mean that all hope of recording queries is lost in the quantum setting? It turns out, surprisingly, that there is a way to record queries when $H$ is a \emph{uniformly random} oracle.

When thinking of an algorithm that queries a uniformly random oracle, it is useful to purify the quantum state of the algorithm via an oracle register (which keeps track of the function that is being queried). An oracle query is then a unitary that acts in the following way on a standard basis element of the query register (where we omit writing normalizing constants):
$$ \ket{x}\ket{y} \sum_{H} \ket{H} \mapsto \sum_{H} \ket{x}\ket{y\oplus H(x)} \ket{H}   \,.$$
It is well known that, up to applying a Hadamard gate on the $y$ register before and after a query, this oracle is equivalent to a ``phase oracle'', which acts in the following way:
\begin{equation}
    \ket{x}\ket{y} \sum_{H} \ket{H} \mapsto \sum_{H} (-1)^{y\cdot H(x)} \ket{x}\ket{y} \ket{H} \label{eq: 45}
\end{equation} 

Now, to get a better sense of what is happening with each query, let's be more concrete about how we represent $H$ using the qubits in the oracle register.

A natural way to represent $H$ is to use $2^n$ qubits, with each qubit representing the output of the oracle at one input, where we take the inputs to be ordered lexicographically. In other words, if $\ket{H} = \ket{t}$, where $t \in \{0,1\}^{2^n}$, then this means that $H(x_i) = t_i$, where $x_i$ is the $i$-th $n$-bit string in lexicographic order. Using this representation, notice that 
$$ \frac{1}{\sqrt{2^n}} \sum_{H} \ket{H} = \ket{+}^{\otimes 2^n} \,.$$
Now, notice that we can write the RHS of Equation (\ref{eq: 45}) as 
$$ \ket{x}\ket{y} \sum_H (-1)^{y \cdot H(x)} \ket{H} \,,$$
i.e. we can equivalently think of the phase in a phase oracle query as being applied to the oracle register.

Thus, when a phase oracle query is made on a standard basis vector of the query register $\ket{x}\ket{y}$, all that happens is
$$ \sum_{H} \ket{H} \mapsto \sum_{H} (-1)^{y\cdot H(x)} \ket{H} \,.$$
Notice that, using the representation for $H$ that we chose above, the latter transformation is:
\begin{itemize}
    \item When $y = 0$,  $$\ket{+}^{\otimes 2^n} \mapsto \ket{+}^{\otimes 2^n} \,. $$
    \item When $y = 1$,  $$\ket{+}^{\otimes 2^n} \mapsto \ket{+}\cdots \ket{+}_{i-1}\ket{-}_i\ket{+}_{i+1}\cdots \ket{+} \,,$$
    where $i$ is such that $x$ is the $i$-th string in lexicographic order.
\end{itemize}
In words, the query does not have any effect when $y=0$, and the query flips the appropriate $\ket{+}$ to a $\ket{-}$ when $y=1$. Then, when we query on a general state $\sum_{x,y} \alpha_{xy} \ket{x}\ket{y}$, the state after the query can be written as:
$$\sum_{x,y} \alpha_{xy} \ket{x}\ket{y} \ket{D_{xy}} \,,$$
where $D_{xy}$ is the all $\ket{+}$ state, except for a $\ket{-}$ corresponding to $x$ if $y=1$.

The crucial observation now is that all of these branches are \emph{orthogonal}, and thus it makes sense to talk about ``the branch on which a particular query was made": the state of the oracle register reveals exactly the query that has been made on that branch. More generally, after $q$ queries, the state will be in a superposition of branches on which at most $q$ of the $\ket{+}$'s have been flipped to $\ket{-}$'s. These locations correspond exactly to the queries that have been made.

Moreover, the good news is that there is a way to keep track of the recorded queries \emph{efficiently}: one does not need to store all of the (exponentially many) $\ket{+}$'s, but it suffices to keep track only of the locations that have flipped to $\ket{-}$ (which is at most $q$). If we know that the oracle algorithm makes at most $q$ queries, then we need merely $n \cdot q$ qubits to store the points that have been queried. We will refer to the set of queried points as the \emph{database}. Formally, there is a well-defined isometry that maps a state on $2^n$ qubits where $q$ of them are in the $\ket{-}$ state, and the rest are $\ket{+}$, to a state on $n \cdot q$ qubits, which stores the $q$ points corresponding to the $\ket{-}$'s in lexicographic order.%

Let $D$ denote an empty database of queried points. Then a query to a uniformly random oracle can be thought of as acting in the following way:
$$ \begin{cases} 
\ket{x}\ket{y} \ket{D} \mapsto \ket{x}\ket{y} \ket{D} \,, \textnormal{ if } y =0\\ 
\ket{x}\ket{y} \ket{D} \mapsto \ket{x}\ket{y} \ket{D \cup \{x\}} \,, \textnormal{ if } y=1 \,.
\end{cases}$$
Such an implementation a uniformly random oracle is referred to as a \emph{compressed phase oracle} simulation \cite{zhandry2019record}.
Formally, the fact that the original and the compressed oracle simulations are \emph{identical} from the point of view of the oracle algorithm (which does not have access to the oracle register) is because at any point in the execution of the algorithm, the states in the two simulations are both purifications of the same mixed state on the algorithm's registers.

We point out that there are two properties of a uniformly random oracle that make a compressed oracle simulation possible:
\begin{itemize}
    \item The query outputs at each point are independently distributed, which means that the state of the oracle register is always a product state across all of the $2^n$ qubits.
    \item Each query output is uniformly distributed. This is important because in general $\alpha \ket{0} + \beta\ket{1} \not \perp \alpha \ket{0} - \beta\ket{1}$ unless $|\alpha| = |\beta|$.
\end{itemize}

Notice that the above compressed oracle simulation does not explicitly keep track of the value of the function at the queried points (i.e. a database is just a set of queried points). In the following slight variation on the compressed oracle simulation, also from \cite{zhandry2019record}, a database is instead a set of pairs $(x,w)$ representing a queried point and the value of the function at that point. This variation will be more useful for our analysis.

Here $D$ is a database of pairs $(x,v)$, which is initially empty. A query acts as follows on a standard basis element $\ket{x}\ket{y} \ket{D}$:
\begin{itemize}
    \item If $y = 0$, do nothing.
    \item If $y = 1$, check if $D$ contains a pair of the form $(x,v)$ for some $v$. 
\begin{itemize}
    \item If it does not, add $(x, \ket{-})$ to the database, where by this we formally mean:
    $D \mapsto \sum_{v} (-1)^v \ket{D \cup (x,v)}$
    \item If it does, apply the unitary that removes $(x,\ket{-})$ from the database.
\end{itemize}
\end{itemize}

One way to understand this compressed simulation is that our database representation only keeps track of pairs $(x, \ket{-})$ (corresponding to the queried points), and it does not keep track of the other unqueried points, which in a fully explicit simulation would correspond to $\ket{+}$'s. One can think of the outputs at the unqueried points as being ``compressed'' in this succinct representation.

It is easy to see that the map above can be extended to a well-defined unitary. In the rest of this overview, we will take this to be our compressed phase oracle. For an oracle algorithm $A$, we will denote by $A^{\comp}$ the algorithm $A$ run with a compressed phase oracle.

\subsubsection{A formal introduction}
\label{sec: comp oracles}
In this subsection, we formally introduce Zhandry's technique for recording queries \cite{zhandry2019record}. This section is loosely based on the explanation in \cite{zhandry2019record}. For a more informal treatment, which carries most of the essence, we suggest starting from the previous section.

\paragraph{Standard and Phase Oracles}
The quantum random oracle, which is the quantum analogue of the classical random oracle, is typically presented in one of two variations: as a \emph{standard} or as a \emph{phase} oracle. 

The standard oracle is a unitary acting on three registers: an $n$-qubit register representing the input to the function, an $m$-qubit register for writing the response, and an $m 2^n$ qubit register representing the truth table of the queried function $H: \{0,1\}^n \rightarrow \{0,1\}^m$. The algorithm that queries the standard oracle has access to the first two registers, while the third register, the oracle's state, is hidden from the algorithm except by making queries. The standard oracle unitary acts in the following way on standard basis states:
$$\ket{x}\ket{y}\ket{H} \mapsto \ket{x}\ket{y \oplus H(x)}\ket{H}\,.$$
For a uniformly random oracle, the oracle register is initialized in the uniform superposition $\frac{1}{\sqrt{m2^n}} \sum_{H} \ket{H}$.
This initialization is of course equivalent to having the oracle register be in a completely mixed state (i.e. a uniformly chosen $H$). This equivalence can be seen by just tracing out the oracle register. We denote the standard (uniformly random) oracle unitary by $\mathsf{StO}$. Moreover, for an oracle algorithm $A$, we will denote by $A^{\mathsf{StO}}$ the algorithm $A$ interacting with the standard oracle, implemented as above.

The phase oracle formally gives a different interface to the algorithm making the queries, but is equivalent to the standard oracle up to Hadamard gates. It again acts on three registers: an $n$-qubit register for the input, an $m$-qubit ``phase'' register, and an $m2^n$-qubit oracle register. It acts in the following way on standard basis states:
$$\ket{x}\ket{s}\ket{H} \mapsto (-1)^{s \cdot H(x)}\ket{x}\ket{s}\ket{H} \,.$$
For a uniformly random oracle, the oracle register is again initialized in the uniform superposition. One can easily see that the standard and phase oracles are equivalent up to applying a Hadamard gate on the phase register before and after a query. We denote the phase oracle unitary by $\mathsf{PhO}$. Moreover, for an oracle algorithm $A$, we will denote by $A^{\mathsf{PhO}}$ the algorithm $A$ interacting with the phase oracle.

\paragraph{Compressed oracle}
The \emph{compressed oracle} technique, introduced by Zhandry \cite{zhandry2019record}, is an equivalent way of implementing a quantum random oracle which (i) is efficiently implementable, and (ii) keeps track of the queried inputs in a meaningful way. This paragraph is loosely based on the explanation in \cite{zhandry2019record}.

In a compressed oracle, the oracle register does not represent the full truth table of the queried function. Instead, it represents a \emph{database} of queried inputs, and the values at those inputs. More precisely, if we have an upper bound $t$ on the number of queries, a database $D$ is represented as an element of the set $S^t$ where $S = (\{0,1\}^n \cup \{\perp\} ) \times \{0,1\}^m$. Each value in $S$ is a pair $(x,y)$: if $x \neq \perp$, then the pair means that the value of the function at $x$ is $y$, which we denote by $D(x) = y$; and if $x = \perp$, then the pair is not currently used, which we denote by $D(x) = \perp$. Concretely, let $l \leq t$. Then, for $x_1 < x_2 < \ldots < x_l$ and $y_1,\ldots, y_l$, the database representing $D(x_i) = y_i$ for $i \in [l]$, with the other $t-l$ points unspecified, is represented as
$$\Big( (x_1, y_1), (x_2,y_2), \ldots, (x_l, y_l), (\perp, 0^m), \ldots, (\perp, 0^m) \Big)$$
where the number of $(\perp, 0^m)$ pairs is $t-l$. We emphasise that in this database representation, the pairs are always ordered lexicographically according to the input value, and the $(\perp, 0^m)$ pairs are always at the end.

In order to define precisely the action of a compressed oracle query, we need to introduce some additional notation. Let $|D|$ denote the number of pairs $(x,y)$ in database $D$ with $x \neq \perp$. Let $t$ be an upper bound on the number of queries. Then, for a database $D$ with $|D|<t$ and $D(x) = \perp$, we write $D \cup (x,y)$ to denote the new database obtained by deleting one of the $(\perp,0^m)$ pairs, and by adding the pair $(x,y)$ to $D$, inserted at the appropriate location (to respect the lexicographic ordering of the input values). 

We also define a ``decompression'' procedure. For $x \in \{0,1\}^n$, $\mathsf{Decomp}_x$ is a unitary operation on the database register. If $D(x) = \perp$, it adds a uniform superposition over all pairs $(x,y)$ (i.e. it ``uncompressed'' at $x$). Otherwise, if $D$ is specified at $x$, and the corresponding $y$ register is in a uniform superposition, $\mathsf{Decomp}$ removes $x$ and the uniform superposition from $D$. If $D$ is specified at $x$, and the corresponding $y$ register is in a state orthogonal to the uniform superposition, then $\mathsf{Decomp}$ acts as the identity. More precisely,
\begin{itemize}
    \item For $D$ such that $D(x) = \perp$ and $|D| < t$, 
    $$ \mathsf{Decomp}_x \ket{D} = \frac{1}{\sqrt{2^m}} \sum_y \ket{D \cup (x,y)} \,.$$
    \item For $D$ such that $D(x) = \perp$ and $D = t$, 
    $$\mathsf{Decomp}_x \ket{D} =\ket{D} \,.$$
    \item For $D$ such that $D(x) \neq \perp$ and $|D|<t$,
   \begin{equation}
    \mathsf{Decomp}_x \left(\frac{1}{\sqrt{2^m}}\sum_y (-1)^{z \cdot y} \ket{D \cup (x,y)} \right) = \begin{cases} \frac{1}{\sqrt{2^m}}\sum_y (-1)^{z \cdot y} \ket{D \cup (x,y)} \textnormal{ if } z \neq 0 \\
    \ket{D}   \textnormal{ if } z = 0
     \end{cases}
   \end{equation}
\end{itemize}
Note that we have specified the action of $\mathsf{Decomp}_x$ on an orthonormal basis of the database register (with a bound of $t$ on the size of the database). Moreover, it is straightforward to verify that $\mathsf{Decomp}_x$ maps this orthonormal basis to another orthonormal basis, and is thus a well-defined unitary. Furthermore, observe that applying $\mathsf{Decomp}_x$ twice gives the identity. Let $\mathsf{Decomp}$ be the related unitary acting on all the registers $x,y, D$ which acts as
$$\mathsf{Decomp} \ket{x,y}\otimes \ket{D} =  \ket{x,y}\otimes \mathsf{Decomp}_x \ket{D} \,.$$

So far, we have considered a fixed upper bound on the number of queries. However, one of the advantages of the compressed oracle technique is that an upper bound on the number of queries does not need to be known in advance. To handle a number of queries that is not fixed, we defined the procedure $\mathsf{Increase}$ which simply increases the upper bound on the size of the database by initialising a new register in the state $\ket{(\perp, 0^n)}$, and appending it to the end. Formally, $\mathsf{Increase} \ket{x,y}\otimes \ket{D} = \ket{x,y}\otimes \ket{D}\ket{(\perp, 0^n)}$.

Now, define the unitaries $\mathsf{CStO}'$ and $\mathsf{CPhO}'$ acting as
\begin{align}
    \mathsf{CStO}' \ket{x,y} \otimes \ket{D} &= \ket{x,y\oplus D(x)} \otimes  \ket{D} \nonumber\\
    \mathsf{CPhO}' \ket{x,y} \otimes \ket{D} &= (-1)^{y \cdot D(x)}\ket{x,y} \otimes  \ket{D}
\end{align}
Finally, we define the compressed standard and phase oracles $\mathsf{CStO}$ and $\mathsf{CPhO}$ as:
\begin{align}
    \mathsf{CStO} &= \mathsf{Decomp} \circ \mathsf{CStO}' \circ \mathsf{Decomp} \circ \mathsf{Increase} \nonumber\\
    \mathsf{CPhO} &= \mathsf{Decomp} \circ \mathsf{CPhO}' \circ \mathsf{Decomp} \circ \mathsf{Increase}
\end{align}

For an oracle algorithm $A$, we denote by $A^{\mathsf{CStO}}$ (resp. $A^{\mathsf{CPhO}}$), the algorithm $A$ run with the compressed standard (resp. phase) oracle, implemented as described above. The following lemma establishes that regular and compressed oracles are equivalent.
\begin{lem}[\cite{zhandry2019record}]
\label{lem: compressed oracle equiv}
For any oracle algorithm $A$, and any input state $\ket{\psi}$, 
$\Pr[A^{\mathsf{StO}}(\ket{\psi}) = 1] = \Pr[A^{\mathsf{CStO}}(\ket{\psi}) = 1]$. Similarly, for any oracle algorithm $B$, $\Pr[B^{\mathsf{PhO}}(\ket{\psi}) = 1] = \Pr[B^{\mathsf{CPhO}}(\ket{\psi}) = 1]$.
\end{lem}

In the rest of the section, we choose to work with \emph{phase} oracles and compressed \emph{phase} oracles. Moreover, to use a more suggestive name, we will denote the compressed phase oracle $\mathsf{CPhO}$ by $\comp$.

\subsection{\texorpdfstring{$\hcollisionhashing d \notin \classCQ d$}{d-hCollisionHashing is not in BPP\^{}QNC\_d}}
\begin{thm}
\label{thm: main}
Let $d \leq \polyA$. Then, any %
$\CQ d$ algorithm solves $\hcollisionhashing d$ with probability at most $1/(1+\frac{C}{3}) + \neglA$, for some negligible function $\neglA$ where $C=1/(2(e^2-1))$.
\end{thm}

The following lemma captures the intuition that the quantum part of a $\CQ d$ algorithm does not have sufficient depth to evaluate $h = H_d \circ \dots \circ H_0$ on its own. We show that, without loss of generality, we can restrict our analysis to potentially unbounded hybrid classical-quantum algorithms where queries to $G_0, G_1, H$ and $H_d$ are polynomially bounded, and moreover the quantum part of the algorithm \emph{does not have access to $H_d$ at all}. To help state this reduction formally, we denote a potentially unbounded hybrid classical-quantum algorithm by $\CQ {\infty}$. In other words, a $\CQ {\infty}$ algorithm has the same structure as a $\CQ {d}$ algorithm except that its classical and quantum parts are computationally unbounded (but they may be query bounded). Then, for $d \leq \polyA$, we denote by $\mathcal{W}_{d}$ be the set of algorithms $B \in \CQ {\infty}$ for $\hcollisionhashing d$ that satisfy the following properties:
\begin{enumerate}
        \item $B$ only makes polynomially many queries to $G_0,G_1,H$, and a (potentially) unbounded number of queries to $H_0,\ldots, H_{d-1}$.
        \item The quantum part of $B$ does not have access to $H_d$.
        \item The classical part of $B$ only makes polynomially many queries to $H_d$.
    \end{enumerate}%
\begin{lem}
\label{lem: wlog}
    Let $d \leq \polyA$. Suppose $A$ is a $\CQ d$ algorithm that solves $\hcollisionhashing d$ with probability $p$. Then, there exists a negligible function $\neglA$ and an algorithm $B \in \mathcal{W}_d$ that solves the same problem with probability at least $p - \neglA$.
\end{lem}

\begin{proof}[Proof sketch]
    Following an argument similar to that in the proof of \Lemref{CQd_hardness} (on $H_0\dots H_d$ which when composed yield $h$), one can show that a circuit in \Figref{CQd_1} behaves like the circuit in \Figref{CQd_3}, i.e. their trace distance is negligible. By inspection, it follows that \Figref{CQd_3} can be simulated by circuit $B$ above. Therefore if $A$ succeeds with $p$ at any task, $B$ succeeds at the same task with probability at least $p-\ngl{\lambda}$.

\end{proof}

From now on, without loss of generality, we restrict to considering algorithms for $\hcollisionhashing d$ that are in $\mathcal{W}_d$. We will show that no such algorithm can solve $\hcollisionhashing d$ with probability greater than $1/(1+\frac{C}{3}) + \neglA$. 

It may be surprising that a seemingly strong class of algorithms $\mathcal{W}_d$ cannot solve $\hcollisionhashing d$ with probability close to $1$. Indeed, the crucial resource that is missing from algorithms in $\mathcal{W}_d$ is that they are unable to maintain coherence while making new queries to $H_d$. This is because, by definition, $H_d$ can only be queried by the classical part.

From here on, we fix a $d \leq \polyA$, and we simply refer to $\hcollisionhashing{d}$ as $\problem$, and to the single-copy version as $\subproblem$. The first step in our proof is to reduce the analysis of algorithms for $\problem$ to algorithms for $\subproblem$. %

\begin{lem}
\label{lem: 102}
Suppose there exists an algorithm $B \in \mathcal{W}_d$ that solves $\problem$ with probability non-negligibly greater than $1/(1+\frac{C}{3})$. Then, there exists an algorithm $A =  \{A_{\lambda}\}_{\lambda \in \mathbb{N}} \in \mathcal{W}_d$ for $\subproblem$, and a non-negligible function $\mathsf{non}\textsf{-}\mathsf{negl}$ such that, for all $\lambda$,
\begin{itemize}
    \item $\Pr[A_{\lambda} \textnormal{ outputs } y \textnormal{ s.t. } y \in \twoone(\G_0, \G_1) ] \geq \mathsf{non}\textsf{-}\mathsf{negl}(\lambda)$, and
    \item $\Pr[A_{\lambda} \textnormal{ wins } | \,y \in \twoone(\G_0, \G_1)] \geq \frac12 + \mathsf{non}\textsf{-}\mathsf{negl}(\lambda)  \,,$
\end{itemize}
where ``$A_{\lambda} \textnormal{ wins}$'' is shorthand for ``$A_{\lambda}$ outputs a valid equation''.
\end{lem}

\begin{proof}
 Let $B = \{B_{\lambda}\}_{\lambda \in \mathbb{N}}$ be a $\mathcal{W}_d$ algorithm that solves $\problem$ with probability non-negligibly greater than $\frac12$. Suppose for a contradiction that the lemma does not hold. This implies that, for all $\mathcal{W}_d$ algorithms $A = \{A_{\lambda}\}_{\lambda \in \mathbb{N}}$ for $\subproblem$, there exists a negligible function $\neglA$ such that, for all $\lambda$,
 \begin{itemize}
 \item $\Pr[A_{\lambda} \textnormal{ outputs } y \textnormal{ s.t. } y \in \twoone(\G_0, \G_1) ] \leq \neglA(\lambda)$, or
\item $\Pr[A_{\lambda} \textnormal{ wins } | \,y \in \twoone(\G_0, \G_1)] \leq \frac12 + \neglA(\lambda)$.
\end{itemize}

Let $B^{i} = \{B^i_{\lambda}\}_{\lambda \in \mathbb{N}}$ be the algorithm for $\subproblem$ that runs algorithm $B$ and returns the $i$-th answer of $B$ as output. Since $B^i$ is a $\mathcal{W}_d$ algorithm, the hypothesis above implies that there exists a negligible function $\neglA_i$ such that, for all $\lambda$,
 \begin{itemize}
 \item $\Pr[B^i_{\lambda} \textnormal{ outputs } y \textnormal{ s.t. } y \in \twoone(\G_0, \G_1) ] \leq \neglA_i(\lambda)$, or
\item $\Pr[B^i_{\lambda} \textnormal{ wins } | \,y \in \twoone(\G_0, \G_1)] \leq \frac12 + \neglA_i(\lambda)$.
\end{itemize}
Let $\neglA = \max_{i} \neglA_i$. This is still a negligible function. Then, we have that, for all $i \in [\lambda]$,
\begin{itemize}
 \item[(i)] $p_{i,\lambda} := \Pr[B^i_{\lambda} \textnormal{ outputs } y \textnormal{ s.t. } y \in \twoone(\G_0, \G_1) ] \leq \neglA(\lambda)$, or
\item[(ii)] $q_{i,\lambda} := \Pr[B^i_{\lambda} \textnormal{ wins } | \,y \in \twoone(\G_0, \G_1)] \leq \frac12 + \neglA(\lambda)$.
\end{itemize}

Let $\mathcal{J}_{\lambda} = \{i: p_{i,\lambda} \leq \neglA(\lambda) \}$, and let $\bar{\mathcal{J}}_{\lambda} := [\lambda] \setminus \mathcal{J}_{\lambda}$.

For brevity, denote by $\mathbf{y} = y_1, \ldots, y_{\lambda}$, and similarly for $\mathbf{r}, \mathbf{m}$. It follows from the above and a union bound that, for all $\lambda$,
\begin{equation}
    \Pr[\exists \,i \in \mathcal{J}_{\lambda} \textnormal{ s.t. } y_i \in \twoone(\G_0, \G_1): (\mathbf{y}, \mathbf{r}, \mathbf{m}) \gets B_{\lambda}] \leq \neglA'(\lambda) \,,
\end{equation}
where $\neglA'(\lambda) = \lambda \cdot \neglA(\lambda)$. We can rewrite the latter as

\begin{equation}
\label{eq: 129}
\Pr[\mathcal{I} \cap \mathcal{J}_{\lambda} \neq \emptyset] \leq \neglA'(\lambda) \,.
\end{equation}

Using the same notation as in the description of $\problem$, we have
\begin{equation}
\label{eq: 128}
    \Pr[B \textnormal{ wins}] \leq \Pr\Big[|\mathcal{I}_{\mathsf{win}}| \geq \frac34 \cdot |\mathcal{I}|\Big] \,,
\end{equation}
since the event ``$B \textnormal{ wins}$'' is a subset of the event ``$|\mathcal{I}_{\mathsf{win}}| \geq \frac34 \cdot |\mathcal{I}|$''. Now, we have
\begin{align}
    \Pr\Big[|\mathcal{I}_{\mathsf{win}}| \geq \frac34 \cdot |\mathcal{I}|\Big] =& \Pr[\mathcal{I} \cap \mathcal{J}_{\lambda} = \emptyset] \cdot \Pr\Big[|\mathcal{I}_{\mathsf{win}}| \geq \frac34 \cdot |\mathcal{I}|\, \Big| \, \mathcal{I} \cap \mathcal{J}_{\lambda} = \emptyset \Big] \nonumber \\
     +& \Pr[\mathcal{I} \cap \mathcal{J}_{\lambda} \neq \emptyset] \cdot \Pr\Big[|\mathcal{I}_{\mathsf{win}}| \geq \frac34 \cdot |\mathcal{I}|\, \Big| \, \mathcal{I} \cap \mathcal{J}_{\lambda} \neq \emptyset \Big] \nonumber\\
     \leq& \Pr[\mathcal{I} \cap \mathcal{J}_{\lambda} = \emptyset] \cdot \Pr\Big[|\mathcal{I}_{\mathsf{win}}| \geq \frac34 \cdot |\mathcal{I}|\, \Big| \, \mathcal{I} \cap \mathcal{J}_{\lambda} = \emptyset \Big] + \neglA'(\lambda)  \nonumber\\
     = & \Pr[\mathcal{I} \cap \mathcal{J}_{\lambda} = \emptyset] \cdot \Pr\Big[|\mathcal{I}_{\mathsf{win}} \cap \bar{\mathcal{J}}_{\lambda} | \geq \frac34 \cdot |\mathcal{I} \cap \bar{\mathcal{J}}_{\lambda}|\, \Big| \, \mathcal{I} \cap \mathcal{J}_{\lambda} = \emptyset \Big] + \neglA'(\lambda) \,.
     \label{eq: 130}
\end{align} 
where the first inequality is implied by Equation (\ref{eq: 129}), and the final equality is because, conditioned on $\mathcal{I} \cap \mathcal{J}_{\lambda} = \emptyset $, we have that $\mathcal{I}_{\mathsf{win}}= \mathcal{I}_{\mathsf{win}} \cap \bar{\mathcal{J}}_{\lambda}$, and $ \mathcal{I}= \mathcal{I} \cap \bar{\mathcal{J}}_{\lambda}  $.

Finally, notice that 
\begin{align}
&\Pr[\mathcal{I} \cap \mathcal{J}_{\lambda} = \emptyset] \cdot \Pr\Big[|\mathcal{I}_{\mathsf{win}} \cap \bar{\mathcal{J}}_{\lambda} | \geq \frac34 \cdot |\mathcal{I} \cap \bar{\mathcal{J}}_{\lambda}|\, \Big| \, \mathcal{I} \cap \mathcal{J}_{\lambda} = \emptyset \Big] \nonumber \\
=& \Pr\Big[|\mathcal{I}_{\mathsf{win}} \cap \bar{\mathcal{J}}_{\lambda} | \geq \frac34 \cdot |\mathcal{I} \cap \bar{\mathcal{J}}_{\lambda}| \, \textnormal{ and } \,\mathcal{I} \cap \mathcal{J}_{\lambda} = \emptyset \Big] \nonumber \\
\leq & \Pr\Big[|\mathcal{I}_{\mathsf{win}} \cap \bar{\mathcal{J}}_{\lambda} | \geq \frac34 \cdot |\mathcal{I} \cap \bar{\mathcal{J}}_{\lambda}| \Big] 
\label{eq: 123}
\end{align}
Combining Equation (\ref{eq: 128}), Equation (\ref{eq: 130}), and Equation (\ref{eq: 123}) gives
\begin{equation}
\label{eq: 133}
\Pr[B \textnormal{ wins}] \leq \Pr\Big[|\mathcal{I}_{\mathsf{win}} \cap \bar{\mathcal{J}}_{\lambda} | \geq \frac34 \cdot |\mathcal{I} \cap \bar{\mathcal{J}}_{\lambda}| \Big] + \neglA'(\lambda) \,.
\end{equation}

Now, notice first that, 
\begin{align}
\Pr\Big[|\mathcal{I} \cap \bar{\mathcal{J}}_{\lambda}| \geq \frac34 C \lambda \Big] 
&\geq \Pr\Big[|\mathcal{I} \cap \bar{\mathcal{J}}_{\lambda}| \geq \frac34 C \lambda \,\textnormal{ and } \, \mathcal{I} \cap \mathcal{J}_{\lambda} = \emptyset \,\Big] \nonumber \\
&= \Pr\Big[|\mathcal{I}| \geq \frac34 C \lambda \,\textnormal{ and } \, \mathcal{I} \cap \mathcal{J}_{\lambda} = \emptyset \,\Big] \nonumber\\
&\geq \Pr\Big[|\mathcal{I}| \geq \frac34 C \lambda\Big] - \neglA'(\lambda) \nonumber \\
&\geq \Pr[B \textnormal{ wins}] -\neglA'(\lambda)  \,,
\end{align}
where the second inequality follows from Equation (\ref{eq: 129}). This implies that 
\begin{equation}
\label{eq: 134}
    \mathbb{E}[|\mathcal{I} \cap \bar{\mathcal{J}}_{\lambda}|] \geq \Pr[B \textnormal{ wins}] \cdot \frac34 C \lambda  - \neglA'(\lambda) \,.
\end{equation} 

Next, we proceed to upper bound
$\Pr\Big[|\mathcal{I}_{\mathsf{win}} \cap \bar{\mathcal{J}}_{\lambda} | \geq \frac34 \cdot |\mathcal{I} \cap \bar{\mathcal{J}}_{\lambda}| \Big]$. Together with Equation (\ref{eq: 133}), this will yield a contradiction.

Notice that, by (i) and (ii), for all $i \in \bar{\mathcal{J}}_{\lambda}$, 
\begin{equation}
    \Pr[B^i_{\lambda} \textnormal{ wins } | \,y_i \in \twoone(\G_0, \G_1)] \leq \frac12 + \neglA(\lambda)\,.
    \label{eq: 125-1}
\end{equation}

Now, for $i \in \bar{\mathcal{J}}_{\lambda}$, define $E_i$ to be the random variable such that:
\begin{equation}
    E_i = \begin{cases}
    1 \textnormal{ if } B_{\lambda}^i \textnormal{ wins and } y_i \in \twoone(G_0, G_1)\\
    0 \textnormal{ otherwise}
    \end{cases}
\end{equation}
Define $F_i$ to be the random variable such that:
\begin{equation}
    F_i = \begin{cases}
    1 \textnormal{ if } y_i \in \twoone(G_0, G_1)\\
    0 \textnormal{ otherwise}
    \end{cases}
\end{equation}
Let $E := \frac{1}{\lambda} \sum_{i \in \bar{\mathcal{J}}_{\lambda}} E_i$, and $F := \frac{1}{\lambda}\sum_{i \in \bar{\mathcal{J}}_{\lambda}} F_i$. Note that $E = |\mathcal{I}_{\mathsf{win}} \cap \bar{\mathcal{J}}_{\lambda}| / \lambda$, and $F = |\mathcal{I} \cap \bar{\mathcal{J}}_{\lambda}|/\lambda$ Then,
\begin{align}
    \mathbb{E}[E] &= \sum_{i \in \bar{\mathcal{J}}_{\lambda}} \mathbb{E}[E_i] \nonumber \\
    &=  \sum_{i \in \bar{\mathcal{J}}_{\lambda}} \Pr[B_{\lambda}^i \textnormal{ wins and } y_i \in \twoone(G_0, G_1)] \nonumber \\
    &= \sum_{i \in \bar{\mathcal{J}}_{\lambda}} \Pr[B^i_{\lambda} \textnormal{ wins } | \,y_i \in \twoone(\G_0, \G_1)] \cdot \Pr[y_i \in \twoone(\G_0, \G_1)] \nonumber\\
    &\leq (\frac12 + \neglA(\lambda)) \cdot \sum_{i \in \bar{\mathcal{J}}_{\lambda}} \Pr[y_i \in \twoone(\G_0, \G_1)]\nonumber \\
    &= (\frac12 + \neglA(\lambda)) \cdot \sum_{i \in \bar{\mathcal{J}}_{\lambda}} \mathbb{E}[F_i] \nonumber\\
     &= (\frac12 + \neglA(\lambda)) \cdot \mathbb{E}[F] \,.
     \label{eq: 138}
\end{align}
We make use of the following:
\begin{restatable}{claim}{claimone}
\label{claim: 1}
Let $E$ and $F$ be random variables taking values in $[0,1]$. Let $\gamma \in [0,1]$. 
$$\Pr[E \geq \gamma \cdot F] \leq 1 - \mathbb{E}(F)\left(1-\frac{\mathbb{E}(E)}{\gamma \cdot \mathbb{E}(F)}\right) $$
\end{restatable}
\begin{proof}
The proof is straightforward and follows from some averaging arguments. It is included in the Appendix for completeness.
\end{proof}

We invoke the claim with $E$ and $F$ defined earlier, and $\gamma = \frac34$. In our case, by Equation (\ref{eq: 134}), $$\mathbb{E}(F) \geq \Pr[B \textnormal{ wins}] \cdot \frac34 C - \neglA'(\lambda) \,,$$
and, by Equation (\ref{eq: 138}), $$\frac{\mathbb{E}(E)}{\mathbb{E}(F)} \leq \frac12 + \neglA(\lambda)\,.$$
Then, by Claim \ref{claim: 1}, we have
\begin{equation}
    \Pr\left[E \geq \frac34 F\right] \leq 1- \frac{\Pr[B \textnormal{ wins}] \cdot C}{3} - \neglA''(\lambda)\,,
    \label{eq: 135}
\end{equation}
for some negligible function $\neglA''$. Combining Equation (\ref{eq: 135}) with Equation (\ref{eq: 133}), and recalling the Definition of $E$ and $F$, we have 
$$ \Pr[B \textnormal{ wins}] - \neglA(\lambda)' \leq 1- \frac{\Pr[B \textnormal{ wins}] \cdot C}{3} - \neglA''(\lambda) \,,$$
which implies 
$$ \Pr[B \textnormal{ wins}] \leq 1/ \big(1+\frac{C}{3}\big) + \neglA'''(\lambda)\,,$$
for some negligible function $\neglA'''$. This is a contradiction.

\end{proof}

\begin{lem}
\label{lem: 103}
Let $A =  \{A_{\lambda}\}_{\lambda \in \mathbb{N}} \in \mathcal{W}_d$ be an algorithm for $\subproblem$. Suppose there exists a function $\eps$ such that, for all $\lambda$, 
\begin{itemize}
\item $\Pr[A_{\lambda} \textnormal{ outputs } y \textnormal{ s.t. } y \in \twoone(\G_0, \G_1) ] \geq \eps(\lambda)$, and
\item $\Pr[A_{\lambda} \textnormal{ wins } | \,y \in \twoone(\G_0, \G_1)] \geq \frac12 + \eps(\lambda)  \,.$
\end{itemize}
Then, there exists a (potentially unbounded) oracle algorithm that makes polynomially many queries to $G_0, G_1$ and outputs a collision with probability at least $\frac{poly(\eps)}{q^3}$, where $q$ is the total number of queries to $G_0, G_1, H, H_d$ made by $A$.
\end{lem}
Recall that $q$ is polynomially bounded, so this quantity is non-negligible when $\eps$ is non-negligible. Lemma \ref{lem: wlog}, Lemma \ref{lem: 102}, and Lemma \ref{lem: 103} together clearly imply Theorem \ref{thm: main}. The rest of the section is dedicated to proving Lemma \ref{lem: 103}.

Algorithm \ref{alg: 1} below is the algorithm that extracts a collision. We introduce some notation before describing it. %
Recall that $A$ alternates classical and quantum circuits. Without loss of generality, we can take $A_{\lambda}$ to be a quantum circuit that applies the unitary: 
$$(\mathsf{CNOT}_{\mathsf{(work,query)} \rightarrow \mathsf{rec}} (U_Q O^G O^H)^{L'} (U_C O^G O^H O^{H_d})^{L})^N\,,$$ where:
\begin{itemize}
 \item $U_{C}$ is a ``classical'' unitary that is diagonal in the standard basis, and acts on registers $\textsf{work}$, $\mathsf{query}$. We assume that $U_C$ also includes (a potentially unbounded number of) queries to oracles $H_0, \ldots, H_{d-1}$.
    \item $U_{Q}$ is a unitary acting on registers $\textsf{work}$, $\textsf{query}$. We again assume that $U_Q$ includes (a potentially unbounded number of) quantum queries to oracles $H_0, \ldots, H_{d-1}$.
    \item $N$ is the total number of quantum circuits. $L$ and $L'$ are respectively the number of oracles calls in each classical and quantum part.
    \item $\mathsf{CNOT}_{(\textsf{work,query}) \rightarrow \mathsf{rec}}$ is a CNOT gate that ``measures'' all of the registers after each $\qnc$ execution by copying them in another register $\mathsf{rec}$.
\end{itemize} 
Note that we are assuming, without loss of generality, that the $\bpp$ and $\qnc$ parts share the same registers, but \emph{all} registers are measured after each $\qnc$ call. 

Finally, for $y$ in the range of $\G_0, \G_1$, and $c_0, c_1 \in \{0,1\}$, denote by $A_{\lambda}^{y, c_0, c_1}$ the algorithm that is identical to $A$, except for the following modification: replace oracle queries $O^H$ with $O^H_{y, c_0, c_1}$ defined as follows:
\begin{equation*}
    O^H_{y, c_0, c_1} \ket{x, w} \ket{z} =  \begin{cases}
     (-1)^{z\cdot H(x)}\ket{x, w} \ket{z}, \textnormal{ if $G_0(x), G_1(x) \neq y$} \\
      (-1)^{z\cdot c_0}\ket{x, w} \ket{z}, \textnormal{ if $G_0(x) = y$} \\ 
      (-1)^{z\cdot c_1}\ket{x, w} \ket{z}, \textnormal{ if $G_0(x) \neq y$ and $G_1(x) = y$}
    \end{cases}
\end{equation*}
where $O^H_{y, c_0, c_1}$ is implemented ``in place'', by querying  $G_0, G_1(x)$, computing in an auxiliary register which of the three cases one is in, applying a controlled unitary based on the value of the control register, and uncomputing everything except the controlled unitary (which returns the auxiliary register to zero). Crucially, $O^H_{y, c_0, c_1}$ can be computed at the cost of one query to $G_0$ and $G_1$.

\begin{algorithm}[Extract a collision]
\label{alg: 1}
$\,$\\
Input: a security parameter $1^{\lambda}$ \\
Oracle access to: $G_0, G_1: \{0,1\}^{n(\lambda)} \rightarrow \{0,1\}^{n(\lambda)}$
\vspace{1mm}

\noindent Run a simulation of the following algorithm, where oracle calls to $H$ are simulated via a compressed oracle simulation, and calls to $H_0, \ldots, H_d$ are simulated inefficiently (by sampling these functions uniformly at random and using a truth table to answer queries). Calls to $G_0$, $G_1$ are made directly to the oracles $G_0, G_1$.
\begin{itemize}
    \item[(i)] Pick $i \leftarrow [N \cdot L]$ (where notice that the latter is the total number of oracle calls that $A_\lambda$ makes to $h$). Let $N_i, L_i$ be such that $i = N_i \cdot L + L_i$, with $0 \leq L_i < N$.
    \item[(ii)] Run $A_{\lambda}$ up until just before the $i$-th query to $H_d$, i.e. apply the unitary $$ (U_C O^G O^H O^{H_d})^{L_i-1} \circ \left(\mathsf{CNOT}_{\mathsf{out}, \mathsf{rec}} (U_Q O^G O^H)^{L'} (U_C O^G O^H O^{H_d})^{L}\right)^{N_i}\,.$$
    Then, measure registers $\mathsf{work}$ and $\mathsf{query}$, and let $adv$ be the outcome. Moreover, let $z$ be the $i$-th query to $H_d$. Let $\mathsf{h}_{\mathsf{data}}^i$ denote the set $\mathsf{h}_{\mathsf{data}}^i = h^{-1}(H_d(z))$ (this set can be computed inefficiently by querying $H_0, \ldots, H_{d-1}$ everywhere). 
    \item[(iii)] Pick $\tilde{y} \leftarrow \mathsf{h}_{\mathsf{data}}^i$. Pick $c_0, c_1, c_0', c_1' \leftarrow \{0,1\}$, and $j, j' \leftarrow [(N - N_i) \cdot (L + L') - L_i] $ (where notice that the latter is the total number of remaining oracle calls to $H$ that the partial run of $A_{\lambda}$ in step (ii) did not perform). Let $V_j$ and $V'_{j'}$ be unitaries corresponding to the continuation of the execution of $A$ from where it stopped in step (ii), for respectively $j$ and $j'$ additional queries to $H$, where we additionally replace oracle calls $O^H$ with oracle calls $O^H_{\tilde{y},c_0, c_1}$ and $O^H_{\tilde{y},c_0', c_1'}$ for $V_j$ and $V'_{j'}$ respectively (we describe these formally after the description of the algorithm). 
    \item[(iv)] Initialize new registers $\mathsf{work'}$ and $\mathsf{query'}$ in the state $\ket{adv}$. Run\footnote{Note that, while the oracle queries in the ``left'' and ``right'' unitaries act on distinct registers $\mathsf{query}$ and $\mathsf{query}'$, one can equivalently replace this unitary with one in which there is a single shared query register, by having one algorithm swap the contents of a local register into the shared query register, query the oracle, and swap out the contents back into the local register.} $$(V_j \otimes V'_{j'})\big(\ket{adv}_{\mathsf{work,query}} \otimes \ket{adv}_{\mathsf{work',query'}}\big)$$ 
    \item[(vi)] Measure the query registers of $H$ in $\mathsf{query}$ and $\mathsf{query'}$ and output a collision if one is found. %
\end{itemize}
\end{algorithm}

To avoid any confusion, we give a formal definition of $V_j$ and $V'_{j'}$. Let $N_j, L_j$ be such that $j = N_j \cdot (L + L') + L_j $, where $0 \leq L_j < L + L'$. Let $\mathsf{CNOT}_{\rightarrow \mathsf{rec}}$ be short for $\mathsf{CNOT}_{(\textsf{work,query}) \rightarrow \mathsf{rec}}$.
    Define $$V_j :=   W_j \circ  \left( (U_C O^G O^H_{\tilde{y}, c_0, c_1} O^{H_d})^{L_i-1} \mathsf{CNOT}_{\rightarrow \mathsf{rec}} (U_Q O^G O^H_{\tilde{y}, c_0, c_1})^{L'} (U_C O^G O^H_{\tilde{y}, c_0, c_1} O^{H_d} )^{L-L_i + 1} \right)^{N_j}\,,$$
    where 
    \begin{equation}
        W_j := \begin{cases} (U_C O^G O^H_{\tilde{y}, c_0, c_1} O^{H_d})^{L_j} \textnormal{ if $L_j \leq L-L_i + 1$} \\
         (U_Q O^G O^H_{\tilde{y}, c_0, c_1})^{L_j- (L-L_i + 1)} (U_C O^G O^H_{\tilde{y}, c_0, c_1} O^{H_d})^{L - L_i + 1} \textnormal{ if $L-L_i + 1 < L_j  \leq L-L_i + 1 + L'$} \\
         (U_C O^G O^H_{\tilde{y}, c_0, c_1} O^{H_d})^{ L_j - (L' + L- L_i + 1)}\mathsf{CNOT}_{\rightarrow \mathsf{rec}} (U_Q O^G O^H_{\tilde{y}, c_0, c_1})^{L'} (U_C O^G O^H_{\tilde{y}, c_0, c_1} O^{H_d})^{L - L_i + 1} \textnormal{   otherwise} \end{cases}
    \end{equation}
$V'_{j'}$ is defined analogously (with $c_0, c_1$ replaced by $c_0', c_1'$).

\subsubsection{A technical lemma}

Let $A$ be an oracle algorithm making $q$ queries to a uniformly random function $H: \{0,1\}^n \rightarrow \{0,1\}$. Denote by $\mathsf{work}$ and $\mathsf{query}$ the registers of $A$, where the former is a work register and the latter a query register to $H$. 

Suppose one runs a compressed oracle simulation of $A$ on some initial state $\ket{\psi}$. We prove an intuitive lemma that directly relates the probability of the final database register containing a particular query $x^*$ to the probability of finding the register $\mathsf{query}$ in the state $x^*$, if this were to be measured before a uniformly selected query. A bit more precisely, we show that if the final compressed oracle state has weight $\Delta$ on databases containing a particular query $x^*$, then if one were to run $A$ and measure register $\mathsf{query}$ before one of the $q$ queries, selected uniformly at random, the measurement outcome would be $x^*$ with probability at least $\Delta/q$. In fact, we show an even more general statement that will be useful in our proof, which lower bounds the probability that measuring a uniformly random query yields $x^*$, and that decompressing the database everywhere yields a particular $H$.

We denote by $\mathsf{Decomp}$ the unitary that decompresses the database at every point. Formally, $\mathsf{Decomp}$ applies $\mathsf{StdDecomp_x}$ for every $x$. For a set $S \subseteq \{0,1\}^n$, denote by $\mathcal{F}(\{0,1\}^n \setminus S, \{0,1\})$ the set of functions from $\{0,1\}^n \setminus S$ to $\{0,1\}$. For $\tilde{H} \in \mathcal{F}(\{0,1\}^n \setminus S, \{0,1\})$, let $\Pi_{\tilde{H}}$, acting on the (decompressed) database register, be the projector onto functions $H$ that are consistent with $\tilde{H}$ outside of $S$. Formally,
$$\Pi_{\tilde{H}} := \sum_{H: \,H|_{\{0,1\}^n \setminus S} = \tilde{H}} \ket{H}\bra{H} \,,$$
where here we are implicitly identifying databases with the functions they specify.

For convenience, we will abuse notation slightly and write $D \ni x$ to mean that $D$ contains a pair $(x, w)$ for some $w$. Moreover, for $x \in \{0,1\}$, let $\Pi_{D \ni x}$, acting on the compressed database register $\mathsf{D}$, be the projector onto databases containing $x$, i.e. 
$$ \Pi_{D \ni x} = \sum_{D \ni x} \ket{D} \bra{D} $$

Without loss of generality, we let $A$ be the algorithm that applies the unitary $(UO)^q$ followed by a measurement of an output register, where $U$ acts on $\mathsf{work}$, $\mathsf{query}$ and $O$ represents the oracle call, which we think of as acting on $\mathsf{query}$, and an ``oracle register'' $\mathsf{O}$ containing the description of $H$. When running a compressed oracle simulation of $A$, the unitary $O$ is replaced by the compressed oracle call $O^{\comp}$, where $\mathsf{Decomp} \circ O^{\comp} = O \circ \mathsf{Decomp}$.

Denote by $\mathcal{X}$ the domain of $H$. In what follows, we use the following notation. For $D \subseteq \mathcal{X}$, we let 
$$ \ket{D} := \sum_{w_x \in \{0,1\}: x\in D} (-1)^{w_x} \ket{\{(x, w_x): x \in D}$$
Denote by $\mathcal{S}_{\mathsf{comp}}$ the set of all (normalized) states of the form:
\begin{equation}
    \sum_{z,x,e, D} \alpha_{z,x,e, D} \ket{z}_{\mathsf{work}}\ket{x,e}_{\mathsf{query}}\ket{D}_{\mathsf{O}} \,.
    \label{eq: scomp}
\end{equation}
These are states that can be reached by running a compressed oracle simulation.

\begin{lem}
\label{lem: technical}
Let $x^* \in \{0,1\}^n$. Let $S \subseteq \{0,1\}^n$ be such that $x^* \in S$. Let $\tilde{H} \in \mathcal{F}(\{0,1\}^n \setminus S, \{0,1\})$. Let $\ket{\Psi_0} \in S_{\comp}$. Let $\ket{\Psi_{final}} = (UO^{\comp})\ket{\Psi_0}$. Let $$\Delta_{0, x^*, \tilde{H}}:= \| \Pi_{\tilde{H}} \mathsf{Decomp} \,\Pi_{D \ni x^*} \ket{\Psi_0}\|^2 \,,$$
and let 
$$\Delta_{final,x^*, \tilde{H}}:= \| \Pi_{\tilde{H}} \mathsf{Decomp} \,\Pi_{D \ni x^*} \ket{\Psi_{final}}\|^2 \,.$$
Then, 
$$ \mathbb{E}_{l \gets \{0,\ldots, q-1\}} \|\ket{x^*}\bra{x^*} (UO)^l \Pi_{\tilde{H}} \mathsf{Decomp} \ket{\Psi_0} \|^2 \geq \frac{1}{q} (\Delta_{final, x^*, \tilde{H}} - \Delta_{0, x^*, \tilde{H}}) \,.$$
\end{lem}

The special case where $S = \{0,1\}^n$ gives the following corollary.
\begin{cor}
Let $x^* \in \{0,1\}^n$. Let $\ket{\Psi_0} \in S_{\comp}$. Let $\ket{\Psi_{final}} = (UO^{\comp})\ket{\Psi_0}$. Let $$\Delta_{0, x^*}:= \|  \Pi_{D \ni x^*} \ket{\Psi_0}\|^2 \,,$$
and let 
$$\Delta_{final,x^*}:= \| \Pi_{D \ni x^*} \ket{\Psi_{final}}\|^2 \,.$$
Then, 
$$ \mathbb{E}_{l \gets \{0,\ldots, q-1\}} \|\ket{x^*}\bra{x^*} (UO)^l \mathsf{Decomp} \ket{\Psi_0} \|^2 \geq \frac{1}{q} (\Delta_{final,x^*}- \Delta_{0, x^*}) \,.$$
\end{cor}

\begin{proof}[Proof of Lemma \ref{lem: technical}] For the rest of the section, we write $\mathbb{E}_l$ as short for $\mathbb{E}_{l \gets \{0,\ldots, q-1\}}$.
Using the fact that $\Pi_{\tilde{H}}$ commutes with both $(UO)^l$ and $\ket{x^*}\bra{x^*}$, that $\mathsf{Decomp} \circ O^{\comp} = O \circ \mathsf{Decomp}$, and that $\ket{x^*}\bra{x^*}$ commutes with $\mathsf{Decomp}$, we have that
\begin{align} 
    &\mathbb{E}_l \|\ket{x^*}\bra{x^*} (UO)^l \Pi_{\tilde{H}} \mathsf{Decomp} \ket{\Psi_0} \|^2 \nonumber \\
   = &\mathbb{E}_l \|\ket{x^*}\bra{x^*} \Pi_{\tilde{H}} \mathsf{Decomp} (UO^{\comp})^l \ket{\Psi_0} \|^2 \nonumber \\
    = &\mathbb{E}_l \| \Pi_{\tilde{H}} \mathsf{Decomp} \ket{x^*}\bra{x^*} (UO^{\comp})^l \ket{\Psi_0} \|^2 \label{eq: 108}
\end{align}

We can write the state $(UO^{\comp})^l \ket{\Psi_0}$ as
\begin{equation}
(UO^{\comp})^l \ket{\Psi_0} = \sum_{\substack{z,x,e\\ D: |D| \leq l}} \alpha^l_{z,x,e,D} \ket{z}_{\mathsf{work}} \ket{x,e}_{\mathsf{query}} \ket{D}\,. 
\end{equation}
for some $\alpha^l_{z,x,e,D}$. For brevity, we will denote by $D^l$ a database with at most $l$ pairs.

Then, by Equation (\ref{eq: 108}), we have
\begin{align}
    &\mathbb{E}_l \|\ket{x^*}\bra{x^*} (UO)^l \Pi_{\tilde{H}} \mathsf{Decomp} \ket{\Psi_0} \|^2\nonumber\\
    = &\mathbb{E}_l \Big\| \Pi_{\tilde{H}} \mathsf{Decomp} \sum_{z,e,D^l} \alpha^l_{z,x^*,e,D^l} \ket{z,x^*,e}\ket{D^l}  \Big\|^2 \,. \label{eq: first}
\end{align}
Now, notice that, for any $D \niton x^*$ and $D' \ni x^*$, we have $\Pi_{\tilde{H}} \mathsf{Decomp} \ket{D} \perp \Pi_{\tilde{H}} \mathsf{Decomp} \ket{D'}$. This is because we can write $$ \mathsf{Decomp} \ket{D} = \bigotimes_{x \in D} \ket{-}_\mathsf{x} \otimes \bigotimes_{x \notin D} \ket{+}_\mathsf{x} =  \ket{+}_\mathsf{x^*} \otimes \bigotimes_{\substack{x \in D \\ x \neq x^*}} \ket{-}_\mathsf{x} \otimes \bigotimes_{\substack{x \notin D \\ x \neq x^*} } \ket{+}_\mathsf{x}\,,$$ and
$$ \mathsf{Decomp} \ket{D'} = \bigotimes_{x \in D'}  \ket{-}_\mathsf{x} \otimes \bigotimes_{x \notin D'} \ket{+}_\mathsf{x} = \ket{-}_\mathsf{x^*} \otimes \bigotimes_{\substack{x \in D' \\ x \neq x^*}} \ket{-}_\mathsf{x} \otimes \bigotimes_{\substack{x \notin D' \\ x \neq x^*} } \ket{+}_\mathsf{x} \,,$$
where $\mathsf{x}$ denotes the sub-register of the decompressed database register corresponding to the value of the oracle at $x$. Finally, notice that $\Pi_{\tilde{H}}$ acts as the identity on register $x^*$, since $\tilde{H} \in \mathcal{F}(\{0,1\}^n \setminus S, \{0,1\})$ and $x^* \in S$. Thus, $\Pi_{\tilde{H}} \mathsf{Decomp} \ket{D}$ and $\Pi_{\tilde{H}} \mathsf{Decomp} \ket{D'}$ are orthogonal, since they are orthogonal on register $\mathsf{x^*}$.
 
Then, we have 
\begin{align}
   \mathrm{ Equation ~(\ref{eq: first})} &= \mathbb{E}_l \Big\|  \Pi_{\tilde{H}} \mathsf{Decomp}   \sum_{\substack{z,e, \\ D^l \ni x^*}} \alpha^l_{z,x^*,e,D^l} \ket{y,x^*,e}\ket{D^l}  \Big\|^2 \nonumber\\
    &+ \mathbb{E}_l \Big\|  \Pi_{\tilde{H}} \mathsf{Decomp}    \sum_{\substack{z,\\ D^l \niton x^*}} \alpha^l_{z,x^*,e=0,D^l} \ket{z,x^*,e=0}\ket{D^l}  \Big\|^2 \nonumber\\
    &+ \mathbb{E}_l \Big\| \Pi_{\tilde{H}} \mathsf{Decomp}   \sum_{\substack{z, \\ D^l \niton x^*}} \alpha^l_{z,x^*,e=1,D^l} \ket{z,x^*,e=1}\ket{D^l}  \Big\|^2 \nonumber \\
    &\geq \mathbb{E}_l \Big\| \Pi_{\tilde{H}} \mathsf{Decomp}   \sum_{\substack{z, \\ D^l \niton x^*}} \alpha^l_{z,x^*,e=1,D^l} \ket{z,x^*,e=1}\ket{D^l}  \Big\|^2 \label{eq: 7}
\end{align}
where the first equality is due to the fact that components with $D \niton x^*$ and with $D \ni x^*$ are orthogonal, and, of course, components with $e=0$ and with $e=1$ are also orthogonal.

We will prove the following lemma.
\begin{lem}
\label{lem: second}
\begin{equation}
\mathbb{E}_l \Big\| \Pi_{\tilde{H}} \mathsf{Decomp}   \sum_{\substack{z \\ D^l \niton x^*}} \alpha^l_{z,x^*,e=1,D^l} \ket{z,x^*,e=1}\ket{D^l}  \Big\|^2 \geq
\frac{1}{q} (\Delta_{final, x^*, \tilde{H}} - \Delta_{0, x^*, \tilde{H}}) \,. \label{eq: 8}
\end{equation}
\end{lem}
Combining Equation (\ref{eq: first}), Equation (\ref{eq: 7}), and Lemma (\ref{lem: second}) immediately yields Lemma \ref{lem: technical}.

Thus, to conclude the proof of Lemma \ref{lem: technical}, we are left with proving Lemma \ref{lem: second}.

\begin{proof}
Notice, via a telescopic sum, that
\begin{align}
    \mathbb{E}_{l \gets \{0,\ldots,q-1\}}\,\, \Bigg[ &\Big\|  \Pi_{\tilde{H}} \mathsf{Decomp}   \sum_{\substack{z,x,e, \\ D^{l+1} \ni x^*}} \alpha^{l+1}_{z,x,e,D^{l+1}} \ket{z,x,e}\ket{D^{l+1}}  \Big\|^2  \nonumber\\
    &- \Big\|  \Pi_{\tilde{H}} \mathsf{Decomp}    \sum_{\substack{z,x,e, \\ D^{l} \ni x^*}} \alpha^{l}_{z,x,e,D^{l}} \ket{z,x,e}\ket{D^{l}}  \Big\|^2 \Bigg] \nonumber\\ &\geq   \frac{1}{q} \bigg( \Big\|  \Pi_{\tilde{H}} \mathsf{Decomp}  \sum_{\substack{z,x,e,\\ D \ni x^*}}   \alpha^q_{z,x,e,D} \ket{z,x,e}\ket{D}  \Big\|^2 \nonumber \\
    &\,\,\,\,\,\quad \,\,\,- \Big\|  \Pi_{\tilde{H}} \mathsf{Decomp}  \sum_{\substack{z,x,e,\\ D \ni x^*}}   \alpha^0_{z,x,e,D} \ket{z,x,e}\ket{D}  \Big\|^2 \bigg) \nonumber\\
    &= \frac{1}{q} (\Delta_{final, x^*, \tilde{H}} - \Delta_{0, x^*, \tilde{H}})
    \,. \label{eq: 9}
\end{align}

For convenience, we will denote the quantities inside the square brackets on the LHS of Equation (\ref{eq: 9}) as $X_{l+1}$ and $X_l$.

Then, 
\begin{align}
    X_{l+1}:= &\Big\|  \Pi_{\tilde{H}} \mathsf{Decomp}   \sum_{\substack{z,x,e, \\ D^{l+1} \ni x^*}} \alpha^{l+1}_{z,x,e,D^{l+1}} \ket{z,x,e}\ket{D^{l+1}}  \Big\|^2\nonumber  \\
    =& \Big\| (U^{-1} \otimes  \Pi_{\tilde{H}} \mathsf{Decomp}   ) \sum_{\substack{z,x,e, \\ D^{l+1} \ni x^*}} \alpha^{l+1}_{z,x,e,D^{l+1}} \ket{z,x,e}\ket{D^{l+1}}  \Big\|^2 \nonumber \\
    =&  \Big\|  \Pi_{\tilde{H}} \mathsf{Decomp}  \circ\,\, O^{\comp} \Bigg[\sum_{\substack{z,x, \\ D^{l} \ni x^*}}  \alpha^{l}_{z,x,e=0,D^{l}} \ket{z,x,e=0}\ket{D^{l}}  \nonumber  \\
    &\,\,\,\,\,\,\,\,\,\,\,\,\,\,\,\,\,\,\,\,\,\,\,\,\,\,\,\,\,\,\,\,+\sum_{\substack{z,x \neq x^*,w, \\ D^{l} \ni x^*}}  \alpha^{l}_{z,x,e=1,D^{l}} \ket{z,x,e=1}\ket{D^{l}}  \nonumber \\
    &\,\,\,\,\,\,\,\,\,\,\,\,\,\,\,\,\,\,\,\,\,\,\,\,\,\,\,\,\,\,\,\,+\sum_{\substack{z, \\ D^{l} \niton x^*}}  \alpha^{l}_{z,x^*,e=1,D^{l}} \ket{z,x^*,e=1}\ket{D^{l}} \Bigg]\Big\|^2 \label{eq: 10}
\end{align}
where the last equality follows from the definition of the compressed oracle call $O^{\comp}$ and the $\alpha^l$ coefficients. In words, the three terms in the last expression correspond to the three ways in which a database containing $x^*$ after the $(l+1)$-th query can originate. 

Using the fact that $\mathsf{Decomp} \circ O^{\comp} = O \circ \mathsf{Decomp}$, and that $\Pi_{\tilde{H}}$ commutes with $O$ (since $\Pi_{\tilde{H}}$ is diagonal in the control basis of $O$), we have
\begin{align}
   Equation (\ref{eq: 10}) =&  \Big\| O\,\circ \left( \Pi_{\tilde{H}}  \mathsf{Decomp}\right)  \Bigg[\sum_{\substack{z,x, \nonumber\\ D^{l} \ni x^*}}  \alpha^{l}_{z,x,e=0,D^{l}} \ket{z,x,e=0}\ket{D^{l}}   \nonumber\\
    &\,\,\,\,\,\,\,\,\,\,\,\,\,\,\,\,\,\,\,\,\,\,\,\,\,\,\,\,\,\,\,\,\,\,\,\,\,\,\,\,\,\,\,\,\,\,\,\,\,\,\,\,\,\,\,\,\,\,\,\,\,\,\,\,\,\,\,\,\,\,\,\,\,\,\,\,\,\,\,\,\,\,\,\,\,\,+\sum_{\substack{y,x \neq x^*,w, \\ D^{l} \ni x^*}}  \alpha^{l}_{z,x,e=1,D^{l}} \ket{z,x,e=1}\ket{D^{l}} \nonumber \\
    &\,\,\,\,\,\,\,\,\,\,\,\,\,\,\,\,\,\,\,\,\,\,\,\,\,\,\,\,\,\,\,\,\,\,\,\,\,\,\,\,\,\,\,\,\,\,\,\,\,\,\,\,\,\,\,\,\,\,\,\,\,\,\,\,\,\,\,\,\,\,\,\,\,\,\,\,\,\,\,\,\,\,\,\,\,\,+\sum_{\substack{z, \\ D^{l} \niton x^*}}  \alpha^{l}_{z,x^*,e=1,D^{l}} \ket{z,x^*,e=1}\ket{D^{l}} \Bigg]\Big\|^2 \nonumber\\
    =&  \Big\| \Pi_{\tilde{H}} \mathsf{Decomp}   \Bigg[\sum_{\substack{z,x, \\ D^{l} \ni x^*}}  \alpha^{l}_{z,x,e=0,D^{l}} \ket{z,x,e=0}\ket{D^{l}}   \nonumber\\
    &\,\,\,\,\,\,\,\,\,\,\,\,\,\,\,\,\,\,+\sum_{\substack{z,x \neq x^*,\\ D^{l} \ni x^*}}  \alpha^{l}_{z,x,e=1,D^{l}} \ket{z,x,e=1}\ket{D^{l}}  \nonumber\\
    &\,\,\,\,\,\,\,\,\,\,\,\,\,\,\,\,\,\,+\sum_{\substack{y,w, \\ D^{l} \niton x^*}}  \alpha^{l}_{z,x^*,e=1,D^{l}} \ket{z,x^*,e=1}\ket{D^{l}} \Bigg]\Big\|^2  \nonumber\\
    =&  \Big\| \Pi_{\tilde{H}} \mathsf{Decomp}  \sum_{\substack{z,x, \nonumber\\ D^{l} \ni x^*}}  \alpha^{l}_{z,x,e=0,D^{l}} \ket{z,x,e=0}\ket{D^{l}}  \Big\|^2 \\
    +&\Big\| \Pi_{\tilde{H}} \mathsf{Decomp}  \sum_{\substack{z,x \neq x^*, \\ D^{l} \ni x^*}}  \alpha^{l}_{z,x,e=1,D^{l}} \ket{z,x,e=1}\ket{D^{l}} \Big\|^2  \nonumber\\
    +&\Big\| \Pi_{\tilde{H}} \mathsf{Decomp}  \sum_{\substack{z, \\ D^{l} \niton x^*}}  \alpha^{l}_{z,x^*,e=1,D^{l}} \ket{z,x^*,e=1}\ket{D^{l}} \Big\|^2  \,, \label{eq: 19}
\end{align}
where the last equality is because the three terms in the sum are orthogonal.

Now, 
\begin{align}
    X_l :=&  \Big\|  \Pi_{\tilde{H}} \mathsf{Decomp}  \sum_{\substack{z,x,e, \\ D^{l} \ni x^*}} \alpha^{l}_{z,x,e,D^{l}} \ket{z,x,e}\ket{D^{l}}  \Big\|^2 \nonumber\\
    =&  \Big\| \Pi_{\tilde{H}} \mathsf{Decomp}  \sum_{\substack{z,x, \\ D^{l} \ni x^*}}  \alpha^{l}_{z,x,e=0,D^{l}} \ket{z,x,e=0}\ket{D^{l}}  \Big\|^2 \nonumber\\
    +&\Big\| \Pi_{\tilde{H}} \mathsf{Decomp}  \sum_{\substack{z,x \neq x^*,\\ D^{l} \ni x^*}}  \alpha^{l}_{z,x,e=1,D^{l}} \ket{z,x,e=1}\ket{D^{l}} \Big\|^2  \nonumber\\
    +&\Big\| \Pi_{\tilde{H}} \mathsf{Decomp}  \sum_{z, D^{l} \ni x^*}  \alpha^{l}_{z,x^*,e=1,D^{l}} \ket{z,x^*,e=1}\ket{D^{l}} \Bigg]\Big\|^2  \label{eq: 23}
\end{align}

Equations Equation (\ref{eq: 19}) and Equation (\ref{eq: 23}) imply
\begin{align}
     &\Big\| \Pi_{\tilde{H}} \mathsf{Decomp}  \sum_{z,D^{l} \niton x^*}  \alpha^{l}_{y,x^*,e=1,D^{l}} \ket{z,x^*,e=1}\ket{D^{l}} \Big\|^2 \\
     &\geq X^{l+1} - X^{l} \,.
\end{align} 
Thus, we have
\begin{align}
    &\mathbb{E}_{l} \Big\| \Pi_{\tilde{H}} \mathsf{Decomp} \sum_{\substack{z \\ D^{l} \niton x^*}}  \alpha^{l}_{y,x^*,e=1,D^{l}} \ket{z,x^*,e=1}\ket{D^{l}} \Big\|^2 \\
    &\geq \mathbb{E}_{l} [X^{l+1} - X^{l} ]\\
    &\geq \frac{1}{q} (\Delta_{final, x^*, \tilde{H}} - \Delta_{0, x^*, \tilde{H}})\,.
\end{align}
where the last line is from Equation (\ref{eq: 9}). This concludes the proof of Lemma \ref{lem: second}, and thus the proof of Lemma \ref{lem: technical}.

\end{proof}
\end{proof}

\subsubsection{The structure of strategies that produce valid equations}

In this section, we prove properties about the structure of strategies that succeed at $\subproblem$. We will later leverage these properties to argue that any algorithm in $\mathcal{W}_d$ (where recall that $\mathcal{W}_d$ was defined before Lemma \ref{lem: wlog}) that succeeds at $\subproblem$ with non-negligible advantage implies there exists an efficient an algorithm to extract collisions of $G_0, G_1$. We emphasize that all of the results in this subsection hold for any algorithm that makes a polynomially-bounded number of queries to $G_0, G_1$ and $H$. Only later in Subsection \ref{sec: structure of Wd strategies}, we will make use of the additional structure of algorithms in $\mathcal{W}_d$.

Let $H: \{0,1\}^n \rightarrow \{0,1\}$ be a uniformly random oracle. Let $S_{\comp}$ be the set of compressed oracle states on registers $\mathsf{Y}, \mathsf{D}, \mathsf{M}, \mathsf{AUX}, \mathsf{O}$, where $\mathsf{Y}$, $\mathsf{D}$, $\mathsf{M}$ correspond to outputs $y, d, m$\footnote{From here on, for the rest of the proof, we switch notation and denote the outputs of the algorithm by $y,d,m$ instead of $y,r,m$.} for $\subproblem$, $\mathsf{AUX}$ includes auxiliary registers, input registers, and query registers, and $\mathsf{O}$ is the compressed database register for $H$. Formally, $S_{\comp}$ is defined as in Equation (\ref{eq: scomp}), except with a different naming of the registers.

Fix oracles $G_0, G_1, h$ for $\subproblem$. Let $y \in \twoone(G_0, G_1)$. For $b \in \{0,1\}$, we denote $\tilde{x}_b^y := (x_b^y, h(y))$.

Let $\Pi_{\mathsf{valid}}$, acting on decompressed databases, be the projector onto valid equations, i.e.
$$ \Pi_{\mathsf{valid}} :=  \sum_{\substack{y,d,m, D:\\  \,m = d \cdot (x_0^y \oplus x_1^y) \oplus D(\tilde{x}_0^y) \oplus D(\tilde{x}_1^y)}} \ket{y,d,m}\bra{y,d,m} \otimes \ket{D}\bra{D} \,.$$ 

We invoke the following ``structure'' theorem, adapted from \cite{coladangelo2022deniable}. We will then extend this structure theorem in Lemma \ref{lem: 106}.

Note that, in general $\tilde{x}_b^y$ could be any function of $y$, and the following structure theorem would hold verbatim. However, for concreteness, we consider $\tilde{x}_b^y = (x_b^y, h(y))$ as this is the relevant choice for $\subproblem$.

When a state $\ket{\Psi} \in \mathcal{S}_{\comp}$ is clear from the context, we denote $$ \Pr[\mathsf{win}] :=  \| \Pi_{\mathsf{valid}} \mathsf{Decomp} \ket{\Psi} \|^2  \,,$$
and we denote
\begin{equation} 
\Pr[\mathsf{win} |\, y] :=  \frac{\| \Pi_{\mathsf{valid}} \mathsf{Decomp} \ket{y}\bra{y}\ket{\Psi} \|^2}{\| \ket{y}\bra{y}\ket{\Psi} \|^2}  \,. \label{eq: pwin}
\end{equation}

As earlier, denote by $O$ the unitary that performs an oracle query, and by $O^{\comp}$ the compressed oracle version of it.

\begin{lem}[Adapted from \cite{coladangelo2022deniable}]
\label{lem: 105}
Fix $G_0, G_1, h$. Let $\ket{\Psi} \in S_{\comp}$. Suppose $$\ket{\Psi} = \sum_{y, d, m, aux, D} \alpha_{y, d,m, aux, D} \ket{y,d,m,aux}\ket{D}\,.$$
Let $y^* \in \twoone(G_0, G_1)$. Let $x_0, x_1$ be such that $G_0(x_0) = G_1(x_1) = y^*$. Let $\tilde{x}_0 = (x_0, h(y^*))$ and $ \tilde{x}_1 = (x_1, h(y^*))$. Let $\epsilon := \Pr[\mathsf{win} | \, y^*] - \frac12$. Suppose, for some $\delta \geq 0$, that  $$\sum_{\substack{d,m,aux,\\ D \ni \tilde{x}_0, \tilde{x}_1}} | \alpha_{y^*,d,m,aux, D}|^2 \leq \delta \cdot \|\ket{y^*}\bra{y^*} \ket{\Psi}\|^2 \,.$$ Then,
\begin{itemize}
    \item[(i)] $$\sum_{\substack{d,m,aux, \\|D \cap \{\tilde{x}_0,\tilde{x}_1\}| = 1}} |\alpha_{y^*,d,m,aux, D}|^2  \geq 2(\eps - \sqrt{\delta}) \cdot \|\ket{y^*}\bra{y^*} \ket{\Psi}\|^2 \,.$$
    \item[(ii)] $$ \sum_{\substack{d,m, aux, \\ D \niton \tilde{x}_0, \tilde{x}_1}} \left| \alpha_{y^*, d,m, aux, D \cup \{\tilde{x}_0\}} - \alpha_{y, d, m,aux, D \cup \{\tilde{x}_1\}} \right|^2 \leq \sum_{\substack{d,m,aux, \\|D \cap \{\tilde{x}_0,\tilde{x}_1\}| = 1}} |\alpha_{y^*,d,m,aux, D}|^2  -2 (\eps - \sqrt{\delta}) \,.$$
\end{itemize}
\end{lem}
\begin{proof}
This is a simple adaptation of the proof of a similar lemma in \cite{coladangelo2022deniable}.
\end{proof}

We now prove a refinement of the structural property about strategies that produce valid equations, by combining Lemma \ref{lem: technical} with Lemma \ref{lem: 105}. The following lemma essentially establishes that strategies that are successful at producing valid equations are such that, with high probability over oracles $H$, the algorithm queries $H$ at a superposition of pre-images the output $y$. In what follows, for $\tilde{H} \in \mathcal{F}(S, \{0,1\})$, we denote by $\ket{\tilde{H}}\bra{\tilde{H}}$ the projector onto oracles $H$ such that $H|_{S} = \tilde{H}$. Let $\Pi_{\tilde{H}} := \mathsf{Decomp}^{-1} \ket{\tilde{H}}\bra{\tilde{H}} \mathsf{Decomp}$. Moreover, recall the notation $\Pr[\mathsf{win}| y]$ from Equation (\ref{eq: pwin}). %

In the following Lemma, $\Xi: [0,\frac12] \times [0,1] \times [0,1] \rightarrow [0,1]$ is a function with the following properties. Suppose $\delta_1, \delta_2: \mathbb{N} \rightarrow [0,1]$ are non-negligible functions. Then,
\begin{itemize}
    \item If $\eps_1:\mathbb{N} \rightarrow [0,\frac12]$ is a non-negligible function, then $1- \Xi(\eps_1, \delta_1, \delta_2)$ is a non-negligible function.
    \item There exists a constant $c>0$ such that, for any $\mu \in [0,\frac12]$, 
    $$\Xi\left(\frac12 - \mu, \delta_1, \delta_2\right) \leq \mu^c \,.$$
\end{itemize}
The exact form of $\Xi$ is given in Equation (\ref{eq: Xi}).
\begin{lem}
\label{lem: 106}
Fix any $G_0, G_1, h$. Let $\ket{\Psi_0} \in S_{\comp}$. Suppose $\ket{\Psi_0} = \sum_{y,d,m, aux, D} \beta_{y, d, m, aux, D} \ket{y,d,m,aux}\ket{D}$. Let $y^* \in \twoone(G_0, G_1)$. Let $x_0, x_1$ be such that $G_0(x_0) = G_1(x_1) = y^*$. Let $\tilde{x}_0 = (x_0, h(y^*))$ and $ \tilde{x}_1 = (x_1, h(y^*))$.
Let $$ \delta_1 :=\sum_{\substack{y,d,m,aux \\ D \ni \, \tilde{x}_0 \textnormal{ or } \tilde{x}_1}} |\beta_{y, d, m, aux, D} |^2\,.$$
Let $U$ be a local unitary, and $O^{\comp}$ a compressed oracle call, and $q \in \mathbb{N}$. Let $$\ket{\Psi_{\textnormal{final}} }= (U O^{\comp})^q \ket{\Psi_0} = \sum_{y, d, m, aux, D} \alpha_{y,d,m,aux,D} \ket{y,d,m,aux}\ket{D} \,.$$
Let $\eps_1 := \Pr[\mathsf{win}| y^*] - \frac12$. Let $\delta_2 := \sum_{\substack{d,m,aux \\ D \ni \tilde{x}_0, \tilde{x}_1}} | \alpha_{y^*,d,m,aux,D} |^2 /  \| \ket{y^*} \bra{y^*} \ket{\Psi_{\textnormal{final}} }\|^2$.
Then, there exists $\mathcal{H}_{good} \subseteq \mathcal{F}(\{0,1\}^n \setminus \{\tilde{x}_0, \tilde{x}_1\}, \{0,1\})$ such that \anote{exact polynomials to be calculated below} 
\begin{itemize}
    \item[(i)] $$\, \sum_{\tilde{H} \in \mathcal{H}_{good}} \frac{\| \Pi_{\tilde{H}} \ket{y^*}\bra{y^*}  \ket{\Psi_{final}}\|^2}{\| \ket{y^*}\bra{y^*} \ket{\Psi_{final}}\|^2} \geq 1 - \Xi(\eps_1,\delta_1,\delta_2) \,.$$ 
    \item[(ii)]  for all $\tilde{H} \in \mathcal{H}_{good}$, $b \in \{0,1\}$,
    \begin{equation*}
        \frac{\mathbb{E}_{l \leftarrow [q]}[\| \ket{x_b}\bra{x_b} (UO)^l \ket{\tilde{H}}\bra{\tilde{H}} \mathsf{Decomp} \ket{\Psi_0} \|^2 ]}{\| \ket{y^*}\bra{y^*} (UO)^q \ket{\tilde{H}}\bra{\tilde{H}} \mathsf{Decomp} \ket{\Psi_0} \|^2} \geq \frac{1}{2q} \cdot (1 - \Xi(\eps_1,\delta_1,\delta_2))\,,
    \end{equation*}
\end{itemize}
\end{lem}

As a special case, Lemma \ref{lem: 106} gives the following characterization of strategies that succeed at the proof of quantumness of \cite{brakerski_simpler_2020}. In the following, denote by $G_0, G_1: \mathcal{X} \rightarrow \mathcal{Y}$ the pair of trapdoor claw-free functions used in the proof of quantumness\footnote{Our characterization applies equally when $G_0, G_1$ are a pair of uniformly random permutations. It suffices for our characterization that it is hard to find collisions between $G_0$ and $G_1$.}. Denote by $H$ the random oracle. For a set $S \subseteq \mathcal{X}$, denote by $\mathcal{F}(S, \{0,1\})$ the set of all functions from $S$ to $\{0,1\}$. Denote by $H|_S$ the restriction of $H$ to domain $S$. Moreover, for an oracle algorithm $A$, and $l \in \mathbb{N}$, denote by $\mathsf{Ext}_l(A)$ the oracle algorithm that runs $A$ up until right before the $l$-th query, and outputs the outcome of measuring the query register. We include a subscript $\lambda$ when we intend to make the dependence on the security parameter explicit. In the following Lemma, $\Xi': [0,\frac12] \rightarrow [0,1]$ are functions with the following properties.
\begin{itemize}
    \item If $\eps:\mathbb{N} \rightarrow [0,\frac12]$ is a non-negligible function, then $1 - \Xi'(\eps)$ is also a non-negligible function.
    \item There exists a constant $c>0$ such that, for any $\mu \in [0,\frac12]$, 
    $$\Xi' \left(\frac12 - \mu \right) \leq \mu^c \,.$$
\end{itemize}

\begin{cor}[Structure theorem for BKVV]
\label{cor: bkvv}
Let $A$ be an algorithm that succeeds at the (single-copy) proof of quantumness of \cite{brakerski_simpler_2020} with probability $1 - \mu$, where $\mu$ is a function of the security parameter such that $1-\mu$ is at least non-negligibly greater than $\frac12$. Then, there exists a negligible function $\neglA$ such that, for all $\lambda$, there exists a set $\mathcal{Y}'_{\lambda} \subseteq \mathcal{Y}_{\lambda}$ such that 
\begin{itemize}
    \item $\Pr[y \in \mathcal{Y}'_{\lambda} :(y,d,m) \gets A_{\lambda}^H] \geq 1 - \Xi'(\mu(\lambda)) - \neglA(\lambda)$.
    \item For all $y \in \mathcal{Y}'_{\lambda}$ the following holds. Let $x_0 = G_0^{-1}(y)$ and $x_1 = G_1^{-1}(y)$.  Let $S = \mathcal{X}\setminus \{x_0,x_1\}$. Then, there exists a set $\mathcal{H}_{good} \subseteq \mathcal{F}(S, \{0,1\})$ such that
    $$ \Pr[H|_S \in \mathcal{H}_{good} | \,A^H_{\lambda} \textnormal{ outputs }y] \geq 1 - \Xi'(\mu(\lambda)) - \neglA(\lambda)\,.$$
    Moreover, for all $\tilde{H} \in \mathcal{H}_{good}$, $b \in \{0,1\}$,
    \begin{equation*}
        \mathbb{E}_{l \leftarrow [q]}\left[ \Pr[H|_S = \tilde{H} \,\, \land \,\, \mathsf{Ext}_l^H(A) \textnormal{ outputs }x_b]\right] \geq \frac{1}{2q} \cdot \big(1 - \Xi'(\mu(\lambda)) - \neglA(\lambda)\big) \cdot \Pr[H|_S = \tilde{H} \,\, \land \,\, A^H_{\lambda} \textnormal{ outputs }y] \,.
    \end{equation*}
\end{itemize}
\end{cor}

\begin{proof}
First, notice that, when considering the proof of quantumness from BKVV, there is no function $h$. So, in Lemma \ref{lem: 106}, one can take $\tilde{x_0} = x_0$ and $\tilde{x_1} = x_1$.

Since by hypothesis $\Pr[A_{\lambda} \textnormal{ wins}] \geq \frac12 + \eps$, we deduce by an averaging argument that, for all $\lambda$, there exists a set $\mathcal{Y}'_{\lambda} \subseteq \mathcal{Y}_{\lambda}$ such that 
\begin{itemize}
    \item[(a)] $\Pr[y \in \mathcal{Y}'_{\lambda} :(y,d,m) \gets A_{\lambda}^H] \geq 1 - \sqrt{1 - 2 \eps}$.
    \item[(b)] For all $y \in \mathcal{Y}'_{\lambda}$, $\Pr[A_{\lambda} \textnormal{ wins} |  A_{\lambda} \textnormal{ outputs } y] \geq   1 - \sqrt{1 - 2 \eps}\,.$
\end{itemize}
Denote by $\mathsf{Y,D,M, AUX}$ the register on which $A$ acts. Consider a compressed oracle simulation of $A$ and additionally denote by $\mathsf{O}$ the compressed oracle register. Let $\ket{\Psi_0} = \ket{0}_{\mathsf{Y,D,M, AUX}}\ket{D = \emptyset}_{\mathsf{O}}$ be the initial state of a compressed oracle simulation of $A$. Let $\ket{\Psi_{final}} = \sum_{y, d, m, aux, D} \alpha_{y,d,m,aux,D} \ket{y,d,m,aux}\ket{D}$ be the final state of a compressed simulation of $A$, right before the final measurement. Notice that there must exist a negligible function $\neglA$, such that, for all $\lambda$, there exists a set $\mathcal{Y}''_{\lambda}\subseteq \mathcal{Y}'_{\lambda}$ such that:
\begin{itemize}
    \item $\Pr[y \in \mathcal{Y}''_{\lambda} :(y,d,m) \gets A_{\lambda}^H] \geq 1 - \sqrt{1 - 2 \eps} - \neglA(\lambda) \,, $
    \item for all $y^* \in \mathcal{Y}''_{\lambda}$,
    \begin{equation}
    \label{eq: negl}
        \sum_{\substack{d,m,aux \\ D \ni \tilde{x}_0, \tilde{x}_1}} | \alpha_{y^*,d,m,aux,D} |^2 /  \| \ket{y^*} \bra{y^*} \ket{\Psi_{\textnormal{final}} }\|^2 \leq \neglA(\lambda) \,.
    \end{equation}
\end{itemize}
Suppose for a contradiction that the above were not the case, then it is easy to see that by running a compressed simulation of $A$, and measuring the database register, one finds a collision with non-negligible probability.

Now, fix any $\lambda$ and any $y \in \mathcal{Y}''_{\lambda}$. Using the notation from Lemma \ref{lem: 106}, we invoke Lemma \ref{lem: 106} with $y^* = y$, and:
\begin{itemize}
    \item $\delta_1 = 0$, which holds since the database is empty in $\ket{\Psi_0}$.
    \item $\eps_1 = 1 - \sqrt{1 - 2 \eps} - \frac12 = \frac12  - \sqrt{1 - 2 \eps} $, which holds by condition (b), since $\mathcal{Y}''_{\lambda} \subseteq \mathcal{Y}'_{\lambda}$, 
    \item $\delta_2 = \neglA(\lambda)$, which we established in (\ref{eq: negl}).
\end{itemize}
It is straightforward to verify that one obtains a function $\Xi'(\eps)$ with the desired properties.
\end{proof} 

The crux in proving Lemma \ref{lem: 106} is to prove the following.

\begin{lem}
\label{lem: 107}
Let $\ket{\Psi} = \sum_{y,d,m, aux, D} \alpha_{y, d, m, aux, D} \ket{y,d,m,aux}\ket{D}$. Fix $y^*, x_0, x_1$. Let $\tilde{x}_0, \tilde{x}_1$ as in Lemma \ref{lem: 105}.
Suppose, for some $\mu_2 >0$,
\begin{equation} 
\sum_{\substack{d,m, aux, \\ D \niton \tilde{x}_0, \tilde{x}_1}} \left| \alpha_{y^*, d,m, aux, D \cup \{\tilde{x}_0\}} - \alpha_{y^*, d, m,aux, D \cup \{\tilde{x}_1\}} \right|^2 \leq (1-\mu_2) \sum_{\substack{d,m,aux, \\|D \cap \{\tilde{x}_0,\tilde{x}_1\}| = 1}} |\alpha_{y^*,d,m,aux, D}|^2 \,. \label{eq: 115}
\end{equation}
For $b \in \{0,1\}$, let $\ket{\phi_{b}}$ be the un-normalized state \begin{align*}
    \ket{\phi_b} :&= \sum_{\substack{d,m,aux\\ D \ni \tilde{x}_b \land D \niton \tilde{x}_{\bar{b}}}} \alpha_{y^*, d, m, aux, D} \ket{y^*, d, m, aux} \ket{D} \\
    &= \sum_{\substack{d,m,aux\\ D \niton \tilde{x}_0, \tilde{x}_1}} \alpha_{y^*, d, m, aux, D \cup \{x_b\}} \ket{y^*, d, m, aux} \ket{D \cup {\tilde{x}_b}} \,.
\end{align*}
Then, there exists $\mathcal{H}_{good} \subseteq \mathcal{F}(\{0,1\}^n \setminus \{\tilde{x}_0, \tilde{x}_1\}, \{0,1\})$ such that
\begin{itemize}
    \item[(i)] $$\, \sum_{\tilde{H} \in \mathcal{H}_{good}} \frac{\|  \Pi_{\tilde{H}} \ket{\phi_0} \|^2 + \|  \Pi_{\tilde{H}} \ket{\phi_1} \|^2}{ \|\ket{\phi_0} \|^2 + \|\ket{\phi_1} \|^2} \geq 1 - \sqrt{1 - \mu_2} \,,$$
    \item[(ii)] for all $\tilde{H} \in \mathcal{H}_{good}$, $b \in \{0,1\}$,
\begin{equation*}
\frac{\big\| \Pi_{\tilde{H}} \ket{\phi_b} \big\|^2}{\|\Pi_{\tilde{H}} \ket{\phi_0} \|^2 + \| \Pi_{\tilde{H}}\ket{\phi_1} \|^2}  \geq \frac12 - \frac{\sqrt{1-(1-\sqrt{1-\mu_2})^2}}{2}
\end{equation*}

\end{itemize}
\end{lem}

Assuming Lemma \ref{lem: 107}, we can prove Lemma \ref{lem: 106}. 

\begin{proof}[Proof of Lemma \ref{lem: 106}]
Since by hypothesis $\epsilon_2 = \Pr[\mathsf{win} | y^*] - \frac12$ and $$\delta_2 := \frac{\sum_{\substack{d,m,aux \\ D \ni \tilde{x}_0, \tilde{x}_1}} | \alpha_{y^*,d,m,aux,D} |^2}{\|\ket{y^*}\bra{y^*} \ket{\Psi_{final}}\|^2}\,,$$ we can apply Lemma \ref{lem: 105} to deduce that 
\begin{itemize}
\item[(a)] 
\begin{equation}
    \sum_{\substack{d,m,aux, \\|D \cap \{\tilde{x}_0,\tilde{x}_1\}| = 1}} |\alpha_{y^*,d,m,aux, D}|^2  \geq 2(\eps_1 - \sqrt{\delta_2}) \cdot \|\ket{y^*}\bra{y^*} \ket{\Psi_{final}}\|^2 \,,
\end{equation}
\item[(b)] $$ \sum_{\substack{d,m, aux, \\ D \niton \tilde{x}_0, \tilde{x}_1}} \left| \alpha_{y^*, d,m, aux, D \cup \{\tilde{x}_0\}} - \alpha_{y, d, m,aux, D \cup \{\tilde{x}_1\}} \right|^2 \leq \sum_{\substack{d,m,aux, \\|D \cap \{\tilde{x}_0,\tilde{x}_1\}| = 1}} |\alpha_{y^*,d,m,aux, D}|^2  -2 (\eps_1 - \sqrt{\delta_2}) \,.$$
\end{itemize}
For $b \in \{0,1\}$, let $\ket{\phi_{b}}$ be the un-normalized state \begin{align*}
    \ket{\phi_b} :&= \sum_{\substack{d,m,aux\\ D \ni \tilde{x}_b \land D \niton \tilde{x}_{\bar{b}}}} \alpha_{y^*, d, m, aux, D} \ket{y^*, d, m, aux} \ket{D} \\
    &= \sum_{\substack{d,m,aux\\ D \niton \tilde{x}_0, \tilde{x}_1}} \alpha_{y^*, d, m, aux, D \cup \{x_b\}} \ket{y^*, d, m, aux} \ket{D \cup {\tilde{x}_b}} \,.
\end{align*}
Then, we can write (a) equivalently as
\begin{equation}
\label{eq: cond a}
   \| \ket{\phi_0} \|^2 + \| \ket{\phi_1} \|^2 \geq 2(\eps_1 - \sqrt{\delta_2}) \cdot \|\ket{y^*}\bra{y^*} \ket{\Psi_{final}}\|^2 \,,
\end{equation}

We can apply Lemma \ref{lem: 107} with $\mu_2 = 2(\eps_1 - \sqrt{\delta_2})$ to deduce that there exists $\mathcal{H}_{good} \subseteq \mathcal{F}(\{0,1\}^n \setminus \{\tilde{x}_0, \tilde{x}_1\}, \{0,1\})$ such that
\begin{itemize}
    \item[(i)]  $$\, \sum_{\tilde{H} \in \mathcal{H}_{good}} \frac{\|  \Pi_{\tilde{H}} \ket{\phi_0} \|^2 + \|  \Pi_{\tilde{H}} \ket{\phi_1} \|^2}{ \|\ket{\phi_0} \|^2 + \|\ket{\phi_1} \|^2} \geq 1 - \sqrt{1-2(\eps_1 - \sqrt{\delta_2})} =: 1 - \xi_1(\eps_1, \delta_2 )\,,$$
    \item[(ii)] for all $\tilde{H} \in \mathcal{H}_{good}$, $b \in \{0,1\}$,
\begin{equation*}
    \frac{\big\| \Pi_{\tilde{H}} \ket{\phi_b} \big\|^2}{\|\Pi_{\tilde{H}} \ket{\phi_0} \|^2 + \| \Pi_{\tilde{H}}\ket{\phi_1} \|^2}  \geq \frac12 - \frac{\sqrt{1-(1-\sqrt{1-2(\eps_1 - \sqrt{\delta_2})})^2}}{2} =: \frac12 - \xi_2(\eps_1, \delta_2)\,.
\end{equation*}
\end{itemize}

Using Lemma \ref{lem: technical}, we get that, for all $\tilde{H} \in \mathcal{H}_{good}$, $b \in \{0,1\}$,
\begin{align}
    &\mathbb{E}_{l \leftarrow [q]}[\| \ket{x_b}\bra{x_b} (UO)^l \ket{\tilde{H}}\bra{\tilde{H}} \mathsf{Decomp} \ket{\Psi_0} \|^2 ] \nonumber\\& \geq \frac{1}{q} \bigg( \big\| \Pi_{\tilde{H}} \ket{\phi_b} \big\|^2  -\| \Pi_{\tilde{H}} \sum_{\substack{y,d,m,aux \\ D \ni \, \tilde{x}_b}} \beta_{y, d, m, aux, D} \ket{y, d, m, aux} \ket{D} \|^2 \bigg)
   \nonumber \\
   &\geq \frac{\frac12 - \xi_2(\eps_1, \delta_2)}{q} \cdot \big(\|\Pi_{\tilde{H}} \ket{\phi_0} \|^2 + \| \Pi_{\tilde{H}}\ket{\phi_1} \|^2\big)
    \nonumber\\&\,\,\,\,\,\,\,\,\,\,\,\,- \frac{1}{q} \|\Pi_{\tilde{H}} \sum_{\substack{y,d,m,aux \\ D \ni \, \tilde{x}_0 \textnormal{ or } \tilde{x}_1}} \beta_{y, d, m, aux, D} \ket{y, d, m, aux} \ket{D} \|^2 \,. \nonumber\\
    &=  \frac{1}{2q} \cdot \bigg( \big(1 - 2\xi_2(\eps_1, \delta_2 \big) \cdot \big(\|\Pi_{\tilde{H}} \ket{\phi_0} \|^2 + \| \Pi_{\tilde{H}}\ket{\phi_1} \|^2\big) \nonumber\\ 
    &\,\,\,\,\,\,\,\,\,\,\,\,\,\,\,\,\,\,\,\,\,\,\,\,\,\,- 2  \|\Pi_{\tilde{H}} \sum_{\substack{y,d,m,aux \\ D \ni \, \tilde{x}_0 \textnormal{ or } \tilde{x}_1}} \beta_{y, d, m, aux, D} \ket{y, d, m, aux} \ket{D} \|^2  \bigg) \nonumber \\
    &:= \frac{1}{2q} \cdot  \Delta_{\tilde{H}} \,.
    \label{eq: 118}
\end{align}
where the second inequality uses (ii) as well as the fact that, for any $\tilde{H} \in \mathcal{F}(\{0,1\}^n \setminus \{\tilde{x}_0, \tilde{x}_1\}, \{0,1\})$, we have $$\| \Pi_{\tilde{H}} \sum_{\substack{y,d,m,aux \\ D \ni \, \tilde{x}_b}} \beta_{y, d, m, aux, D} \ket{y, d, m, aux} \ket{D} \|^2 \leq \| \Pi_{\tilde{H}} \sum_{\substack{y,d,m,aux \\ D \ni \, \tilde{x}_0 \textnormal{ or } \tilde{x}_1}} \beta_{y, d, m, aux, D} \ket{y, d, m, aux} \ket{D} \|^2 \,.$$ 
Now, notice that 
\begin{align}
    \sum_{\tilde{H} \in \mathcal{H}_{good}} \Delta_{\tilde{H}} & \geq \big(1 - \xi_1(\eps_1,  \delta_2) - 2\xi_2(\eps_1,  \delta_2) + \xi_1(\eps_1,  \delta_2) \cdot \xi_2(\eps_1,  \delta_2) - 2\delta_1\big)  \cdot \big( \| \ket{\phi_0} \|^2 + \| \ket{\phi_1} \|^2\big) \nonumber \\
    &\geq \big(1 - \xi_1(\eps_1,  \delta_2) - 2\xi_2(\eps_1,  \delta_2)  - 2\delta_1\big)  \cdot \big( \| \ket{\phi_0} \|^2 + \| \ket{\phi_1} \|^2\big) \,. \label{eq: 262}
\end{align}
where the first inequality uses (i) and the definition of $\delta_1$.

We can rewrite Equation (\ref{eq: 262}) as
\begin{equation}
 \sum_{\tilde{H} \in \mathcal{H}_{good}} \frac{\|  \Pi_{\tilde{H}} \ket{\phi_0} \|^2 + \|  \Pi_{\tilde{H}} \ket{\phi_1} \|^2}{ \|\ket{\phi_0} \|^2 + \|\ket{\phi_1} \|^2} \cdot \Delta'_{\tilde{H}} \, \geq 1 - \xi_1(\eps_1,  \delta_2) - 2\xi_2(\eps_1,  \delta_2)  - 2\delta_1 \,. \label{eq: 263}
\end{equation}
where $$\Delta'_{\tilde{H}} := 1- 2\xi_2(\eps_1, \delta_2) - 2 \cdot \frac{\Big\|\Pi_{\tilde{H}} \sum_{\substack{y,d,m,aux \\ D \ni \, \tilde{x}_0 \textnormal{ or } \tilde{x}_1}} \beta_{y, d, m, aux, D} \ket{y, d, m, aux} \ket{D} \Big\|^2}{\|  \Pi_{\tilde{H}} \ket{\phi_0} \|^2 + \|  \Pi_{\tilde{H}} \ket{\phi_1} \|^2} \,.$$

An averaging argument applied to Equation (\ref{eq: 263}) implies that there exists a set $\mathcal{H}'_{good} \subseteq \mathcal{H}_{good}$ such that:
\begin{itemize}
    \item[(i')] $$\, \sum_{\tilde{H} \in \mathcal{H}'_{good}} \frac{\|  \Pi_{\tilde{H}} \ket{\phi_0} \|^2 + \|  \Pi_{\tilde{H}} \ket{\phi_1} \|^2}{ \|\ket{\phi_0} \|^2 + \|\ket{\phi_1} \|^2} \geq  1 - \sqrt{ \xi_1(\eps_1,  \delta_2) + 2\xi_2(\eps_1,  \delta_2)  + 2\delta_1} \,,$$ 
    \item[(ii')] for all $\tilde{H} \in \mathcal{H}'_{good}$, 
    \begin{equation}
        \Delta'_{\tilde{H}} \geq 1 - \sqrt{ \xi_1(\eps_1,  \delta_2) + 2\xi_2(\eps_1,  \delta_2)  + 2\delta_1} \label{eq: 264}
    \end{equation}
\end{itemize}

Since $\mathcal{H}'_{good} \subseteq \mathcal{H}_{good}$, we can plug the latter bound on $\Delta'_{\tilde{H}}$ into Equation (\ref{eq: 118}) to obtain that, for all $\tilde{H} \in \mathcal{H}'_{good}$,
\begin{align}
    &\mathbb{E}_{l \leftarrow [q]}[\| \ket{x_b}\bra{x_b} (UO)^l \ket{\tilde{H}}\bra{\tilde{H}} \mathsf{Decomp} \ket{\Psi_0} \|^2 ]\nonumber \\ &\geq \frac{1}{2q} \cdot \big(\|  \Pi_{\tilde{H}} \ket{\phi_0} \|^2 + \|  \Pi_{\tilde{H}} \ket{\phi_1} \|^2 \big) \cdot \big(1 - \sqrt{ \xi_1(\eps_1,  \delta_2) + 2\xi_2(\eps_1,  \delta_2)  + 2\delta_1} \big) \label{eq: 266}
\end{align}

Using Equation (\ref{eq: cond a}), we can rewrite (i') as 
\begin{align}
    \sum_{\tilde{H} \in \mathcal{H}'_{good}} \frac{\|  \Pi_{\tilde{H}} \ket{\phi_0} \|^2 + \|  \Pi_{\tilde{H}} \ket{\phi_1} \|^2}{ \| \ket{y^*}\bra{y^*} \ket{\Psi_{final}}\|^2} &\geq  (1 - \sqrt{ \xi_1(\eps_1,  \delta_2) + 2\xi_2(\eps_1,  \delta_2)  + 2\delta_1} )\cdot 2(\eps_1 - \sqrt{\delta_2}) \nonumber\\
    & = (1 - \sqrt{ \xi_1(\eps_1,  \delta_2) + 2\xi_2(\eps_1,  \delta_2)  + 2\delta_1} )\cdot (1 - (1-2(\eps_1 - \sqrt{\delta_2})) ) \nonumber\\
    & = (1 - \sqrt{ \xi_1(\eps_1,  \delta_2) + 2\xi_2(\eps_1,  \delta_2)  + 2\delta_1} )\cdot (1 - \xi_1^2(\eps_1, \delta_2)) \nonumber\\
    &\geq 1 - \sqrt{ \xi_1(\eps_1,  \delta_2) + 2\xi_2(\eps_1,  \delta_2)  + 2\delta_1} - \xi_1^2(\eps_1, \delta_2) \,. \label{eq: 267}
\end{align}

We can further rewrite Equation (\ref{eq: 267}) as 
\begin{align}
     &\sum_{\tilde{H} \in \mathcal{H}'_{good}} \frac{\| \Pi_{\tilde{H}} \ket{y^*}\bra{y^*}  \ket{\Psi_{final}}\|^2}{\| \ket{y^*}\bra{y^*} \ket{\Psi_{final}}\|^2} \cdot \frac{\|  \Pi_{\tilde{H}} \ket{\phi_0} \|^2 + \|  \Pi_{\tilde{H}} \ket{\phi_1} \|^2}{ \| \Pi_{\tilde{H}} \ket{y^*}\bra{y^*}  \ket{\Psi_{final}}\|^2} \nonumber\\
     &\geq 1 - \sqrt{ \xi_1(\eps_1,  \delta_2) + 2\xi_2(\eps_1,  \delta_2)  + 2\delta_1} - \xi_1^2(\eps_1, \delta_2) \,.
\end{align}
By an averaging argument, there exists a set $\mathcal{H}''_{good} \subseteq \mathcal{H}'_{good} $ such that
\begin{itemize}
    \item[(i'')] $$\sum_{\tilde{H} \in \mathcal{H}''_{good}} \frac{\| \Pi_{\tilde{H}} \ket{y^*}\bra{y^*}  \ket{\Psi_{final}}\|^2}{\| \ket{y^*}\bra{y^*} \ket{\Psi_{final}}\|^2} \geq 1 - \sqrt{\sqrt{ \xi_1(\eps_1,  \delta_2) + 2\xi_2(\eps_1,  \delta_2)  + 2\delta_1} - \xi_1^2(\eps_1, \delta_2)} \,,$$
    \item[(ii'')] for all $\tilde{H} \in \mathcal{H}''_{good}$, 
    \begin{equation}
       \frac{\|  \Pi_{\tilde{H}} \ket{\phi_0} \|^2 + \|  \Pi_{\tilde{H}} \ket{\phi_1} \|^2}{\| \Pi_{\tilde{H}} \ket{y^*}\bra{y^*}  \ket{\Psi_{final}}\|^2} \geq  1 - \sqrt{\sqrt{ \xi_1(\eps_1,  \delta_2) + 2\xi_2(\eps_1,  \delta_2)  + 2\delta_1} - \xi_1^2(\eps_1, \delta_2)} \,. \label{eq: 269}
    \end{equation}
\end{itemize}

Plugging (\ref{eq: 269}) into (\ref{eq: 266}), we obtain 
\begin{itemize}
    \item[(ii''')] For all $\tilde{H} \in \mathcal{H}''_{good}$, 
\begin{align}
    &\frac{\mathbb{E}_{l \leftarrow [q]}[\| \ket{x_b}\bra{x_b} (UO)^l \ket{\tilde{H}}\bra{\tilde{H}} \mathsf{Decomp} \ket{\Psi_0} \|^2 ]}{\| \Pi_{\tilde{H}} \ket{y^*}\bra{y^*}  \ket{\Psi_{final}}\|^2}  \nonumber \\ &\geq \frac{1}{2q} \cdot \big(1 - \sqrt{\sqrt{ \xi_1(\eps_1,  \delta_2) + 2\xi_2(\eps_1,  \delta_2)  + 2\delta_1} - \xi_1^2(\eps_1, \delta_2)} \big) \cdot \big(1 - \sqrt{ \xi_1(\eps_1,  \delta_2) + 2\xi_2(\eps_1,  \delta_2)  + 2\delta_1} \big) \\
    &\geq \frac{1}{2q} \cdot (1 - \Xi(\eps_1, \delta_2, \delta_1))\,, \label{eq: 270}
\end{align}
\end{itemize}
where 
\begin{equation}
    \Xi(\eps_1, \delta_2, \delta_1) := \sqrt{\sqrt{ \xi_1(\eps_1,  \delta_2) + 2\xi_2(\eps_1,  \delta_2)  + 2\delta_1} - \xi_1^2(\eps_1, \delta_2)} + \sqrt{ \xi_1(\eps_1,  \delta_2) + 2\xi_2(\eps_1,  \delta_2)  + 2\delta_1} \,. \label{eq: Xi}
\end{equation}

Finally, using the facts that:
\begin{itemize}
    \item $\mathsf{Decomp}$ acts only on the oracle register and $ \mathsf{Decomp} \circ O^{\comp} = O \circ \mathsf{Decomp}$,
    \item $\ket{\tilde{H}}\bra{\tilde{H}}$ commutes with the local unitary evolution and any local measurement,
\end{itemize} 
we have that, for any $\tilde{H}, y^*$, 
\begin{align}
    \|\Pi_{\tilde{H}} \ket{y^*}\bra{y^*}\ket{\Psi_{final}} \|^2 &= \| \ket{\tilde{H}}\bra{\tilde{H}} \mathsf{Decomp} \ket{y^*}\bra{y^*} (UO^{\comp})^q \ket{\Psi_0} \|^2 \\
    &= \| \ket{y^*}\bra{y^*} (UO)^q \ket{\tilde{H}}\bra{\tilde{H}} \mathsf{Decomp} \ket{\Psi_0} \|^2 \,.
\end{align} 

Thus, we can replace the denominator in the LHS of (\ref{eq: 270}) with $\| \ket{y^*}\bra{y^*} (UO)^q \ket{\tilde{H}}\bra{\tilde{H}} \mathsf{Decomp} \ket{\Psi_0} \|^2$. 

Then, $\mathcal{H}''_{good}$ is the desired set and (i'') and (ii''') are the desired conditions. It is also straightforward to check that $\Xi$ as defined above satisfies the desired properties.
\end{proof}

We are left with proving Lemma \ref{lem: 107}.

\begin{proof}[Proof of Lemma \ref{lem: 107}]
For $b \in \{0,1\}$, let $\ket{\phi_{b}}$ be the un-normalized state \begin{align*}
    \ket{\phi_b} :&= \sum_{\substack{d,m,aux\\ D \ni x_b \land D \niton x_{\bar{b}}}} \alpha_{y^*, d, m, aux, D} \ket{y^*, d, m, aux} \ket{D} \\
    &= \sum_{\substack{d,m,aux\\ D \niton \tilde{x}_0, \tilde{x}_1}} \alpha_{y^*, d, m, aux, D \cup \{x_b\}} \ket{y^*, d, m, aux} \ket{D \cup {x_b}} \,.
\end{align*}
First notice that, for any $\tilde{H} \in \mathcal{F}(\{0,1\}^n \setminus \{\tilde{x}_0, \tilde{x}_1\}, \{0,1\})$,
\begin{align}
     \| \Pi_{\tilde{H}} \ket{\phi_1} \|^2 &= \sum_{d,m,aux} \| \Pi_{\tilde{H}} \sum_{D \niton \tilde{x}_0, \tilde{x}_1} \alpha_{y^*, d, aux, D \cup \{\tilde{x}_1\} } \ket{y^*, d, m, aux} \ket{D \cup \{\tilde{x}_1\} }\|^2  \nonumber\\
     &= \sum_{d,aux}  \|\Pi_{\tilde{H}} \sum_{D \niton \tilde{x}_0, \tilde{x}_1} (-1)^{m + d \cdot (x_0 \oplus x_1)} \alpha_{y^*, d, aux, D \cup \{\tilde{x}_0\} } \ket{y^*, d, m, aux} \ket{D \cup \{\tilde{x}_0\} }\|^2 \nonumber\\
     &= \|\Pi_{\tilde{H}} \ket{\phi_1'} \|^2 \label{eq: 120}
\end{align}
where $\ket{\phi_1'} := \sum_{\substack{d,m,aux\\ D \niton \tilde{x}_0, \tilde{x}_1}} (-1)^{m + d \cdot (x_0 \oplus x_1)}\alpha_{y^*, d, m, aux, D \cup \{\tilde{x}_1\}} \ket{y^*, d, m, aux} \ket{D \cup {\tilde{x}_0}}$.
The second equality in Equation (\ref{eq: 120}) holds because the unitary $\ket{D \cup \{\tilde{x}_1\}} \mapsto \ket{D \cup \{\tilde{x}_0\}}$ commutes with $\Pi_{\tilde{H}}$ and thus does not affect the norm, and the phase $(-1)^{m + d \cdot (x_0 \oplus x_1)} $ clearly also does not affect the norm.
Hence, we have 
\begin{equation}
    \| \Pi_{\tilde{H}} \ket{\phi_1} \| = \| \Pi_{\tilde{H}} \ket{\phi'_1} \| \,. \label{eq: 22}
\end{equation}
In the following calculation, the sum is over $\tilde{H} \in \mathcal{F}(\{0,1\}^n \setminus \{\tilde{x}_0, \tilde{x}_1\}, \{0,1\})$. Notice that
\begin{align}
&\sum_{\tilde{H}} (\| \Pi_{\tilde{H}} \ket{\phi_0} \| - \| \Pi_{\tilde{H}} \ket{\phi_1} \| )^2 \nonumber\\
&= \sum_{\tilde{H}} (\| \Pi_{\tilde{H}} \ket{\phi_0} \| - \| \Pi_{\tilde{H}} \ket{\phi'_1} \| )^2 \quad \quad \textnormal{ by Equation (\ref{eq: 22})}  \nonumber\\
& \leq  \sum_{\tilde{H}} \| \Pi_{\tilde{H}} (\ket{\phi_0} - \ket{\phi'_1}) \|^2 \quad \quad \quad \textnormal{ by the triangle inequality} \nonumber \\
&=\|\ket{\phi_0} - \ket{\phi'_1} \|^2 \nonumber \\
&= \sum_{\substack{d, m, aux, \\ D \niton \tilde{x}_0, \tilde{x}_1 }}  | \alpha_{y^*, d, m, aux, D \cup \{\tilde{x}_0\}} - (-1)^{m+ d \cdot (x_0 \oplus x_1)} \alpha_{y^*, d, m, aux, D \cup \{\tilde{x}_1\}}  |^2  \nonumber \\
&\leq  \sum_{\substack{d,m,aux, \\ D \niton \tilde{x}_0, \tilde{x}_1, b \in \{0,1\}}} |\alpha_{y^*, d, m, aux, D \cup \{x_b\}} |^2 \cdot (1 - \mu_2) \quad \quad \textnormal{ by Equation (\ref{eq: 115})} \nonumber \\
&= (\| \ket{\phi_0} \|^2 + \| \ket{\phi_1}\|^2) \cdot (1 - \mu_2) \label{eq: 125}
\end{align}

We can equivalently rewrite Equation (\ref{eq: 125}) as
$$\sum_{\tilde{H}} p_{\tilde{H}} \cdot \delta_{\tilde{H}} \leq (1-\mu_2) \,, $$
where $$p_{\tilde{H}} := \frac{\| \Pi_{\tilde{H}} \ket{\phi_0} \|^2  + \| \Pi_{\tilde{H}} \ket{\phi_1} \|^2}{ \| \ket{\phi_0}\|^2 +  \| \ket{\phi_1} \|^2} \,,$$
and $$\delta_{\tilde{H}} := \frac{( \| \Pi_{\tilde{H}} \ket{\phi_0} \| - \| \Pi_{\tilde{H}} \ket{\phi_1}\|)^2}{\|\Pi_{\tilde{H}} \ket{\phi_0} \|^2 +  \|\Pi_{\tilde{H}} \ket{\phi_1} \|^2} \,.$$
Note that $\sum_{\tilde{H}} p_{\tilde{H}} = 1$. Then, by an averaging argument, there must exist $\mathcal{H}_{good} \subseteq \mathcal{F}(\{0,1\}^n \setminus \{\tilde{x}_0, \tilde{x}_1\}$ such that
\begin{itemize}
    \item[(a)] \begin{equation}
        \sum_{\tilde{H} \in \mathcal{H}_{good}} p_{\tilde{H}} \geq 1 - \sqrt{1-\mu_2} \,,
    \end{equation} 
    \item[(b)] for all $\tilde{H} \in \mathcal{H}_{good}$,
    \begin{equation}
    \label{eq: 127}
    \delta_{\tilde{H}} \leq \sqrt{1-\mu_2} \,.
    \end{equation}
\end{itemize}

We use the following lemma omitting the proof. 
\begin{lem}
\label{lem: vectors}
Let $0\leq \gamma \leq 1$, and $v, w$ vectors in a Hilbert space.
If $\frac{(\|v\| - \| w\|)^2}{\|v\|^2 + \|w\|^2} \leq \gamma $, then $$ \frac{\min(\|v \|^2, \| w\|^2)}{\| v\|^2+ \|w \|^2} \geq \frac12 - \frac{\sqrt{1-(1-\gamma)^2}}{2} \,.$$
\end{lem}

Using Lemma \ref{lem: vectors} we have that (b) implies \begin{itemize}
    \item[(b')] for all $H \in \mathcal{H}_{good}$, for $b \in \{0,1\}$,
    \begin{equation}
        \frac{\| \Pi_{\tilde{H}} \ket{\phi_b} \|^2}{\| \Pi_{\tilde{H}} \ket{\phi_0} \|^2 + \| \Pi_{\tilde{H}} \ket{\phi_1} \|^2 } \geq \frac12 - \frac{\sqrt{1-(1-\sqrt{1-\mu_2})^2}}{2} \,.
    \end{equation}
\end{itemize}
(a) and (b') are the desired conditions.
\end{proof}

We now state Lemma \ref{lem: 106} in a form which will be useful in our proof later on. Let $\Xi$ be the same function as in Lemma \ref{lem: 106}.
\begin{cor} \label{cor: 120}
Suppose the hypothesis of Lemma \ref{lem: 106} holds.
Then, there exists $\mathcal{H}_{good} \subseteq \mathcal{F}(\{0,1\}^n \setminus \{\tilde{x}_0, \tilde{x}_1\}, \{0,1\})$ such that
\begin{itemize}
    \item[(i)] $$\, \sum_{\tilde{H} \in \mathcal{H}_{good}} \| \Pi_{\tilde{H}} \ket{\Psi_0} \|^2 \geq 1 - \sqrt{\Xi(\eps_1,\delta_1,\delta_2)} \cdot \| \ket{y^*}\bra{y^*} \ket{\Psi_{final}}\|^2 \,.$$ 
    \item[(ii)]  for all $\tilde{H} \in \mathcal{H}_{good}$, $b \in \{0,1\}$,
    \begin{equation}
        \frac{\mathbb{E}_{l \leftarrow [q]}[\| \ket{x_b}\bra{x_b} (UO)^l \ket{\tilde{H}}\bra{\tilde{H}} \mathsf{Decomp} \ket{\Psi_0} \|^2 ]}{\| \Pi_{\tilde{H}} \ket{\Psi_0} \|^2} \geq  \frac{1}{2q} \big(1 - \sqrt{\Xi(\eps_1,\delta_1,\delta_2)} - \Xi(\eps_1,\delta_1,\delta_2) \big) \cdot \| \ket{y^*}\bra{y^*} \ket{\Psi_{final}}\|^2 \,.
    \end{equation}
\end{itemize}
\end{cor}

\begin{proof}
From Lemma \ref{lem: 106}, we have that there exists $\mathcal{H}_{good} \subseteq \mathcal{F}(\{0,1\}^n \setminus \{\tilde{x}_0, \tilde{x}_1\}, \{0,1\})$ such that
\begin{itemize}
    \item[(i)] $$\, \sum_{\tilde{H} \in \mathcal{H}_{good}} \frac{\| \Pi_{\tilde{H}} \ket{y^*}\bra{y^*}  \ket{\Psi_{final}}\|^2}{\| \ket{y^*}\bra{y^*} \ket{\Psi_{final}}\|^2} \geq 1 - \Xi(\eps_1,\delta_1,\delta_2) \,.$$ 
    \item[(ii)]  for all $\tilde{H} \in \mathcal{H}_{good}$, $b \in \{0,1\}$,
    \begin{equation*}
        \frac{\mathbb{E}_{l \leftarrow [q]}[\| \ket{x_b}\bra{x_b} (UO)^l \ket{\tilde{H}}\bra{\tilde{H}} \mathsf{Decomp} \ket{\Psi_0} \|^2 ]}{\| \ket{y^*}\bra{y^*} (UO)^q \ket{\tilde{H}}\bra{\tilde{H}} \mathsf{Decomp} \ket{\Psi_0} \|^2} \geq \frac{1}{2q} \cdot \big(1 - \Xi(\eps_1,\delta_1,\delta_2)\big)
    \end{equation*}
\end{itemize}

We can rewrite (i) as 
$$ \sum_{\tilde{H} \in \mathcal{H}_{good}} p^{tot}_{\tilde{H}} \cdot \frac{\| \Pi_{\tilde{H}} \ket{y^*}\bra{y^*}  \ket{\Psi_{final}}\|^2} {p^{tot}_{\tilde{H}}  } \geq 1 - \Xi(\eps_1,\delta_1,\delta_2) \cdot \| \ket{y^*}\bra{y^*} \ket{\Psi_{final}}\|^2 \,,$$
where $p^{tot}_{\tilde{H}}  := \| \Pi_{\tilde{H}} \ket{\Psi_0} \|^2$. Then, by an averaging argument, there exists $\mathcal{H}'_{good} \subseteq \mathcal{H}_{good}$ such that 
\begin{itemize}
    \item[(a)] \begin{equation} \sum_{\tilde{H} \in \mathcal{H}'_{good}} p^{tot}_{\tilde{H}}  \geq 1 - \sqrt{\Xi(\eps_1,\delta_1,\delta_2)} \,,
    \end{equation}
    \item[(b)]  for all $H \in \mathcal{H}'_{good}$,
   \begin{equation}
      \frac{\| \Pi_{\tilde{H}} \ket{y^*}\bra{y^*}  \ket{\Psi_{final}}\|^2}{p^{tot}_{\tilde{H}} } \geq \frac{1}{2q} \Big(1 - \sqrt{\Xi(\eps_1,\delta_1,\delta_2)}\Big) \cdot \Big(1 - \Xi(\eps_1,\delta_1,\delta_2) \Big) \cdot \| \ket{y^*}\bra{y^*} \ket{\Psi_{final}}\|^2
   \end{equation}
\end{itemize}
Now, notice that $$\| \ket{y^*}\bra{y^*} (UO)^q \ket{\tilde{H}}\bra{\tilde{H}} \mathsf{Decomp} \ket{\Psi_0} \|^2 = \| \Pi_{\tilde{H}} \ket{y^*}\bra{y^*}  \ket{\Psi_{final}}\|^2\,.$$

Then, since $\mathcal{H}'_{good} \subseteq \mathcal{H}_{good}$, (ii) and (b) together imply
\begin{itemize}
    \item[(b')] for all $H \in \mathcal{H}'_{good}$, $b \in \{0,1\}$,
    \begin{equation}
        \frac{\mathbb{E}_{l \leftarrow [q]}[\| \ket{x_b}\bra{x_b} (UO)^l \ket{\tilde{H}}\bra{\tilde{H}} \mathsf{Decomp} \ket{\Psi_0} \|^2 ]}{p^{tot}_{\tilde{H}}} \geq \frac{1}{2q} \big(1 - \sqrt{\Xi(\eps_1,\delta_1,\delta_2)} - \Xi(\eps_1,\delta_1,\delta_2) \big) \cdot \| \ket{y^*}\bra{y^*} \ket{\Psi_{final}}\|^2
    \end{equation}
\end{itemize}
$\mathcal{H}'_{good}$ is the desired set, and (a) and (b') are the desired conditions.
\end{proof}

We state a simple consequence of Corollary \ref{cor: 120}, which we will use directly in our proof later on. Let $\tilde{x}_0 \neq \tilde{x}_1$ be in the domain of $H$, and let $c_0, c_1 \in \{0,1\}$. Let $O^H_{(\tilde{x}_0, c_0), (\tilde{x}_1, c_1)}$ be defined as:

\begin{equation*}
    O^H_{(\tilde{x}_0, c_0), (\tilde{x}_1, c_1)} \ket{x} \ket{z} =  \begin{cases}
     O^H \ket{x} \ket{z}, \textnormal{ if $x \neq \tilde{x}_0,\tilde{x}_1$} \\
     (-1)^{z\cdot c_0} \ket{x} \ket{z}, \textnormal{ if $x = \tilde{x}_0$} \\ 
       (-1)^{z\cdot c_1} \ket{x} \ket{z}, \textnormal{ if $x = \tilde{x}_1$}
    \end{cases}
\end{equation*}

\begin{cor}
\label{cor: 102}
Suppose the conditions of Lemma \ref{lem: 106} hold. Then, there exists $\mathcal{H}_{good} \subseteq \mathcal{F}(\{0,1\}^n \setminus \{\tilde{x}_0, \tilde{x}_1\}, \{0,1\})$ such that
\begin{itemize}
    \item[(i)] $\, \sum_{\tilde{H} \in \mathcal{H}_{good}} \| \Pi_{\tilde{H}} \ket{\Psi_0} \|^2  \geq \big( 1 - \sqrt{\Xi(\eps_1,\delta_1,\delta_2)} \big) \cdot \| \ket{y^*}\bra{y^*} \ket{\Psi_{final}}\|^2 \,,$
    \item[(ii)] for all $\tilde{H} \in \mathcal{H}_{good}$, $b \in \{0,1\}$,
\begin{align}
    &\frac{\mathbb{E}_{c_0, c_1 \leftarrow \{0,1\}}\mathbb{E}_{l \leftarrow [q]}[\| \ket{\tilde{x}_b}\bra{\tilde{x}_b} (UO_{(\tilde{x}_0, c_0), (\tilde{x}_1, c_1)})^l \ket{\tilde{H}}\bra{\tilde{H}} \mathsf{Decomp} \ket{\Psi_0} \|^2 ]}{\|\Pi_{\tilde{H}} \ket{\Psi_0} \|^2} \\ \geq & \frac{1 - \sqrt{\Xi(\eps_1,\delta_1,\delta_2)} - \Xi(\eps_1,\delta_1,\delta_2)}{8 q} \cdot \| \ket{y^*}\bra{y^*} \ket{\Psi_{final}}\|^2 \,.
\end{align}
\end{itemize}
\end{cor}

\begin{proof}
This is immediate since one of the four possible assignments of oracle outputs at $\tilde{x}_0,\tilde{x}_1$ agrees with $H(\tilde{x}_0), H(\tilde{x}_1)$.
\end{proof}

\subsubsection{The structure of successful $\CQ d$ strategies}
\label{sec: structure of Wd strategies}
By Lemma \ref{lem: wlog}, recall that it suffices to restrict our analysis to $\mathcal{W}_d$ strategies that are successful at $\subproblem$, in the sense of Lemma \ref{lem: 102}.
Let $A$ be a $\mathcal{W}_d$ strategy for $\subproblem$. %
Since $H_d$ is only queried by classical circuits, we assume that all queries to $H_d$ are recorded by measuring the query register, without any disturbance to the state of the algorithm. We denote by $Q$ the total number of queries to $H_d$ made by the algorithm. %
We moreover assume that all points at which $H_d$ is queried by the algorithm are distinct.

For a security parameter $\lambda \in \mathbb{N}$, we denote by $\mathcal{G}_{\lambda}$ the set of all possible functions $G_b$, by $\mathcal{Y}_{\lambda}$ the co-domain of such functions, and by $\mathcal{H}'_{\lambda}$ the set of all possible functions $h$. We denote by $\mathcal{V}_{\lambda}$ the set of all possible outcomes $adv$ that one can obtain by measuring the registers $\mathsf{work}$, $\mathsf{query}$ of $A$. We omit writing $\lambda$ when this is clear from the context. 

Throughout the section, $poly$ denotes a polynomial, not always the same one. Let $Z^i$ be a random variable for the $i$-th (classical) query to $H_d$. In the following theorem, we denote by $\mathsf{h}_{\mathsf{data}}^{i}$ the random variable representing the set $h^{-1}(H_d(Z^i))$. Moreover, for simplicity and ease of notation, we assume that $\mathsf{h}_{\mathsf{data}}^{i}$ consists of a single element, and we identify the set with that element. The argument is virtually unchanged without this assumption since $h$ is injective with overwhelming probability. %

\anote{let $Q \leq q$ be total number of classical queries to $h$}

\begin{lem}
\label{lem: 109}
Let $\eps: \mathbb{N} \rightarrow \mathbb{R}$. Suppose, for all $\lambda$, $\Pr[y \in \twoone(G_0, G_1): y,d,m \gets A_{\lambda}] > \eps$, and $\Pr[A \textnormal{ wins} | y \in \twoone(G_0, G_1),  y,d,m \gets A_{\lambda}] > \frac12 + \epsilon(\lambda)$. Consider a simulation $\tilde{A}$ of $A$, where calls to $H$ are simulated via a compressed oracle.

For $i \in [Q]$, let $\mathsf{h}_{\mathsf{data}}^{i}$ be the random variable for the $i$-th classical query to $h$ made by $\tilde{A}$. Let $D^i$ be a random variable obtained by measuring the compressed database (in the standard basis) just before the $i$-th query to $h$, and $Y_{out}$ be a random variable for the $y$ output of $\tilde{A}$. The random variables are implicitly functions of $\lambda$, but we omit writing this.

Then, there exists a negligible function $\neglA$ such that, for all $\lambda$, there exists $i^* \in [Q]$, and $\mathcal{W} \subseteq \mathcal{V}_{\lambda} \times \mathcal{Y}_{\lambda} \times \mathcal{G}_{\lambda} \times \mathcal{G}_{\lambda} \times \mathcal{H}'_{\lambda}$ such that:
$$\Pr[(adv, \mathsf{h}_{\mathsf{data}}^{i^*}, G_0, G_1, h) \in \mathcal{W}] \geq poly(\eps)\,.$$
Moreover, for all $(\tilde{adv},\tilde{y}, \tilde{G}_0, \tilde{G}_1, \tilde{h}) \in \mathcal{W}$:
\begin{itemize}
    \item $\tilde{y} \in \twoone(\tilde{G}_0, \tilde{G}_1)\,,$
    \item $\Pr[(x,\tilde{h}(y)) \in D^{i^*} \textnormal{ for some }x| (adv,G_0, G_1, h) = (\tilde{adv},\tilde{G}_0,\tilde{G}_1, \tilde{h}) \,\land \, \mathsf{h}_{\mathsf{data}}^{i^*} = \tilde{y}]  = \neglA(\lambda) \,,$
    \item $\Pr[Y_{out} = \tilde{y} | (adv, G_0, G_1, h) = (\tilde{adv},\tilde{G}_0,\tilde{G}_1, \tilde{h})  \,\land \, \mathsf{h}_{\mathsf{data}}^{i^*} = \tilde{y}] \geq poly(\eps) \,.$
    \item $\Pr[\tilde{A} \textnormal{ wins} | (adv, G_0, G_1, h) = (\tilde{adv},\tilde{G}_0,\tilde{G}_1, \tilde{h})  \,\land \, Y_{out} = \mathsf{h}_{\mathsf{data}}^{i^*} = \tilde{y} ] \geq \frac12 + poly(\eps(\lambda))\,.$
\end{itemize}
\end{lem}

We will prove the following Lemma first. 
\begin{lem}
\label{lem: 110}
Suppose the hypothesis of Lemma \ref{lem: 109} holds.
Then, either there exists a negligible function $\neglA$ such that $$ \Pr[Y_{out} \notin \mathsf{h}_{\mathsf{data}} \land Y_{out} \in \twoone(G_0, G_1) ] =  \neglA(\lambda) \,,$$ 
or there exists a negligible function such that
$$ \Pr[\tilde{A} \textnormal{ wins} \,|\, Y_{out} \notin \mathsf{h}_{\mathsf{data}} \,\land \, Y_{out} \in \twoone(G_0, G_1) ] \leq \frac12 + \neglA(\lambda) \,.$$ 
\end{lem}

\begin{proof}
Suppose there exists a non-negligible function $\textsf{non-negl}$ such that, for all $\lambda$, $$\Pr[Y_{out} \notin \mathsf{h}_{\mathsf{data}} \, \land \, Y_{out} \in \twoone(G_0, G_1) ] = \textsf{non-negl}(\lambda)\,.$$ Suppose for a contradiction that there exists a non-negligible function $\textsf{non-negl}'$ such that, for all $\lambda$, 
$$\Pr[\tilde{A} \textnormal{ wins} | Y_{out} \notin \mathsf{h}_{\mathsf{data}} \land Y \in \twoone(G_0, G_1)]  > \frac12 + \textsf{non-negl'}(\lambda)\,.$$ 
Fix $\lambda$. Let $\mathcal{Q}$ be the set of possible values that $\mathsf{h}_{\mathsf{data}}$ can take. Then, by an averaging argument, there exists $\mathcal{S} \subseteq \mathcal{Y} \times \mathcal{Q} \times  \mathcal{G}_0 \times \mathcal{G}_1$ such that $\Pr[(Y_{out}, \mathsf{h}_{\mathsf{data}}, G_0, G_1) \in \mathcal{S}] = \textsf{non-negl}(\lambda)$, and moreover, for all $(\tilde{y}, \tilde{\mathsf{h}}_{\mathsf{data}}, \tilde{G}_0, \tilde{G}_1) \in \mathcal{S}$,
\begin{itemize}
    \item $\tilde{y} \notin \tilde{\mathsf{h}}_{\mathsf{data}}$, and $\tilde{y} \in \twoone(\tilde{G}_0, \tilde{G}_1)$.
    \item $\Pr[\tilde{A} \textnormal{ wins} \,| \, (Y_{out}, \mathsf{h}_{\mathsf{data}}, G_0, G_1) = (\tilde{y}, \tilde{\mathsf{h}}_{\mathsf{data}}, \tilde{G}_0, \tilde{G}_1)] \geq \frac12 + \textsf{non-negl}(\lambda) \,.$
\end{itemize}
Let $D^{final}$ be a random variable for the the outcome of measuring the final state of the compressed database (in the standard basis). Then, we can apply Lemma \ref{lem: 105} to deduce that 
\begin{align*}
    \Pr[D^{final} \ni (x, h(Y_{out})) \, \land \, Y_{out} \notin \mathsf{h}_{\mathsf{data}} \, \land \, Y_{out} \in \twoone(G_0, G_1) \, \land \, x \in &\{G^{-1}_0(Y_{out}), G^{-1}_1(Y_{out})\}] \\
    &= \textsf{non-negl''}(\lambda) \,,
\end{align*}
for some non-negligible function \textsf{non-negl''}.
This straightforwardly implies that there exist an algorithm that only makes classical queries to $h$, and correctly predicts the value of $h$ at an unqueried point with non-negligible probability. This is a contradiction.
\end{proof}

\begin{proof}[Proof of Lemma \ref{lem: 109}]
Fix $\lambda$. For the rest of the proof, we omit writing any $\lambda$ dependence. By hypothesis, $\Pr[Y_{out} \in \twoone(G_0, G_1)] > poly(\eps)$, and $\Pr[A \textnormal{ wins} | Y_{out} \in \twoone(G_0, G_1)] > \frac12 + \eps$. Then, using Lemma \ref{lem: 110}, it is straightforward to see that there exists a negligible function $\neglA$ such that:
\begin{itemize}
    \item[(i)] $\Pr[Y_{out} \in \mathsf{h}_{\mathsf{data}} | Y_{out} \in \twoone(G_0, G_1)] > \min(poly(\eps),1 -\neglA)$, and
    \item[(ii)] $\Pr[\tilde{A} \textnormal{ wins} | Y_{out} \in \mathsf{h}_{\mathsf{data}} \land Y_{out} \in \twoone(G_0, G_1)] \geq \frac12 + poly(\eps)$
\end{itemize}
In what follows, for ease of notation, we denote the event ``$Y_{out} \in \twoone(G_0, G_1)$'' as $E$.
We can equivalently rewrite (ii) as 
\begin{equation}
    \sum_{i \in [Q]} \frac{\Pr[Y_{out} = \mathsf{h}_{\mathsf{data}}^i \,| \,E]}{\Pr[Y_{out} \in \mathsf{h}_{\mathsf{data}} \,|\, E]} \cdot \frac{\Pr[\tilde{A} \textnormal{ wins} \, \land \, Y_{out} = \mathsf{h}_{\mathsf{data}}^i\,|\,E ]}{ \Pr[Y_{out} = \mathsf{h}_{\mathsf{data}}^i \,| \,E]} \geq \frac12 + poly(\eps)
\end{equation}
By an averaging argument, there exists $i^* \in [Q]$ such that 
\begin{itemize}
    \item $\Pr[Y_{out} = \mathsf{h}_{\mathsf{data}}^{i^*} \,| \,E]  \geq  \Pr[Y_{out} \in \mathsf{h}_{\mathsf{data}} \,|\, E] \cdot \frac{poly(\eps)}{Q} \geq \frac{poly(\eps)}{Q} $, and 
    \item  $\Pr[\tilde{A} \textnormal{ wins} \,|\, Y_{out} = \mathsf{h}_{\mathsf{data}}^{i^*} \,\land \, E ] > \frac12 + poly(\eps)$
\end{itemize}

Since $\Pr[E] \geq \eps$ by hypothesis, the former implies $\Pr[Y_{out} = \mathsf{h}_{\mathsf{data}}^{i^*} \,\land \,E]  \geq \frac{poly(\eps)}{Q} \cdot \Pr[E] \geq \frac{poly(\eps)}{Q} $.

By another averaging argument, there exists $\mathcal{W} \subseteq \mathcal{V} \times \mathcal{Y} \times \mathcal{G} \times \mathcal{G} \times \mathcal{H}'$ such that:
$$\Pr[Y_{out}= \mathsf{h}_{\mathsf{data}}^{i^*} \,\land \, (adv, \mathsf{h}_{\mathsf{data}}^{i^*}, G_0, G_1, h) \in \mathcal{W}] \geq \frac{poly(\eps)}{Q}\,.$$
Moreover, for all $(\tilde{adv}, \tilde{y}, \tilde{G}_0, \tilde{G}_1, \tilde{h}) \in \mathcal{W}$:
\begin{itemize}
    \item $\tilde{y} \in \twoone(\tilde{G}_0, \tilde{G}_1)\,,$
       \item $\Pr[\tilde{A} \textnormal{ wins} | Y_{out} = \mathsf{h}_{\mathsf{data}^{i^*}} = \tilde{y}  \,\land \, (adv, G_0, G_1, h) = (\tilde{adv},\tilde{G}_0,\tilde{G}_1, \tilde{h})  ] \geq \frac12 + poly(\eps)\,.$
\end{itemize}
Notice that, trivially, for any $\tilde{y}$, the distribution of $h(\tilde{y})$ is uniform, conditioned on the values of $H_d$ at any subset of points that does not contain $H_{d-1} \circ \cdots \circ H_0 \tilde{y}$. Thus, since we assumed without loss of generality that $A$ never queries $H_d$ at the same point twice, this clearly implies that, for all $(\tilde{adv}, \tilde{y}, \tilde{G}_0, \tilde{G}_1, \tilde{h})$,
\begin{itemize}
    \item $\Pr[(x, h(\tilde{y})) \in D^{i^*} \textnormal{ for some }x| (G_0, G_1) = (\tilde{G}_0,\tilde{G}_1) \,\land \, \mathsf{h}_{\mathsf{data}}^{i^*} = \tilde{y}]  = \neglA(\lambda) \,,$
\end{itemize}
This implies that there exists $\mathcal{W}' \subseteq \mathcal{W}$ such that:
$$\Pr[Y_{out}= \mathsf{h}_{\mathsf{data}}^{i^*} \,\land \, (adv, \mathsf{h}_{\mathsf{data}}^{i^*}, G_0, G_1, h) \in \mathcal{W}] \geq \frac{poly(\eps)}{Q}\,.$$
Moreover, for all $(\tilde{y}, \tilde{G}_0, \tilde{G}_1, \tilde{h}) \in \mathcal{W}'$:
\begin{itemize}
    \item $\tilde{y} \in \twoone(\tilde{G}_0, \tilde{G}_1)\,,$
       \item $\Pr[\tilde{A} \textnormal{ wins} | Y_{out} = \mathsf{h}_{\mathsf{data}^{i^*}} = \tilde{y}  \,\land \, (adv, G_0, G_1, h) = (\tilde{adv},\tilde{G}_0,\tilde{G}_1, \tilde{h})  ] \geq \frac12 + poly(\eps)\,,$
\end{itemize}
and, for all $(\tilde{adv},\tilde{y}, \tilde{G}_0, \tilde{G}_1, h) \in \mathcal{W}'$,
\begin{itemize}
    \item $\Pr[(x, \tilde{h}(\tilde{y})) \in D^{i^*} \textnormal{ for some }x| (adv, G_0, G_1, h) = (\tilde{adv},\tilde{G}_0,\tilde{G}_1, \tilde{h}) \,\land \, \mathsf{h}_{\mathsf{data}}^{i^*} = \tilde{y}]  = \neglA(\lambda) \,,$
\end{itemize}

By one final averaging argument, there exists $\mathcal{W}'' \subseteq \mathcal{W}'$ such that 
$$\Pr[(adv, \mathsf{h}_{\mathsf{data}}^{i^*}, G_0, G_1, h) \in \mathcal{W}''] \geq \frac{poly(\eps)}{Q}\,,$$ and for all $(\tilde{adv},\tilde{y}, \tilde{G}_0, \tilde{G}_1, \tilde{h}) \in \mathcal{W}''$,
\begin{itemize}
    \item $\Pr[Y_{out}= \mathsf{h}_{\mathsf{data}}^{i^*}| (adv, \mathsf{h}_{\mathsf{data}}^{i^*}, G_0, G_1, h) = (\tilde{adv},\tilde{y}, \tilde{G}_0, \tilde{G}_1, \tilde{h})] \geq \frac{poly(\eps)}{Q}\,,$
    \item $\tilde{y} \in \twoone(\tilde{G}_0, \tilde{G}_1)\,,$
     \item $\Pr[\tilde{A} \textnormal{ wins} | Y_{out} = \mathsf{h}_{\mathsf{data}}^{i^*} = \tilde{y}  \,\land \, (adv, G_0, G_1, h) = (\tilde{adv}, \tilde{G}_0,\tilde{G}_1, \tilde{h})  ] \geq \frac12 + poly(\eps)\,.$
     \item $\Pr[(x, h(\tilde{y})) \in D^{i^*} \textnormal{ for some }x| (adv, G_0, G_1, h) = (\tilde{adv}, \tilde{G}_0,\tilde{G}_1, \tilde{h}) \,\land \, \mathsf{h}_{\mathsf{data}}^{i^*} = \tilde{y}]  = \neglA(\lambda) \,.$
\end{itemize}
This concludes the proof of Lemma \ref{lem: 109}.

\end{proof}

\subsubsection{Putting things together}
In this subsection, we complete the proof of Lemma \ref{lem: 103}.
Let $\eps: \mathbb{N} \rightarrow [0,1]$. Suppose, for all $\lambda$, $\Pr[y \in \twoone(G_0, G_1): y,d,m \gets A_{\lambda}] \geq \eps(\lambda)$, and $\Pr[A \textnormal{ wins} | y \in \twoone(G_0, G_1),  y,d,m \gets A_{\lambda}] \geq \frac12 + \epsilon(\lambda)$. Let $q$ be the total number of queries made by $A$.

We will show that Algorithm \ref{alg: 1} extracts a collision with probability at least $poly(\eps, 1/q)$.

We can apply Lemma \ref{lem: 109}. Using the notation of Lemma \ref{lem: 109}, we have that there exists a negligible function $\neglA$ such that, for all $\lambda$, there exists $i^* \in [Q]$, and $\mathcal{W} \subseteq \mathcal{V} \times \mathcal{Y} \times \mathcal{G} \times \mathcal{G} \times \mathcal{H}'$ such that:
\begin{equation}
    \Pr[(adv, \mathsf{h}_{\mathsf{data}}^{i^*}, G_0, G_1, h) \in \mathcal{W}] \geq poly(\eps)\,. \label{eq: 300}
\end{equation}
Moreover, for all $(\tilde{adv}, \tilde{y}, \tilde{G}_0, \tilde{G}_1, \tilde{h}) \in \mathcal{W}$:
\begin{itemize}
    \item[(i)] $\tilde{y} \in \twoone(\tilde{G}_0, \tilde{G}_1)\,,$
    \item[(ii)] $\Pr[(x, h(y)) \in D^{i^*} \textnormal{ for some }x| (adv, G_0, G_1, h) = (\tilde{adv}, \tilde{G}_0,\tilde{G}_1, \tilde{h}) \,\land \, \mathsf{h}_{\mathsf{data}}^{i^*} = \tilde{y}]  = \neglA(\lambda) \,,$
    \item[(iii)] $\Pr[Y_{out} = \tilde{y} | (adv , G_0, G_1, h) = (\tilde{adv}, \tilde{G}_0,\tilde{G}_1, \tilde{h})  \,\land \, \mathsf{h}_{\mathsf{data}^{i^*}} = \tilde{y}] \geq poly(\eps) \,.$
    \item[(iv)] $\Pr[\tilde{A} \textnormal{ wins} | (adv, G_0, G_1, h) = (\tilde{adv},\tilde{G}_0,\tilde{G}_1, \tilde{h})  \,\land \, Y_{out} = \mathsf{h}_{\mathsf{data}^{i^*}} = \tilde{y} ] \geq \frac12 + poly(\eps)\,.$
\end{itemize}

Notice then that, at step (ii) of Algorithm \ref{alg: 1},
$$\Pr[i = i^* \, \land \, (adv, \mathsf{h}_{\mathsf{data}}^{i^*}, G_0, G_1, h) \in \mathcal{W}] \geq \frac{poly(\eps)}{Q} \,.$$

Fix $(\tilde{adv},\tilde{y}, \tilde{G_0}, \tilde{G_1}, \tilde{h}) \in \mathcal{W}$. Let $\ket{\Psi_0}$ be the state of the compressed oracle simulation after step (ii), conditioned on $i = i^*$ and $\tilde{adv},\tilde{y}, \tilde{G_0}, \tilde{G_1}, \tilde{h}$. Let $x_0 = G_0^{-1}(\tilde{y})$, and $x_1 = G_1^{-1}(\tilde{y})$.

Let $\ket{\Psi_{final}}$ be the final state of the compressed oracle simulation (continuing from $\ket{\Psi_0}$), i.e. $\ket{\Psi_{final}} = (U'    O^{\comp})^{q} \ket{\Psi_0}$, where $\tilde{q}$ denotes the number of remaining queries to $H$, and we are absorbing in $U'$ all queries to $G_0, G_1, h$ as well as the unitaries $U_C$ and $U_Q$. 

Condition (ii) implies that $$\delta_1 := \| \Pi_{D \cap \{(x, h(\tilde{y}): x \in \mathcal{X}\} \neq \emptyset} \ket{\Psi_0} \|^2 = \neglA(\lambda)\,.$$
Condition (iii) implies that 
\begin{equation}
    \|\ket{\tilde{y}}\bra{\tilde{y}} \ket{\Psi_{final}}\|^2 \geq poly(\eps)\,, \label{eq: poly weight}
\end{equation}
and condition (iv) implies that 
$$\eps_1 := \Pr[\tilde{A} \textnormal{ wins }| (adv, G_0, G_1,h) = (\tilde{adv}, \tilde{G}_0,\tilde{G}_1, \tilde{h})  \,\land \, Y_{out} = \mathsf{h}_{\mathsf{data}^{i^*}} = \tilde{y} ] - \frac12 \geq poly(\eps) \,.$$

Finally notice that there exists a negligible function $\neglA'$, such that except with $\neglA'$ probability over $(\tilde{adv}, \tilde{y}, \tilde{G_0}, \tilde{G_1}, \tilde{h}) \in \mathcal{W}$, it must be that $\delta_2 := \| \Pi_{D \ni (x_0, h(\tilde{y})),(x_1, h(\tilde{y}))} \ket{\Psi_{final}}\|^2 \leq \neglA'(\lambda)$. Otherwise, the algorithm that simply runs a compressed oracle simulation of $A$ and measures the compressed database at the end, recovers a collision with non-negligible probability. We restrict to this ``good'' subset of $\mathcal{W}$ from here on.

We are now ready to apply Corollary \ref{cor: 102} with $\eps_1, \delta_1,$ and $\delta_2$ as above. We deduce that there exists $\mathcal{H}_{good} \subseteq \mathcal{F}(\{0,1\}^n \setminus \{\tilde{x}_0, \tilde{x}_1\}, \{0,1\})$ such that
\begin{itemize}
    \item[(a)] $\,  \sum_{\tilde{H} \in \mathcal{H}_{good}}  \| \Pi_{\tilde{H}} \ket{\Psi_0} \|^2 \geq  \big( 1 - \sqrt{\Xi(\eps_1,\delta_1,\delta_2)}\big) \cdot  \| \ket{\tilde{y}}\bra{\tilde{y}} \ket{\Psi_{final}}\|^2 \geq poly(\eps) - \neglA''(\lambda) \,,$
    \item[(b)] for all $\tilde{H} \in \mathcal{H}_{good}$, $b \in \{0,1\}$,
\begin{align*}
   & \frac{\mathbb{E}_{c_0, c_1 \leftarrow \{0,1\}}\mathbb{E}_{l \leftarrow [q]}[\| \ket{\tilde{x}_b}\bra{\tilde{x}_b} (U'O^H_{(x_0, c_0), (x_1, c_1)})^l \ket{\tilde{H}}\bra{\tilde{H}} \mathsf{Decomp} \ket{\Psi_0} \|^2 ]}{ \| \Pi_{\tilde{H}} \ket{\Psi_0} \|^2} \\ &\geq \frac{1 - \sqrt{\Xi(\eps_1,\delta_1,\delta_2)} - \Xi(\eps_1,\delta_1,\delta_2)}{8 \tilde{q}} \cdot  \| \ket{\tilde{y}}\bra{\tilde{y}} \ket{\Psi_{final}}\|^2 \geq \frac{poly(\eps)}{\tilde{q}} - \neglA''(\lambda)  \,.
\end{align*}
\end{itemize}
where $\neglA''$ is a non-negligible function. To obtain the final inequalities in (a) and (b) we used the bounds on $\eps_1, \delta_1,$ and $\delta_2$ and Equation (\ref{eq: poly weight}).

Now, notice that, at step (ii) of Algorithm \ref{alg: 1}, conditioned on $(\tilde{adv}, \tilde{y}, \tilde{G_0}, \tilde{G_1}, \tilde{h})$, the state $\ket{\Psi_0}$ takes the form $\ket{\Psi_0} = \ket{ adv}_{\mathsf{work},\mathsf{query}} \otimes \ket{\Phi}_\mathsf{O}$, where $\ket{\Phi}_\mathsf{O}$ is some state on the compressed database register for $H$ that can depend on $\tilde{adv}, \tilde{y}, \tilde{G_0}, \tilde{G_1}, \tilde{h}$. Now, let $\mathcal{H}_{good} \subseteq \mathcal{F}(\{0,1\}^n \setminus \{\tilde{x}_0, \tilde{x}_1\}, \{0,1\})$ be the set that is guaranteed to exist from the argument above. In the following calculation, we abbreviate $\mathcal{F}(\{0,1\}^n \setminus \{\tilde{x}_0, \tilde{x}_1\}, \{0,1\})$ as $\mathcal{F}_{\tilde{x}_0, \tilde{x}_1}$. 

Then, for $b,b' \in \{0,1\}$,
\begin{align*}
    &\Pr[\textnormal{Algorithm } \ref{alg: 1} \textnormal{ outputs $x_b, x_{b'}$} \,|\, i = i^* \,\land \, (adv, \mathsf{h}_{\mathsf{data}}^{i^*}, G_0, G_1, h) = (\tilde{adv}, \tilde{y}, \tilde{G}_0, \tilde{G}_1, \tilde{h})] \\
    &= \sum_{\tilde{H}\in \mathcal{F}_{x_0, x_1}}  \Bigg( \| \Pi_{\tilde{H}} \ket{adv} \otimes \ket{\Phi} \|^2 \\ &\cdot  \mathbb{E}_{\substack{c_0, c_1  \leftarrow \{0,1\} \\c_0, c_1' \leftarrow \{0,1\} \\ j, j' \leftarrow [q]}} \Big[\| \big(\ket{\tilde{x}_b}\bra{\tilde{x}_b} \otimes \ket{\tilde{x}_{b'}}\bra{\tilde{x}_{b'}}\big)  \big( (U'O^{\tilde{H}}_{(x_0, c_0), (x_1, c_1)})^{j} \otimes  (U'O^{\tilde{H}}_{(x_0, c_0'), (x_1, c_1')})^{j'}\big) \ket{adv}_{\mathsf{work, query}} \otimes  \ket{adv}_{\mathsf{work', query'}} \|^2 \Big] \Bigg) \\
    &= \sum_{\tilde{H}\in \mathcal{F}_{x_0, x_1}} \Bigg( \| \Pi_{\tilde{H}} \ket{adv} \otimes \ket{\Phi} \|^2 \cdot  \mathbb{E}_{c_0, c_1 \leftarrow \{0,1\}}\mathbb{E}_{j, j' \leftarrow [q]}\Big[\| \ket{\tilde{x}_b}\bra{\tilde{x}_b} (U'O^{\tilde{H}}_{(x_0, c_0), (x_1, c_1)})^{j} \ket{adv} \|^2 \Big] \\
  &\quad \quad \quad \quad \quad \quad \quad \quad \quad \quad \quad \quad \quad \quad \quad\quad\quad\cdot  \mathbb{E}_{c_0', c_1' \leftarrow \{0,1\}}\mathbb{E}_{j' \leftarrow [q]}\Big[\| \ket{\tilde{x}_{b'}}\bra{\tilde{x}_{b'}} (U'O^{\tilde{H}}_{(x_0, c_0), (x_1, c_1)})^{j'} \ket{adv}\|^2 \Big] \Bigg) \\
    &= \sum_{\tilde{H}\in \mathcal{F}_{x_0, x_1}}  \Bigg( \| \Pi_{\tilde{H}} \ket{\Psi_0} \|^2 \cdot  \frac{\mathbb{E}_{c_0, c_1 \leftarrow \{0,1\}}\mathbb{E}_{j \leftarrow [q]}[\| \ket{\tilde{x}_b}\bra{\tilde{x}_b} (U'O^{\tilde{H}}_{(x_0, c_0), (x_1, c_1)})^{j} \ket{\tilde{H}}\bra{\tilde{H}} \mathsf{Decomp} \ket{\Psi_0} \|^2 ]}{\|\Pi_{\tilde{H}}\ket{\Psi_0}  \|^2} \\
  &\quad \quad \quad \quad \quad \quad \quad \quad \quad \quad \quad \quad \quad \quad \cdot  \frac{\mathbb{E}_{c_0', c_1' \leftarrow \{0,1\}}\mathbb{E}_{j' \leftarrow [q]}[\| \ket{\tilde{x}_{b'}}\bra{\tilde{x}_{b'}} (U'O^{\tilde{H}}_{(x_0, c_0), (x_1, c_1)})^{j'} \ket{\tilde{H}}\bra{\tilde{H}} \mathsf{Decomp} \ket{\Psi_0} \|^2 ]}{\|\Pi_{\tilde{H}} \ket{\Psi_0}\|^2} \Bigg) \\
  &\geq \sum_{\tilde{H} \in \mathcal{H}_{good}} \| \Pi_{\tilde{H}} \ket{\Psi_0} \|^2 \cdot \left(\frac{poly(\eps)}{\tilde{q}}\right)^2 - \neglA'' \quad\quad \textnormal{  using (b)}\\
  &\geq \frac{poly(\eps)}{\tilde{q}^2} - \neglA'' \quad \quad \textnormal{ using (a)}\,,
\end{align*}
where the first equality implicitly uses the equivalence between compressed an uncompressed simulations.
All in all, we have
\begin{align*}
    &\Pr[\textnormal{Algorithm } \ref{alg: 1} \textnormal{ outputs $x_b, x_{b'}$}] \\ 
    &\geq  \Pr[\textnormal{Algorithm } \ref{alg: 1} \textnormal{ outputs $x_b, x_{b'}$} \,|\, i = i^* \, \land \, (adv, \mathsf{h}_{\mathsf{data}}^{i^*}, G_0, G_1, h) \in \mathcal{W}] \cdot \Pr[ i = i^* \,\land \, (adv, \mathsf{h}_{\mathsf{data}}^{i^*}, G_0, G_1, h) \in \mathcal{W}] \\
    &\geq \frac{poly(\eps)}{Q \cdot \tilde{q}^2} - \neglA'' \,,
\end{align*}
where we used Equation (\ref{eq: 300}).

When $b \neq b'$, we get that Algorithm \ref{alg: 1} outputs a collision with probability $\frac{poly(\eps)}{Q \cdot \tilde{q}^2} - \neglA'' \geq \frac{poly(\eps)}{q^3} - \neglA''$, where $q$ is the total number of queries made by $A$ to $G_0, G_1, H,$ and $H_d$.

\pagebreak{}

\printbibliography

\pagebreak{}

\appendix
\part*{Appendix}
\section{The O2H lemma\label{sec:O2Happendix}}
The following proofs of the O2H lemma (due originally to~\cite{ambainis_quantum_2018,chia_need_2020-1}) as used in our setting are taken almost verbatim from~\cite{arora_oracle_2022}.

\subsection{Proof of \Lemref{O2H}}

\branchcolor{blue}{\begin{proof}[Proof of \Lemref{O2H}]
  We begin by assuming that $\calL $ and $S$ are fixed (and so
  is $\calG$). In that case, we can assume $\rho$ is pure. If
  not, we can purify it and absorb it in the work register. (The general
  case should follow from concavity). From \Remref{psiphi0phi1}, we
  have
  \begin{align*}
      \left|\psi_{L}\right\rangle                                        & :=\calL U\left|\psi\right\rangle _{Q'}\overset{\prettyref{rem:psiphi0phi1}}{=}\left|\phi_{0}\right\rangle _{Q'}+\left|\phi_{1}\right\rangle _{Q'}. \\
      \calL U_{S}U\left|\psi\right\rangle _{Q'}\left|0\right\rangle _{B} & =\left|\phi_{0}\right\rangle _{Q'}\left|0\right\rangle _{B}+\left|\phi_{1}\right\rangle _{Q'}\left|1\right\rangle _{B}
  \end{align*}
  where $Q'$ is a shorthand for $QRW$. Similarly let
  \begin{align*}
      \left|\psi_{G}\right\rangle & :=\calG U\left|\psi\right\rangle _{Q'}=\left|\phi_{0}\right\rangle _{Q'}+\left|\phi_{1}^{\perp}\right\rangle _{Q'}
  \end{align*}
  where note that
  \begin{equation}
      \left\langle \phi_{1}|\phi_{1}^{\perp}\right\rangle _{QRW}=0\label{eq:phisPerp}
  \end{equation}
  because $\left|\phi_{1}\right\rangle $ and $\left|\phi_{1}^{\perp}\right\rangle $
  are the states where the queries were made on $S$, and on $S$ $\calG$
  responds with $\perp$ while $\calL $ does not. Further, we
  analogously have
  \[
      \calG U_{S}U\left|\psi\right\rangle _{Q'}\left|0\right\rangle _{B}=\left|\phi_{0}\right\rangle _{Q'}\left|0\right\rangle _{B}+\left|\phi_{1}^{\perp}\right\rangle _{Q'}\left|1\right\rangle _{B}.
  \]
  We show that the difference between $\left|\psi_{L}\right\rangle $
  and $\left|\psi_{G}\right\rangle $ is bounded by $P_{{\rm find}}(\calL ,S):=\Pr[{\rm find}:U^{\calL \backslash S},\rho]$,
  which in turn can be used to bound the quantity in the statement of
  the lemma.
  \begin{align*}
      \left\Vert \left|\psi_{L}\right\rangle -\left|\psi_{G}\right\rangle \right\Vert ^{2} & =\left\Vert \left|\phi_{1}\right\rangle -\left|\phi_{1}^{\perp}\right\rangle \right\Vert ^{2}                                                                                                                                                                                                                                                           \\
                                                                                          & \overset{\prettyref{eq:phisPerp}}{=}\left\Vert \left|\phi_{1}\right\rangle \right\Vert ^{2}+\left\Vert \left|\phi_{1}^{\perp}\right\rangle \right\Vert ^{2}                                                                                                                                                                                             \\
                                                                                          & =2\left\Vert \left|\phi_{1}\right\rangle \right\Vert ^{2}                                                                                                   & \because\left\Vert \left|\phi_{1}\right\rangle \right\Vert ^{2}=\left\Vert \left|\phi_{1}^{\perp}\right\rangle \right\Vert ^{2}=1-\left\Vert \left|\phi_{0}\right\rangle \right\Vert ^{2} \\
                                                                                          & =2P_{{\rm find}}(\calL ,S).
  \end{align*}
  If $\calL $ and $S$ are random variables drawn from a (possibly)
  joint distribution $\Pr(\calL ,S)$, the analysis can be generalised
  as follows. Let
  \begin{align*}
      \rho_{L} & :=\sum_{\calL ,S}\Pr(\calL ,S)\left|\psi_{L}\right\rangle \left\langle \psi_{L}\right| \\
      \rho_{G} & :=\sum_{\calL ,S}\Pr(\calL ,S)\left|\psi_{G}\right\rangle \left\langle \psi_{G}\right|
  \end{align*}
  where $\left|\psi_{G}\right\rangle $ is fixed by $\calL $ and
  $S$ because $G$ itself is fixed once $\calL $ and $S$ is
  fixed (by assumption). One can then use  monotonicity of fidelity
  to obtain
  \begin{align*}
      F(\rho_{L},\rho_{G}) & \ge\sum_{L,S}\Pr(\calL ,S)F(\left|\psi_{L}\right\rangle ,\left|\psi_{G}\right\rangle )                                                                                                                                                                                      \\
                          & \ge1-\frac{1}{2}.\sum_{L,S}\Pr(\calL ,S)\left\Vert \left|\psi_{L}\right\rangle -\left|\psi_{G}\right\rangle \right\Vert ^{2} & \because1-\frac{1}{2}F(\left|a\right\rangle ,\left|b\right\rangle )\ge\left\Vert \left|a\right\rangle -\left|b\right\rangle \right\Vert ^{2} \\
                          & \ge1-\cancel{\frac{1}{2}}\sum_{L,S}\Pr(\calL ,S)\cancel{2}P_{{\rm find}}(\calL ,S)                                                                                                                                                                                          \\
                          & =1-P_{{\rm find}}
  \end{align*}
  where $P_{{\rm find}}$ is the expectation of $P_{{\rm find}}(\calL ,S)$
  over $\calL $ and $S$. It is known that the trace distance
  bounds the LHS of the Lemma and the trace distance itself is bounded
  by $\sqrt{2-2F}$.

\end{proof}
} 

\subsection{Proof of \Lemref{boundPfind}}
\branchcolor{blue}{\begin{proof}[Proof of \Lemref{boundPfind}]
  We resume the use of boldface for the query and
response registers as they do play an active role in the discussion.
  Let us begin with the case where the oracle is applied only once,
  i.e. $\boldsymbol{Q}$ is a single query register $Q$. Since the
  $RW$ registers don't play any significant role, we denote it by $L$.
  Let
  \begin{align*}
      U\left|\psi\right\rangle               & =\sum_{q,l}\psi(q,l)\left|q,l,0\right\rangle _{QLB}                                                                                                                                                        \\
      \implies U_{S}U\left|\psi\right\rangle & =\sum_{q\notin S}\left(\sum_{r,l}\psi(q,l)\left|q,l\right\rangle _{QL}\right)\left|0\right\rangle _{B}+\sum_{q\in S}\left(\sum_{r,l}\psi(q,l)\left|q,l\right\rangle _{QL}\right)\left|1\right\rangle _{B}.
  \end{align*}
  Since $\calL $ leaves registers $QB$ unchanged, %
  \begin{align*}
      \tr[\mathbb{I}_{QL}\otimes\left|1\right\rangle \left\langle 1\right|_{B}\left(\calL \circ U_{S}\circ U\circ\left|\psi\right\rangle \left\langle \psi\right|\right)] & =\tr[\mathbb{I}_{QL}\otimes\left|1\right\rangle \left\langle 1\right|_{B}\left(U_{S}\circ U\circ\left|\psi\right\rangle \left\langle \psi\right|\right)] \\
                                                                                                                                                                          & =\sum_{q}\left|\psi(q)\right|^{2}\chi_{S}(q)
  \end{align*}
  where $\psi(q)=\sum_{l}\psi(q,l)$ and $\chi_{S}$ is the characteristic
  function for $S$, i.e.
  \[
      \chi_{S}(q)=\begin{cases}
          1 & q\in S     \\
          0 & q\notin S.
      \end{cases}
  \]
  We are yet to average over the random variable $S$. Clearly, $\mathbb{E}(\chi_{S}(q))=\Pr[q\in S]\le p$,
  yielding
  \[
      \Pr[{\rm find}:U^{\calL \backslash S},\rho]\le p.
  \]
  In the general case, everything goes through unchanged except the
  string $q$ is now a set of strings $\boldsymbol{q}$ and
  \[
      \chi_{S}(\boldsymbol{q})=\begin{cases}
          1 & \boldsymbol{q}\cap S\neq\emptyset \\
          0 & \boldsymbol{q}\cap S=\emptyset.
      \end{cases}
  \]
  Consequently, one evaluates $\mathbb{E}(\chi_{S}(\boldsymbol{q}))=\Pr[\boldsymbol{q}\cap S\neq\emptyset]\le\left|\boldsymbol{q}\right|\cdot p=\bar{q}\cdot p$,
  by the union bound, yielding
  \[
      \Pr[{\rm find}:U^{\calL \backslash S},\rho]\le\bar{q}\cdot p.
  \]
\end{proof}
}

\section{Misc calculations}

\subsection{\label{subsec:sillySteps_in_claim_x_in_S_CQ_d}Proof of \texorpdfstring{\Claimref{x_in_S_QC_d}}{Shadow Oracle Properties}
  | Deferred steps}

\subsubsection{Proof of \texorpdfstring{\Eqref{x_in_S_ii_2_delta}}{Shadow Oracle Equation}}

First, note that $2^{\delta}\frac{\left|S_{i-1}\right|}{|S_{i-1,i}|}\le2^{\delta}\frac{1}{|\Sigma|}$
because $|S_{i-1}|=\Sigma$ and $1/|S_{i-1,i}|\le1/|\Sigma|^{2}$
by construction (see \Algref{S_ij_given_beta} or \Algref{setMatrix}
for simplicity). \\
Second, observe that
\begin{align*}
  \frac{|S_{ii}|-|S_{i}|}{|S_{i-1,i}|-|S_{i}|} & \le\frac{|S_{ii}|}{|S_{i-1,i}|-|S_{i}|}                                        \\
                                               & =\frac{\frac{|S_{ii}|}{|S_{i-1,i}|}}{1-\frac{|S_{i}|}{|S_{i-1,i}|}}            \\
                                               & =\frac{1}{|\Sigma|}\left(1-\frac{|S_{i}|}{|S_{i-1,i}|}\right)^{-1}             \\
                                               & \le\frac{1}{|\Sigma|}\left(1+\frac{|S_{i}|}{|S_{i-1,i}|}+\epsilon\right)       \\
                                               & \le\frac{1}{|\Sigma|}\left(1+\frac{1}{\left|\Sigma\right|^{2}}+\epsilon\right)
\end{align*}
where $\epsilon$ is a small fixed constant $\epsilon$ and we used
the fact that the inequality $(1-x)^{-1}\le1+x+\epsilon$ holds for
a small enough $0\le x$. Combining these, and recalling $|\Sigma|=2^{{\lambda}^{\Theta(1)}}$
and $n=\Theta(\lambda)$, one obtains \Eqref{x_in_S_ii_2_delta}.

\subsection{Proof of Claim \ref{claim: 1}}
We prove the following.
\claimone*
\begin{proof}
Let $p := \Pr[E \geq \gamma \cdot F]$. Then, 
\begin{equation}
\label{eq: 184}
    \mathbb{E}[E] = p \cdot \mathbb{E}[E \,|\, E \geq \gamma \cdot F] + (1-p) \cdot \mathbb{E}[E \,|\, E < \gamma \cdot F]
\end{equation}
Notice that $\mathbb{E}[E \,|\, E\geq \gamma F] \geq \gamma \cdot \mathbb{E}[F \,|\, E\geq \gamma F ]$. Plugging this into (\ref{eq: 184}), we get
\begin{align}
    \mathbb{E}[E] &= p \cdot \gamma \cdot \mathbb{E}[F \,|\, E \geq \gamma \cdot F] + (1-p) \cdot \mathbb{E}[E \,|\, E < \gamma \cdot F] \nonumber\\
    &\geq  p \cdot \gamma \cdot \mathbb{E}[F \,|\, E \geq \gamma \cdot F] \,. \label{eq: 185}
\end{align}
Now, notice that 
\begin{equation}
     \mathbb{E}[F] = p \cdot \mathbb{E}[F \,|\, E \geq \gamma \cdot F] + (1-p) \cdot \mathbb{E}[F \,|\, E < \gamma \cdot F]
\end{equation}
Rearranging the latter, and plugging this into (\ref{eq: 185}) gives
\begin{align}
      \mathbb{E}[E] &\geq \gamma \cdot \left(\mathbb{E}[F] - (1-p) \cdot \mathbb{E}[F \,|\, E < \gamma \cdot F]   \right) \, \nonumber\\
      &\geq \gamma \cdot \left(\mathbb{E}[F] - (1-p) \right) \nonumber
\end{align}
Solving for $p$ gives the desired inequality. 
\end{proof}

\section{Sampling argument for Permutations\label{sec:Appendix_Sampling-argument-for-Permutations}}
\ifthenelse{\boolean{keepOldProofs}}{To keep the proof self-contained, we include the proof of the sampling argument for permutations, taken almost verbatim\footnote{We fixed some notation.} from \cite{arora_oracle_2022}. The key idea has been adapted from~\cite{EC:CorettiDGS18} and slightly generalised.
\subsection{Sampling argument for Uniformly Distributed Permutations\label{subsec:Sampling-argument-for}}

\subsubsection{Convex Combination of Random Variables}

We first make the notion of ``convex combination of random variables''
precise. Consider a function $f$ which acts on a random permutation,
say $t$, to produce an output, i.e. $f(t)=r$ where $r$ is an element
in the range of $f$.\footnote{The function will later be interpreted as an algorithm and the random
  permutation accessed via an oracle.} This range can be arbitrary. We say a \emph{convex combination} $\sum_{i}p_{i}t_{i}$
of random variables $t_{i}$ is \emph{equivalent} to $t$ if for all
functions $f$, and all outputs $s$ in its range, $\sum_{i}p_{i}\Pr[f(t_{i})=s]=\Pr[f(t)=s]$.
This relation is denoted by $\sum_{i}p_{i}t_{i}\equiv t$.

\subsubsection{The ``parts'' notation}

While permutations are readily defined as an ordered set of distinct
elements, it would nonetheless be useful to introduce what we call
the ``parts'' notation which allows one to specify parts of the
permutation.
\begin{notation}
  \label{nota:permPaths}Consider a permutation $t$ over $N$ elements,
  labelled $\{0,1\dots N-1\}$.
  \begin{itemize}
    \item \emph{Parts:} Let $S=\{(x_{i},y_{i})\}_{i=1}^{M}$ denote the mapping
          of $M\le N$ elements under some permutation, i.e. there is some permutation
          $t$, such that $t(x_{i})=y_{i}$. Call any such set $S$ a ``part''
          and its constituents ``paths''.
          \begin{itemize}
            \item Denote by $\Omega_{{\rm parts}}(N)$ the set of all such ``parts''.
            \item Call two parts $S=\{(x_{i},y_{i})\}_{i}$ and $S'=\{(x'_{i'},y'_{i'})\}_{i'}$
                  \emph{distinct} if for all $i,i'$ (a) $x_{i}\neq x_{i'}$, and (b)
                  there is a permutation $t$ such that $t(x_{i})=y_{i}$ and $t(x_{i'})=y_{i'}$.
            \item Denote by\footnote{We use $\Omega_{{\rm parts}}$ because the symbol $\Omega$ is often
                    used for the sample space; for $\parts$, $\Omega_{\parts}$ plays
                    an analogous role.} $\Omega_{{\rm parts}}(N,S)$ the set of all parts $S'\in\Omega_{{\rm parts}}(N)$
                  such that $S'$ is distinct from $S$.
          \end{itemize}
    \item \emph{Parts in $t$:} The probability that $t$ maps the elements
          as described by $S$ may be expressed as $\Pr[\land_{i=1}^{M}(t(x_{i})=y_{i})]=\Pr[S\subseteq{\rm \paths}(t)]$
          where ${\rm \paths}(t):=\{(x,t(x))\}_{x=0}^{N-1}$.
    \item \emph{Conditioning $t$ based on parts:} Finally, use the notation
          $t_{S}$ to denote the random variable $t$ conditioned on $S\subseteq{\rm \paths}(t)$.
  \end{itemize}
  \branchcolor{purple}{To clarify the notation, consider the following simple example.}
\end{notation}

\begin{example}
  Let $N=2$. Then $\Omega_{{\rm parts}}(N)=\{\{(0,0)\},\{(1,1)\},\{(0,0),(1,1)\},\{(0,1)\},\{(1,0)\},\{(0,1),(1,0)\}\}$
  and there are only two permutations, $t(x)=x$ and $t'(x)=x\oplus1$
  for all $x\in\{0,1\}$. An example of a part $S$ is $S=\{(0,0)\}$.
  A part (in fact the only part) distinct from $S$ is $(1,1)$, i.e.
  $\Omega_{{\rm parts}}(N,S)=\{(1,1)\}$.
\end{example}

\subsubsection{$\delta$ non-uniform distributions}

\branchcolor{purple}{Using the ``parts'' notation (see \Notaref{permPaths}), we define
  uniform distributions over permutations and a notion of being $\delta$
  non-uniform---distributions which are at most $\delta$ ``far from''
  being being uniform.\footnote{Clarification to a possible conflict in terms: We use the word uniform
    in the sense of probabilities---a uniformly distributed random variable---and
    not quite in the complexity theoretic sense---produced by some Turing
    Machine without advice.}}
\begin{defn}[uniform and $\delta$ non-uniform distributions]
  \label{def:uniformDistr}Consider the set, $\Omega(N)$, of all
  possible permutations of $N$ objects labelled $\{0,1,2\dots N-1\}$.
  Let $\mathbb{F}$ be a distribution over $\Omega$. Call $\mathbb{F}$
  a \emph{uniform distribution} if for $u\sim\mathbb{F}$, $\Pr[S\subseteq\paths(u)]=\frac{\left(N-M\right)!}{N!}$
  for all parts $S\in\Omega_{\parts}(N)$ where we are using \Notaref{permPaths}.

  An arbitrary distribution $\mathbb{F}^{\delta}$ over $\Omega$ is
  $\delta$ non-uniform if it satisfies for $t\sim\mathbb{F}^{\delta}$
  \[
    \Pr[S\subseteq\paths(t)]\le2^{\delta\left|S\right|}\cdot\Pr[S\subseteq\paths(u)]
  \]
  for all parts $S\in\Omega_{\parts}(N)$.

  Finally, $\mathbb{F}^{p,\delta}$ over $\Omega$ is $(p,\delta)$
  non-uniform if there is a subset of parts $S$ of size $|S|\le p$
  such that the distribution conditioned on $S$ specifying a part of
  the permutation, becomes $\delta$ non-uniform over parts distinct
  from $S$. Formally, let $t'\sim\mathbb{F}^{p,\delta}$. Then $t'$
  is $(p,\delta)$ non-uniformly distributed if $t'_{S}$ is $\delta$
  non-uniformly distributed over all $S'\in\Omega_{\parts}(N,S)$ (see
  \Notaref{permPaths}), i.e.
  \begin{equation}
    \Pr[S'\subseteq\paths(t')|S\subseteq\paths(t')]\le2^{\delta|S'|}\cdot\Pr[S'\subseteq\paths(u)|S\subseteq\paths(u)].\label{eq:p_delta_uniform}
  \end{equation}
\end{defn}

\branchcolor{purple}{In \Eqref{p_delta_uniform}, we are conditioning a uniform distribution
  using the ``paths/parts'' notation which may be confusing. The following
  should serve as a clarification.}
\begin{note}
  Let $u\sim\mathbb{F}$ as above. Then, we have $\Pr[S'\subseteq\paths(u)|S\subseteq\paths(u)]=\frac{\left(N-\left|S\right|-\left|S'\right|\right)!}{\left(N-\left|S\right|\right)!}$
  where $S'\in\Omega_{\parts}(N,S)$ and $S\in\Omega_{\parts}(N)$.
  Let $S=\{(x_{i},y_{i})\}_{i=1}^{|S|}$. Then, the conditioning essentially
  specifies that the $|S|$ elements in $X=(x_{i})_{i=1}^{|S|}$ must
  be mapped to $Y=(y_{i})_{i=1}^{|S|}$ by $u$, i.e. $u(x_{i})=y_{i}$,
  but the remaining elements $\{0,1\dots N-1\}\backslash X$ are mapped
  uniformly at random to $\{0,1\dots N-1\}\backslash Y$.
\end{note}

\branchcolor{purple}{Clearly, for $\delta=0$, the $\delta$ non-uniform distribution becomes
a uniform distribution. However, this can be achieved by relaxing
the uniformity condition in many ways. The $\delta$ non-uniform distribution
is defined the way it is to have the following property. Notice that
$|S|$ appears in a form such that the product of two probabilities,
$\Pr[S_{1}\subseteq\parts(t)]$ and $\Pr[S_{2}\subseteq\parts(t)]$
yields $|S_{1}|+|S_{2}|$, e.g. $(1+\delta)^{|S|}$ instead of $2^{\delta|S|}$
would also have worked.\footnote{The former was chosen by \textcite{chia_need_2020-1} while the latter by \textcite{EC:CorettiDGS18}
  and possibly others.} This property plays a key role in establishing that in the main decomposition
(as described informally in \Subsecref{Sampling-argument-for}), the
number of ``paths'' (in the informal discussion it was bits) fixed
is small. We chose the pre-factor $2^{|S|}$ for convenience---unlikely
events in our analysis are those which are exponentially suppressed,
and we therefore take the threshold parameter to be $\gamma=2^{-m}$.
These choices result in a simple relation between $|S|$ and $m$.}
\begin{notation}
  \label{nota:DeltaUniformlyDistributed}To avoid double negation, we
  use the phrase ``$t$ is more than $\delta$ non-uniform'' to mean
  that $t$ is not $\delta$ non-uniform. Similarly, we use the phrase
  ``$t$ is at most $\delta$ non-uniform'' to mean that $t$ is $\delta$
  non-uniform.
\end{notation}

\branchcolor{purple}{As shall become evident, the only property of a uniform distribution
  we use in proving the main proposition of this section, is the following.
  It not only holds for all distributions over permutations, but also
  for $d$-Shuffler. We revisit this later.}
\begin{note}
  Let $t$ be a permutation sampled from an arbitrary distribution $\mathbb{F}'$
  over $\Omega(N)$. Let $S,S'\subseteq\Omega_{\parts}(N)$ be \emph{distinct}
  parts (see \Notaref{permPaths}). Then,
  \[
    \Pr[S\subseteq\paths(t)\land S'\subseteq\paths(t)]=\Pr[S\cup S'\subseteq\paths(t)].
  \]
  If $S\cap S'=\emptyset$ and the parts are not distinct, then both
  expressions vanish.
\end{note}

\subsubsection{Advice on uniform yields $\delta$ non-uniform}

\branchcolor{purple}{We are now ready to state and prove the simplest variant of the main
  proposition of this section.}
\begin{prop}[$\mathbb{F}|r'\equiv{\rm conv}(\mathbb{F}^{p,\delta})$]
  \label{prop:sumOfDeltaNonUni_perm} Premise:
  \begin{itemize}
    \item Let $u\sim\mathbb{F}$ where $\mathbb{F}$ is a uniform distribution
          over all permutations, $\Omega$, on $\{0,1\dots N-1\}$, as in \Defref{uniformDistr}
          with $N=2^{n}$.
    \item Let $r$ be a random variable which is arbitrarily correlated to $u$,
          i.e. let $r=g(u)$ where $g$ is an arbitrary function.
    \item Fix any $\delta>0$, $\gamma=2^{-m}>0$ ($m$ may be a function of
          $n$) and some string $r'$.
    \item Suppose
          \begin{equation}
            \Pr[r=r']\ge\gamma.\label{eq:conditionOn_gamma_and_l}
          \end{equation}
    \item Let $t$ denote the variable $u$ conditioned on $r=r'$, i.e. let
          $t=u|(g(u)=r')$.
  \end{itemize}
  Then, $t$ is ``$\gamma$-close'' to a convex combination of finitely
  many $(p,\delta)$ non-uniform distributions, i.e.
  \[
    t\equiv\sum_{i}\alpha_{i}t_{i}+\gamma't'
  \]
  where $t_{i}\sim\mathbb{F}_{i}^{p,\delta}$ and $\mathbb{F}_{i}^{p,\delta}$
  is $(p,\delta)$ non-uniform with $p=\frac{2m}{\delta}$. The permutation
  $t'$ is sampled from an arbitrary (but normalised) distribution over
  $\Omega$ and $\gamma'\le\gamma$.
\end{prop}

\branchcolor{blue}{\begin{proof}
    Suppose that $t$ is more than $\delta$ non-uniformly distributed
    (see \Defref{uniformDistr} and \Notaref{DeltaUniformlyDistributed}),
    otherwise then there is nothing to prove (set $\alpha_{1}$ to $1$,
    and $t_{i}$ to $t$, remaining $\alpha_{i}$s and $\gamma'$ to zero).
    Recall $\Omega_{\parts}(N)$ is the set of all parts (see \Notaref{permPaths}).
    Let the subset $S\in\Omega_{\parts}(N)$ be the maximal subset of
    paths (i.e. subset with the largest size) such that
    \begin{equation}
      \Pr[S\subseteq\paths(t)]>2^{\delta\cdot|S|}\cdot\Pr[S\subseteq\paths(u)].\label{eq:S_not-delta-non-uniform}
    \end{equation}
    \begin{claim}
      \label{claim:easyConditionDeltaUniform}Let $S$ and $t$ be as described
      above. The random variable $t$ conditioned on being consistent with
      the paths in $S\in\Omega_{\parts}(N)$, i.e. $t_{S}$, is $\delta$
      non-uniformly distributed over $S'\subseteq\Omega_{\parts}(N,S)$,
      is $\delta$ non-uniformly distributed.
    \end{claim}

    We prove \Claimref{easyConditionDeltaUniform} by contradiction. Suppose
    that $t_{S}$ is ``more than'' $\delta$ non-uniform. Then, there
    exists some $S'\in\Omega_{\parts}(N,S)$ such that
    \begin{align}
      \Pr[S'\subseteq\paths(t_{S})] & =\Pr[S'\subseteq\paths(t)|S\subseteq\paths(t)]>2^{\delta\cdot|S'|}\cdot\Pr[S'\subseteq\paths(u)|S\subseteq\paths(u)].\label{eq:_S'_not_delta}
    \end{align}
    Since $S'$ violates the $\delta$ non-uniformity condition for $t_{S}$,
    the idea is to see if the union $S\cup S'$ violates the $\delta$
    non-uniformity condition for $t$. If it does, we have a contradiction
    because $S$ was by assumption the maximal subset satisfying this
    property. Indeed,
    \begin{align*}
      \Pr[S\cup S'\subseteq\paths(t)] & =\Pr[S\subseteq\paths(t)\land S'\subseteq\paths(t)]                                                  & \text{\ensuremath{\because} \ensuremath{S} and \ensuremath{S'} are distinct}                \\
                                      & =\Pr[S\subseteq\paths(t)]\Pr[S'\subseteq\paths(t)|S\subseteq\paths(t)]                               & \text{conditional probability}                                                              \\
                                      & >2^{\delta\cdot(|S|+|S'|)}\cdot\Pr[S\subseteq\paths(u)]\Pr[S'\subseteq\paths(u)|S\subseteq\paths(u)] & \text{\text{using \prettyref{eq:S_not-delta-non-uniform} and \prettyref{eq:_S'_not_delta}}} \\
                                      & =2^{\delta\cdot|S\cup S'|}\cdot\Pr[S\cup S'\subseteq\paths(u)]                                       & \text{\ensuremath{\because} \ensuremath{S} and \ensuremath{S'} are disjoint}
    \end{align*}
    which completes the proof.

    \Claimref{easyConditionDeltaUniform} shows how to construct a $\delta$
    non-uniform distribution after conditioning but we must also bound
    $|S|$. This is related to how likely is the $r'$ we are conditioning
    upon, i.e. the probability of $g(u)$ being $r'$.
    \begin{claim}
      \label{claim:sizeS}One has
      \[
        \left|S\right|<\frac{m}{\delta}.
      \]
    \end{claim}

    While \Eqref{S_not-delta-non-uniform} lower bounds $\Pr[S\subseteq\paths(t)]$,
    the upper bound is given by
    \begin{align}
      \Pr[S\subseteq\paths(t)] & =\Pr[S\subseteq\paths(u)|(g(u)=r')]\nonumber                       \\
                               & =\Pr[S\subseteq\paths(u)\land g(u)=r']/\Pr[g(u)=r']\nonumber       \\
                               & \le\Pr[S\subseteq\paths(u)\land g(u)=r']\cdot\gamma^{-1}\nonumber  \\
                               & \le\Pr[S\subseteq\paths(u)]\cdot\gamma^{-1}\label{eq:S_upperBound}
    \end{align}
    where recall that $\gamma=2^{-m}$. Combining these, we have $2^{\delta\cdot|S|}<2^{m}$,
    i.e., $|S|<\frac{m}{\delta}$.

    Using Bayes rule on the event that $S\subseteq\paths(t)$ we conclude
    that
    \[
      t\equiv\alpha_{1}t_{1}+\alpha'_{1}t'_{1}
    \]
    where $\alpha_{1}=\Pr[S\subseteq\paths(t)]$, $t_{1}=t_{S}$, i.e.
    $t$ conditioned on $S\subseteq\paths(t)$, $\alpha'_{1}=1-\alpha_{1}$
    and $t'_{1}$ is $t$ conditioned on $S\nsubseteq\paths(t)$. Further,
    while $t_{1}$ is $(p,\delta)$ non-uniform (from \Claimref{easyConditionDeltaUniform}
    and \Claimref{sizeS}), $t_{1}'$ may not be. Proceeding as we did
    for $t$, if $t_{1}'$ is itself $\delta$ non-uniform, there is nothing
    left to prove (we set $\alpha_{2}=\alpha_{1}'$ and $t_{2}=t_{1}'$
    and the remaining $\alpha_{i}$s and $\gamma'$ to zero). Also assume
    that $\alpha_{1}'>\gamma$ because otherwise, again, there is nothing
    to prove.

    Therefore, suppose that $t'_{1}$ is not $\delta$ non-uniform. Note
    that the proof of \Claimref{easyConditionDeltaUniform} goes through
    for any permutation which is not $\delta$ non-uniform. Thus, the
    claim also applies to $t_{1}'$ where we denote the maximal set of
    parts by $S_{1}$. Let $t_{2}$ be $t_{1}'$ conditioned on $S_{1}\subseteq\paths(t_{1}')$
    and $t_{2}'$ be $t_{1}'$ conditioned on $S_{1}\nsubseteq\paths(t_{1}')$.
    Using Bayes rule as before, we have
    \[
      t\equiv\alpha_{1}t_{1}+\alpha_{2}t_{2}+\alpha_{2}'t_{2}'.
    \]
    Adapting the statement of \Claimref{easyConditionDeltaUniform} (with
    $t_{1}'$ playing the role of $t$ and $S_{1}$ playing the role of
    $S$) to this case, we conclude that $t_{2}$ is $\delta$ non-uniform
    but we still need to show that $|S_{1}|\le p$. We need the analogue
    of \Claimref{sizeS} which we assert is essentially unchanged.
    \begin{claim}
      \label{claim:S_k_bound_general}One has
      \begin{equation}
        \left|S_{i}\right|<\frac{2m}{\delta}.\label{eq:BoundS1}
      \end{equation}
    \end{claim}

    The proof is deferred to \Subsecref{tech_res_non-uniform}. The factor
    of two appears because for the general case, we use both $\alpha_{i}'>\gamma$
    and $\Pr[g(u)=r']>\gamma$. One can iterate the argument above. Suppose
    \begin{equation}
      t\equiv\alpha_{1}t_{1}+\dots\alpha_{j}t_{j}+\alpha_{j}'t'_{j}\label{eq:generalSumT}
    \end{equation}
    where $t_{1},\dots t_{j}$ are $(p,\delta)$ non uniformly distributed
    while $t'_{j}$ is not and for $\alpha_{j}':=\Pr[S\nsubseteq\paths(t)\land\dots\land S_{j-1}\nsubseteq\paths(t)]$
    it holds that $\alpha'_{j}>\gamma$ (else one need not iterate). Let
    $S_{j}$ be the maximal set such that $t_{j+1}:=t'_{j}|S_{j}\subseteq\paths(t_{j}')$
    is $\delta$ non-uniform (which must exist from \Claimref{easyConditionDeltaUniform})
    and let $t_{j+1}':=t_{j}'|S_{j}\nsubseteq\paths(t_{j}')$. Let $\alpha'_{j+1}:=\Pr[S_{j}\nsubseteq\paths(t'_{j})]$
    which equals $\Pr[S\nsubseteq\paths(t)\land\dots\land S_{j}\nsubseteq\paths(t)]$.
    From \Claimref{S_k_bound_general}, $\left|S_{j}\right|<2m/\delta\le p$
    therefore $t_{j+1}$ is $(p,\delta)$ non-uniform.

    We now argue that the sum in \Claimref{S_k_bound_general} contains
    finitely many terms. At every iteration, $\alpha'_{i}$ strictly decreases
    because at each step, more constraints are added; $S_{i}\neq S_{j}$
    for all $i\neq j$ (otherwise conditioning on $S_{j}$ (if $j\ge i$)
    as in \Claimref{easyConditionDeltaUniform} could not have any effect).
    Since $\Omega_{\parts}(N)$ is finite, the decreasing sequence $\alpha'_{1}\dots\alpha'_{i}$
    must, for some integer $i$, satisfy $\alpha_{i}\le\gamma$ after
    finitely many iterations.
  \end{proof}
}

\subsubsection{Iterating advice and conditioning on uniform distributions | $\delta$
  non-$\beta$-uniform distributions}

\branchcolor{purple}{Once generalised to the $d$-Shuffler (which, as we shall, see is
surprisingly simple), recall that the way we intend to use the above
result is to repeatedly get advice from a quantum circuit, a role
played by $g$ in the previous discussion. However, the way it is
currently stated, one starts with a uniformly distributed permutation
$u$ for which some advice $g(u)$ is given but one ends up with $(p,\delta)$
non-uniform distributions. We want the result to apply even when we
start with a $(p,\delta)$ non-uniform distribution.

As should become evident shortly, the right generalisation of \Propref{sumOfDeltaNonUni_perm}
for our purposes is as follows. Assume that the advice being conditioned
occurs with probability at least $\gamma=2^{-m}$ and think of $m$
as being polynomial in $n$; $\delta>0$ is some constant and $p=2m/\delta$.
\begin{itemize}
  \item Step 1: Let $t\sim\mathbb{F}^{\delta'}$ be $\delta'$ non-uniform\footnote{Notation: When I say $t$ is $\delta$ non-uniform, it is implied
          that $t$ is sampled from a $\delta$ non-uniform distribution.} and $s\sim\mathbb{F}^{\delta'}|r$ be $t|(g(t)=r)$. Then it is straightforward
        to show that $s\equiv\sum_{i}\alpha_{i}s_{i}$ where $s_{i}$ are
        $(p,\delta+\delta')$ non-uniform, which we succinctly write as
        \[
          \mathbb{F}^{\delta'}|r\equiv{\rm conv}(\mathbb{F}^{p,\delta+\delta'}).
        \]
\end{itemize}
Observation: If $t\sim\mathbb{F}^{p,\delta}$ is $(p,\delta)$ non-uniform,
then there is some $S$ of size at most $p$ such that $t\sim\mathbb{F}^{\delta|\beta}$
is $\delta$ non-$\beta$-uniform where\footnote{The conditioning is in superscript because it is non-standard; standard
  would be $S\subseteq\parts(t)$ which is too long.} $\beta:=(S)$. A $\beta$-uniform distribution is simply a uniform
distribution conditioned on having $S$ as parts. This amounts to
basically making the conditioning explicit. Having this control will
be of benefit later.
\begin{itemize}
  \item Step 2: It is not hard to show that Step 1 goes through unchanged
        if non-uniform is replaced with non-$\beta$-uniform for an arbitrary
        $\beta$.
\end{itemize}
These combine to yield the following. Let $t\sim\mathbb{F}^{\delta'|\beta}$
be a $\delta'$ non-$\beta$-uniform distribution and $s\sim\mathbb{F}^{\delta'|\beta}|r$
be $t|(g(t)=r)$. Then $t\equiv\sum_{i}\alpha_{i}s_{i}$ where $s_{i}\sim\mathbb{F}^{p,\delta+\delta'|\beta}$
are $(p,\delta+\delta')$ non-$\beta$-uniform,\footnote{The last term with $\alpha_{k}<\gamma$ is suppressed for clarity
  in this informal discussion.} which we briefly express as
\[
  \mathbb{F}^{\delta'|\beta}|r\equiv{\rm conv}(\mathbb{F}^{p,\delta+\delta'|\beta}).
\]
Observe that this composes well,
\begin{equation}
  \mathbb{F}^{p,\delta+\delta'|\beta}|r\equiv{\rm conv}(\mathbb{F}^{2p,2\delta+\delta'|\beta}).\label{eq:comp}
\end{equation}
To see this, consider the following:
\begin{itemize}
  \item For some $S_{i}$, $s_{i}$ (as defined in the statement above) is
        $\delta'':=\delta+\delta'$ non-$\beta'$-uniform where $\beta':=(S\cup S_{i})$
        if $\beta=(S)$.
  \item With $t$ set to $s_{i}$, $\beta$ set to $\beta'$, one can apply
        the above to get $s_{i}|(h(s_{i})=r')\equiv\sum_{i}\alpha'_{i}q_{i}$
        where $q_{i}$ are $(p,\delta+\delta'')$ non-$\beta'$-uniform.
  \item Note that $q_{i}$ are also $(2p,2\delta+\delta')$ non-$\beta$-uniform;
        which we succinctly denoted as $\mathbb{F}^{2p,2\delta+\delta'|\beta}$.
\end{itemize}
Clearly, if this procedure is repeated $\tilde{n}\le\poly$ times,
starting from $\delta=0$ and $\beta=(\emptyset)$, then the final
convex combination would be over $\mathbb{F}^{\tilde{n}p,\tilde{n}\delta}$.
As we shall see, for our use, it suffices to ensure that $\tilde{n}\delta$
is a small constant and that $\tilde{n}p=\frac{\tilde{n}m}{\delta}\le\poly$.
Choosing $\delta=\Delta/\tilde{n}$ for some small fixed $\Delta>0$
yields $\tilde{n}\delta=\Delta$ and $\tilde{n}p=\frac{\tilde{n}^{2}m}{\Delta}$
which is indeed bounded by $\poly$ (recall $m$ and $\tilde{n}$
are bounded by $\poly$).}

\branchcolor{purple}{One can define a notion of closeness to any arbitrary distribution,
  as we did for closeness to uniform. To this end, first consider the
  following.}
\begin{defn}[$\delta$ non-$\mathbb{G}$ distributions---$\mathbb{G}^{\delta}$]
  \label{def:nonG}Let $s$ be sampled from an arbitrary distribution,
  $\mathbb{G}$, over the set of all permutations $\Omega(N)$ of $N$
  objects and fix any $\delta>0$.

  Then, a distribution $\mathbb{G}^{\delta}$ is\emph{ $\delta$ non-$\mathbb{G}$}
  if $s'\sim\mathbb{G}^{\delta}$ satisfies
  \[
    \Pr[S\subseteq\paths(s')]\le2^{\delta|S|}\cdot\Pr[S\subseteq\paths(s)]
  \]
  for all $S\in\Omega_{\parts}(N)$.

  Similarly, a distribution $\mathbb{G}^{p,\delta}$ is $(p,\delta)$
  non-$\mathbb{G}$ if there is a subset $S'\in\Omega_{\parts}(N)$
  of size at most $\left|S'\right|\le p$ such that conditioned on $S'\subseteq\parts(s)$,
  $s''\sim\mathbb{G}^{p,\delta}$ satisfies
  \[
    \Pr[S\subseteq\paths(s'')|S'\subseteq\paths(s'')]\le2^{\delta|S'|}\cdot[S\subseteq\paths(s)|S'\subseteq\paths(s)]
  \]
  for all $S\in\Omega_{\parts}(N,S')$, i.e. conditioned on $S'$ is
  a part of both $s$ and $s''$, $s''$ is $\delta$ non-$\mathbb{G}$.
\end{defn}

\branchcolor{purple}{We now define $\beta$-uniform as motivated above and using the previous
  definition, define $\delta$ non-$\beta$-uniform.}
\begin{defn}[$\beta$-uniform and $\delta$ non-$\beta$-uniform distributions---$\mathbb{F}^{|\beta}$
    and $\mathbb{F}^{\delta|\beta}$]
  \label{def:beta-uniform}Let $u\sim\mathbb{F}(N)$ be sampled from
  a uniform distribution over all permutations, $\Omega(N)$, of $\{0,1\dots N-1\}$
  as in \Notaref{DeltaUniformlyDistributed}. A permutation $s\sim\mathbb{F}^{|\beta}(N)$
  sampled from a $\beta$-uniform distribution is $s=u|(S\subseteq\paths(u))$
  where\footnote{As alluded to earlier, we define $\beta$ to be a redundant-looking
    ``one-tuple'' $(S)$ here but this is because later when we generalise
    to $d$-Shufflers, we set $\beta=(S,T)$ where $T$ encodes paths
    not in $u$.} $\beta=:(S)$ and $S\in\Omega_{\parts}(N)$.

  A distribution $\mathbb{F}^{\delta|\beta}$ is $\delta$ non-$\beta$-uniform
  if it is $\delta$ non-$\mathbb{G}$ with $\mathbb{G}$ set to a $\beta$-uniform
  distribution (see \Defref{nonG}, above). Similarly, a distribution
  $\mathbb{F}^{p,\delta|\beta}$ is $(p,\delta)$ non-$\beta$-uniform
  if it is $(p,\delta)$ non-$\mathbb{G}$ with $\mathbb{G}$, again,
  set to a $\beta$-uniform distribution.
\end{defn}

\branchcolor{purple}{We now state the general version of \Propref{sumOfDeltaNonUni_perm}.}
\begin{prop}[$\mathbb{F}^{\delta'|\beta}|r'={\rm conv}(\mathbb{F}^{(p,\delta+\delta')|\beta})$]
  \label{prop:composableP_Delta_non_beta_uniform-1}Let $t\sim\mathbb{F}^{\delta'|\beta}(N)$
  be sampled from a $\delta'$ non-$\beta$-uniform distribution with
  $N=2^{n}$. Fix any $\delta>0$ and let $\gamma=2^{-m}$ be some function
  of $n$. Let $s\sim\mathbb{F}^{\delta'|\beta}|r$, i.e. $s=t|(h(t)=r)$
  and suppose $\Pr[h(t)=r]\ge\gamma$ where $h$ is an arbitrary function
  and $r$ some string in its range. Then $s$ is ``$\gamma$-close''
  to a convex combination of finitely many $(p,\delta+\delta')$ non-$\beta$-uniform
  distributions, i.e.
  \[
    s\equiv\sum_{i}\alpha_{i}s_{i}+\gamma's'
  \]
  where $s_{i}\sim\mathbb{F}_{i}^{p,\delta+\delta'|\beta}$ with $p=2m/\delta$.
  The permutation $s'$ may have an arbitrary distribution (over $\Omega(2^{n})$)
  but $\gamma'\le\gamma$.
\end{prop}

The proof follows from minor modifications to that of \Propref{sumOfDeltaNonUni_perm}
(see below).

\subsection{Technical results for $\delta$ non-uniform distributions\label{subsec:tech_res_non-uniform}}

\branchcolor{blue}{\begin{proof}[Proof of \Claimref{S_k_bound_general}]
 To see this for $S_{1}$, we proceed as before and recall the lower
bound $\Pr[S_{1}\subseteq\paths(t'_{1})]>2^{\delta|S_{1}|}\Pr[S_{1}\subseteq\paths(u)]$.
The upper bound may be evaluated as 
\begin{align*}
\Pr[S_{1}\subseteq\paths(t_{1}')] & =\Pr[S_{1}\subseteq\paths(t)|S\nsubseteq\paths(t)]\\
 & =\frac{\Pr[S_{1}\subseteq\paths(t)\land S\nsubseteq\paths(t)]}{\Pr[S\nsubseteq\paths(t)]}\\
 & =\frac{\Pr[S_{1}\subseteq\paths(u)\land S\nsubseteq\paths(u)\land g(u)=r']}{\Pr[S\nsubseteq\paths(t)]\Pr[g(u)=r']}\\
 & \le\Pr[S_{1}\subseteq\paths(u)]\cdot\gamma^{-2}
\end{align*}
where we used $\alpha'_{1}=1-\Pr[S\subseteq\paths(t)]=\Pr[S\nsubseteq\paths(t)]\ge\gamma$,
and $\Pr[g(u)=r']\ge\gamma$. In the general case, suppose $t'_{i}$s,
$t_{i}$s and $S_{i}$s are as described in the proof of \Propref{sumOfDeltaNonUni_perm}.
Then, one would have 
\begin{align}
\Pr[S_{i}\subseteq\paths(t_{i}')] & =\frac{\Pr[S_{i}\subseteq\paths(u)\land S_{i-1}\nsubseteq\paths(u)\land\dots S\nsubseteq\paths(u)\land g(u)=r']}{\Pr[S_{i-1}\nsubseteq\paths(t)\land\dots S\nsubseteq\paths(t)]\Pr[g(u)=r']}\label{eq:S_i_upperbound}\\
 & \le\Pr[S_{i}\subseteq\paths(u)]\cdot\gamma^{-2}\nonumber 
\end{align}
where $\alpha'_{i}=\Pr[S_{i-1}\nsubseteq\paths(t)\land\dots S\nsubseteq\paths(t)]>\gamma$
is assumed (else there is nothing to prove).
\end{proof}
}

\begin{prop*}[\Propref{composableP_Delta_non_beta_uniform} restated with slightly
different parameters]
 Let $t\sim\mathbb{F}^{\delta'|\beta}(N)$ be sampled from a $\delta'$
non-$\beta$-uniform distribution with $N=2^{n}$. Fix any $\delta>\delta'$
and let $\gamma=2^{-m}$ be some function of $n$. Let $s=t|(h(t)=r')$
and suppose $\Pr[h(t)=r']\ge\gamma$ where $h$ is an arbitrary function
and $r'$ some string in its range. Then $s$ is ``$\gamma$-close''
to a convex combination of finitely many $(p,\delta)$ non-$\beta$-uniform
distributions, i.e. 
\[
s\equiv\sum_{i}\alpha_{i}s_{i}+\gamma's'
\]
where $s_{i}\sim\mathbb{F}_{i}^{p,\delta|\beta}$ with $p=2m/(\delta-\delta')$.
The permutation $s'$ may have an arbitrary distribution (over $\Omega(2^{n})$)
but $\gamma'\le\gamma$.
\end{prop*}
\branchcolor{blue}{\begin{proof}
While redundant, we follow the proof of \Propref{sumOfDeltaNonUni_perm}
adapting it to this general setting and omitting full details this
time. 

(For comparison: We replace $t$ with $s$ and $u$ with $b$)

\textbf{Step A:} Lower bound on $\Pr[S\subseteq\paths(s)]$.

Let $b\sim\mathbb{F}^{|\beta}(N)$. Suppose $s$ is not $\delta$
non-$\beta$-uniform. Then consider the largest $S\in\Omega_{\parts}(N)$
such that

\begin{equation}
\Pr[S\subseteq\paths(s)]>2^{\delta\cdot|S|}\cdot\Pr[S\subseteq\paths(b)].\label{eq:_S_not_delta}
\end{equation}
 
\begin{claim}
Let $S$ and $s$ be as described. The random variable $s$ conditioned
on being consistent with the paths in $S\in\Omega_{\parts}(N)$, i.e.
$s_{S}:=s|(S\subseteq\parts(s))$, is $\delta$ non-$\beta$-uniformly
distributed. 
\end{claim}

We give a proof by contradiction. Suppose $s_{S}$ is ``more than''
$\delta$ non-$\beta$-uniform. Then there exist some $S'\in\Omega_{\parts}(N,S)$
such that 
\[
\Pr[S'\subseteq\paths(s)|S\subseteq\paths(s)]>2^{\delta\cdot|S'|}\Pr[S'\subseteq\paths(b)|S\subseteq\paths(b)].
\]
Then 
\begin{align*}
\Pr[S\cup S'\subseteq\paths(s)] & =\Pr[S\subseteq\paths(s)]\Pr[S'\subseteq\paths(s)|S\subseteq\paths(s)]\\
 & >2^{\delta\cdot|S\cup S'|}\cdot\Pr[S\cup S'\subseteq\paths(b)]
\end{align*}
using \Eqref{_S_not_delta} and \Eqref{_S'_not_delta}. That's a contradiction
to $S$ being maximal.

\textbf{Step B:} Upper bound on $\Pr[S\subseteq\paths(s)]$.
\begin{claim}
One has $\left|S\right|<m/(\delta-\delta')$. 
\end{claim}

To see this, observe that

\begin{align*}
\Pr[S\subseteq\paths(s)] & =\Pr[S\subseteq\paths(t)\land h(t)=r']\cdot\Pr[h(t)=r']\\
 & \le\Pr[S\subseteq\paths(t)]\cdot\gamma^{-1}\\
 & \le2^{\delta'|S|}\Pr[S\subseteq\paths(b)]\cdot\gamma^{-1}
\end{align*}
and comparing this with the lower bound, one obtains $\left|S\right|<m/(\delta-\delta')$. 

The remaining proof \Propref{sumOfDeltaNonUni_perm} similarly generalises
by proceeding in the same vein. More concretely, suppose $S_{i}$,
$s_{i}$, $s_{i}'$ are defined analogously. Then the lower bound
goes through almost unchanged while for the upper bound, the analogue
of \Eqref{S_i_upperbound} becomes 
\begin{align*}
\Pr[S_{i}\subseteq\paths(s'_{i})] & =\frac{\Pr[S_{i}\subseteq\paths(s)\land S_{i-1}\nsubseteq\paths(s)\land\dots S\nsubseteq\paths(s)]}{\Pr[S_{i-1}\nsubseteq\paths(s)\land\dots S\nsubseteq\paths(s)]}\\
 & \le\Pr[S_{i}\subseteq\paths(t)\land S_{i-1}\nsubseteq\paths(t)\land\dots S\nsubseteq\paths(t)|h(t)=r']\cdot\gamma^{-1}\\
 & \le\frac{\Pr[S_{i}\subseteq\paths(t)]}{\Pr[h(t)=r']}\cdot\gamma^{-1}\le2^{\delta'|S_{i}|}\Pr[S_{i}\subseteq\paths(b)]\cdot\gamma^{-2}.
\end{align*}
\end{proof}
}

}{See the GitHub version (\url{http://atulsingharora.github.io/instaDepth}) of this article or refer to Section 6.1 of \cite{arora_oracle_2022}.}

\newpage
\end{document}